\documentclass[aps,prx,reprint,superscriptaddress,longbibliography]{revtex4-2}

\usepackage{graphicx}
\graphicspath{{./}{./figure/}}

\usepackage{amsmath,amssymb,bm,braket}
\usepackage{enumitem}

\usepackage{xcolor}

\usepackage{orcidlink}

\usepackage{hyperref}
\hypersetup{colorlinks=true, citecolor=magenta, linkcolor=blue, urlcolor=cyan}

\usepackage{amsthm}
\newtheorem{definition}{Definition}
\newtheorem{theorem}{Theorem}
\newtheorem{proposition}{Proposition}
\newtheorem{corollary}{Corollary}
\newtheorem{lemma}{Lemma}

\newtheorem{example}{Example}
\newtheorem{assumpGibbs}{Assumption}

\newtheorem{assumpThermalization}{Assumption}

\newtheorem{mainresult}{Main result}
\newtheorem{TDLaw}{Planck's principle}

\newcommand{\maceq}{\stackrel{\mathrm{mac}}{\sim}}
\newcommand{\supp}{\mathrm{Supp}}
\newcommand{\initH}{H_L}
\newcommand{\sumst}[2]{\sum_{\substack{#1 \ \mathrm{s.t.} \ #2}}}
\newcommand{\Cell}[1][]{C^{\ell}_{#1}}
\newcommand{\SetAdditive}{\mathcal{A}}
\newcommand{\aLocal}{\alpha}



\begin{document}

\title{Second law of thermodynamics in closed quantum many-body systems}

\author{Yuuya Chiba\,\orcidlink{0009-0007-9064-7981}}
\email{yuya.chiba@riken.jp}
\affiliation{Nonequilibrium Quantum Statistical Mechanics RIKEN Hakubi Research Team, Pioneering Research Institute (PRI), RIKEN, 2-1 Hirosawa, Wako, Saitama 351-0198, Japan}

\author{Yasushi Yoneta\,\orcidlink{0000-0002-2334-4565}}
\email{yasushi.yoneta@riken.jp}
\affiliation{Center for Quantum Computing, RIKEN, 2-1 Hirosawa, Wako, Saitama 351-0198, Japan}

\author{Ryusuke Hamazaki\,\orcidlink{0000-0003-3793-6016}}
\email{ryusuke.hamazaki@riken.jp}
\affiliation{Nonequilibrium Quantum Statistical Mechanics RIKEN Hakubi Research Team, Pioneering Research Institute (PRI), RIKEN, 2-1 Hirosawa, Wako, Saitama 351-0198, Japan}
\affiliation{RIKEN Center for Interdisciplinary Theoretical and
Mathematical Sciences (iTHEMS), RIKEN, Wako, Saitama 351-0198, Japan}

\author{Akira Shimizu\,\orcidlink{0009-0009-8926-7898}}
\email{firstname-lastname@g.ecc.u-tokyo.ac.jp}
\affiliation{Institute for Photon Science and Technology, The University of Tokyo, 7-3-1 Hongo, Bunkyo-ku, Tokyo 113-0033, Japan}
\affiliation{Center for Quantum Information and Quantum Biology, The University of Osaka, Toyonaka, Osaka 560-0043, Japan}

\date{\today}

\begin{abstract}

The second law of thermodynamics for adiabatic operations --- constraints on state transitions in closed systems under external control --- is one of the fundamental principles of thermodynamics. On the other hand, recent studies of thermalization have established that even pure quantum states can represent thermal equilibrium. However, pure quantum states do not satisfy the second law in that they are not passive, i.e., work can be extracted from them if arbitrary unitary operations are allowed, and that various entropy formulas, such as the von Neumann entropy, deviate from thermodynamic entropy in such states. It therefore remains unresolved how thermal equilibrium represented by a pure quantum state can be reconciled with thermodynamics. Here, based on our key quantum-mechanical notions of thermal equilibrium and adiabatic operations, we address the emergence of the second law of thermodynamics in closed quantum many-body systems. We first introduce infinite-observable macroscopic thermal equilibrium (iMATE); a quantum state, including pure states, is said to represent iMATE if the expectation values of \emph{all} additive observables, which correspond to additive quantities in thermodynamics, agree with their equilibrium values. We also introduce a macroscopic operation as unitary evolution generated by a \emph{time-dependent} additive Hamiltonian, which is regarded as corresponding to adiabatic operations. Employing these concepts, we show Planck's principle: no extensive work can be extracted from any quantum state representing iMATE through any cyclic macroscopic operations with the operation times independent of the system size. Furthermore, we introduce a quantum-mechanical form of entropy density such that it agrees with thermodynamic entropy density for any quantum state representing iMATE. We then prove the law of increasing entropy: for any initial state representing iMATE, this entropy density cannot be decreased by any macroscopic operations with the operation times independent of the system size, 
followed by a relaxation process governed by a time-independent Hamiltonian.
Our theory thus proves two different forms of the second law, which are quantum mechanically inequivalent to each other, and demonstrates how thermodynamics emerges from quantum mechanics by adopting macroscopically reasonable classes of observables, equilibrium states, and operations.

\end{abstract}

\maketitle

\section{\label{sec:introduction}Introduction}

The second law of thermodynamics~\cite{Kittel1980,Landau1980,Callen1985,Lieb1999,Pusz1978,Lenard1978,Gorecki1980,Daniels1981,Tasaki2000,Baba2023,Brandao2015SecondLaw,Tasaki2000Stat,Santos2011,Ikeda2015Second,Tasaki2016SecondLaw,Kaneko2019,Munoz-Arias2022,Hokkyo2025,Meier2025} is the origin of irreversibility of many macroscopic phenomena and provides a strong constraint on the possible state transitions in such systems. 
State transitions occur due to various causes, such as a sudden quench
that induces relaxation 
(governed by a \emph{time-independent} Hamiltonian)
in an \emph{isolated system}.
Among them, transitions due to adiabatic thermodynamic operations
on \emph{closed systems} 
(governed by a \emph{time-dependent} Hamiltonian)
are particularly important, since, e.g., thermodynamics can be rigorously constructed from axioms about adiabatic accessibility and equilibrium states~\cite{Lieb1999}.

While there are many different forms of the second law, the representative
forms that correspond to adiabatic operations are
\emph{Planck's principle}~\cite{Lieb1999,Tasaki2016SecondLaw,Hokkyo2025} and \emph{the law of increasing entropy}~\cite{Kittel1980,Landau1980,Callen1985,Lieb1999}.
Here, 
Planck's principle states that the energy of a closed system cannot be decreased by any adiabatic thermodynamic operation in which the initial and final values of the control parameters, such as volume and magnetic field, coincide. 
The law of increasing entropy states that thermodynamic entropy of a closed system cannot be decreased by any adiabatic operation.
These laws are well established in thermodynamics. 
However, their understanding from microscopic physical laws is still insufficient.

Planck's principle has motivated researchers to introduce the notion of \emph{passivity}~\cite{Pusz1978,Lenard1978,Gorecki1980,Daniels1981}.
A quantum state is said to be passive for a given Hamiltonian if the expectation value of the Hamiltonian cannot be decreased by applying any unitary time evolution.
It has been shown~\cite{Lenard1978} that the necessary and sufficient condition for passivity is that the quantum state can be written as a classical mixture of energy eigenstates with probabilities that are in descending order with respect to the energy eigenvalues. 
This implies that 
any mixed quantum states with nondescending-order probabilities, 
including pure quantum states other than the ground states,
are not passive.

On another front, in the study of thermalization of isolated quantum systems~\cite{Neumann1929,Deutsch1991,Srednicki1994,Rigol2008,
Kim2014,Beugeling2014,Biroli2010,Mori2016,Iyoda2017,Srednicki1996,Arad2016,Brandao2019,Essler2024,Tasaki1998,Reimann2008,Linden2009,Short2012,Reimann2012,Farrelly2017,Eisert2015,D'Alessio2016,Gogolin2016,Mori2018,Kinoshita2006,Gring2012,Trotzky2012,Kaufman2016,
Goldstein2015,Tasaki2016,Goldstein2017,Goldstein2006,Popescu2006,Sugita2006,Reimann2007,Sugiura2012,Sugiura2013,Rigol2007,Sotiriadis2014,Wouters2014,Pozsgay2014,Ilievski2015,Mierzejewski2015,Ilievski2015a,Ilievski2016,Doyon2017,Kuwahara2020ETH,Reimann2015,Sugimoto2021,Sugimoto2022,sugimoto2023bounds,Serbyn2026}, it is widely believed that not only Gibbs ensembles but also other quantum states, including pure quantum states, can represent thermal equilibrium. For instance, the eigenstate thermalization hypothesis (ETH)~\cite{Neumann1929,Deutsch1991,Srednicki1994,Rigol2008} states that every energy eigenstate located in the bulk of the energy spectrum represents thermal equilibrium, and is believed to hold for typical nonintegrable systems~\cite{D'Alessio2016,Neumann1929,Deutsch1991,Srednicki1994,Rigol2008,Kim2014,Beugeling2014,Reimann2015,Sugimoto2021,Sugimoto2022,sugimoto2023bounds}.
However, given that general pure quantum states do not satisfy the condition for passivity, Planck's principle does not seem to hold for thermal equilibrium represented by pure quantum states.

One reason for the above inconsistency is due to the fact that some unitaries with which the energy is decreased from nonpassive states do not correspond to adiabatic thermodynamic operations, i.e., thermodynamic operations in systems that exchange no heat with the environment (not to be confused with the adiabatic theorem in quantum mechanics).  If the system is isolated from the environment, the dynamics is governed by a unitary time evolution generated by a time-dependent Hamiltonian. However, the Hamiltonian of thermodynamic systems is usually assumed to consist only of local interactions, and thermodynamic operations must be completed in a usual timescale that should not exceed, e.g., 
the lifetimes of constituent particles
or the age of the universe. 
Under such restrictions, 
it will be difficult to generate all unitaries in macroscopic systems. In other words, it is not obvious what unitaries correspond to adiabatic thermodynamic operations.

Furthermore, there is the possibility that, to resolve the above inconsistency, we need to adopt the appropriate definition of quantum states representing thermal equilibrium. In the context of thermalization, two definitions are often used: macroscopic thermal equilibrium (MATE)~\cite{Neumann1929,Goldstein2015,Tasaki2016,Goldstein2017,Mori2018} and microscopic thermal equilibrium (MITE)~\cite{Goldstein2015,Goldstein2017,Mori2018}. 
In the notion of MATE, a quantum state is said to be in MATE if, for a \emph{finite} number of fixed \emph{macroscopic} observables, the measurement outcomes almost surely agree with their equilibrium values~\footnote{For noncommutative macroscopic observables, the measurement outcomes are not defined naively. To overcome this difficulty, one considers their approximation by commuting observables. For more details, see Sec.~\ref{sec:OrdinaryMATE}}.
On the other hand, a quantum state is said to be in MITE if the expectation values of all \emph{local} observables agree with those in the Gibbs state.
Note that the Kubo–Martin–Schwinger (KMS) condition  \cite{Haag1996,Bratteli2002,Bratteli2002a} corresponds to the thermodynamic limit of MITE (Sec.~\ref{sec:MITE}).
Since the standard thermodynamics focuses on macroscopic quantities~\cite{Callen1985}, MATE is considered to be more compatible with thermodynamics than MITE. 
However, we stress that both definitions only require that the value of each observable in that state agrees with the equilibrium value described by statistical mechanics; it is nontrivial whether  \textit{state transitions} starting from that state are consistent with transitions allowed in thermodynamics.

Moreover, it is rather nontrivial how the law of increasing entropy is reconciled with the nonpassive nature of pure quantum states representing thermal equilibrium. To address this question, it is desirable to define entropy for each quantum state
such that it agrees
with thermodynamic entropy whenever the state represents thermal equilibrium. 
[Unfortunately, 
the von Neumann entropy does not work in general since 
it is zero for a pure quantum state even if the state represents thermal equilibrium
and does not agree with thermodynamic entropy.] 
It seems challenging to define such an entropy purely quantum mechanically.
In addition, even if such a definition is possible, it is not clear whether the law of increasing entropy is satisfied.

In fact, there are several studies that considered similar questions on the second law of thermodynamics via unitary dynamics~\footnote{ Some readers would be concerned about the relation to results by Iyoda, Kaneko, and Sagawa~\cite{Iyoda2017}. They showed $\Delta S-\beta Q\ge 0$, a form of the second law for \emph{isothermal} operations, in a system attached to a heat reservoir. By contrast, we are interested in the second law for \emph{adiabatic} operations}.
The passivity in closed quantum many-body systems has been addressed in a pioneering work by Hokkyo and Ueda~\cite{Hokkyo2025} 
employing  
MITE as thermal equilibrium states; however, many of the problems remain unsolved (see Sec.~\ref{sec:Hokkyo}).
Regarding the law of increasing entropy, it
has been discussed based on certain forms of entropy, such as extensions of the Boltzmann entropy~\cite{Tasaki2000Stat}, diagonal entropy~\cite{Santos2011,Ikeda2015Second,Munoz-Arias2022}, and observable entropy~\cite{Neumann1929,Meier2025}. However, they face at least one of the following two difficulties in addressing the above 
questions. One is that they do not apply to pure quantum states, particularly energy eigenstates. For instance, diagonal entropy~\cite{Santos2011,Ikeda2015Second} becomes zero when applied to energy eigenstates. 
In addition, they only consider the initial state satisfying the condition for passivity~\cite{Tasaki2000Stat} or having a large effective dimension~\cite{Meier2025} (see also Sec.~\ref{sec:Meier}), neither of which is satisfied by energy eigenstates. 
The other difficulty is that they do not consider 
the case of \emph{closed} quantum systems where the time evolution is generated by a \emph{time-dependent} Hamiltonian; 
rather, they consider the case of isolated quantum systems where the time evolution is generated by a time-independent Hamiltonian alone. 
This distinction is thermodynamically crucial because general adiabatic thermodynamic operations are often time dependent.

In this paper,
we address the above issues by resolving the  following fundamental questions:
(Q1) How should quantum states representing thermal equilibrium be defined, so that they are consistent with thermodynamics, including state transitions?
(Q2) How should adiabatic thermodynamic operations be expressed quantum mechanically?
(Q3) Can we resolve the inconsistency between 
thermodynamics and the fact that any pure quantum states 
can break the passivity even when they represent thermal equilibrium states?
(Q4) Is there a formula for entropy that can be determined for each quantum state and agrees with thermodynamic entropy for any states representing thermal equilibrium, including pure states? 
(Q5) Can we derive the law of increasing entropy 
from quantum mechanics
for any quantum state representing thermal equilibrium?

Our answers to these questions are summarized as follows:
(A1) We introduce a wide class of quantum state representing thermal equilibrium,
called \emph{infinite-observable macroscopic thermal equilibrium} (iMATE).
They are indistinguishable from the Gibbs state 
by the expectation value of \textit{any} additive observable,
 which correspond to additive quantities in thermodynamics.
(A2) We represent adiabatic thermodynamic operations
by unitary time evolutions generated by time-dependent Hamiltonians that are 
restricted to additive observables. 
We call such operations \emph{macroscopic operations}.
(A3) We prove a macroscopic version of passivity for any quantum states representing iMATE, including pure quantum states, under 
any macroscopic operation.
(A4) We introduce a quantum-mechanical entropy density by using a characterization of quantum states in terms of additive observables. 
We call it \emph{quantum macroscopic entropy density}.
It agrees with thermodynamic entropy density for any states representing iMATE, including pure quantum states.
(A5) We study the time evolution in which the initial quantum state is 
prepared in iMATE, and a macroscopic operation is applied to it, 
followed by a relaxation process induced by 
a time-independent final Hamiltonian. Then we prove that the quantum macroscopic entropy density does 
not decrease by any macroscopic operation.
Overall, (A3) and (A5) demonstrate that the second law of thermodynamics emerges in closed quantum many-body systems, given macroscopically reasonable notions of equilibrium states (A1), operations (A2), and entropy (A4).

This paper is structured as follows.
In Sec.~\ref{sec:SummaryResults}, 
we summarize 
our main results 
for the simplified case where states and macroscopic operations are macroscopically uniform.
In Sec.~\ref{sec:Setup}, we introduce additive observables. In Sec.~\ref{sec:MacroEquiv}, we define macroscopic states and their equivalence relation. Using this, we introduce iMATE in Sec.~\ref{sec:MacroEquilibrium}. In Sec.~\ref{sec:MacroOp_MacroEquiv}, we formulate unitary time evolution corresponding to adiabatic thermodynamic operations and provide a critical result stating that the equivalence relation between macroscopic states is preserved by such a time evolution.
Using this, we show the macroscopic version of passivity in Sec.~\ref{sec:MacroPassivity}.
In Sec.~\ref{sec:reduced_rho}, we introduce reduced density matrices that characterize the values of additive observables in macroscopic states, which are basic elements of our entropy formula given in Sec.~\ref{sec:Entropy_s^mac}. 
In Sec.~\ref{sec:Entropy_s^mac_eq}, we prove that this entropy formula gives thermodynamic entropy for any quantum state representing iMATE. In Sec.~\ref{sec:EntropyIncrease}, we show the law of increasing entropy applicable to both our entropy and thermodynamic entropy. In Sec.~\ref{sec:ComparisonStudies}, we compare our results with the existing studies, i.e., MITE, MATE, and the results by Hokkyo and Ueda \cite{Hokkyo2025} and those by Meier et al.~\cite{Meier2025}. In Sec.~\ref{sec:Counterexamples_Timescale}, we provide pivotal counterexamples to the second law when operation time diverges in the thermodynamic limit, which indicates the optimality of our results. In Sec.~\ref{sec:Discussion}, we discuss possible candidates for extending our results while avoiding such counterexamples.
Sec.~\ref{sec:summary} summarizes the paper.

\section{\label{sec:SummaryResults}Summary of results for macroscopically uniform case}

The results of this paper reduce to simpler ones in the case where states are macroscopically uniform and operations are induced by 
external fields (or control parameters) that are 
uniform throughout the system, 
which we refer to as the \emph{macroscopically uniform case}.
In this section, we provide an informal but physically-intuitive summary of our main results in that case. The mathematically rigorous formulation of results for the case where macroscopically nonuniform states and operations are considered will be given in the subsequent sections.
We take $k_{\text{B}}=\hbar=1$ throughout this paper.

\subsection{Macroscopic equivalence and infinite-observable macroscopic thermal equilibrium}

The MATE \cite{Neumann1929,Goldstein2015,Tasaki2016,Goldstein2017,Mori2018}
is one of the ordinary notions of quantum states representing thermal 
equilibrium. However, 
as will be pointed out in Sec.~\ref{sec:OrdinaryMATE}, 
it is not fully consistent with thermodynamics. 
We first resolve this problem by introducing its 
improved 
version that is expected to be consistent with thermodynamics.

We consider 
quantum spin 
systems on a $d$-dimensional hypercubic lattice $\Lambda_{L}$ with $N=L^d$ sites.
Following the philosophy of thermodynamics, we focus on ``additive observables''.
For simplicity, this section focuses only on additive observables on the whole system, while the following sections focus also on those on macroscopic subsystems, which are $\Theta(N)$-sized~\cite{LandauSymbols} subsets of $\Lambda_L$.
Then, additive observables are simply defined as follows.
Let $\aLocal$ be an observable whose support is contained in the hypercube $\Cell$ of side length $\ell$ centered at site $\boldsymbol{0}\in\Lambda_L$. Such an observable and its translations are called \emph{$\ell$-local observables}
(see Fig.~\ref{fig:Summary_Additive} (a)).
We call  $A_L$  an \emph{additive observable} composed of $\ell$-local observables
if 
\begin{align}
    A_L=\sum_{\boldsymbol{r}\in\Lambda_{L}}\aLocal_{\boldsymbol{r}},
    \label{eq:Summary_Additive}
\end{align}
where $\aLocal_{\boldsymbol{r}}$ is the translation of 
$\aLocal$ by $\boldsymbol{r}$
(see Fig.~\ref{fig:Summary_Additive} (b)).
\begin{figure*}
    \centering
    \includegraphics[width=0.9\linewidth]{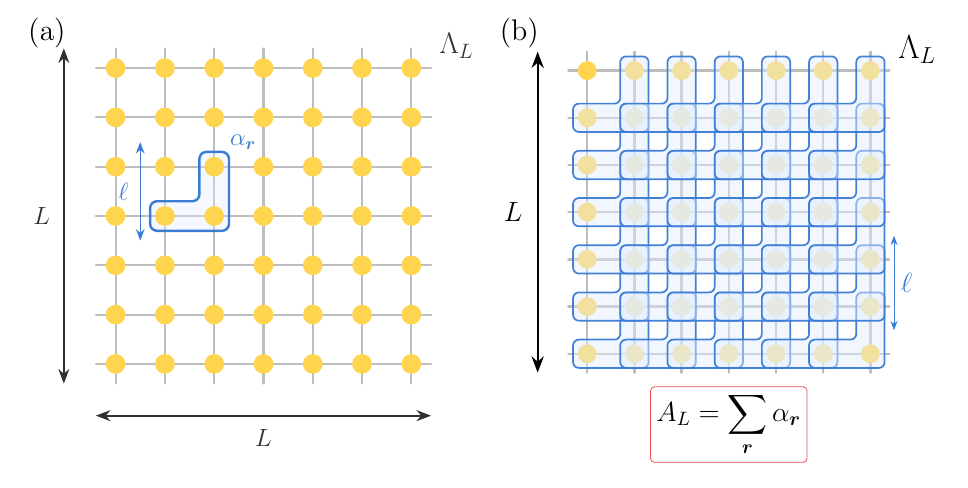}
    \caption{Schematic illustrations of (a) an $\ell$-local observable $\aLocal_{\bm{r}}$ and 
    (b) an additive observable composed of $\ell$-local observables~$A_L$, Eq.~\eqref{eq:Summary_Additive},
    for the macroscopically uniform case. The case with $\ell=2$ and $L=7$ is shown.}
    \label{fig:Summary_Additive}
\end{figure*}

In spin-$1/2$ chains, for instance,
$\aLocal=\sigma_{0}^x$ and $\aLocal=\sigma^{z}_{-1}\sigma^{x}_{0}\sigma^{z}_{2}$ are examples of $1$- and $3$-local observables, and hence $A_L=\sum_{j=1}^{L}\sigma_{j}^x$ and $A_L=\sum_{j=1}^{L}\sigma^{z}_{j-1}\sigma^{x}_{j}\sigma^{z}_{j+2}$ are additive observables composed of $1$- and $3$-local observables, respectively.
Note that as $\ell$ is increased, the number of independent $\ell$-local observables increases, and so does the number of independent additive observables $A_L$.

By utilizing \emph{all} additive observables, 
we introduce a new way of characterizing a quantum state macroscopically:
we say that two quantum states $\rho_{L}$ and $\sigma_{L}$ are \emph{macroscopically equivalent}, denoted by 
\begin{align}
    \rho_L\maceq\sigma_L,
    \label{eq:Summary_MacroEquiv}
\end{align}
if, for any positive integer $\ell\in \mathbb{N}$, it holds that
\begin{align}
    \label{eq:Summary_CondMacroEquiv}\lim_{L\to\infty}\mathrm{Tr}\Bigl[\rho_{L}\frac{A_L}{N}\Bigr]=\lim_{L\to\infty}\mathrm{Tr}\Bigl[\sigma_{L}\frac{A_L}{N}\Bigr]
\end{align}
for every additive observable $A_L$ composed of $\ell$-local observables.
This means that we focus on all additive observables composed of $\ell$-local observables where $\ell$ is taken arbitrarily large but independent of $L$.
See Fig.~\ref{fig:MacroEquiv} for a schematic illustration.
For a rigorous and generalized definition, see Definition~\ref{definition:MacroEquiv}. 

\begin{figure}
    \centering
    \includegraphics[width=\linewidth]{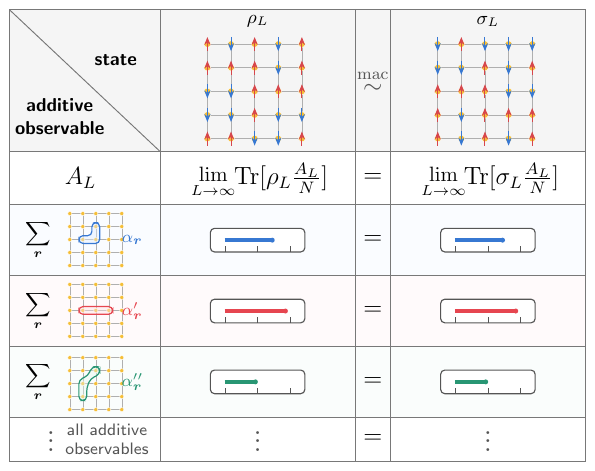}
    \caption{Schematic illustration of macroscopic equivalence, Eq.~\eqref{eq:Summary_MacroEquiv} and Definition~\ref{definition:MacroEquiv}.
    For all additive observables given in Eq.~\eqref{eq:Summary_Additive}, two states $\rho_L$ and $\sigma_L$ give the same expectation values of their densities in the thermodynamic limit, as described in Eq.~\eqref{eq:Summary_CondMacroEquiv}.} 
    \label{fig:MacroEquiv}
\end{figure}

Using this notion of macroscopic equivalence,
we define an \emph{infinite-observable macroscopic thermal equilibrium} (iMATE), which will be shown to be consistent with the second law of thermodynamics, as follows.
Suppose that the system is described by a Hamiltonian $H_L$ that is an additive observable defined above.
We say that a state $\rho_{L}$ represents an iMATE
if $\rho_{L}$ is macroscopically equivalent to the canonical Gibbs state 
\begin{align}
    \rho_{L}^{\mathrm{can}}(\beta|H_L)=e^{-\beta H_L}/Z
    \label{eq:DEF_rho^can}
\end{align}
with some finite inverse temperature $-\infty<\beta<\infty$, where $Z$ is the partition function.
Table~\ref{tbl:Comparison_Equilibrium} shows a quick comparison between 
our notion of thermal equilibrium and those in previous literature.
\begin{table*}
\caption[Equilibrium]{\label{tbl:Comparison_Equilibrium}Comparison of the notions of thermal equilibrium when the system is described by a translation-invariant Hamiltonian. Lower rows include 
more states than the upper ones, indicating that our iMATE includes 
more states than the statistical ensembles and MITE. 
}
\begin{ruledtabular}
\begin{tabular}{llll}
    Notion\quad &Characterization\quad &Examples\quad &Spatial structure\\
    \colrule
    Statistical ensembles\quad & -\quad &Gibbs states\quad & Translation invariant~\eqref{eq:rho_TranslationInv}\\
    MITE (Def.~\ref{definition:MITE})\quad & Arbitrary local observables\quad &Typicality~\cite{Goldstein2006,Popescu2006,Sugita2006}, TPQ states~\cite{Sugiura2012,Sugiura2013} & Locally uniform (Def.~\ref{definition:LocallyUniform})\\
    Our iMATE (Def.~\ref{definition:MacroEqState})\quad & Arbitrary additive observables\quad &METTS (Example~\ref{example:TypicalMETTS_represents_iMATE}), 
    Examples~\ref{example:NonMITE_MacroEqState}, \ref{example:NonMITE_MacroEqState_FiniteTemp}\quad & Macroscopically uniform (Def.~\ref{definition:MacroUniform})\\
    Ordinary MATE\quad & Finite macro-observables\quad & Example~\ref{example:ProblemFiniteObs}\quad & -
\end{tabular}
\end{ruledtabular}
\end{table*}

We will generalize the macroscopic equivalence and iMATE in the following sections 
by focusing also on additive observables on macroscopic subsystems. Such a more general definition reduces to the above one if the states $\rho_L$ and $\sigma_L$ are macroscopically uniform, 
meaning that the expectation values of the density of each additive observable on any macroscopic subsystems coincide.
For the precise definition (Definition~\ref{definition:MacroUniform})
and physical motivations behind considering all additive observables,
see Sec.~\ref{sec:MacroEquilibrium_iMATE}.

\subsection{Macroscopic operations}
As another important ingredient to discuss the second law, 
we next explain what type of quantum-mechanical operations should be considered.
The second law 
concerns transitions of thermal equilibrium states caused by thermodynamic operations, which are macroscopic in a certain 
sense. Although the quantum-mechanical representations of such thermodynamic operations are nontrivial in general, 
we adopt the following quantum-mechanical operations.

Let $m$ be some positive integer independent of $L$.
Suppose that we can control $m$ external fields (or control parameters) 
in a time-dependent manner as $f^{1}(t)$, $f^{2}(t)$, ..., $f^{m}(t)\in\mathbb{R}$,
which couple to $m$ additive observables $B_L^{1}$, $B_L^{2}$, ..., $B_L^{m}$
of the system, respectively. Hence, the
system obeys a unitary time evolution $U_{L}(t,0)$
generated by the time-dependent Hamiltonian
\begin{align}
    H_L(t)=H_L-\sum_{\mu=1}^{m}f^{\mu}(t)B_L^{\mu}.
    \label{eq:SummaryResults_H(t)}
\end{align}
We call such a unitary operation $U_{L}(t,0)$ a \emph{macroscopic operation} of operation time $t$.
A schematic illustration is given in Fig.~\ref{fig:Summary_MacroOp}.
\begin{figure}
    \centering
    \includegraphics[width=\linewidth]{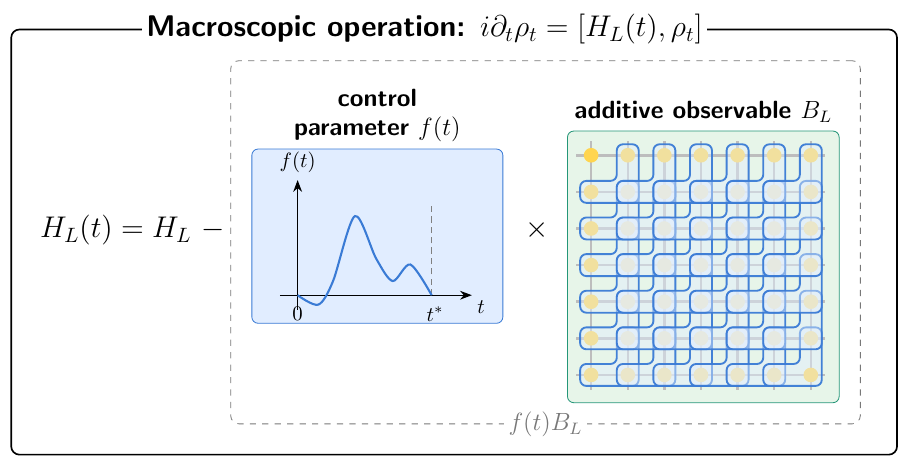}
    \caption{Schematic illustration of macroscopic operations, Eq.~\eqref{eq:SummaryResults_H(t)},
    for the macroscopically uniform case. We here show the case with $m=1$, where the Hamiltonian is driven by a time-dependent control field $f(t)$ coupled to an additive observable $B_L$.}
    \label{fig:Summary_MacroOp}
\end{figure}
See also Table~\ref{tbl:Comparison_Operation}, where our operation is compared with the operations treated in the previous literature.

\begin{table*}
\caption[Operation]{\label{tbl:Comparison_Operation}Consistency to Planck's principle for each combination of the class of initial states and of operations. Lower rows include broader states and narrower operations than the upper ones. 
Importantly, a straightforward extension of our setup given in the fourth row to broader operations (resp. states) while keeping the class of states (resp. operations) as in the third row (resp. fifth row)
causes inconsistency to Planck's principle. }
\begin{ruledtabular}
\begin{tabular}{llll}
    Reference & Initial states\quad &Operations\quad &Consistency to Planck's principle\\
    \colrule
    \cite{Lenard1978} & Passive state \quad &Arbitrary unitary operation\quad &Passivity~\cite{Lenard1978}\\
    \cite{Hokkyo2025} & MITE (Def.~\ref{definition:MITE}) \quad &Local control (Def.~\ref{definition:LocalControl}) of timescale~$O(L^0)$\quad & Thermodynamic passivity~\cite{Hokkyo2025}\\
    - & iMATE (Def.~\ref{definition:MacroEqState}) \quad &Local control (Def.~\ref{definition:LocalControl}) of timescale~$O(L^0)$\quad & Inconsistent as in Example~\ref{example:Problem_CombiningiMATEandLocalControl}\\
    This paper & iMATE (Def.~\ref{definition:MacroEqState}) \quad &Macroscopic operation (Def.~\ref{definition:macroscopic-operation}) of timescale~$O(L^0)$\quad & Macroscopic passivity~\eqref{eq:macroscopic-passivity}\\
    - & Ordinary MATE  \quad &Macroscopic operation (Def.~\ref{definition:macroscopic-operation}) of timescale~$O(L^0)$\quad & Inconsistent as in Example~\ref{example:ProblemFiniteObs}
\end{tabular}
\end{ruledtabular}
\end{table*}

Note that, in the following sections we also consider macroscopic operations whose $B_L^\mu$'s are additive observables on macroscopic subsystems, instead of the entire system. In such a case, macroscopic uniformity of states is not preserved under the operations, and therefore, we need to focus on additive observables on macroscopic subsystems to discuss macroscopic equivalence after the operation.

\subsection{Main results}

The four main results of this paper are summarized as follows in the macroscopically uniform case.

The first main result is about the consistency between macroscopic equivalence and macroscopic operations.
Using the Lieb-Robinson bound~\cite{Gong2020}, we show that macroscopic equivalence is preserved by any macroscopic operation within an operation time that is independent of $L$:
\begin{mainresult}[Macroscopic equivalence after a macroscopic operation, Theorem~\ref{theorem:JAIVTMXN_3} for the macroscopically uniform case]
Suppose that states $\rho_{L}$ and $\sigma_{L}$ are macroscopically equivalent, $\rho_L\maceq \sigma_L$.
Take these states as initial states and consider a time evolution by a macroscopic operation $U_{L}(t,0)$,
$\rho_{L}(t)=U_{L}(t,0)\rho_{L}U_{L}^{\dagger}(t,0)$ and $\sigma_{L}(t)=U_{L}(t,0)\sigma_{L}U_{L}^{\dagger}(t,0)$.
Then, for any operation time $t^*$ independent of $L$, $t^*=O(L^0)$, the states $\rho_{L}(t^*)$ and $\sigma_{L}(t^*)$ remain macroscopically equivalent,
\begin{align}
    \rho_L(t^*)\maceq \sigma_{L}(t^*).
\end{align}
\end{mainresult}
\noindent
Note that this result holds regardless of whether $\rho_L$ and $\sigma_L$ represent iMATE.
In other words, this result also applies to nonequilibrium states.

The second main result corresponds to a quantum-mechanical version 
of Planck's principle~\cite{Lieb1999,Tasaki2016SecondLaw,Hokkyo2025}, which is one of various expressions of the second law 
in closed systems. In particular, thanks to the restriction of operations to macroscopic operations, we succeed in obtaining 
the following result, 
which applies to any quantum states representing iMATE, including pure quantum states:
\begin{mainresult}[Macroscopic passivity of iMATE, Corollary~\ref{corollary:macroscopic-passivity} for the macroscopically uniform case]
Let $H_L$ and $\rho_L$ be the initial Hamiltonian and the initial state at $t=0$, respectively. 
Suppose that $\rho_{L}$ represents iMATE described by $H_L$ at a nonnegative inverse temperature $\beta\ge 0$.
After any macroscopic operation $U_{L}(t^*,0)$ with an operation time $t^*>0$ independent of $L$, 
the expectation value of $H_L$ does not decrease extensively,
\begin{align}
    \lim_{L\to\infty}\mathrm{Tr}\Bigl[\rho_{L}(t^*)\frac{H_L}{N}\Bigr]\ge \lim_{L\to\infty}\mathrm{Tr}\Bigl[\rho_{L}\frac{H_L}{N}\Bigr],
    \label{eq:Summary_MacroPassive}
\end{align}
where  $\rho_{L}(t^*)=U_{L}(t^*,0)\rho_{L}U_{L}^{\dagger}(t^*,0)$.
\end{mainresult}
\noindent
Since the system does not exchange heat with the environment, the energy change of the system $\mathrm{Tr}\bigl[\bigl(\rho_{L}(t^*)-\rho_{L}\bigr)H\bigr]$ equals 
the work done by the external world.
In other words, an extensive work cannot be extracted from iMATE by any macroscopic operation of operation time $t^*=O(L^0)$.
In the special case where
$\rho_L$ represents iMATE at $\beta=0$, inequality~\eqref{eq:Summary_MacroPassive} becomes equality.
[By contrast, at $\beta>0$, an extensive work can be done on iMATE, i.e., there exist macroscopic operations with $t^*=O(L^0)$ such that the strict inequality holds in \eqref{eq:Summary_MacroPassive}.]
We will also show that the constraint $t^*=O(L^0)$ on the operation time is optimal by providing a counterexample to Eq.~\eqref{eq:Summary_MacroPassive} at any longer timescale. See Proposition~\ref{proposition:breakdown_macroscopic-passivity} for details.

To explain the third and fourth results, we introduce an appropriate spatial average of the reduced density matrix.
Let $\Cell[\bm{r}]:=\Cell+\boldsymbol{r}$ be the translation of 
a hypercube
$\Cell$ by $\boldsymbol{r}$.
Consider the reduced density matrix $\mathrm{Tr}_{\Lambda_L\setminus\Cell[\bm{r}]}[\rho_L]$ on $\Cell[\bm{r}]$,
and move it to $\Cell$ 
(so that it is defined on the same region $\Cell$ for any $\bm{r}$),
as schematically shown in Fig.~\ref{fig:Summary_Entropy} (a).
We call this reduced density matrix an \emph{$\ell$-local density matrix} around $\boldsymbol{r}$ and denote as $\rho_{L|\ell}^{\bm{r}}$.
\begin{figure*}
    \centering
    \includegraphics[width=0.9\linewidth]{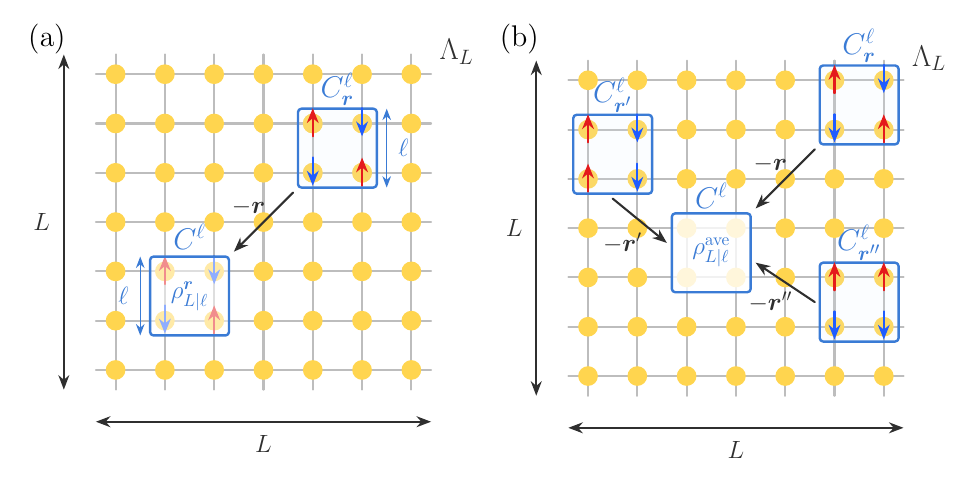}
    \caption{Schematic illustrations of (a) $\ell$-local density matrix $\rho_{L|\ell}^{\bm{r}}$ and 
    (b) its spatial average $\rho_{L|\ell}^{\mathrm{ave}}$, Eq.~\eqref{eq:Summary_rho^ave},
     for the macroscopically uniform case.
     (a) We first define $C_{\bm{r}}^\ell$, which is a region obtained from the translation of $C^\ell$ by $\bm{r}$. The density matrix $\rho_{L|\ell}^{\bm{r}}$ is obtained by the reduced density matrix of $\rho$ onto $C_{\bm r}^\ell$, followed by the translation by $-\bm{r}$. (b) The state $\rho_{L|\ell}^\mathrm{ave}$ is obtained by averaging $\rho_{L|\ell}^{\bm{r}}$  for all $\bm{r}$.
     }
    \label{fig:Summary_Entropy}
\end{figure*}
Even when $\rho_L$ is macroscopically uniform, it can be nonuniform microscopically, and thus $\rho_{L|\ell}^{\bm{r}}$ can differ for different $\bm{r}$.
To extract a thermodynamic property, we
consider its spatial average given by
(see Fig.~\ref{fig:Summary_Entropy} (b))
\begin{align}
    \rho_{L|\ell}^{\mathrm{ave}}:=
    \frac{1}{N}\sum_{\boldsymbol{r}\in\Lambda_{L}}\rho_{L|\ell}^{\boldsymbol{r}}.
    \label{eq:Summary_rho^ave}
\end{align}
This quantity characterizes the expectation values of additive observables composed of $\ell$-local observables in $\rho_L$ because it satisfies $\mathrm{Tr}[\rho_{L|\ell}^{\mathrm{ave}} \aLocal]=\mathrm{Tr}[\rho_LA_L]/N$.
Consequently, its $L\to\infty$ limit becomes a quantity common to all macroscopically equivalent states:
if $\rho_{L}$ and $\sigma_{L}$ are macroscopically equivalent, $\rho_L\maceq\sigma_L$, then 
$\lim_{L\to\infty}\rho_{L|\ell}^{\mathrm{ave}}=\lim_{L\to\infty}\sigma_{L|\ell}^{\mathrm{ave}}$ holds
for all $\ell\in\mathbb{N}$ (see Proposition~\ref{proposition:Equiv_rho^red} for details).

The third main result states that 
a suitably defined form of entropy, which we refer to as \emph{quantum macroscopic entropy density}
$s^{\mathrm{mac}}_{\ell}$ (Definition \ref{definition:MacroEntropy}),
agrees with thermodynamic entropy density $s^{\mathrm{TD}}$
if the state represents iMATE. For macroscopically uniform states, 
$s^{\mathrm{mac}}_{\ell}$ is simply 
given by the von Neumann entropy density of $\rho_{L|\ell}^{\mathrm{ave}}$, and hence the result can be stated as follows:
\begin{mainresult}[\label{mainresult:EntropyFormula}Entropy formula for iMATE, Theorem~\ref{theorem:s^red=s^TD} for the macroscopically uniform case]
Suppose that $\rho_{L}$ is in iMATE of the system described by $H$ at inverse temperature $\beta$.
Let $\rho_{L|\ell}^{\mathrm{ave}}$ be its spatical average, 
defined by \eqref{eq:Summary_rho^ave},
and $S_{\mathrm{vN}}[\rho_{L|\ell}^{\mathrm{ave}}]$ 
be the von Neumann entropy of $\rho_{L|\ell}^{\mathrm{ave}}$.
Then, the density of $S_{\mathrm{vN}}[\rho_{L|\ell}^{\mathrm{ave}}]$
approaches the thermodynamic entropy density
$s^{\mathrm{TD}}(\beta|H)$
$(:=\lim_{L\to\infty}S_{\mathrm{vN}}[\rho_{L}^{\mathrm{can}}(\beta|H)]/N)$
in the following iterated limit,
\begin{align}
    \lim_{\ell\to\infty}\lim_{L\to\infty}\frac{S_{\mathrm{vN}}[\rho_{L|\ell}^{\mathrm{ave}}]}{\ell^d}
    =s^{\mathrm{TD}}(\beta|H).
    \label{eq:Summary_EntropyFormula}
\end{align}
\end{mainresult}
\smallskip

\noindent
Importantly, this result states that Eq.~\eqref{eq:Summary_EntropyFormula} holds for any state representing iMATE.
This should be contrasted with 
 the naive von Neumann entropy density~$S_{\mathrm{vN}}[\rho_L]/N$ and the entanglement entropy density, 
which, in general, disagree with the thermodynamic entropy density
(see Secs.~\ref{sec:Entropy_s^mac_localUniform} and \ref{sec:numerical_entropy} for details). 
Furthermore, this result implies that thermodynamic entropy density can be obtained 
by measurements of additive observables, 
as will be explained in 
Sec.~\ref{sec:method.measure.sTD}. These measurements only require preparations of quantum states representing iMATE at a \emph{single} inverse temperature $\beta$, in contrast to entropy measurements in thermodynamics, which require preparations of multiple thermal equilibrium states at different inverse temperatures.

The fourth result is to derive the law of increasing entropy, another renowned formulation of the second law of thermodynamics under closed systems following adiabatic operations.
To explain this result, we introduce a process consisting of a macroscopic operation and a relaxation process, which is regarded as corresponding to the adiabatic operation in thermodynamics.
Suppose that, at time $t=0$, the system is described by the initial Hamiltonian $H_{0}$ and that the state of the system $\rho_{L}$ represents an iMATE at inverse temperature $\beta_{0}$.
Until time $t^*$, we perform an arbitrary macroscopic operation $U_{L}(t,0)$, where $t^*$ is assumed to be independent of $L$.
At time $t\ge t^*$, the Hamiltonian is taken to be a time-independent one $H_{1}$, and the system obeys the ``relaxation process'' defined by $U_{L}(t,t^*)=e^{-iH_{1}(t-t^*)}$.
That is,
the time dependence of the Hamiltonian is given by 
\begin{align}
    H(t)=\begin{cases}
        H_{0} & (t= 0)\\
        H_{0}-\sum_{\mu=1}^{m}f^{\mu}(t)B^{\mu} & (0<t<t^*)\\
        H_{1} & (t\ge t^*)
    \end{cases},
    \label{eq:Summary_EntropyIncrease_H(t)}
\end{align}
where $H(t)$ for $0<t<t^*$ takes the form of Eq.~\eqref{eq:SummaryResults_H(t)}.
We interpret the long-time average of the state at $t\ge t^*$ as the final state of the relaxation process, as is often done in the studies of thermalization,
\begin{align}
    &\overline{\rho_{L}(t>t^*)}
    \notag\\
    &\ :=\lim_{\mathcal{T}\to\infty}\frac{1}{\mathcal{T}}\int_{t^*}^{t^*+\mathcal{T}}dt \ e^{-i(t-t^*)H_{1}}\rho_{L}(t^*)e^{i(t-t^*)H_{1}},
    \label{eq:Summary_FinalState}
\end{align}
where $\rho_{L}(t^*)=U_{L}(t^*,0)\rho_{L}U_{L}^{\dagger}(t^*,0)$.

We show that the entropy density given in the LHS of Eq.~\eqref{eq:Summary_EntropyFormula} does not decrease in such a process:
\begin{mainresult}[\label{mainresult:IncreasingEntropy}Law of increasing entropy, Theorem~\ref{theorem:EntropyIncrease_WithoutThermalization} for the macroscopically uniform case]
For the process mentioned above, 
let $\sigma_L:=\overline{\rho_{L}(t>t^*)}$ be the final state evaluated by the long time average~\eqref{eq:Summary_FinalState}.
Under some mild assumption about the existence of a unique thermodynamic limit (for details, see Theorem~\ref{theorem:EntropyIncrease_WithoutThermalization}), 
the quantum macroscopic entropy density does not decrease through this process,
\begin{align}
    \lim_{\ell\to\infty}\lim_{L\to\infty}\frac{S_{\mathrm{vN}}[\rho_{L|\ell}^{\mathrm{ave}}]}{\ell^d}
    \le 
    \lim_{\ell\to\infty}\lim_{L\to\infty}\frac{S_{\mathrm{vN}}[\sigma_{L|\ell}^{\mathrm{ave}}]}{\ell^d}
    \label{eq:Summary_EntropyIncrease}
\end{align}
as long as $t^*=O(L^0)$.
\end{mainresult}
\noindent
This result should be distinguished from other results on the law of increasing entropy in \emph{isolated} quantum  systems~\cite{Neumann1929,Santos2011,Ikeda2015Second,Meier2025}, i.e., systems evolving via a time-independent Hamiltonian. 
This is because our result includes not only relaxation processes but also \textit{arbitrary} macroscopic operations described by time-dependent Hamiltonians, consistent with the setting of the second law of thermodynamics under adiabatic operations.
We can also show that the constraint $t^*=O(L^0)$ on the operation time is optimal by providing a counterexample to Eq.~\eqref{eq:Summary_EntropyIncrease} at any longer timescale. See Proposition~\ref{proposition:breakdown_entropy} for details.

We note that the assumption mentioned in Main result~\ref{mainresult:IncreasingEntropy} is about the relaxation process governed by $H_1$, and it is weaker than the assumption that thermalization occurs, 
which means that
the final state~\eqref{eq:Summary_FinalState} represents an iMATE.
If we further assume this assumption of thermalization,
Main results~\ref{mainresult:EntropyFormula} and \ref{mainresult:IncreasingEntropy} lead to 
\begin{align}
    s^{\mathrm{TD}}(\beta_{0}|H_{0})\le s^{\mathrm{TD}}(\beta_{1}|H_{1}),
\end{align}
where $\beta_1$ is the inverse temperature of iMATE represented by the final state.
This inequality indeed corresponds to the law of increasing entropy in thermodynamics,
upon which thermodynamics can be rigorously constructed~\cite{Lieb1999}.

\section{Setup} \label{sec:Setup}

In this section, we introduce our setup in a formal manner.
Note that mathematically rigorous treatment of our formulations is far from trivial because multiple length scales appear such as the sizes of local regions, subsystems, and the whole system.
We need to take the thermodynamic limit with appropriately controlling the relations among these length scales.
To clarify this point, we consider 
sequences of systems, subsystems, states and observables
rather than their members, in the following sections. This enables us to suitably define important concepts, such as iMATE.
For the reader's convenience,
we summarize the notations of our symbols in Appendix~\ref{sec:notation}.

\subsection{\label{sec:Setup_Lattice}Systems and subsystems}

We consider a quantum spin system on a $d$-dimensional hypercubic lattice 
$\Lambda_{L}=(\mathbb{Z}_{L})^d=\{-\lceil L/2\rceil+1,-\lceil L/2\rceil+2,...,\lfloor L/2\rfloor\}^d$ 
with side length $L$ and the number of sites $N=L^d$~\footnote{Generalization to other lattices is straightforward.}.
For simplicity, we impose the periodic boundary conditions, while our results presented in the following sections are not affected by the choice of boundary conditions.
Suppose that the total Hilbert space is given by a tensor product of local Hilbert spaces $\mathcal{H}_{\bm{r}}$ that are isomorphic to $\mathbb{C}^D$ and 
let $T_{\bm{r}}$ be the translation operator defined by 
\begin{align}
  T_{\bm{r}}\bigotimes_{\bm{r}^\prime\in\Lambda_{L}}\ket{j_{\bm{r}^\prime}}_{\bm{r}^\prime}
    =
    \bigotimes_{\bm{r}^\prime\in\Lambda_{L}}\ket{j_{\bm{r}^\prime}}_{\bm{r}^\prime+\bm{r}},
\end{align}
where $\{\ket{j}_{\bm{r}}\}_{j=0}^{D-1}$ 
is an orthonormal basis of $\mathcal{H}_{\bm{r}}$.
We write the Manhattan distance between two sites $\bm{r}$ and $\bm{r}^\prime$ by $|\bm{r}-\bm{r}^{\prime}|$.

In Sec.~\ref{sec:SummaryResults}, we have considered only additive observables on the whole system; however, to detect macroscopic nonuniformity of nonequilibrium states, such as that of convecting fluid in a box, and to distinguish these states from thermal equilibrium states, we should also focus on additive observables on subsystems. These subsystems should be small relative to the whole system so that 
finer spatial structures can be detected, yet sufficiently large to ensure that thermodynamics applies. To this end, we introduce \emph{macroscopic subsystems} as follows.

\begin{figure*}
    \centering
    \includegraphics[width=\linewidth]{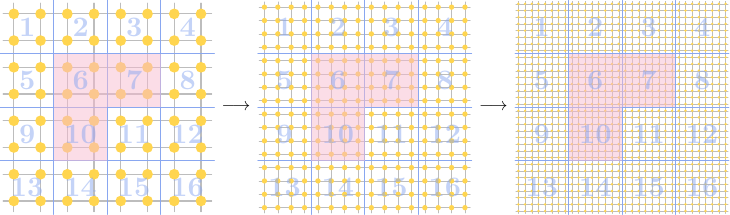}
    \caption{Illustration of how each primitive macroscopic subsystem grows as $L$ is increased ($L=8$, $16$, $32$ from left to right)
    in the case of $d=2$ and $K=4$. We can define a macroscopic subsystem $\mathcal{S}_L$ as a union of certain primitive macroscopic subsystems. For example, we can consider a proper macroscopic subsystem corresponding to the shaded region as a union of $\mathcal{S}_L^{(6)},\mathcal{S}_L^{(7)},$ and $\mathcal{S}_L^{(10)}$.} 
    \label{fig:MacroSubsystem}
\end{figure*}

We take a large integer $K$ independent of $L$~\footnote{If one would like the theory to be independent of the choice of $K$, one can take the limit $K\to\infty$ after the thermodynamic limit $L\to\infty$.
See Sec.~\ref{sec:Kdependence} for details.
}, and partition the total lattice $\Lambda_L$ into $K^d$ subsets $\mathcal{S}^{(1)}_{L}, \dots, \mathcal{S}^{(K^d)}_{L}$ that have almost equal side lengths $\lfloor L/K\rfloor$ or $\lceil L/K\rceil$.
We call each of these subsets a \emph{primitive macroscopic subsystem}. 
The union of primitive macroscopic subsystems is called a \emph{macroscopic subsystem}.
In Fig.~\ref{fig:MacroSubsystem}, we illustrate 
primitive macroscopic subsystems with $K=4$ and
a macroscopic subsystem $\mathcal{S}_L$ composed of three primitive macroscopic subsystems, $\mathcal{S}^{(6)}_{L}, \mathcal{S}^{(7)}_{L}$ and $\mathcal{S}^{(10)}_{L}$, for $L=8,16,32$.
When taking the thermodynamic limit $L\to\infty$, we consider such proper sequences of macroscopic subsystems composed of the same primitive macroscopic subsystems: 
\begin{definition}[\label{definition:ProperMacroSubsystem}Proper sequences of macroscopic subsystems]
 Primitive macroscopic subsystems $\mathcal{S}^{(k)}_{L}$ are labeled by $k\in\{1,...,K^d\}$ for each $L$ such that one in the same position has the same label~$k$, as illustrated in Fig.~\ref{fig:MacroSubsystem}.
A sequence of macroscopic subsystems $(\mathcal{S}_L)_{L\in\mathbb{N}}$ is said to be proper if each member $\mathcal{S}_L$ is composed of the same primitive macroscopic subsystems,
\begin{align}
    \mathcal{S}_L=\bigcup_{j=1}^{m}\mathcal{S}^{(k_j)}_{L}\quad\text{ for every }L\in\mathbb{N},
\end{align}
where $\{k_1,...,k_m\}\subset\{1,...,K^d\}$ is an arbitrary subset composed of $m$ labels with $m=1,...,K^d$.
\end{definition}
\noindent
Note that since $K$ is kept independent of $L$, the number of sites in $\mathcal{S}_{L}$ for any 
proper sequence~$(\mathcal{S}_L)_{L\in\mathbb{N}}$ grows in the same order as the total volume, $|\mathcal{S}_{L}|=\Theta(L^d)$.
Note also that, by taking $m=K^d$, we can see that $(\Lambda_L)_{L\in\mathbb{N}}$ is a proper sequence of macroscopic subsystem.

Since we will only consider proper sequences when discussing macroscopic subsystems, 
we often omit the term ``proper sequences'' hereafter.

\subsection{Additive observables}\label{sec:additive_obs}

In the spirit of thermodynamics, we are interested in observables that are spatially averaged in a macroscopic subsystem. 
Such observables correspond to additive observables defined below. 

To introduce them, we first define local observables.
\begin{definition}[\label{definition:LocalObs}Local observables]
Let $\Cell$ be a $d$-dimensional hypercube with side length $\ell$, $\{-\lceil \ell/2\rceil+1,-\lceil \ell/2\rceil+2,...,\lfloor \ell/2\rfloor\}^d$, composed of $\ell^d$ sites, and $\Cell[\bm{r}]:=\Cell+\bm{r}$ be its translation by $\bm{r}$. 
An observable supported on $\Cell[\bm{r}]$ for some $\bm{r}\in\Lambda_{L}$ is called an $\ell$-local observable.
\end{definition}
\noindent 
In this paper, we are interested in $\ell$-local observables whose $\ell$ is independent of $L$, $\ell=O(L^0)$, because we would like to regard an $\ell$-local observable $\aLocal$ as defined independently of the thermodynamic limit.

Using this, we introduce an additive observable on a macroscopic subsystem as follows
(see Fig.~\ref{fig:Additive_Subsystem}):

\begin{definition}[\label{definition:Additive}Additive observables]
Let $\mathcal{S}_{L}$ be a macroscopic subsystem and $\aLocal$ be an $\ell$-local observable.
The following observable is said to be the additive observable on $\mathcal{S}_{L}$ obtained from $\aLocal$,
\begin{align}
    A_{\mathcal{S}_{L}}(\aLocal)
    := \sumst{\bm{r}}{\supp(\aLocal_{\bm{r}}) \subset \mathcal{S}_{L}} \aLocal_{\bm{r}},
    \label{eq:Additive_S}
\end{align}
where $\aLocal_{\bm{r}}=T_{\bm{r}} \aLocal T_{\bm{r}}^{\dagger}$ is the translation of $\aLocal$ by $\bm{r}$ and
$\supp(\aLocal_{\bm{r}})$ means the support of the operator $\aLocal_{\bm{r}}$. 
An observable $A_L$ is said to be an additive observable composed of $\ell$-local observables if there are an $\ell$-local observable $\aLocal$ and a macroscopic subsystem $\mathcal{S}_{L}$ such that $A_L=A_{\mathcal{S}_{L}}(\aLocal)$.
\end{definition}
\noindent 
Physically, the density of the additive observable on $\mathcal{S}_{L}$ obtained from $\aLocal$, $A_{\mathcal{S}_{L}}(\aLocal)/|\mathcal{S}_{L}|$, corresponds to the spatial average of $\aLocal$ on the macroscopic subsystem $\mathcal{S}_{L}$.

\begin{figure}
    \centering
    \includegraphics[width=0.8\linewidth]{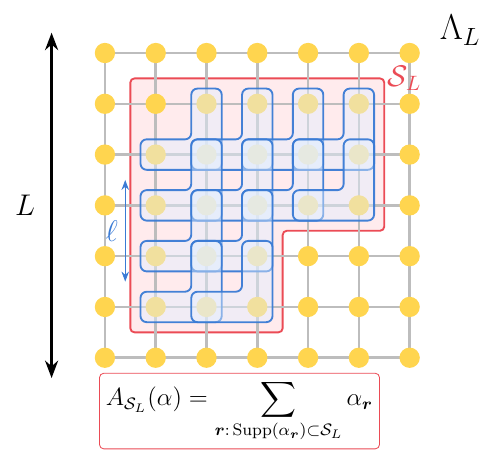}
    \caption{Schematic illustration of the additive observable
on a macroscopic subsystem $\mathcal{S}_{L}$ obtained from an $\ell$-local observable $\aLocal$,
Definition~\ref{definition:Additive}.
The special case where the subsystem is taken as the total lattice, $\mathcal{S}_{L} = \Lambda_{L}$,
was given in Fig.~\ref{fig:Summary_Additive} (b).}
    \label{fig:Additive_Subsystem}
\end{figure}

When taking the thermodynamic limit $L\to\infty$, we consider proper sequences of additive observables, each of which is characterized by an $\ell$-local observable $\aLocal$ and a proper sequence of macroscopic subsystems as follows: 
\begin{definition}[\label{definition:Proper_Additive}Proper sequences of additive observables]
A sequence of observables $(A_L)_{L\in\mathbb{N}}$ is said to be a proper sequence of additive observables if there are a positive integer $\ell\in\mathbb{N}$, an $\ell$-local observable $\aLocal$, and a proper sequence of macroscopic subsystems $(\mathcal{S}_{L})_{L\in\mathbb{N}}$ such that each element $A_L$ is given by the additive observable on $\mathcal{S}_L$ obtained from $\aLocal$,
\begin{align}
    A_L=A_{\mathcal{S}_{L}}(\aLocal)\quad\text{ for every }L\in\mathbb{N}.
    \label{eq:Proper_Additive}
\end{align}
The set of all proper sequences of additive observables is denoted by $\SetAdditive$.
\end{definition}
\noindent
Note that 
$\ell$ and $\aLocal$ are independent of $L$ in this definition.
In particular, for any proper sequence $(A_{L})_{L\in\mathbb{N}}$, the operator norm~$\|\bullet\|$ of $A_{L}$ satisfies $\| A_{L}\|\le |\mathcal{S}_{L}|\times\|\aLocal\|=O(L^d)$, which indicates the extensive property of such a sequence.
Since we will only consider proper sequences when discussing additive observables, 
we often omit the term ``proper sequences'' hereafter.

Let us clarify the meaning of $\mathcal{A}$, i.e., the set of \emph{all} (proper sequences of) additive observables.
Here, ``\emph{all}'' refers to the following arbitrarinesses of their choices: 
\begin{enumerate}[label={(\roman*)},ref={\roman*}]
    \item\label{enum:choice_ell} $\ell$ is an arbitrary positive integer, $\ell\in\mathbb{N}$;
    \item\label{enum:choice_a} $\aLocal$ is an arbitrary $\ell$-local observable on $\Cell$~\footnote{Because a translation of $\aLocal$ causes no change in $A_{\mathcal{S}_L}(\aLocal)$, we can restrict the support of $\aLocal$ to $\Cell$ without loss of generality.};
    \item\label{enum:choice_S} $(\mathcal{S}_{L})_{L\in\mathbb{N}}$ is any
    proper sequence of macroscopic subsystems.
\end{enumerate}
Note that, since arbitrariness~(\ref{enum:choice_ell}) states that $\ell$ can be taken arbitrarily large, $\mathcal{A}$ is an infinite set.
However,
for each $\ell$, the number of independent 
 additive observables $(A_L)_{L\in\mathbb{N}}$
is bounded from above by 
the total number of choices~(\ref{enum:choice_a}) and (\ref{enum:choice_S}), $D^{2\ell^d}\times 2^{K^d}$ 
(recall that $D$ is the dimension of the local Hilbert space), which increases with $\ell$ but is independent of $L$.
This means that $\mathcal{A}$ is much narrower than the set of sequences of arbitrary observables, as the number of independent observables on $\Lambda_L$ grows as $D^{2L^d}$ with respect to $L$.

Now, we give an easy example showing the difference between $\mathcal{A}$ and the set of sequences of arbitrary observables.
Let $U_{L}$ be the translation-invariant unitary operator in spin-$1/2$ systems defined by $U_{L}=\prod_{\bm{r}\in\Lambda_{L}}\sigma^{z}_{\bm{r}}$, where $\sigma_{\bm{r}}^{x,y,z}$ is the Pauli operator on site $\bm{r}$. 
To represent it in the form of Eq.~\eqref{eq:Additive_S}, we need to take $\ell=L$ and $\mathrm{Supp}(\aLocal)=\Lambda_{L}$. 
This means that the sequence of 
these 
unitaries $(U_{L})_{L\in\mathbb{N}}$ cannot be written in the form of Eq.~\eqref{eq:Proper_Additive}
with $\ell=O(L^0)$ 
and is not a proper sequence.

The additive observables defined above are important because they satisfy the additivity property, Eq.~\eqref{eq:Additive_S_approximation} below, and hence, they are a natural candidate for quantum-mechanical representations of additive quantities in thermodynamics~\footnote{
Here, `additive quantities in thermodynamics' mean additive \emph{mechanical} quantities in thermodynamics such as the total magnetization. 
That is, we exclude \emph{genuine} thermodynamic quantities
such as the free energy.}.
Suppose that 
a macroscopic subsystem $\mathcal{S}_L$ is composed of $m$ $(=O(L^0))$ primitive macroscopic subsystems $\{\mathcal{S}^{(k_j)}_L\}_{j=1}^{m}$ ($k_{j}\in\{1,2,...,K^d\}$).
An additive observable on $\mathcal{S}_{L}$ can be well approximated by the sum of additive observables on these primitive macroscopic subsystems because it satisfies
\begin{align}
    \Bigl\| A_{\mathcal{S}_{L}}(\aLocal) - \sum_{j=1}^{m} A_{\mathcal{S}^{(k_j)}_{L}}(\aLocal) \Bigr\| = O(L^{d-1}).
    \label{eq:Additive_S_approximation}
\end{align}
Here, the $O(L^{d-1})$ term in the right-hand side (RHS) corresponds to the contributions from the boundaries of 
$\mathcal{S}^{(k_{j})}_{L}$'s.
This property~\eqref{eq:Additive_S_approximation} corresponds to the additivity property in thermodynamics because it indicates that the expectation values of densities, $A_{\mathcal{S}_{L}}(\aLocal)/|\mathcal{S}_{L}|$ and $\sum_{j=1}^{m} A_{\mathcal{S}^{(k_j)}_{L}}(\aLocal)/|\mathcal{S}_{L}|$, coincide in the thermodynamic limit.
Therefore, we regard that additive observables defined by Eq.~\eqref{eq:Additive_S} correspond to additive quantities in thermodynamics.

\section{\label{sec:MacroEquiv}Macroscopic equivalence}

In this section, we introduce the notion of `macroscopic equivalence'~\footnote{A similar concept has been investigated in Ref.~\cite{Mori2016MacroEquiv} and references therein.}.
It is applicable to both equilibrium and nonequilibrium states.
We will utilize it to define macroscopic thermal equilibrium in the next section.

Let $\rho_{L}$ be a density matrix of the system of size $L$.
To take the thermodynamic limit, it should be considered as a member of an appropriate sequence $(\rho_{L})_{L\in\mathbb{N}}$ of density matrices.
Since we are interested in additive observables, we characterize appropriate sequences of density matrices in terms of additive observables as follows:
\begin{definition}[Macroscopic state]\label{definition:MacroState}
A sequence $(\rho_{L})_{L\in\mathbb{N}}$ of density matrices
$\rho_{L}$ defined on $\Lambda_{L}$
is said to represent a macroscopic state if, 
for any proper sequence of additive observables $(A_L)_{L\in\mathbb{N}}$, the expectation value of $A_{L}$ divided by the total number of sites converges in the limit $L\to\infty$, i.e.,
\begin{align}
    \lim_{L\to\infty} \mathrm{Tr} \bigl[ \rho_L  A_{L} / N \bigr]\text{ exists for every }(A_{L})_{L\in\mathbb{N}}\in \SetAdditive.
    \label{eq:MacroState}
\end{align}
Here, $\SetAdditive$ is the set of all proper sequences of additive observables defined in Definition~\ref{definition:Proper_Additive}.
\end{definition}
\noindent

In the special case where 
$\rho_{L}$ is translation invariant, 
$\mathrm{Tr} \bigl[ \rho_L  A_L\bigr] / N=\mathrm{Tr}[\rho_{L} \aLocal] \, |\mathcal{S}_{L}|/N$ holds for any $A_L=A_{\mathcal{S}_{L}}(\aLocal)$,
and hence, 
the above condition~\eqref{eq:MacroState} 
reduces to 
the existence of 
the $L \to \infty$ limit of $\mathrm{Tr}[\rho_{L} \aLocal]$ 
for every local observable~$\aLocal$.

Furthermore, we define an equivalence relation between sequences of density matrices as follows:
\begin{definition}[\label{definition:MacroEquiv}Macroscopic equivalence]
Let $(\rho_{L})_{L\in\mathbb{N}}$ and $(\sigma_{L})_{L\in\mathbb{N}}$ 
represent macroscopic states. We say that they are macroscopically equivalent and denote 
\begin{align}
    (\rho_{L})_{L\in\mathbb{N}}\maceq (\sigma_{L})_{L\in\mathbb{N}},
\label{eq:maceq}\end{align}
if
\begin{align}
    \lim_{L\to\infty} \mathrm{Tr} \Bigl[ \rho_L \frac{A_{L}}{N} \Bigr]
    = \lim_{L\to\infty} \mathrm{Tr} \Bigl[ \sigma_L \frac{A_{L}}{N} \Bigr]
    \notag\\
    \text{for all }(A_{L})_{L\in\mathbb{N}}\in \SetAdditive.
    \label{eq:MacroscopicEquiv}
\end{align}
For macroscopically equivalent states
$(\rho_{L})_{L\in\mathbb{N}}$ and $(\sigma_{L})_{L\in\mathbb{N}}$, we also say that they represent the same macroscopic state.
\end{definition}
\noindent
The relation \eqref{eq:maceq} indeed satisfies the definition of an equivalence relation.
See Fig.~\ref{fig:MacroEquiv} for a schematic illustration.

As a simple example, let $\rho_{L}$ be a pure product state of the form $(\ket{0}\bra{0})^{\otimes \Lambda_{L}}$ in a spin-$1/2$ system, and $\sigma_{L}$ be its spin flip at site $\bm{0}$,
$\sigma_{L}=\sigma_{\bm{0}}^{x}\rho_{L}\sigma_{\bm{0}}^{x}$.
Then both sequences of density matrices $(\rho_{L})_{L\in\mathbb{N}}$ and $(\sigma_{L})_{L\in\mathbb{N}}$ represent macroscopic states and they are macroscopically equivalent.
Hence, they represent the same macroscopic state.

In addition, it is usually considered
in thermodynamics that
any additive quantity should take a macroscopically definite value, and hence a single-shot measurement suffices to obtain its equilibrium value. 
In other words, variances of additive quantities should be macroscopically negligible.
When a representation of a macroscopic state satisfies this property for all additive observables, we call it a representation of a \emph{normal} macroscopic state:
\begin{definition}[Normal macroscopic state]\label{definition:MacroState_Normal}
Let $(\rho_{L})_{L\in\mathbb{N}}$ represent a macroscopic state. It is said to represent a normal macroscopic state if 
the variances of 
additive observables in all proper sequences of additive observables
are macroscopically negligible, that is,
\begin{align}
    \mathrm{Tr} \bigl[ \rho_L (A_{L} /N)^2 \bigr]
    - \bigl( \mathrm{Tr} \bigl[ \rho_L A_{L} /N \bigr] \bigr)^2
    = o(L^0)
    \notag\\
    \text{for all }(A_{L})_{L\in\mathbb{N}}\in \SetAdditive.
    \label{eq:MacroscopicallyNormal}
\end{align}
\end{definition}

\section{\label{sec:MacroEquilibrium}Infinite-observable macroscopic thermal equilibrium}

In this section, we introduce the notion of thermal equilibrium used in this paper, which is characterized by all additive observables.

\subsection{\label{sec:Setup_Hamiltonian}Hamiltonian}

Thermodynamics 
treats \emph{simple systems} as the fundamental building blocks of a macroscopic system~\cite{Callen1985,Lieb1999}.
In equilibrium statistical mechanics, such a simple system corresponds to a many-body system described by a translation-invariant Hamiltonian.
Therefore, when investigating equilibrium states (e.g., before or after thermodynamic operations), we consider a system described by a Hamiltonian~$H$ that is an additive observable on the whole system~$\Lambda_{L}$, unless otherwise stated.
That is, the Hamiltonian is translation invariant and can be written as
\begin{align}
    H_L =A_{\Lambda_{L}}(h)= \sum_{\bm{r}\in\Lambda_{L}} h_{\bm{r}}
    \label{eq:Hamiltonian}
\end{align}
where $h$ is some $\ell$-local observable, with $\ell$ independent of $L$, and $h_{\bm{r}}=T_{\bm{r}} h T_{\bm{r}}^{\dagger}$.

For example, if we 
choose $d = 1$ and $h = J \sigma_{0}^z \sigma_{1}^z - g \sigma_{0}^x$
then Eq.~\eqref{eq:Hamiltonian} defines the following additive observable:
\begin{align}
    H_L= J \sum_{j=1}^L \sigma_{j}^z \sigma_{j+1}^z - g \sum_{j=1}^L \sigma_{j}^x.
\label{eq:H_transverse_Ising}
\end{align}
This is nothing but the Hamiltonian of a one-dimensional transverse-field Ising model with the periodic boundary condition. 
As is seen in this example, Eq.~\eqref{eq:Hamiltonian} defines general translation-invariant systems with short-range interactions.

\subsection{Gibbs states}\label{sec:Gibbs_states}

In statistical mechanics, thermal equilibrium states are often described by statistical mechanical ensembles, such as the microcanonical and the canonical Gibbs states.
In this paper, we employ 
the canonical Gibbs state \eqref{eq:DEF_rho^can}
as a reference state for a general representation of thermal equilibrium states.
This implies that we consider thermal equilibrium states specified 
(labeled) only by inverse temperature $\beta$ and $N$.
It will be straightforward to extend our theory 
to more general cases where additional variables are necessary to specify equilibrium states.
Furthermore, we exclude systems that exhibit phase coexistence.
Although the extension to such systems does not seem difficult, 
it is out of the scope of this paper.

In order to utilize the canonical Gibbs state as a reference state, 
we assume that its sequence
represents a macroscopic state in the sense defined in Sec.~\ref{sec:MacroEquiv}:
\begin{assumpGibbs}\label{assumption:GibbsState}
The sequence of canonical Gibbs states $\bigl(\rho_{L}^{\mathrm{can}}(\beta|H_L)\bigr)_{L\in\mathbb{N}}$ represents a  macroscopic state, that is, it satisfies Eq.~\eqref{eq:MacroState}.
\end{assumpGibbs}
\noindent
We consider that this assumption will be reasonable aside from the first-order phase transition point for the following reason.
For any additive observable $A$, its expectation value in $\rho_{L}^{\mathrm{can}}(\beta|H_L)$ is given by the derivative of a generalized free energy,
\begin{align}
    \mathrm{Tr}\bigl[\rho_{L}^{\mathrm{can}}(\beta|H_L) A\bigr]=\Bigl(\frac{1}{\beta}\frac{\partial}{\partial h}\log \mathrm{Tr}\bigl[e^{-\beta(H_L-hA)}\bigr]\Bigr)\Bigm|_{h=0}.
\end{align}
It was shown that the thermodynamic limit of this free energy exists~\cite{Ruelle1999,Tasaki2018}.
Furthermore, it is usually considered that this derivative and the thermodynamic limit can be interchanged if $\beta$ is not at the first-order phase transition point \footnote{
This was proven in the Supplemental Material of \cite{Fujikura2016} under the same assumption as the 
quantum central limit theorem (QCLT) of the characteristic function.}.
This indicates that, for such a value of $\beta$, the thermodynamic limit of the expectation value of $A$ should coincide with the derivative of the thermodynamic limit of the above free energy, meaning that Assumption~\ref{assumption:GibbsState} is satisfied.

In some parts of this paper, where the fluctuations of additive observables are crucial for our analysis, we will also use the following assumption:
\begin{assumpGibbs}\label{assumption:GibbsState_Normal}
The sequence of canonical Gibbs states $\bigl(\rho_{L}^{\mathrm{can}}(\beta|H_L)\bigr)_{L\in\mathbb{N}}$ represents a normal macroscopic state, that is, it satisfies both Eqs.~\eqref{eq:MacroState} and \eqref{eq:MacroscopicallyNormal}.
\end{assumpGibbs}
\noindent
The second assumption~\eqref{eq:MacroscopicallyNormal} is 
ensured by
the assumption of the decay of correlations between distant local observables.
Therefore, it will be reasonable except when $\beta$ is at the first-order phase transition point.
Note that this assumption will hold even at the second-order phase transition point because an arbitrarily slow decay of correlations is sufficient for this assumption.

\subsection{\label{sec:MacroEquilibrium_iMATE}Thermal equilibrium in this paper}

In recent 
studies of the foundation of statistical mechanics, two notions of thermal equilibrium are well known: ``macroscopic thermal equilibrium (MATE)'' and ``microscopic thermal equilibrium (MITE)''~\cite{Goldstein2015,Goldstein2017,Mori2018}.
The notion of MATE (resp. MITE) is defined focusing on macroscopic (resp. microscopic) observables.
The notion of MATE applies not only to quantum systems but also to classical systems; MATE has a classical counterpart~\cite{Goldstein2017,Mori2018}, and we can discuss whether each point in the classical phase space satisfies it, whereas MITE does not work well in such situations.
For example, consider a classical gas in thermal equilibrium.
Most of the snapshots are obviously in the same equilibrium state.
However, since such a snapshot is not uniform microscopically
(though uniform macroscopically), 
each of them is not in MITE but in MATE.
The situation is the same in quantum systems;
equilibrium states can be non-uniform microscopically, which is 
not in MITE but in MATE.
In addition, MATE will have a stronger connection to thermodynamics than MITE, as macroscopic observables used in MATE will correspond to macroscopic physical quantities in thermodynamics.
Considering its generality and consistency with thermodynamics, 
we adopt the approach based on MATE to discuss a notion of thermal equilibrium.

In the conventional formulation of MATE, whether a state represents MATE is determined solely from the probability distributions of a \emph{finite} number of given macroscopic observables~\cite{Goldstein2015,Tasaki2016,Goldstein2017}, as will be explained in Sec.~\ref{sec:OrdinaryMATE}.
However, for a certain choice of a finite number of macroscopic observables, we can provide a simple counterexample to the second law of thermodynamics, as will be explained in Example~\ref{example:ProblemFiniteObs}.
Furthermore, we expect that even if we focused on more generic but a \emph{finite} number of observables~\cite{sugimoto2023bounds}, we would encounter almost the same problem.

To resolve such 
inconsistency with the second law, we here
characterize 
thermal equilibrium states by focusing on all 
(thus, as discussed in Sec.~\ref{sec:additive_obs}, an \emph{infinite} number of) 
additive observables. We call such a notion of thermal equilibrium 
\emph{infinite-observable macroscopic thermal equilibrium} (iMATE):
\begin{definition}[iMATE] \label{definition:MacroEqState}
Let $(\rho_{L})_{L\in\mathbb{N}}$ represent a macroscopic state (cf. Definition~\ref{definition:MacroState}).
It is said to represent iMATE of a system described by a translation-invariant Hamiltonian $H_L$ 
if there exists a finite inverse temperature $-\infty<\beta<+\infty$
such that
$(\rho_{L})_{L\in\mathbb{N}}$ is macroscopically equivalent to $\bigl(\rho_{L}^{\mathrm{can}}(\beta|H_L)\bigr)_{L\in\mathbb{N}}$,
\begin{align}
    \bigl(\rho_{L}\bigr)_{L\in\mathbb{N}}
    \maceq
    \bigl(\rho_{L}^{\mathrm{can}}(\beta|H_L)\bigr)_{L\in\mathbb{N}},
    \label{eq:MacroEqState}
\end{align}
where $\maceq$ is defined by Eq.~\eqref{eq:MacroscopicEquiv}.
\end{definition}
\noindent
In other words, we say that $(\rho_{L})_{L\in\mathbb{N}}$ represents an iMATE 
if the expectation value of the density of any additive observable coincides with its equilibrium value in the thermodynamic limit.

Furthermore, in thermodynamics, it is usually considered that additive physical quantities do not have macroscopically large fluctuations.
When we impose this behavior in addition to Definition~\ref{definition:MacroEqState}, we call such an equilibrium state a \emph{normal iMATE} as follows:
\begin{definition}[Normal iMATE] \label{definition:MacroEqState_Normal}
A sequence of density matrices $(\rho_{L})_{L\in\mathbb{N}}$ is said to represent normal iMATE 
if it represents a normal macroscopic state [Definition~\ref{definition:MacroState_Normal}], and at the same time iMATE [Definition~\ref{definition:MacroEqState}].
\end{definition}

We discuss physical motivations behind 
defining thermal equilibrium by iMATE.
As an example, let us consider a classical dilute gas.
Its thermal equilibrium states are specified only by 
the internal energy $U$, volume $V$, and particle number $N$.
This means that, within the \emph{space of thermal equilibrium states}, each thermal equilibrium state is labeled by $(U,V,N)$.
Suppose that an experimenter prepares a state of the gas in such a way that 
$(U,V,N)$ take some specified values.
In order to judge whether the gas is in thermal equilibrium or not,
he/she should also measure other additive observables \footnote{Any thermodynamic quantities can be measured by measuring additive quantities, 
since they are functions of additive observables.} to explore the 
\emph{space of all macroscopic states}, 
since there are various nonequilibrium states 
that have the same values of $(U,V,N)$ as the equilibrium state.
For instance, 
the total kinetic energy 
will take the unique value 
satisfying equipartition in the equilibrium state, 
whereas it can take other values in nonequilibrium states.
More generally, if, in some state, some additive observable 
deviates from the equilibrium value, we can measure this observable, at least in principle, and can detect a difference from the thermal equilibrium state. Reflecting this arbitrariness of additive observables measured in real experiments, 
we have defined iMATE, focusing on all additive observables.

Some readers might still be mystified by the fact that iMATE 
thus defined focuses 
on an infinite number of additive observables.
We point out that iMATE requires focusing on 
\emph{fewer} observables than MITE, 
a commonly used notion of thermal equilibrium in the study of thermalization in isolated quantum systems~\cite{Goldstein2015,Goldstein2017,Mori2018}.
For instance, in spin-$1/2$ chains, the number of additive observables composed of $\ell$-local observables that are used to define iMATE scales as $4^\ell$,
which is independent of $L$.
In comparison, the number of $\ell$-local observables 
that are used to define MITE
scales as $4^\ell \times L$ (see Sec.~\ref{sec:MITE} for details).
Hence, the number of observables considered in iMATE is much smaller than that in MITE.
As a result, 
iMATE includes all states that represent MITE
(see Sec.~\ref{sec:MITE}).
It also includes equilibrium states that are {\em not}
included in MITE,
as explained in Table~\ref{tbl:Comparison_Equilibrium} and in Sec.~\ref{sec:Example_iMATE}.

\subsection{\label{sec:Example_iMATE}Examples of iMATE}

Let us provide examples of a sequence of states that represents iMATE, but that is not in MITE.

The first example is the 
minimally entangled typical thermal states (METTS)~\cite{White2009,Stoudenmire2010},
which provides a natural illustration.
Let $\ket{i} \ (i=0,1, \dots, D^N-1)$ 
denote a basis vector of a 
computational basis, which can be written as
\begin{align}
    \bigotimes_{\bm{r}\in\Lambda_{L}} \ket{j_{\bm{r}}}_{\bm{r}},
\end{align}
where $\{\ket{j_{\bm{r}}}_{\bm{r}}\}_{j_{\bm{r}}=0}^{D-1}$ is an orthonormal basis of the local Hilbert space $\mathcal{H}_{\bm{r}}$.
The METTS are a family of states obtained by imaginary-time evolution of these computational basis vectors,
\begin{align}
    \ket{\phi(i)} = \frac{1}{\sqrt{P(i)}} e^{- \frac{1}{2} \beta H_L} \ket{i},
    \label{eq:METTS}
\end{align}
where $P(i) = \braket{i| e^{- \beta H_L} |i}$. 
The canonical Gibbs state can be decomposed as a classical mixture of METTS:
\begin{align}
    \rho_{L}^{\mathrm{can}}(\beta|H_L)
    = \sum_{i} \frac{P(i)}{Z} \ket{\phi(i)}\bra{\phi(i)}.
    \label{eq:METTS_canonical}
\end{align}
Therefore, the average of the expectation value of any observable in $\ket{\phi(i)}$'s for all $i$ according to the probability $P(i)/Z$
agrees with that in the canonical Gibbs state.

Since METTS are obtained by evolving a product state containing no spatial entanglement for an imaginary time of order $O(N^0)$, they 
retain only a small amount of entanglement~\cite{Kusuki2024} and exhibit short-range correlations.
As a consequence, the self-averaging property is expected to hold, and for any additive observable, the expectation value of its density in a single typical METTS converges in the thermodynamic limit to that in the canonical Gibbs state, without ensemble averaging.
Therefore, a sequence of typical METTS sampled according to the probability distribution $\frac{P(i)}{Z}$ is macroscopically equivalent to the canonical Gibbs state and hence represents iMATE.
In fact, we can rigorously show the following
(for the precise statement and the rigorous proof, see Propositions~\ref{proposition:METTS_represent_iMATE}--\ref{proposition:METTS_Justification_Decay_pL} in Appendix~\ref{sec:Proof_METTS}).

\begin{example}[\label{example:TypicalMETTS_represents_iMATE}Typical sequence of METTS represents iMATE, informal version of Proposition~\ref{proposition:METTS_Justification_Decay_pL}]
We choose $i_L\in\{0,1,...,D^N-1\}$ according to the probability distribution $P(i_L)/Z$ and define $\rho_L=\ket{\phi(i_L)}\bra{\phi(i_L)}$ for each $L\in\mathbb{N}$.
Under a reasonable 
assumption for the canonical Gibbs state $\rho_{L}^{\mathrm{can}}(\beta|H_L)$ (See Assumption~\ref{assumption:GibbsState_Concentration}.), we can show that the sequence of states $(\rho_L)_{L\in\mathbb{N}}$ satisfies $(\rho_L)_{L\in\mathbb{N}}\maceq\bigl(\rho_{L}^{\mathrm{can}}(\beta|H_L)\bigr)_{L\in\mathbb{N}}$ with probability $1$.

\end{example}
\noindent
On the other hand, METTS are in general not locally uniform (see Def.~\ref{definition:LocallyUniform} for a precise definition) and thus excluded from MITE.

The following second and third examples of iMATE contain
no random variables (in contrast to METTS, which contains $i_L$ as a random variable).
Since the finite temperature case is technically cumbersome, we first 
give an example at infinite temperature, $\beta = 0$:
\begin{example}[\label{example:NonMITE_MacroEqState}Infinite temperature iMATE without random variables, informal version of Proposition~\ref{proposition:NonMITE_MacroEqState}]
For simplicity, we consider a one-dimensional spin-$1/2$ system ($d=1,D=2$) with $L=n2^n$ for some $n\in\mathbb{N}$.
Let $\{\ket{b_n(i)}\}_{i=0}^{2^n-1}$ be the computational basis on consecutive $n$ sites, as shown in Eq.~\eqref{eq:basis_binary}.
We consider a state $\ket{\psi_L}$ given by the tensor product of the basis vectors, $\ket{\psi_L}=\bigotimes_{i=0}^{2^n-1}\ket{b_n(i)}$.
Because of the spatial structure of $\ket{\psi_L}$, the spatial average of the expectation value of any $\ell$-local observable $\aLocal$ in $\ket{\psi_L}$ is proportional to its trace,
\begin{align}
    \bra{\psi_L}\frac{1}{L}\sum_{j=1}^{L}\aLocal_j\ket{\psi_L}\simeq \frac{1}{n}\sum_{j=1}^{n}\frac{\mathrm{Tr}_{n}[\aLocal_j]}{2^n}
    \simeq \frac{\mathrm{Tr}_{n}[\aLocal]}{2^n},
\end{align}
where $\mathrm{Tr}_n[\bullet]$ is the trace on $n$ consecutive sites.
The above approximate equalities $\simeq$ become exact after taking the thermodynamic limit $n\to\infty$.
Therefore, 
the sequence of states $(\rho_L)_{L\in\mathbb{N}}$ given by $\rho_L=\ket{\psi_L}\bra{\psi_L}$ is macroscopically equivalent to the sequence of maximally mixed states, $(\rho_L)_{L\in\mathbb{N}}\maceq (I_{\Lambda_L}/2^N)_{L\in\mathbb{N}}$, which means it represents iMATE at inverse temperature $\beta=0$ for any Hamiltonian.
\end{example}
\noindent
Note that the expectation values of Pauli-$z$ operators $\sigma_{j}^{z}$ in $\ket{\psi_L}$ depend on the site, and hence, this state does not represent MITE described by a translation-invariant Hamiltonian.

Finally, we provide an example of iMATE at finite temperature that contains no random variables:
\begin{example}[\label{example:NonMITE_MacroEqState_FiniteTemp}Finite temperature iMATE without random variables, informal version of Proposition~\ref{proposition:NonMITE_MacroEqState_FiniteTemp}]
For simplicity, we consider a one-dimensional spin-$1/2$ system $(d=1, D=2)$ with $L=n2^n$ for some $n\in\mathbb{N}$. Let $H_n$ be the sum of local-Hamiltonians $h_{\bm{r}}$ restricted to consecutive $n$ sites, $\{1,2,...,n\}$. We can construct a sequence of states representing iMATE at inverse temperature $\beta$ by the following steps: 
(1)~Take the eigenstates of $H_{n}$, denoted by $\ket{E_\nu}$ 
$(\nu\in\{1,2,...,2^n\})$. 
(2)~Perform the discrete Fourier transform, $\ket{q}\propto \sum_{\nu}e^{2\pi iq\nu/2^n}\ket{E_\nu}$ ($q\in\{0,1,...,2^n-1\}$). 
(3)~Evolve them for an imaginary time $\beta/2$ by $H_n$, $\ket{\varphi(q)}= e^{-\beta H_n/2}\ket{q}/\sqrt{Z_n}$ with $Z_n=\mathrm{Tr}[e^{-\beta H_n}]$. Importantly, the normalization constants take the same value (i.e., $\sqrt{Z_n}$) for all $q$. 
(4)~Take a tensor product of $2^n$ number of them, $\ket{\psi_L}=\bigotimes_{q}\ket{\varphi(q)}$. 

Then, the sequence of states $(\rho_L)_{L\in\mathbb{N}}$ with $\rho_L=\ket{\psi_L}\bra{\psi_L}$ is macroscopically equivalent to the sequence of the tensor product of $2^n$ number of the canonical Gibbs states on $n$ sites, $(e^{-\beta H_n}/Z_n)^{\otimes 2^n}$. The latter sequence can be shown to be macroscopically equivalent to $\bigl(\rho_{L}^{\mathrm{can}}(\beta|H_L)\bigr)_{L\in\mathbb{N}}$.
\end{example}
\noindent
Note that, since the state in this example
is a tensor product of $2^n$ number of $n$-qubit states, it has no correlation between local observables on two adjacent $n$-qubit blocks. By contrast, the Gibbs state at $\beta>0$ usually has nonzero correlation. These imply that the state in 
this example
is not in MITE.

\section{\label{sec:MacroOp_MacroEquiv}Macroscopic operations and equivalence after such operations}

\subsection{Macroscopic operations}

In this section, we introduce macroscopic operations, which are unitary operators regarded as corresponding to adiabatic thermodynamic operations.
After that, we provide our crucial theorem stating that macroscopic equivalence defined in the previous section is preserved after a macroscopic operation.

As discussed in the Introduction, we should not consider arbitrary operations when discussing consistency with thermodynamics; rather, we should restrict them to those feasible within the framework of thermodynamics. 
Although quantum-mechanical operations corresponding to thermodynamic operations have not been fully identified yet (to the 
authors' best knowledge), we here 
consider the following class of quantum operations, 
which seems most interesting in many physical situations.

A closed system is manipulated through some fields (or control parameters), such as a magnetic field. 
In thermodynamic systems, such fields are usually conjugate to additive quantities.
Therefore, we define macroscopic operations as unitary time evolutions induced by the external fields coupled to additive observables, including those on macroscopic subsystems~\footnote{In this paper, we only consider the case where the back-reaction effects from the system to the fields, such as electric or magnetic field, are negligible.},
as shown schematically in Fig.~\ref{fig:MacroOp}.
More precisely:

\begin{figure}
    \centering
    \includegraphics[width=\linewidth]{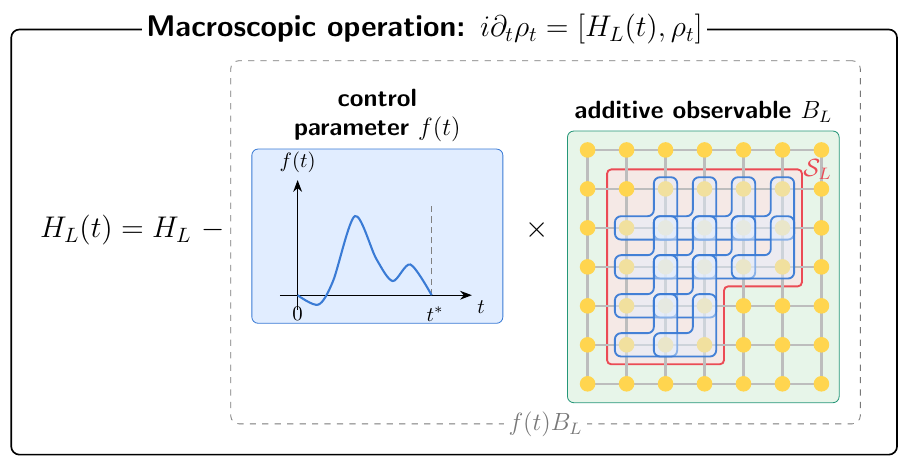}
    \caption{Schematic illustration of macroscopic operations, Definition~\ref{definition:macroscopic-operation}, where we consider the case with $m=1$.
    The special case where $\mathcal{S}_{L} = \Lambda_{L}$
was given in Fig.~\ref{fig:Summary_MacroOp}.
}
    \label{fig:MacroOp}
\end{figure}

\begin{definition}[Macroscopic operation] \label{definition:macroscopic-operation}
Let $m$ be a positive integer independent of $L$,
$(B^1_L,B^2_L,\dots,B^m_L)$ be an arbitrary set of 
proper sequences of
additive observables on some macroscopic subsystems, and $(f^1(t),f^2(t),\dots,f^m(t))$ be an arbitrary set of bounded real functions independent of $L$~\footnote{Our results also hold even when $f^\mu(t)$'s depend on $L$ but their magnitudes are bounded by a constant independent of $L$. }.
We define a `macroscopic operation' starting from $t=0$
as a sequence of unitary-operator-valued functions $(U_L(\bullet, 0))_{L\in\mathbb{N}}$, where $U_L(t, 0)$ is the time evolution generated by the Hamiltonian~\footnote{When an external field $\mathcal{F}(t)$ couples to an additive observable $B$
as $-g(\mathcal{F}(t)) B$ in $H_L(t)$, where $g(\mathcal{F})$ is a nonlinear function of $\mathcal{F}$, take $g(\mathcal{F}(t))$ as $f(t)$.}
\begin{align}
    H_L(t) = \initH - \sum_{\mu=1}^m f^\mu(t) B^\mu_L \quad (t>0).
    \label{eq:general_H}
\end{align}
Then, the macroscopic operation $U_L(t, 0)$ is determined by
\begin{align}
    i \frac{\partial}{\partial t} U_L(t, 0)
    = H_L(t) U_L(t, 0), \quad U_L(0,0) = 1,
\end{align}
and changes a quantum state $\rho_L$ as
\begin{align}
    \rho_L(0) = \rho_L
    \longmapsto \rho_L(t) = U_L(t, 0) \rho_L {U_L(t, 0)}^\dagger
    \quad (t>0).
\end{align}
\end{definition}
\noindent

As expected from thermodynamics, we can show that 
any representation of a macroscopic state remains a representation of some macroscopic state
under macroscopic operations, when the operation time is $L$-independent:
\begin{proposition}[State after a macroscopic operation also represents a macroscopic state] \label{proposition:JAIVTMXN_1}
Let $(\rho_L)_{L\in\mathbb{N}}$
represent 
a macroscopic state. Then, for any macroscopic operation $(U_L(\bullet, 0))_{L\in\mathbb{N}}$ and operation time $t^* \in \mathbb{R}_{>0}$ independent of $L$, $(\rho_L(t^*))_{L\in\mathbb{N}}$
also represents a macroscopic state, 
where $\rho_L(t^*) = U_L(t^*, 0) \rho_L {U_L(t^*, 0)}^\dagger$.
\end{proposition}
\begin{proof}
See Appendix~\ref{sec:JAIVTMXN}.
\end{proof}
\noindent
Moreover, if the initial state
represents 
a normal macroscopic state, the normality is also preserved under macroscopic operations:
\begin{proposition}[State after a macroscopic operation also represents a normal macroscopic state] \label{proposition:JAIVTMXN_2}
Let $(\rho_L)_{L\in\mathbb{N}}$
represent 
a normal macroscopic state. Then, for any macroscopic operation $(U_L(\bullet, 0))_{L\in\mathbb{N}}$ and operation time $t^* \in \mathbb{R}_{>0}$ independent of $L$, $(\rho_L(t^*))_{L\in\mathbb{N}}$
also represents a normal macroscopic state, 
where $\rho_L(t^*) = U_L(t^*, 0) \rho_L {U_L(t^*, 0)}^\dagger$.
\end{proposition}
\begin{proof}
See Appendix~\ref{sec:JAIVTMXN}.
\end{proof}

\subsection{Macroscopic equivalence after macroscopic operations}
We now discuss the macroscopic equivalence of states after a macroscopic operation.
Suppose that the same macroscopic operation is applied
to two different quantum states that are macroscopically equivalent. 
Since the resulting states obtained by applying the macroscopic operation
also represent macroscopic states 
according to Proposition~\ref{proposition:JAIVTMXN_1}, we can ask whether they are macroscopically equivalent.
To this question, we obtain the following theorem:
\begin{theorem}[Macroscopic equivalence after a macroscopic operation] \label{theorem:JAIVTMXN_3}
Let $(\rho_L)_{L\in\mathbb{N}}$ and $(\sigma_L)_{L\in\mathbb{N}}$
represent macroscopic states. Suppose that $(\rho_L)_{L\in\mathbb{N}}$ and $(\sigma_L)_{L\in\mathbb{N}}$ are macroscopically equivalent,
$(\rho_L)_{L\in\mathbb{N}} \maceq (\sigma_L)_{L\in\mathbb{N}}$.
Then, for any macroscopic operation $(U_L(\bullet, 0))_{L\in\mathbb{N}}$ and operation time $t^* \in \mathbb{R}_{>0}$ independent of $L$,
it holds that
\begin{align}
    (\rho_L(t^*))_{L\in\mathbb{N}} \maceq (\sigma_L(t^*))_{L\in\mathbb{N}},
    \label{eq:Preserve_MacroEquiv}
\end{align}
where $\rho_L(t^*) = U_L(t^*, 0) \rho_L {U_L(t^*, 0)}^\dagger$ and $\sigma_L(t^*) = U_L(t^*, 0) \sigma_L {U_L(t^*, 0)}^\dagger$.
\end{theorem}
\noindent
Therefore, if every additive observable initially takes  
the macroscopically same value
in two states $\rho_L$ and $\sigma_L$ [as \eqref{eq:MacroscopicEquiv}], 
then this macroscopic coincidence is kept for arbitrary operation time $t^*=O(L^0)$, while its value depends on $t^*$.
\begin{proof}[Proof outline]
Consider the time evolution of an additive observable $A_L$ by macroscopic operation $U_L(t,0)$ in the Heisenberg picture, $A_L^H(t)$. Using the Lieb-Robinson bound, for $t^*=O(L^0)$, $A_L^H(t^*)$ can be well approximated by a certain additive observable~\cite{Gong2020,Gong2020L,Gong2022}. Combining this with the definition of macroscopic equivalence, Definition~\ref{definition:MacroEquiv}, we can show that the expectation value of $A_L^H(t^*)$ in $\rho_L$ coincides with that in $\sigma_L$. Since this holds for any additive observable $A_L$, we have Eq.~\eqref{eq:Preserve_MacroEquiv}.
For the details, see Appendix~\ref{sec:JAIVTMXN}.
\end{proof}

Specifically, by taking the canonical Gibbs state $\rho_L^\mathrm{can}(\beta|\initH)$ as $\sigma_L$, we obtain the following corollary:
\begin{corollary}[\label{corollary:macroscopically-equivalence_equilibrium-state}Response of iMATE]
Let $(\rho_L)_{L\in\mathbb{N}}$
represent 
iMATE described by the Hamiltonian $\initH$ at the inverse temperature $\beta$.
Then, for any macroscopic operation $(U_L(\bullet, 0))_{L\in\mathbb{N}}$ and operation time $t^* \in \mathbb{R}_{>0}$ independent of $L$,
it holds that
\begin{align}
    (\rho_L(t^*))_{L\in\mathbb{N}}
    \maceq (U_L(t^*, 0) \rho_L^\mathrm{can}(\beta|\initH) {U_L(t^*, 0)}^\dagger)_{L\in\mathbb{N}},
\end{align}
where $\rho_L(t^*) = U_L(t^*, 0) \rho_L {U_L(t^*, 0)}^\dagger$.
\end{corollary}
\noindent
In other words, on a timescale independent of $L$, the response of additive observables in a system prepared in iMATE to macroscopic operations does not depend on the choice of the representation of iMATE.
This is consistent with the experimental fact that macroscopic responses of 
macroscopic systems in thermal equilibrium are insensitive to microscopic details of the 
states of the systems.

\subsection{Lower bound on the timescale of thermalization}

By using Corollary~\ref{corollary:macroscopically-equivalence_equilibrium-state}, we can readily obtain the following corollary, 
which states that 
the timescale of thermalization is bounded from below.
\begin{corollary}[\label{corollary:lower_bound_thermalization}Lower bound on the timescale of thermalization]
Let the Hamiltonian $H_L$ be an additive observable on the whole system
and $(\rho_L)_{L\in\mathbb{N}}$ be an arbitrary representation of a macroscopic state  that does not represent iMATE of the system described by $H_L$.
Now, consider the time evolution generated by $H_L$, i.e., $\rho_L(t):=e^{-i H_L t}\rho_L e^{i H_L t}$.
Then, for any time $t^*$ independent of $L$, the sequence $\bigl(\rho_L(t^*)\bigr)_{L\in\mathbb{N}}$ does not represent iMATE.
\end{corollary}
\noindent
This means that, for any initial nonequilibrium states, relaxation to thermal equilibrium in the sense of Definition~\ref{definition:MacroEqState} can occur only at the timescale of $\omega(L^0)$~\cite{LandauSymbols}.
This is also consistent with experiments on macroscopic systems in thermal equilibrium.
Note that this timescale is much longer than that derived by using typicality arguments~\cite{Goldstein2015Extremely,Reimann2016}. This difference arises because our Hamiltonian consists of local interactions, 
and we focus on \emph{all} additive observables, 
both of which respect locality,
in contrast to the typicality arguments.
We also note that the above result is distinct from the $\omega(L^0)$-bound on timescale of relaxation provided in Ref.~\cite{Hamazaki2022}, which considered locality of dynamics but basically focused on a restricted class of observables corresponding to, e.g., transport of particles.

\begin{proof}
We prove this corollary by contradiction.
Assume that $\bigl(\rho_L(t^*)\bigr)_{L\in\mathbb{N}}$ represents iMATE.
This means that there is some inverse temperature~$\beta$ such that $\bigl(\rho_L(t^*)\bigr)_{L\in\mathbb{N}}\maceq \bigl(\rho_L^{\mathrm{can}}(\beta|H_L)\bigr)_{L\in\mathbb{N}}$.
Because the inverse $U_L(t,0)=e^{iHt}$ of the original time evolution is a macroscopic operation, Corollary~\ref{corollary:macroscopically-equivalence_equilibrium-state} implies 
\begin{align}
    (\rho_L)_{L\in\mathbb{N}}&=\bigl(U_L(t^*,0)\rho_L(t^*)U_L^\dagger(t^*,0)\bigr)_{L\in\mathbb{N}}\notag\\
    &\maceq \bigl(U_L(t^*,0)\rho_L^{\mathrm{can}}(\beta|H_L)U_L^\dagger(t^*,0)\bigr)_{L\in\mathbb{N}}\notag\\
    &=\bigl(\rho_L^{\mathrm{can}}(\beta|H_L)\bigr)_{L\in\mathbb{N}},
\end{align}
which contradicts the premise that $(\rho_L)_{L\in\mathbb{N}}$ does not represent iMATE.
\end{proof}

\section{\label{sec:MacroPassivity}Macroscopic passivity}

\subsection{\label{sec:MacroPassivity_beta>=0}Macroscopic passivity at \texorpdfstring{$\beta\ge 0$}{β≥0}}

According to thermodynamics, one cannot extract work from thermal equilibrium states by any adiabatic operation in which the final values of control parameters coincide with the initial values, a form of the second law known as Planck's principle~\cite{Lieb1999}.
By combining Corollary~\ref{corollary:macroscopically-equivalence_equilibrium-state} and the passivity of the canonical Gibbs state~\cite{Pusz1978,Lenard1978}, we immediately obtain the following result, which is a quantum-mechanical expression of Planck's principle on the timescale independent of $L$
and is valid for any (possibly pure) quantum states representing iMATE. 
\begin{corollary}[Macroscopic passivity of iMATE] \label{corollary:macroscopic-passivity}
Let $(\rho_L)_{L\in\mathbb{N}}$
represent 
iMATE described by a translation-invariant Hamiltonian $\initH$ at a nonnegative inverse temperature $\beta \geq 0$.
Then, for any macroscopic operation $(U_L(\bullet, 0))_{L\in\mathbb{N}}$ and operation time $t^* \in \mathbb{R}_{>0}$ independent of $L$, it holds that
\begin{align}
    \lim_{L \to \infty} \mathrm{Tr} \left[ \rho_L(t^*) \initH / N \right]
    \geq \lim_{L \to \infty} \mathrm{Tr} \left[ \rho_L(0) \initH / N \right],
    \label{eq:macroscopic-passivity}
\end{align}
where $\rho_L(t^*) = U_L(t^*, 0) \rho_L {U_L(t^*, 0)}^\dagger$.
\end{corollary}
\noindent
Note that the energy change through a macroscopic operation corresponds to the total amount of work done on the system by the external fields during the operation. 
Therefore, Corollary~\ref{corollary:macroscopic-passivity} states that it is impossible to extract an extensive work from iMATE through a macroscopic operation with a finite operation time.
This corresponds to Planck's principle in thermodynamics (see Appendix~\ref{sec:PlanckPrinciple} on the differences from Planck's principle given in Ref.~\cite{Lieb1999}).
In this way, our definitions of thermal equilibrium and macroscopic operations are consistent with thermodynamics.

\begin{proof}[Proof of Corollary~\ref{corollary:macroscopic-passivity}.]
Since $(\rho_L)_{L\in\mathbb{N}}$
represents 
iMATE at the inverse temperature $\beta$, by applying Corollary~\ref{corollary:macroscopically-equivalence_equilibrium-state}, we get $(\rho_L(t^*))_{L\in\mathbb{N}} \maceq (U_L(t^*, 0) \rho_L^\mathrm{can}(\beta|\initH) {U_L(t^*, 0)}^\dagger)_{L\in\mathbb{N}}$. Therefore, we have
\begin{align}
    &\lim_{L\to\infty} \mathrm {Tr} \left[ \rho_L(t^*) \initH / N \right] \nonumber\\
    &= \lim_{L\to\infty} \mathrm{Tr} \left[ U_L(t^*, 0) \rho_L^\mathrm{can}(\beta|\initH) {U_L(t^*, 0)}^\dagger \initH / N \right].
\end{align}
Combining this with the passivity of the canonical Gibbs state with nonnegative inverse temperature~\cite{Pusz1978,Lenard1978}, i.e.,
\begin{align}
    \mathrm{Tr} \left[ V \rho_L^\mathrm{can}(\beta|\initH) V^\dagger \initH / N \right]
    \geq \mathrm{Tr} \left[ \rho_L^\mathrm{can}(\beta|\initH) \initH / N \right]
\end{align}
for an arbitrary unitary operator $V$,
we have
\begin{align}
    \lim_{L\to\infty} \mathrm {Tr} \left[ \rho_L(t^*) \initH / N \right]
    \geq \lim_{L\to\infty} \mathrm{Tr} \left[ \rho_L^\mathrm{can}(\beta|\initH) \initH / N \right].
\end{align}
Again, using the assumption that $(\rho_L)_{L\in\mathbb{N}}$
represents 
iMATE, we obtain
\begin{align}
    \lim_{L\to\infty} \mathrm {Tr} \left[ \rho_L(t^*) \initH / N \right]
    \geq \lim_{L\to\infty} \mathrm{Tr} \left[ \rho_L(0) \initH / N \right].
\end{align}
\end{proof}
Some readers might be concerned that allowing only finite-time operations could leave the final state out of equilibrium, potentially undermining the physical relevance of our results to thermodynamics. Originally, Planck's principle, which motivates the concept of passivity, states that one can extract no work through any adiabatic operation in which the final values of control parameters coincide with the initial values (See Appendix~\ref{sec:PlanckPrinciple} for details). In our quantum-mechanical setup, this corresponds to the setting that the final Hamiltonian is the same as the initial Hamiltonian, $\initH$. 
Thus, in this setting, the time evolution generated by the final Hamiltonian conserves $\initH$. Consequently, even if we add a relaxation process following macroscopic operations and wait for an arbitrarily long time
to reach thermal equilibrium, the expectation value of $\initH$ remains invariant, and Eq.~\eqref{eq:macroscopic-passivity} holds exactly the same.

\subsection{Macroscopic passivity at \texorpdfstring{$\beta=0$}{β=0} and difference from \texorpdfstring{$\beta>0$}{β>0}}

Interestingly, 
in the special case 
$\beta=0$, Corollary~\ref{corollary:macroscopic-passivity} additionally prohibits an extensive \emph{increase} of energy. This is because if $(\rho_L)_{L\in\mathbb{N}}$ represents iMATE described by $H_L$ at $\beta=0$, it also represents iMATE described by $-H_L$, which implies an inequality opposite to Eq.~\eqref{eq:macroscopic-passivity}. In other words, not only can no extensive work be extracted from such a state, but also no extensive work can be done on such a state through any macroscopic operation with a finite operation time. In addition, since such a $(\rho_L)_{L\in\mathbb{N}}$ represents iMATE described by $A_L$ at $\beta=0$ for any additive observable $A_L$, the same argument implies that the expectation value of $A_L$ cannot be increased nor decreased extensively through any finite time macroscopic operation. Thus, we obtain the following corollary:
\begin{corollary}[\label{corollary:MacroPassivity_beta0}Stability of iMATE at $\beta=0$ to any macroscopic operation]
Let $(\rho_L)_{L\in\mathbb{N}}$ represent iMATE (described by some translation-invariant Hamiltonian) at zero inverse temperature $\beta =0$.
Then, for any macroscopic operation $(U_L(\bullet, 0))_{L\in\mathbb{N}}$ and operation time $t^* \in \mathbb{R}_{>0}$ independent of $L$, 
the sequence $\bigl(\rho_L(t^*)\bigr)_{L\in\mathbb{N}}$ given by $\rho_L(t^*) = U_L(t^*, 0) \rho_L {U_L(t^*, 0)}^\dagger$ remains to represent iMATE at $\beta=0$, $\bigl(\rho_L(t^*)\bigr)_{L\in\mathbb{N}}\maceq (\rho_L)_{L\in\mathbb{N}}$.
In other words, it holds that
\begin{align}
    \lim_{L \to \infty} \mathrm{Tr} \left[ \rho_L(t^*) A_L / N \right]
    = \lim_{L \to \infty} \mathrm{Tr} \left[ \rho_L(0) A_L / N \right]\notag\\
    \text{for all }(A_L)_{L\in\mathbb{N}}\in\mathcal{A}.
    \label{eq:macroscopic-passivity_beta0}
\end{align}

\end{corollary}
\begin{proof}[Proof outline]
This corollary directly follows from Corollary~\ref{corollary:macroscopically-equivalence_equilibrium-state} and the stability of the maximally mixed state $I_{\Lambda_L}/D^N$ to unitary operations, $U_L(t^*, 0) (I_{\Lambda_L}/D^N) {U_L(t^*, 0)}^\dagger=(I_{\Lambda_L}/D^N)$.
\end{proof}
By contrast, at positive inverse temperature $\beta>0$, an extensive work can be done on iMATE at $\beta$ even within a finite operation time. In other words, the strict inequality in \eqref{eq:macroscopic-passivity} naturally holds for typical macroscopic operations with $t^*=O(L^0)$.
An example is the case where 
$H_L=M^z$, $\beta>0$, and $U_L(t,0)=e^{-\pi iM^x t}$ for $t\notin\mathbb{Z}$, where $M^{x,y,z}=\sum_{\bm{r}\in\Lambda_L}\sigma^{x,y,z}_{\bm{r}}$.

\subsection{Straightforward extension of macroscopic passivity is impossible}

In this subsection, we explain that we cannot straightforwardly extend macroscopic passivity to a longer operation time, to a broader class of initial thermal equilibrium states, or to a broader class of operations.

First, we briefly comment on the timescale of operation time in Corollaries~\ref{corollary:macroscopic-passivity} and \ref{corollary:MacroPassivity_beta0}.
We will show in Sec.~\ref{sec:Counterexamples_Timescale} that the constraint on the operation time, $t^*=O(L^0)$, is optimal by constructing a counterexample of the macroscopic passivity for any longer operation time $t^*=\omega(L^0)$. In other words, the timescale of $O(L^0)$ is the longest one in that macroscopic passivity can be proven under the setting considered in this paper.

Next, we comment on the extension to a broader class of initial thermal equilibrium states. We emphasize that
it is crucial in Corollaries~\ref{corollary:macroscopic-passivity} and \ref{corollary:MacroPassivity_beta0} to focus on 
\emph{all} additive observables for characterizing thermal equilibrium.
If we focused on a \emph{finite} number of additive observables instead, we would encounter an inconsistency with the second law,
as the following example clearly shows:
\begin{example}[\label{example:ProblemFiniteObs}Contradiction to Planck's principle when thermal equilibrium were characterized by only a few additive observables, informal version of Proposition~\ref{proposition:ProblemFiniteObs}]
Let us consider the XY model on the square lattice ($d=2$),
\begin{align}
    H_L&=H^{XY}_{L}-h^xM_L^x-h^yM_L^y,
    \label{eq:Counterex_XYmodel}
\end{align}
where
\begin{align}
    H^{XY}_{L}&=-\sum_{j,k=1}^{L}\Bigl(J^x(\sigma_{j,k}^{x}\sigma_{j+1,k}^{x}+\sigma_{j,k}^{x}\sigma_{j,k+1}^{x})\notag\\
    &\hspace{40pt}+J^y(\sigma_{j,k}^{y}\sigma_{j+1,k}^{y}+\sigma_{j,k}^{y}\sigma_{j,k+1}^{y})\Bigr),\\
    M_L^{x,y}&=\sum_{j,k=1}^{L}\sigma_{j,k}^{x,y}.
\end{align}
We assume $J^{x},J^{y}> 0$ and $J^x\neq J^y$, while $h^x$ and $h^y$ can be taken arbitrarily including zero.
In this case, local conserved quantities, i.e., additive observables commuting with $H_L$, are restricted only to $H_L$ itself~\cite{ShiraishiTasaki2024} and symmetry breaking is absent at sufficiently high temperature.
Then, thermal equilibrium state at such temperature can be described by the canonical Gibbs state $\rho_{L}^{\mathrm{can}}(\beta|H_L)$, which is expected to satisfy Assumption~\ref{assumption:GibbsState} and \ref{assumption:GibbsState_Normal}.

Suppose that we only focus on three additive observables~$\tilde{\mathcal{A}}=\{M_L^x, M_L^y, H^{XY}_{L}\}$, which seem to be a natural choice because they appear in Eq.~\eqref{eq:Counterex_XYmodel}. 
We take $h^x$ as a control parameter and consider a macroscopic operation of operation time $O(L^0)$.
Then, there is a state $\rho_L$ that is indistinguishable from the canonical Gibbs state $\rho_L^{\mathrm{can}}(\beta|H_L)$ with respect to $\tilde{\mathcal{A}}$ and represents the ordinary MATE~\cite{Tasaki2016} but violates Eq.~\eqref{eq:macroscopic-passivity}.
\end{example}

In other words, the choice $\tilde{\mathcal{A}}$ is inappropriate in the sense that it cannot distinguish a macroscopically-nonpassive state $\rho_L$ from $\rho_{L}^{\mathrm{can}}(\beta|H_L)$.
This example suggests the difficulty of characterizing thermal equilibrium using only a finite number of additive observables,
as anticipated from the physical motivations discussed in
Sec.~\ref{sec:MacroEquilibrium_iMATE}.
Instead, by focusing on all additive observables through iMATE, we 
have successfully proved 
the  
consistency with the second law.

Finally, we comment on the extension to a broader class of operations.
In Corollaries~\ref{corollary:macroscopic-passivity} and \ref{corollary:MacroPassivity_beta0}, the restriction of the class of operations to macroscopic operations is 
crucial.
This can be seen from the following example, which shows a contradiction to Planck's principle when 
allowed to apply
local controls --- a class of operations broader than macroscopic operations
--- to iMATE.
\begin{example}[\label{example:Problem_CombiningiMATEandLocalControl}Contradiction to Planck's principle when combining iMATE and local control of $O(L^0)$ time]
Local controls are defined as unitary time evolution in which parameters conjugate to local observables are controlled 
(see Definition~\ref{definition:LocalControl} in Appendix~\ref{sec:Local_control}
for the precise definition),
and are broader operations than macroscopic operations.
Take the initial Hamiltonian as
\begin{align}
    H_L&=M_L^z=\sum_{j=1}^{L}\sigma_{j}^z,
    \label{eq:Problem_CombiningiMATEandLocalControl_initH}
\end{align}
and the initial state $\rho_L$ as that given in Example~\ref{example:NonMITE_MacroEqState}, where the sequence $(\rho_L)_{L\in\mathbb{N}}$ represents iMATE at $\beta=0$.
There exists
a local control $\bigl(U_L(t^*,0)\bigr)_{L\in\mathbb{N}}$ of operation time $t^*=O(L^0)$ such that all up spins in $\rho_L$ are flipped to downward.
As a result, Planck's principle is violated as
\begin{align}
    \lim_{L \to \infty} \mathrm{Tr} \left[ \rho_L(t^*) \initH / N \right]
    < \lim_{L \to \infty} \mathrm{Tr} \left[ \rho_L \initH / N \right],
    \label{eq:Problem_CombiningiMATEandLocalControl_nonPassive}
\end{align}
where $\rho_L(t^*) = U_L(t^*, 0) \rho_L {U_L(t^*, 0)}^\dagger$.
\end{example}

\begin{proof}[Proof of Eq.~\eqref{eq:Problem_CombiningiMATEandLocalControl_nonPassive}]
See Appendix~\ref{sec:Proof_Problem_CombiningiMATEandLocalControl}.
\end{proof}

This means that the setup of ``iMATE $+$ local control'' violates Planck's principle even for operation time independent of $L$.
By contrast, as mentioned in Table~\ref{tbl:Comparison_Operation}, 
the setups of ``iMATE $+$ macroscopic operation'' and ``MITE $+$ local control'' for such an operation timescale are consistent with Planck's principle.
These facts suggest a trade-off between the notion of thermal equilibrium and the class of operations consistent with the second law of thermodynamics.
For more details, see also Appendix~\ref{sec:Local_control}.

\section{\label{sec:reduced_rho}Characterization of macroscopic states by reduced density matrices}

In this section, we introduce a reduced density matrix that 
completely characterizes macroscopically equivalent states.
We then use it to clearly define macroscopic uniformity.
Furthermore, in the next section, we will use it to define an entropy 
that is consistent with thermodynamics.
For readers' convenience, we summarize in Table \ref{tbl:entropy} the definitions and symbols used in this paper.
\begin{table}[h]
\caption{A quick list of symbols used in Sec.~\ref{sec:reduced_rho} and the following sections.}
\label{tbl:entropy}
\centering
\begin{ruledtabular}
\begin{tabular}{lll}
    Symbol & Definition & Meaning\\\hline
    $\Lambda_L$ & Sec.~\ref{sec:Setup_Lattice} & Total lattice with side length $L$\\
    $\mathcal{S}_L^{(k)}$ & Def.~\ref{definition:ProperMacroSubsystem} & $k$-th primitive macroscopic subsystem\\
    $\Cell$ & Def.~\ref{definition:LocalObs} & Hypercube of side length $\ell$ at $\bm{0}$\\
    $\Cell[\bm{r}]$ & Def.~\ref{definition:LocalObs} & Translation of $\Cell$ by $\bm{r}$\\
    $\rho_L$ & - & Density matrix on $\Lambda_L$\\
    $\rho_{L|\ell}^{\boldsymbol{r}}$ & Eq.~\eqref{eq:ReducedState_rho^r} & Reduced density matrix of $\rho_L$ on $\Cell[\bm{r}]$ \\
    $\rho_{L|\ell}^{(k)}$ & Eq.~\eqref{eq:DEF_rho_L^red} & Spatial average of $\rho_{L|\ell}^{\boldsymbol{r}}$ around $\mathcal{S}_L^{(k)}$\\
    $s^{\mathrm{mac}}_{\ell}[(\rho_{L})_{L\in\mathbb{N}}]$ & Def.~\ref{definition:MacroEntropy} &Quantum macroscopic entropy density\\
    $\tilde{s}_{\ell}[(\rho_{L})_{L\in\mathbb{N}}]$ & Eq.~\eqref{eq:tilde_s} &An upper bound of $s^{\mathrm{mac}}_{\ell}[(\rho_{L})_{L\in\mathbb{N}}]$\\
    $s^{\mathrm{TD}}(\beta|H_L)$ & Eq.~\eqref{eq:DEF_s^TD} &Thermodynamic entropy
\end{tabular}
\end{ruledtabular}
\end{table}

\subsection{\label{sec:Entropy_CGLS}\texorpdfstring{$\ell$-local}{l-local} density matrices and their spatial average}

Recall that $\Cell$ is a $d$-dimensional hypercube with an $L$-independent side length $\ell$ centered at site $\boldsymbol{0}\in\Lambda_{L}$.
For an arbitrary quantum state $\rho_{L}$ on $\Lambda_{L}$, 
consider its reduced density matrix on $\Cell[\bm{r}]=\Cell+\boldsymbol{r}$, 
and move it to $\Cell$ (so that it has the same support for any $\boldsymbol{r}$).
The density matrix obtained is 
called an \emph{$\ell$-local density matrix} 
around $\boldsymbol{r}$
and denoted as
\begin{align}
    \rho_{L|\ell}^{\boldsymbol{r}}:=\mathrm{Tr}_{\Lambda_{L}\setminus \Cell}\bigl[T_{\boldsymbol{r}}^{\dagger}\rho_{L} T_{\boldsymbol{r}}\bigr].
    \label{eq:ReducedState_rho^r}
\end{align}

To extract macroscopic properties of the system, we introduce 
the spatial average of $\rho_{L|\ell}^{\boldsymbol{r}}$ on
a primitive macroscopic subsystem $\mathcal{S}^{(k)}$, 
\begin{align}
    \rho_{L|\ell}^{(k)}:=
    \frac{1}{|\mathcal{S}^{(k)}|}\sum_{\boldsymbol{r}\in\mathcal{S}^{(k)}}
    \rho_{L|\ell}^{\boldsymbol{r}},
    \label{eq:DEF_rho_L^red}
\end{align}
and its $L\to\infty$ limit:
\begin{definition}[Spatial average of $\ell$-local density matrices]
\label{def:sp.av.ell.dens.mat}
Let $(\rho_{L})_{L\in\mathbb{N}}$ represent a macroscopic state.
For any positive integer $\ell\in\mathbb{N}$, 
and for a primitive macroscopic subsystem $\mathcal{S}^{(k)}$, 
we define
\begin{align}
    \rho_{\infty|\ell}^{(k)}:=
    \lim_{L\to\infty}\rho_{L|\ell}^{(k)}.
    \label{eq:DEF_rho_infty^red}
\end{align}
We call it the 
spatial average of $\ell$-local density matrices around $\mathcal{S}^{(k)}$.
\end{definition}
\noindent
We can show that the limit~\eqref{eq:DEF_rho_infty^red} exists as follows.
\begin{proof}[Proof of well-definedness of Eq.~\eqref{eq:DEF_rho_infty^red}.]
For any observable $\aLocal$ on $\Cell$, we have 
\begin{align}
    \mathrm{Tr}\bigl[\rho_{L|\ell}^{(k)} \aLocal\bigr]
    =\mathrm{Tr}\Bigl[\rho_{L}\frac{A_{\mathcal{S}^{(k)}}(\aLocal)}{|\mathcal{S}^{(k)}|}\Bigr]+O(1/L),
    \label{eq:<a>^red=<A_Sigma>/Sigma}
\end{align}
where $A_{\mathcal{S}^{(k)}}(\aLocal)$ is the additive observable on $\mathcal{S}^{(k)}$ obtained from 
$\aLocal$
[which is local since $\ell=O(L^0)$]
and the last term of $O(1/L)$ corresponds to the contributions from the boundary of $\mathcal{S}^{(k)}$. 
Since $(\rho_{L})_{L\in\mathbb{N}}$ is a macroscopic state, Eq.~\eqref{eq:MacroState} guarantees that there exists the $L\to\infty$ limit of the first term of the RHS of the above equation. 
Thus there exists $\lim_{L\to\infty}\mathrm{Tr}\bigl[\rho_{L|\ell}^{(k)} \aLocal\bigr]$.
By taking $\aLocal$ as a member of a basis of operators on $\Cell$, the above statement implies the existence of the $L\to\infty$ limit of $\rho_{L|\ell}^{(k)}$.
\end{proof}

Importantly, the spatial average 
$(\rho_{\infty|\ell}^{(k)})_{\ell\in\mathbb{N}, k\in\{1,...,K^d\}}$
completely characterizes macroscopically equivalent states as follows:
\begin{proposition}[\label{proposition:Equiv_rho^red}Characterization of macroscopically equivalent states]

Any two representations of macroscopic states $(\rho_{L})_{L\in\mathbb{N}}$ and $(\sigma_{L})_{L\in\mathbb{N}}$ are macroscopically equivalent, $(\rho_{L})_{L\in\mathbb{N}}\maceq (\sigma_{L})_{L\in\mathbb{N}}$, if and only if
\begin{align}
    \rho_{\infty|\ell}^{(k)}=\sigma_{\infty|\ell}^{(k)}
    \label{eq:Equiv_rho^red}
\end{align}
holds for all $k\in\{1,...,K^d\}$ and for all $\ell\in\mathbb{N}$.
\end{proposition}
\noindent
This means that all macroscopically equivalent representations of a macroscopic state are characterized only by 
a set of spatial averages of $\ell$-local density matrices,
$\rho_{\infty|\ell}^{(k)}$ for $k\in\{1,...,K^d\}$ and $\ell\in\mathbb{N}$.
\begin{proof}[Proof outline]
From Eq.~\eqref{eq:<a>^red=<A_Sigma>/Sigma}, it holds that 
\begin{align}
    \mathrm{Tr}\bigl[\rho_{\infty|\ell}^{(k)}\aLocal\bigr]
    =\lim_{L\to\infty}\mathrm{Tr}\Bigl[\rho_{L}\frac{A_{\mathcal{S}^{(k)}}(\aLocal)}{|\mathcal{S}^{(k)}|}\Bigr]
\end{align}
for all primitive macroscopic subsystems $\mathcal{S}^{(k)}$ and an $\ell$-local observable $\aLocal$ on $\Cell$.
This means that the equality between $\mathrm{Tr}\bigl[\rho_{\infty|\ell}^{(k)}\aLocal\bigr]$ and $\mathrm{Tr}\bigl[\sigma_{\infty|\ell}^{(k)}\aLocal\bigr]$ is equivalent to the equality of the expectation value of $A_{\mathcal{S}^{(k)}}(\aLocal)/|\mathcal{S}^{(k)}|$ in $\rho_L$ and $\sigma_L$. 
Combining this with Eq.~\eqref{eq:Additive_S_approximation}, we can show that Eq.~\eqref{eq:Equiv_rho^red} is equivalent to macroscopic equivalence~\eqref{eq:MacroscopicEquiv}, $(\rho_{L})_{L\in\mathbb{N}}\maceq (\sigma_{L})_{L\in\mathbb{N}}$.
For the details, see Appendix~\ref{sec:Proof_Equiv_rho^k}.
\end{proof}

\subsection{\label{sec:Entropy_MacroUniform}Macroscopic uniformity}

The term ``macroscopically uniform" is often used ambiguously in the literature. We here define it clearly by using the spatial average of $\ell$-local density matrices.

In the above discussion, Eq.~\eqref{eq:<a>^red=<A_Sigma>/Sigma} tells us that the expectation value 
with respect to 
$\rho_{\infty|\ell}^{(k)}$ (the spatial average of $\ell$-local density matrices around $\mathcal{S}^{(k)}$)
corresponds to the coarse graining 
in a primitive macroscopic subsystem $\mathcal{S}^{(k)}$. 
From the viewpoint of macroscopic theory, when such 
coarse-grained values do 
not depend on the choice of $\mathcal{S}^{(k)}$, we should regard the state as macroscopically uniform. Therefore, using $\rho_{\infty|\ell}^{(k)}$, we can formulate \emph{macroscopically uniform states} as follows:
\begin{definition}[Macroscopically uniform state]\label{definition:MacroUniform}
A representation of a macroscopic state $(\rho_{L})_{L\in\mathbb{N}}$ is said to represent a macroscopically uniform state when 
\begin{align}
    \rho_{\infty|\ell}^{(k)}=\rho_{\infty|\ell}^{(k^\prime)}
    \label{eq:MacroUniform}
\end{align}
holds for all $k,k^\prime\in\{1,...,K^d\}$ and for all $\ell\in\mathbb{N}$.
\end{definition}
\noindent
In other words, if $(\rho_L)_{L\in\mathbb{N}}$ represents a macroscopically uniform state, 
the density of the additive observable obtained from an $\ell$-local observable $\aLocal$ on any macroscopic subsystem $(\mathcal{S}_L)_{L\in\mathbb{N}}$ coincides with that on the total system,
\begin{align}
    \lim_{L\to\infty}\mathrm{Tr}\Bigl[\rho_{L}\frac{A_{\mathcal{S}_L}(\aLocal)}{|\mathcal{S}_L|}\Bigr]
    =\lim_{L\to\infty}\mathrm{Tr}\Bigl[\rho_{L}\frac{A_{\Lambda_L}(\aLocal)}{N}\Bigr],
\end{align}
for any $\aLocal$ and $\ell\in\mathbb{N}$.

Furthermore, from Proposition~\ref{proposition:Equiv_rho^red}, we can say that whether a macroscopic state is macroscopically uniform is independent of the choice of its representations in the following sense:
\begin{corollary}[Macroscopic uniformity depends only on a macroscopic state but not on its representation]\label{cor:mac.un.dep.only.on.mcrostate}
For any two macroscopically equivalent representations of a macroscopic state,
$(\rho_{L})_{L\in\mathbb{N}}\maceq (\sigma_{L})_{L\in\mathbb{N}}$,
one of them $(\rho_{L})_{L\in\mathbb{N}}$ represents a macroscopically uniform state
if and only if 
the other $(\sigma_{L})_{L\in\mathbb{N}}$ represents a macroscopically uniform state.
\end{corollary}
\noindent
In particular, since the canonical Gibbs state $\bigl(\rho_{L}^{\mathrm{can}}(\beta|H_L)\bigr)_{L\in\mathbb{N}}$ of a system described by 
a translation invariant Hamiltonian $H_L$ is translation invariant [in the sense of Eq.~\eqref{eq:rho_TranslationInv} in Sec.~\ref{sec:Entropy_LocalUniform}],
\emph{any representation of iMATE, Definition~\ref{definition:MacroEqState}, represents a macroscopically uniform state}:
\begin{corollary}[\label{corollary:MacroscopicallyUniform_iMATE}Macroscopic uniformity of iMATE]
Any representation of an iMATE represents a macroscopically uniform state.
\end{corollary}
\noindent
This is consistent with the fact that we exclude first-order phase transition points.
We also note that iMATE needs not to be locally uniform (defined below), unlike MITE.
See Table~\ref{tbl:Comparison_Equilibrium} for a quick comparison.

\subsection{\label{sec:Entropy_LocalUniform}Local uniformity}

We also introduce a sense of uniformity stronger than the macroscopic uniformity.
We call it \emph{local uniformity}:
\begin{definition}[\label{definition:LocallyUniform}Locally uniform]
A sequence of quantum states $(\rho_L)_{L\in\mathbb{N}}$ is said to represent a locally uniform macroscopic state if it satisfies the following two conditions:
(i) For any positive integer $\ell\in\mathbb{N}$, there exists the thermodynamic limit of the $\ell$-local density matrix around $\boldsymbol{r}=\boldsymbol{0}$,
\begin{align}
    \lim_{L\to\infty}\rho_{L|\ell}^{\boldsymbol{r}=\boldsymbol{0}}.
    \label{eq:ExistLocalState}
\end{align}
(ii) For any positive integer $\ell\in\mathbb{N}$, the $\ell$-local density matrix around any site $\boldsymbol{r}$ becomes the same in the thermodynamic limit,
\begin{align}
    \limsup_{L\to\infty}\max_{\boldsymbol{r},\boldsymbol{r}^\prime\in\Lambda_L}\|\rho_{L|\ell}^{\boldsymbol{r}}-\rho_{L|\ell}^{\boldsymbol{r}^\prime}\|_{\mathrm{tr}}=0\quad \text{ for all }\ell\in\mathbb{N},
    \label{eq:LocallyUniform}
\end{align}
where $\|\bullet\|_{\mathrm{tr}}$ is the trace norm.
\end{definition}
\noindent 
Importantly, the class of locally uniform macroscopic states is narrower than the class of macroscopically uniform states.
In other words, if a sequence of quantum states $(\rho_L)_{L\in\mathbb{N}}$ represents a locally uniform macroscopic state, then it represents a macroscopically uniform state. 

It is obvious that if a sequence of quantum states $(\rho_L)_{L\in\mathbb{N}}$ 
representing a macroscopic state is translation invariant in a microscopic scale,
i.e., if
\begin{align}
    T_{\boldsymbol{r}}\rho_{L}T_{\boldsymbol{r}}^{\dagger}=\rho_{L}\quad \text{ for all }\boldsymbol{r}\in\Lambda_{L},
    \label{eq:rho_TranslationInv}
\end{align}
then it represents a locally uniform macroscopic state.
However, the converse is not necessarily true, 
as exemplified by 
typical state vectors on an energy shell~\cite{Goldstein2006,Popescu2006,Sugita2006} and thermal pure quantum (TPQ) states~\cite{Sugiura2012,Sugiura2013}.
Note that 
locally uniform iMATE is equivalent to MITE in  Definition~\ref{definition:MITE}.

\section{\label{sec:Entropy_s^mac}Quantum macroscopic entropy density}

In this section, 
we introduce an entropy density 
that takes the same value for any macroscopically equivalent states.
In particular, 
we will show in the next section
that it agrees with the thermodynamic entropy density
for any macroscopic states representing iMATE.
Given these significant features, 
we call it the \emph{quantum macroscopic entropy density}.

\subsection{Definition and basic property}

\begin{definition}[\label{definition:MacroEntropy}Quantum macroscopic entropy density]
For any representation of a macroscopic state $(\rho_{L})_{L\in\mathbb{N}}$, the quantum macroscopic entropy density is defined by
\begin{align}
    s^{\mathrm{mac}}_{\ell}[(\rho_{L})_{L\in\mathbb{N}}]
    &:=\frac{1}{K^d}\sum_{k=1}^{K^d}
    \frac{S_{\mathrm{vN}}[\rho_{\infty|\ell}^{(k)}]}{\ell^d}
    \label{eq:DEF_s^red_ell}\\
    &=\lim_{L\to\infty}\frac{1}{K^d}\sum_{k=1}^{K^d}
    \frac{S_{\mathrm{vN}}[\rho_{L|\ell}^{(k)}]}{\ell^d}
\end{align}
for any $\ell\in\mathbb{N}$.

\end{definition}
\noindent
In Eq.~\eqref{eq:DEF_s^red_ell}, 
the term 
$S_{\mathrm{vN}}[\rho_{\infty|\ell}^{(k)}]/\ell^d$ 
is the average of the 
von Neumann entropy density of the average of the $\ell$-local density matrices
over the primitive macroscopic subsystem~$\mathcal{S}^{(k)}$, 
and $s^{\mathrm{mac}}_{\ell}[(\rho_{L})_{L\in\mathbb{N}}]$ 
is its average
over the whole system.
Here, recall that the whole system contains $K^d$ primitive macroscopic subsystems with almost equal sizes.

Since $s^{\mathrm{mac}}_{\ell}[\bullet]$ 
depends only on the spatial average of 
$\rho_{\infty|\ell}^{(k)}$,
Proposition~\ref{proposition:Equiv_rho^red} immediately gives the following:
\begin{corollary}[Equivalent macroscopic states have the same value of $s^{\mathrm{mac}}_{\ell}$]\label{corollary:Equiv_s^red}
For two representations of macroscopic states $(\rho_{L})_{L\in\mathbb{N}}$ and $(\sigma_{L})_{L\in\mathbb{N}}$, it holds that 
\begin{align}
    &(\rho_{L})_{L\in\mathbb{N}}\maceq (\sigma_{L})_{L\in\mathbb{N}}
    \notag\\
    &\quad\Rightarrow\quad
    s^{\mathrm{mac}}_{\ell}[(\rho_{L})_{L\in\mathbb{N}}]=s^{\mathrm{mac}}_{\ell}[(\sigma_{L})_{L\in\mathbb{N}}]\quad \text{for all }\ell\in\mathbb{N}.
    \label{eq:Equiv_s^red}
\end{align}
\end{corollary}
\noindent
That is, $s^{\mathrm{mac}}_{\ell}[\bullet]$
takes a unique value for all macroscopically equivalent states.
Such a strong universality has been obtained by taking the spatial average and the limit $L\to\infty$ in the definition of $\rho_{\infty|\ell}^{(k)}$.

Since $s^{\mathrm{mac}}_{\ell}[\bullet]$ depends only on a macroscopic state,
it is independent of
other information, such as the Hamiltonian of the system.
Therefore, 
when two systems described by different Hamiltonians $H_{1}$ and $H_{2}$ are in the same quantum state $\rho_{L}$ (which does not represent thermal equilibrium in at least one of the systems), the quantum macroscopic entropy density for these two cases coincides. By contrast, other forms of entropy often depend on the details of the Hamiltonian. 
For instance, the Boltzmann entropy at energy coinciding with the energy expectation value in $\rho_{L}$ will generally differ between these two systems described by $H_{1}$ and $H_{2}$.

\subsection{\label{sec:Entropy_s^mac_Uniform}Quantum macroscopic entropy density of macroscopically uniform states}

In this subsection, we explain how the quantum macroscopic entropy density can be simplified when the macroscopic state of interest is macroscopically 
uniform in the sense explained in subsection~\ref{sec:Entropy_MacroUniform}. 

First, we give a useful upper bound for 
$s^{\mathrm{mac}}_{\ell}[\bullet]$.
\begin{corollary}[Upper bound for the quantum macroscopic entropy density]
For any representation of a macroscopic state $(\rho_{L})_{L\in\mathbb{N}}$,
the quantum macroscopic entropy density satisfies
\begin{align}
    &s^{\mathrm{mac}}_{\ell}[(\rho_{L})_{L\in\mathbb{N}}]\le \tilde{s}_{\ell}[(\rho_{L})_{L\in\mathbb{N}}]
    \label{eq:s^red<=tilde_s}
\end{align}
where 
\begin{align}
    &\tilde{s}_{\ell}[(\rho_{L})_{L\in\mathbb{N}}]:=
    \frac{S_{\mathrm{vN}}[\rho_{\infty|\ell}^{\mathrm{ave}}]}{\ell^d}
    =\lim_{L\to\infty}
    \frac{S_{\mathrm{vN}}[\rho_{L|\ell}^{\mathrm{ave}}]}{\ell^d},
    \label{eq:tilde_s}
\end{align}
and
$\rho_{\infty|\ell}^{\mathrm{ave}}$ is the density matrix on $\Cell$ defined by 
the $L\to\infty$ of 
\begin{align}
    \rho_{L|\ell}^{\mathrm{ave}}
    &:=\frac{1}{N}\sum_{\boldsymbol{r}\in\Lambda_{L}}\mathrm{Tr}_{\Lambda_{L}\setminus \Cell}\bigl[T_{\boldsymbol{r}}^{\dagger}\rho_{L} T_{\boldsymbol{r}}\bigr]
    \label{eq:DEF_rho_infty^tot}\\
    &=\frac{1}{K^d}\sum_{k=1}^{K^d}
    \rho_{L|\ell}^{(k)}.
    \label{eq:DEF_rho_infty^tot_1}
\end{align}
\end{corollary}
\begin{proof}
First, Eq.~\eqref{eq:DEF_rho_infty^tot_1} follows from the definition. 
The inequality~\eqref{eq:s^red<=tilde_s} follows from Eq.~\eqref{eq:DEF_rho_infty^tot_1} and the concavity of the von Neumann entropy~\cite{NielsenChuang2010}. 
\end{proof}
\noindent
This upper bound $\tilde{s}_{\ell}[(\rho_{L})_{L\in\mathbb{N}}]$ will be easier to calculate than the quantum macroscopic entropy density because $\tilde{s}_{\ell}[(\rho_{L})_{L\in\mathbb{N}}]$ does not depend on the definition of primitive macroscopic subsystems $\mathcal{S}^{(k)}$.
Roughly speaking, 
$\tilde{s}_{\ell}$ is $S_{\mathrm{vN}}/\ell^d$ of the 
average of $\rho_{L|\ell}^{(k)}$'s over primitive macroscopic subsystems,
whereas 
$s^{\mathrm{mac}}_{\ell}$
is the average of $S_{\mathrm{vN}}[\rho_{L|\ell}^{(k)}]/\ell^d$'s 
over primitive macroscopic subsystems.

Furthermore, the following corollary readily follows from Eqs.~\eqref{eq:MacroUniform} and \eqref{eq:DEF_rho_infty^tot_1} and Corollary~\ref{corollary:MacroscopicallyUniform_iMATE}:
\begin{corollary}[\label{corollary:Entropy_MacroUniform}Quantum macroscopic entropy density of a macroscopically uniform state]
For any representation of a macroscopically uniform state $(\rho_{L})_{L\in\mathbb{N}}$, its quantum macroscopic entropy density coincides with the upper bound~\eqref{eq:tilde_s},
\begin{align}
    &s^{\mathrm{mac}}_{\ell}[(\rho_{L})_{L\in\mathbb{N}}]=\tilde{s}_{\ell}[(\rho_{L})_{L\in\mathbb{N}}].
\end{align}
In particular, any representation of iMATE satisfies the above equality.
\end{corollary}
\noindent
This corollary states that if $(\rho_{L})_{L\in\mathbb{N}}$ represents a macroscopically uniform state, 
the quantum macroscopic entropy density can be calculated from the simpler quantity $\tilde{s}_{\ell}[(\rho_{L})_{L\in\mathbb{N}}]$.
In addition, the corollary also states that the quantum macroscopic entropy density of $(\rho_{L})_{L\in\mathbb{N}}$ representing iMATE can be calculated from $\tilde{s}_{\ell}[(\rho_{L})_{L\in\mathbb{N}}]$.

\subsection{\label{sec:Entropy_s^mac_localUniform}Quantum macroscopic entropy density of locally uniform macroscopic states}

If $(\rho_{L})_{L\in\mathbb{N}}$ is locally uniform
in the sense explained in subsection~\ref{sec:Entropy_LocalUniform},
then $\rho_{L|\ell}^{(k)}=\rho_{L|\ell}^{\boldsymbol{r}=\boldsymbol{0}}$ follows from
Definition~\ref{definition:LocallyUniform}.
Combining this with the continuity of the von Neumann entropy~\cite{NielsenChuang2010}, 
the quantum macroscopic entropy density can be calculated without spatial average as follows:
\begin{corollary}[\label{corollary:Entropy_LocallyUniform}Quantum macroscopic entropy density of a locally uniform macroscopic state]
For any representation of 
a locally uniform macroscopic state (see Definition~\ref{definition:LocallyUniform}),
the quantum macroscopic entropy density satisfies
\begin{align}
    s^{\mathrm{mac}}_{\ell}[(\rho_{L})_{L\in\mathbb{N}}]=
    \lim_{L\to\infty}\frac{S_{\mathrm{vN}}\bigl[
    \rho_{L|\ell}^{\boldsymbol{r}=\boldsymbol{0}}
    \bigr]}{\ell^d}.
    \label{eq:s^red_TransInv}
\end{align}
\end{corollary}

Now we explain the relation between the quantum macroscopic entropy density and a well-studied quantity.
When distinguishing whether energy eigenstates are thermal (which means in thermal equilibrium) or not, many existing studies investigate entanglement entropy~\cite{Turner2018a,Moudgalya2018a}. In such studies, an energy eigenstate is often regarded as nonthermal
if its entanglement entropy between a subsystem of side length~$\ell$ and the rest deviates from the thermodynamic entropy for $1\ll \ell \ll L$. From the formula~\eqref{eq:s^red_TransInv}, we can see that, when $(\rho_{L})_{L\in\mathbb{N}}$ is a sequence of translation-invariant energy eigenstates, its quantum macroscopic entropy density coincides with its entanglement entropy 
density in an intermediate size subsystem, $1\ll\ell\ll L$.
Therefore, 
for such states
we can justify the above-mentioned usage of entanglement entropy by showing that the quantum macroscopic entropy density of iMATE coincides with thermodynamic entropy, which will be given later as Theorem~\ref{theorem:s^red=s^TD}.

Although the half-chain entanglement entropy is also well used to characterize the thermality of given pure quantum states, our results do not guarantee the correctness of such a usage. In fact, we can find examples of representations of iMATE whose half-chain entanglement entropy does not agree with thermodynamic entropy
even for locally uniform macroscopic states
while the quantum macroscopic entropy density does:
\begin{example}[Thermal pure quantum state with zero half-chain entanglement]\label{example:HalfChainEntanglement}
For simplicity, we here consider one-dimensional spin-$1/2$ systems with an even number of sites, $d=1$, $D=2$, and even $L$.
Let $\{M_{L}\}_{L\in\mathbb{N}}$ be a sequence of integers that diverges as $L\to\infty$ with the rate of $o(L)$, e.g., $M_{L}=\lfloor\sqrt{L}\rfloor$.
Take $\rho_L = \ket{\psi_L}\bra{\psi_L}$, where
\begin{align}
    \ket{\psi_L} &= \bigotimes_{j=1}^{L/2} \ket{\Phi}_{2j-1-2M_L,\ 2j+2M_L}.
\end{align}
Here, 
$\ket{\Phi}_{r,r'}$ is one of the Bell states on sites $r$ and $r'$,
such as $\ket{\Phi}_{r,r'} = \frac{1}{\sqrt{2}} (\ket{0}_r \ket{1}_{r'} -\ket{1}_r \ket{0}_{r'})$, 
and the addition of sites, $2j-1-2M_L$ and $2j+2M_L$, is defined modulo $L$.

The reduced density matrix of $\rho_L$ on any subset with a diameter less than $4M_{L}$ is maximally mixed, and hence $\rho_{\infty|\ell}^{(k)}=(I/2)^{\otimes \Cell}$ holds because of $\ell \ll 4M_{L}$.
From Proposition~\ref{proposition:Equiv_rho^red}, this implies that $(\rho_L)_{L\in\mathbb{N}}$
represents
iMATE at infinite temperature $\beta=0$ (for any Hamiltonian).
In addition, 
it is a locally uniform macroscopic state, and 
its quantum macroscopic entropy density is given by $\log 2$, which coincides with thermodynamic entropy density.
On the other hand, the half-chain entanglement entropy density of $\rho_L$ is given by $(8M_{L}/L)\log 2$, which becomes zero as $L\to\infty$ and does not agree with thermodynamic entropy density.
\end{example}
\noindent
This example and the above discussion show that
for locally uniform macroscopic states
the entanglement entropy for $1\ll \ell\ll L$ can detect nonthermality of a pure state, while the half-chain entanglement cannot in general.

\subsection{\label{sec:measurement_s^macro}Quantum macroscopic entropy density is measurable with
macroscopic measurements}

Interestingly,
the quantum macroscopic entropy density 
can be obtained by measuring additive observables
(without measuring other observables)
as follows:
\begin{corollary}[\label{corollary:Measurement_s^mac}Measurement of 
quantum macroscopic entropy density]
Let $D$ be the dimension of the local Hilbert space on each site and $\{\ket{j}_{\boldsymbol{r}}\}_{j=0}^{D-1}$ be an orthonormal basis on each site $\boldsymbol{r}$.
Given an arbitrary positive integer $\ell\in\mathbb{N}$.
The basis vector on $\Cell$ is written as $\ket{\boldsymbol{j}}_{\Cell}:=\bigotimes_{s=1}^{\ell^d}\ket{j_{s}}_{\boldsymbol{r}_s}$, where $j_s\in\{1,...,D\}$, $\boldsymbol{j}=(j_1,...,j_{\ell^d})$, and the symbol $s$ labels the sites in $\Cell$ as $\Cell=\{\boldsymbol{r}_s\}_{s=1}^{\ell^d}$.

Prepare the system in a state $\rho_L$ whose sequence $(\rho_L)_{L\in\mathbb{N}}$ represents a macroscopic state. Measure the additive observable on a primitive macroscopic subsystem $\mathcal{S}^{(k)}$ obtained from an $\ell$-local observable $\ket{\bm{j}}\bra{\bm{j}^\prime}$, defined by Eq.~\eqref{eq:Additive_S},
\begin{align}
    A_{\mathcal{S}^{(k)}}(\ket{\bm{j}}\bra{\bm{j}^\prime})=\sum_{\bm{r}\text{ s.t. }\Cell+\bm{r}\subset \mathcal{S}^{(k)}}
    \ket{\boldsymbol{j}}_{\Cell+\boldsymbol{r}}
    \bra{\boldsymbol{j}^\prime}_{\Cell+\boldsymbol{r}}
\end{align}
where $\ket{\boldsymbol{j}}_{\Cell+\boldsymbol{r}}=\bigotimes_{s=1}^{\ell^d}\ket{j_{s}}_{\boldsymbol{r}_s+\bm{r}}$.
Repeating this measurement many times and taking the sample average, we obtain the expectation value of the density $A_{\mathcal{S}^{(k)}}(\ket{\bm{j}}\bra{\bm{j}^\prime})/|\mathcal{S}^{(k)}|$.
Here, note that if $(\rho_L)_{L\in\mathbb{N}}$ represents a normal macroscopic state, a single measurement suffices to evaluate the expectation value in a macroscopically negligible precision.
Suppose that we can 
reasonably extrapolate the sample average to $L\to\infty$.
Then, we obtain a matrix element of $\rho_{\infty|\ell}^{(k)}$ through Eq.~\eqref{eq:<a>^red=<A_Sigma>/Sigma},
\begin{align}
    \lim_{L\to\infty}\mathrm{Tr}\Bigl[\rho_L \frac{A_{\mathcal{S}^{(k)}}(\ket{\bm{j}}\bra{\bm{j}^\prime})}{|\mathcal{S}^{(k)}|}\Bigr]
    =\bra{\boldsymbol{j}^\prime}_{\Cell}
    \rho_{\infty|\ell}^{(k)}
    \ket{\boldsymbol{j}}_{\Cell}.
    \label{eq:Measurement_s^mac_Add=rho}
\end{align}
By performing this measurement for all $\bm{j},\bm{j}^\prime$, we obtain $\rho_{\infty|\ell}^{(k)}$ for each $k\in\{1,...,K^d\}$. 
By substituting the obtained $\rho_{\infty|\ell}^{(k)}$'s into Eq.~\eqref{eq:DEF_s^red_ell}, we have $s^{\mathrm{mac}}_{\ell}[(\rho_{L})_{L\in\mathbb{N}}]$.
\end{corollary}
\noindent
We give several comments on this result: first, 
this measurement protocol 
might look 
similar to quantum state tomography; however, 
only the additive observables need to be measured here,
and their number  is independent of $L$ (while it increases with $\ell$).
Second, this protocol enables us to measure not only quantum macroscopic entropy density, $s^{\mathrm{mac}}_{\ell}[(\rho_{L})_{L\in\mathbb{N}}]$, but also the spatial average of $\ell$-local density matrices, $\rho_{\infty|\ell}^{(k)}$. In this sense, $\rho_{\infty|\ell}^{(k)}$ and $s^{\mathrm{mac}}_{\ell}[(\rho_{L})_{L\in\mathbb{N}}]$ are macroscopic quantities that can be obtained directly by macroscopic measurements.
Third, even 
in a practical experimental situation where
the states prepared for repeated measurements, $\rho_L^1$, $\rho_L^2$, $\rho_L^3,...$, are not the same each time, $\rho_L^i\neq \rho_L^{i^\prime}$, the above measurement protocol works, as long as all sequences of these states represent the same macroscopic state, $(\rho_L^i)_{L\in\mathbb{N}}\maceq (\rho_L^{i^\prime})_{L\in\mathbb{N}}$. 
This is because, for all $i$, the expectation value of $A_{\mathcal{S}^{(k)}}(\ket{\bm{j}}\bra{\bm{j}^\prime})/|\mathcal{S}^{(k)}|$ in $\rho_{L}^{i}$ is the same, and hence, the law of large numbers applies to the sample average with a large sample size.

In Corollary~\ref{corollary:Measurement_s^mac}, we have assumed that 
the expectation value of the density of additive observables can be reasonably extrapolated to $L\to\infty$. 
However, the measurement outcomes inevitably deviate from the $L\to\infty$ limit due to fluctuations of additive observables, measurement errors, and finite-size effects. 
For completeness, we will show in Appendix~\ref{sec:MeasurementDeviation} that
such deviations do not affect the value of the quantum macroscopic entropy density.

\section{Quantum macroscopic entropy density of thermal equilibrium}
\label{sec:Entropy_s^mac_eq}

In this section, we show that 
the quantum macroscopic entropy density agrees 
with the thermodynamic entropy density for all representations of iMATE.
We also discuss methods of calculating it and measuring it,
and present a numerical demonstration.

\subsection{\label{sec:agree_with_tde}Agreement with thermodynamic entropy density}

We have shown that the quantum macroscopic entropy density
takes the same value among equivalent macroscopic states, irrespective
of their quantum-mechanical representations. This means 
that it extracts only a macroscopic property from the density matrices.
Hence, we expect that it should capture some thermodynamic property
when macroscopic states represent thermal equilibrium states.
By employing iMATE for thermal equilibrium states,
we can show that this is indeed the case.
That is, the quantum macroscopic entropy density agrees 
with the thermodynamic entropy density for all representations of iMATE:

\begin{theorem}[quantum macroscopic entropy density of iMATE]\label{theorem:s^red=s^TD}
For any representation of an iMATE, $(\rho_{L})_{L\in\mathbb{N}}$, 
there exists the $\ell\to\infty$ limit of $s^{\mathrm{mac}}_{\ell}[(\rho_{L})_{L\in\mathbb{N}}]$, and it agrees 
with the thermodynamic entropy density $s^{\mathrm{TD}}(\beta|H)$, 
\begin{align}
    \lim_{\ell\to\infty}s^{\mathrm{mac}}_{\ell}[(\rho_{L})_{L\in\mathbb{N}}]=s^{\mathrm{TD}}(\beta|H),
    \label{eq:s^red_ell=s^TD}
\end{align}
where $s^{\mathrm{TD}}(\beta|H)$ is given by
\begin{align}
    s^{\mathrm{TD}}(\beta|H):=\lim_{L\to\infty}\frac{S_{\mathrm{vN}}[\rho^{\mathrm{can}}_{L}(\beta|H)]}{N}.
    \label{eq:DEF_s^TD}
\end{align}
\end{theorem}
\noindent
Since $s^{\mathrm{mac}}_{\ell}[(\rho_{L})_{L\in\mathbb{N}}]$ is defined 
in the $L\to\infty$ limit, 
Eq.~\eqref{eq:s^red_ell=s^TD} 
takes the $L\to\infty$ limit before the $\ell\to\infty$ limit. 
This 
corresponds to the iterated limit in Eq.~\eqref{eq:Summary_EntropyFormula}.

\begin{proof}
First, we remark that the existence of the thermodynamic limit~\eqref{eq:DEF_s^TD} can be verified from that of free energy and of the expectation value of energy density. 
The former is proved, for instance, in Refs.~\cite{Ruelle1999,Tasaki2018}. The latter follows from Assumption~\ref{assumption:GibbsState}.
From Corollary~\ref{corollary:Equiv_s^red}, it is sufficient to prove 
the above theorem for the canonical Gibbs state $\bigl(\rho^{\mathrm{can}}_{L}(\beta|H)\bigr)_{L\in\mathbb{N}}$,
that is, there exists the $\ell\to\infty$ limit of 
$s^{\mathrm{mac}}_{\ell}\bigl[\bigl(\rho^{\mathrm{can}}_{L}(\beta|H)\bigr)_{L\in\mathbb{N}}\bigr]$, and it agrees with thermodynamic entropy,
\begin{align}
    \lim_{\ell\to\infty}s^{\mathrm{mac}}_{\ell}\bigl[\bigl(\rho^{\mathrm{can}}_{L}(\beta|H)\bigr)_{L\in\mathbb{N}}\bigr]=s^{\mathrm{TD}}(\beta|H).
    \label{eq:s^red_rho^can=s^TD}
\end{align}
For the proof of this equality, see Appendix~\ref{sec:Proof_s^red=s^TD}.

\end{proof}

This result should be contrasted with the
von Neumann entropy density~$\lim_{L\to\infty}S_{\mathrm{vN}}[\rho_{L}]/N$
of the total system.
In general, it does not agree
with the thermodynamic entropy density $s^{\mathrm{TD}}$. 
For instance, it takes zero if $(\rho_{L})_{L\in\mathbb{N}}$ is a sequence of pure quantum states representing iMATE, 
such as thermal energy eigenstates~\cite{Neumann1929,Deutsch1991,Srednicki1994,Rigol2008}, the thermal pure quantum (TPQ) states~\cite{Sugiura2012,Sugiura2013}, and the imaginary-time-evolved entangled antipodal pair states~\cite{Chiba2024a,Yoneta2024}.

Note that the von Neumann entropy density of
a reduced density matrix 
(for a pure quantum state, this is nothing but the entanglement entropy density)
of a state representing iMATE
does not necessarily agree with $s^{\mathrm{TD}}$, either.
In fact, that of
the $\ell$-local density matrix,  
$S_{\mathrm{vN}}[\rho_{L|\ell}^{\boldsymbol{r}}]/\ell^d$,
of iMATE 
does not necessarily agree with $s^{\mathrm{TD}}$, 
even if we take the $\ell \to \infty$ limit.
We will give a vivid illustration of this fact 
taking the METTS as an example
in subsection~\ref{sec:numerical_entropy}.

While the METTS in that example does not represent locally uniform iMATE,
one might expect that 
if $\rho_L$ represents locally uniform iMATE
(which is then 
equivalent to MITE of Definition~\ref{definition:MITE})
$S_{\mathrm{vN}}[\rho_{L|\ell}^{\boldsymbol{r}}]/\ell^d$
would agree with $s^{\mathrm{TD}}$.  For example, it
has been often argued~\cite{Goldstein2006,Iwaki2021,Meier2025} that the reduced density matrix of a translation-invariant representation of thermal equilibrium on a small hypercube $\rho_{L|\ell}^{\bm{r}=\bm{0}}$ is ``close'' to the canonical Gibbs state on the same hypercube $ \rho^{\mathrm{can}}_{\ell}(\beta|H_\ell)$ [not to be confused with the reduced density matrix of $\rho^{\mathrm{can}}_{L}(\beta|H_L)$ on $\Cell$].
Combining this argument with Eq.~\eqref{eq:s^red_TransInv}, 
one might 
expect that 
Theorem~\ref{theorem:s^red=s^TD} would readily follow at least for representations of iMATE that are locally uniform. 
However, it is unclear that in what sense $\rho_{L|\ell}^{\bm{r}=\bm{0}}$ is close to $\rho^{\mathrm{can}}_{\ell}(\beta|H_\ell)$ in such an argument and whether the deviation between them does not affect the value of the von Neumann entropy.
For instance, the expectation values of operators near the boundary of $\Cell$ in the state $\rho_{L|\ell}^{\bm{r}=\bm{0}}$ differ from those in the state $\rho^{\mathrm{can}}_{\ell}(\beta|H_{\ell})$, as manifested in difference between the Gibbs states under periodic boundary and open boundary.
This means that $\rho_{L|\ell}^{\bm{r}=\bm{0}}$ and $\rho^{\mathrm{can}}_{\ell}(\beta|H_{\ell})$ are not close in the sense of trace norm, 
and hence closeness of $S_{\mathrm{vN}}\bigl[\rho_{L|\ell}^{\bm{r}=\bm{0}}\bigr]$ and $S_{\mathrm{vN}}[\rho^{\mathrm{can}}_{\ell}(\beta|H_{\ell})]$ cannot be proved by naively applying the Fannes-Audenaert inequality~\cite{Audenaert2007}.
Therefore, the proof (Appendix~\ref{sec:Proof_s^red=s^TD}) of the above Theorem is nontrivial 
even for the locally uniform case.

\subsection{Remarks on calculating thermodynamic entropy}

Equation \eqref{eq:s^red_ell=s^TD} gives a method 
for calculating thermodynamic entropy. 
It has two notable features among various methods:
it depends only on the representation of a macroscopic state $(\rho_L)_{L\in\mathbb{N}}$, and 
it provides thermodynamic entropy correctly for all representations of iMATE. 

For the first point, other methods 
(including the Boltzmann entropy) often require information about the Hamiltonian of the system, while our quantum macroscopic entropy density does not require any information other than $(\rho_L)_{L\in\mathbb{N}}$, as explained in Sec.~\ref{sec:Entropy_s^mac}. 

For the second point, 
the ordinary von Neumann entropy does not work for all representations of iMATE, 
as mentioned in subsection~\ref{sec:agree_with_tde}.
The same can be said for 
the entanglement-entropy-like quantity in the RHS of formula~\eqref{eq:s^red_TransInv}.
In fact, it provides thermodynamic entropy only for representations of iMATE that are locally uniform, but not for other representations of iMATE. As an example, let us consider the state in Example~\ref{example:NonMITE_MacroEqState}. Since it is not locally uniform, formula~\eqref{eq:s^red_TransInv} does not hold, and the RHS of Eq.~\eqref{eq:s^red_TransInv} becomes zero, 
and hence it does not agree with thermodynamic entropy. By contrast,
our macroscopic entropy density predicts $\log 2$, 
which agrees with 
thermodynamic entropy density.

\subsection{Numerical illustration}
\label{sec:numerical_entropy}

Now we numerically demonstrate that our quantum macroscopic entropy density 
$s_{\ell}^{\mathrm{mac}}[(\rho_L)_{L\in\mathbb{N}}]$
approaches the thermodynamic entropy density
$s^{\mathrm{TD}}(\beta|H)$ with increasing  $\ell$,
as Eq.~\eqref{eq:s^red_ell=s^TD},
whereas the density of entanglement entropy does not, for a pure quantum state that represents iMATE (but not MITE).

We consider the one-dimensional $(d=1)$ transverse-field Ising model,
whose Hamiltonian is given by Eq.~\eqref{eq:H_transverse_Ising}.
We impose the open boundary condition for ease of calculation, and take $J=g=1$ and the inverse temperature as $\beta=1$. 

Using the matrix product state representation,
we generate an ensemble of METTS samples $\rho_L = \ket{\phi(i_n)}\bra{\phi(i_n)}$ ($n=1,2,\dots$),
defined by Eq.~\eqref{eq:METTS}, which
are sampled according to the probability distribution $P(i)/Z$.
Although $(\rho_L)_{L\in\mathbb{N}}$ is not MITE, it typically represents iMATE, 
as discussed immediately after Eq.~\eqref{eq:METTS_canonical},
Therefore, from Corollary~\ref{corollary:Entropy_MacroUniform},
its quantum macroscopic entropy density 
$s_{\ell}^{\mathrm{mac}}[(\rho_L)_{L\in\mathbb{N}}]$
should agree with a simpler quantity $\tilde{s}_{\ell}[(\rho_L)_{L\in\mathbb{N}}] = \lim_{L\to\infty} S_{\mathrm{vN}}[\rho_{L|\ell}^{\mathrm{ave}}]/\ell^d$,
for each $\ell=1,2,\cdots$.
Hence, we calculate 
the von Neumann entropy density $S_{\mathrm{vN}}[\rho_{L|\ell}^{\mathrm{ave}}]/\ell^d$ for these METTS samples.
We then compare their average with thermodynamic entropy density $s^{\mathrm{TD}}(\beta|H)$ defined by Eq.~\eqref{eq:DEF_s^TD},
which is here calculated using exact analytical solutions.
The length of the total lattice is set to $L=512$, 
which is sufficiently large to ensure that finite-size effects are negligible.

In Fig.~\ref{fig:MacroscopicEntropydensity}, we plot $S_{\mathrm{vN}}[\rho_{L|\ell}^{\mathrm{ave}}]/\ell^d$ (green filled circles) averaged over $10^5$ METTS,
where the error bars indicate the standard deviation.
It is confirmed that $S_{\mathrm{vN}}[\rho_{L|\ell}^{\mathrm{ave}}]/\ell^d$ 
approaches $s^{\mathrm{TD}}(\beta|H)$ (light-blue line) as $\ell$ is increased.

\begin{figure}
    \centering
    \includegraphics[width=\linewidth]{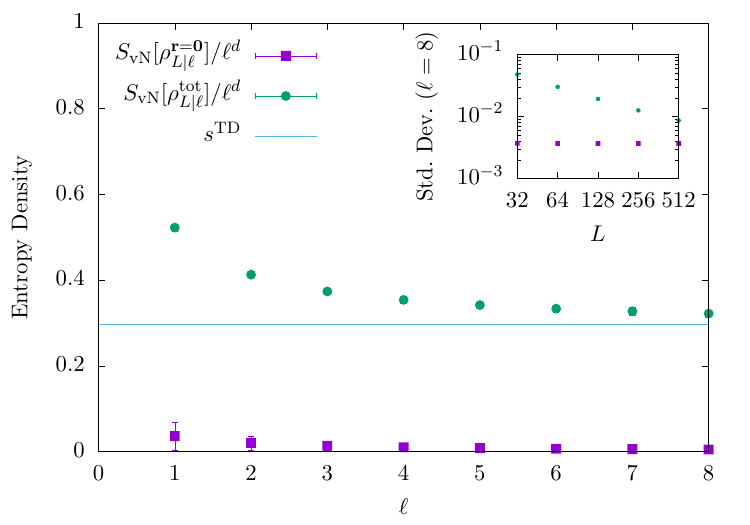}
    \caption{
        Locality~$(\ell)$ dependence of 
        $S_{\mathrm{vN}}[\rho_{L|\ell}^{\mathrm{ave}}]/\ell^d$ (green filled circles) 
        and the entanglement entropy density 
        $S_{\mathrm{vN}}[\rho_{L|\ell}^{\bm{r}=\bm{0}}]/\ell^d$ (purple square)
        in the one-dimensional transverse-field Ising model at $J=g=1$ and $\beta=1$
        with $L=512$ (see text for details). 
        Both quantities are obtained from 
        $10^5$ METTS samples, where the error bars represent the standard deviation.
        The light-blue line indicates the exact thermodynamic entropy density in the thermodynamic limit.
        The inset shows the $L$ dependence of the standard deviations at fixed $\ell=8$ for both entropy densities.
        Imaginary-time evolution is carried out using a second-order Trotter decomposition with step size $\delta\beta=0.01$ and truncation threshold $10^{-14}$.
    }
    \label{fig:MacroscopicEntropydensity}
\end{figure}

To justify that the behavior of the sample average reflects each typical METTS at the thermodynamic limit, we also examine the $L$ dependence of sample-to-sample fluctuations for $L=32$-$512$.
The inset shows the standard deviation as a function of $L$ at fixed $\ell=8$.
It can be seen that the standard deviation of $S_{\mathrm{vN}}[\rho_{L|\ell}^{\mathrm{ave}}]/\ell^d$ decreases as a power law in $L$ toward zero, indicating self-averaging. Therefore, $S_{\mathrm{vN}}[\rho_{L|\ell}^{\mathrm{ave}}]/\ell^d$ becomes independent of the METTS sample as $L\to\infty$ for fixed $\ell$.
These results demonstrate Eq.~\eqref{eq:s^red_ell=s^TD}
in Theorem~\ref{theorem:s^red=s^TD},
which takes the iterated limit, 
$L\to\infty$ followed by $\ell\to\infty$,
as in Eq.~\eqref{eq:Summary_EntropyFormula}.

For comparison, we also plot 
in Fig.~\ref{fig:MacroscopicEntropydensity}
the density of the entanglement entropy,
$S_{\mathrm{vN}}[\rho_{L|\ell}^{\bm{r}=\bm{0}}]/\ell^d$ (purple square), 
which is a measure of entanglement
between $\Cell$ taken from the center of the chain and the rest.
We observe that this quantity is well separated from $S_{\mathrm{vN}}[\rho_{L|\ell}^{\mathrm{ave}}]/\ell^d$ (green filled circles) and
$s^{\mathrm{TD}}(\beta|H)$ (blue line), 
and approaches $0$ with increasing $\ell$, which is consistent with the result in Ref.~\cite{Kusuki2024}.
Furthermore, its standard deviation (inset) approaches a constant independent of $L$, in contrast to that of $S_{\mathrm{vN}}[\rho_{L|\ell}^{\mathrm{ave}}]/\ell^d$ (also in the inset), which decays as a power law. 
These results clearly demonstrate that 
the entanglement entropy density for METTS does not agree with the thermodynamic entropy density.
This is because $\rho_L$ is not locally uniform, and consequently, Eq.~\eqref{eq:s^red_TransInv} in Corollary~\ref{corollary:Entropy_LocallyUniform} does not hold for METTS, although it represents 
thermal equilibrium.

It is worth mentioning that Fig.~\ref{fig:MacroscopicEntropydensity} also demonstrates that 
thermal equilibrium does not necessarily imply the volume law of the entanglement entropy,
or equivalently, the subvolume law does not necessarily imply nonequilibrium.

\subsection{\label{sec:method.measure.sTD}Method for measuring thermodynamic entropy}

Our results provide a novel method for measuring thermodynamic entropy.
In thermodynamics, 
entropy can only be measured through a certain thermodynamic process, 
which involves multiple different thermal equilibrium states.
A typical example of such a process is a heating process, 
in which the system traverses multiple 
thermal equilibrium states at different temperatures.
By measuring the specific heat at many points in such a process, 
one obtains the (difference of) entropy using thermodynamic relations.

By contrast, Theorem~\ref{theorem:s^red=s^TD} and Corollary~\ref{corollary:Measurement_s^mac} show that thermodynamic entropy density can be obtained by the measurement protocol in Corollary~\ref{corollary:Measurement_s^mac}, which only requires preparations of quantum states representing iMATE at a \emph{single} inverse temperature $\beta$.
Although it is safe to repeatedly prepare the states in order to eliminate the effect of the backactions of quantum measurement~\footnote{According to Refs.~\cite{Fujikura2016,Shimizu2017}, 
if the measurement is performed appropriately, the backactions on the densities of additive observables are macroscopically negligible. If so, it is sufficient to prepare only a single quantum state representing iMATE.}, our protocol contains only iMATE at a single $\beta$, not multiple $\beta$, in contrast to the conventional thermodynamic protocol.

Furthermore, to obtain thermodynamic entropy, we do not need to measure additive observables on primitive subsystems $\mathcal{S}^{(k)}$ for all $k\in\{1,...,K^d\}$; instead, measurements of additive observables on the total system $\Lambda_L$ suffice.
This is because, from Corollary~\ref{corollary:Entropy_MacroUniform}, quantum macroscopic entropy density coincides with $S_{\mathrm{vN}}[\rho_{\infty|\ell}^{\mathrm{ave}}]/\ell^d$ for states representing iMATE. Here, we can obtain $S_{\mathrm{vN}}[\rho_{\infty|\ell}^{\mathrm{ave}}]$ by measurements of additive observables on the total system since the density matrix $\rho_{L|\ell}^{\mathrm{ave}}$ is defined by the spatial average of reduced density matrices over the total lattice, and hence, it is characterized by the additive observables on the total system.

\section{\label{sec:EntropyIncrease}Law of increasing entropy}

In Sec.~\ref{sec:MacroPassivity}, we have given Corollary~\ref{corollary:macroscopic-passivity},
which corresponds to a quantum-mechanical derivation of Planck's principle.
It is one 
of many forms of the second law of thermodynamics 
and applies to any adiabatic operation in which the initial and final values of 
$(f^1,..., f^{m})$ coincide.
In thermodynamics, it is sometimes said to be equivalent to different forms of the second law; however, other laws of thermodynamics were assumed 
when the equivalence was proved, 
and therefore, they are not shown to be equivalent to each other from quantum mechanics alone. 

In this section, we 
address
this problem and will give a 
quantum-mechanical
proof of another form of the second law of thermodynamics. 
It is the law of increasing entropy, which  
asserts that an adiabatic thermodynamic operation cannot decrease the thermodynamic entropy of the system.
This law applies to any adiabatic operation 
including the case where the initial and the final values of $(f^1,..., f^{m})$ are different.
Lieb and Yngvason~\cite{Lieb1999} showed that thermodynamics
can be rigorously constructed from the adiabatic accessibility that 
corresponds to this law.

\subsection{Time evolution representing adiabatic thermodynamic operation}

In order to prove 
the law of increasing entropy 
by quantum mechanics, we need to formulate thermodynamic entropy and adiabatic thermodynamic operations in terms of quantum mechanics. 
For entropy, we will use the quantum macroscopic entropy density introduced 
in Sec.~\ref{sec:Entropy_s^mac} because it
corresponds to thermodynamic entropy as shown in Sec.~\ref{sec:Entropy_s^mac_eq}. 
For adiabatic thermodynamic operations, 
we 
refine 
the macroscopic operations introduced in Sec.~\ref{sec:MacroOp_MacroEquiv}
as follows.

In thermodynamics, an adiabatic thermodynamic operation generally consists of three parts:
(i) preparation of a thermal equilibrium state, 
(ii) operation 
by changing the values of 
the control parameters 
$(f^1,...,f^m)$ in a time-dependent manner,
and 
(iii)
waiting for relaxation to another thermal equilibrium state to occur.

In Sec.~\ref{sec:MacroOp_MacroEquiv}, macroscopic operations have been formulated as consisting of only parts~(i) and (ii), but excluding part~(iii). This is because part~(iii) does not affect the macroscopic passivity~\eqref{eq:macroscopic-passivity}, as mentioned in the last paragraph of Sec.~\ref{sec:MacroPassivity_beta>=0}, and can be omitted. 
However, in thermodynamics, the law of increasing entropy requires part~(iii) because the state after part~(ii) is generally a nonequilibrium state and its thermodynamic entropy is not well-defined.
Therefore, in this section, we consider the time evolution containing part~(iii) in addition to macroscopic operations, i.e., parts~(i) and (ii).

Precisely, 
we consider the time evolution of the following form (see also Table~\ref{tbl:Operation_EntropyIncrease}):
(i)
At $t= 0$, the system is prepared in the state $\rho_{L}(0)$, whose sequence 
$\bigl(\rho_L(0)\bigr)_{L\in\mathbb{N}}$ represents iMATE of the system described by 
the initial Hamiltonian~$H_{0}$.
We take $H_{0}$ as an additive observable on the whole system~$\Lambda_{L}$, which is translation invariant on $\Lambda_{L}$.
(ii) For $0<t<t^*$, where $t^*$ is the operation time, 
$(f^1,...,f^m)$ are changed in a time-dependent manner. 
Consequently, the Hamiltonian becomes a time-dependent one, $H(t)$, which induces the unitary time evolution. We take $H(t)$
as a general one of the form~\eqref{eq:general_H}, which may not be translation invariant on $\Lambda_{L}$ 
(recall that $B^\mu_L$'s are 
additive observables on some macroscopic subsystems, not necessarily on the whole system $\Lambda_{L}$).
Then, the time evolution in this operation part is a general macroscopic operation $\bigl(U_{L}(\bullet,0)\bigr)_{L\in\mathbb{N}}$ given in Definition~\ref{definition:macroscopic-operation}. 
(iii) For $t\ge t^*$, we stop the operation 
and wait for relaxation. 
We take
the Hamiltonian~$H_1$ in this final stage
as an additive observable on the whole system~$\Lambda_{L}$ again
(but can be different from $H_0$).
The system relaxes due to the unitary evolution with 
the time evolution operator $e^{-iH_{1}(t-t^*)}$.
\begin{table}[h]
\caption{\label{tbl:Operation_EntropyIncrease}Unitary time evolution corresponding to adiabatic thermodynamic operations.}
\centering
\begin{ruledtabular}
\begin{tabular}{lcc}
time regime & $H(t)$ & $\rho_{L}(t)$ \\[1pt]\hline
$t= 0$& $H_{0}$ & $\rho_{L}(0)$ $\maceq \rho_{L}^{\mathrm{can}}(\beta_0|H_0)$\\[1pt]
$0<t< t^*$ (operation)& $H(t)$ & $U_{L}(t,0)\rho_{L}(0)U_{L}^{\dagger}(t,0)$\\[1pt]
$t\ge t^*$ (relaxation)& $H_1$ & $e^{-iH_1 (t-t^*)}\rho_{L}(t^*)e^{iH_1 (t-t^*)}$\\[1pt]
$t\to\infty$ (thermalization)& $H_1$ & $\overline{\rho_{L}(t>t^*)}\maceq \rho_{L}^{\mathrm{can}}(\beta_1|H_1)$
\end{tabular}
\end{ruledtabular}
\end{table}

In summary, the time evolution can be written as follows.
\begin{definition}[Time evolution 
regarded as corresponding
to adiabatic thermodynamic operations]
Unitary time evolution that we adopt to describe adiabatic thermodynamic operations is generated by the following time-dependent Hamiltonian,
\begin{align}
    H(t)=\begin{cases}
        H_{0} & (t= 0);\\
        H_0 - \sum_{\mu=1}^m f^\mu(t) B^\mu & (0<t<t^*);\\
        H_{1} & (t\ge t^*),
    \end{cases}
    \label{eq:EntropyIncrease_H(t)}
\end{align}
where $H_0$ and $H_1$ are additive observables on the whole system $\Lambda_L$ while $H(t)$ $(0<t<t^*)$ takes the form of Eq.~\eqref{eq:general_H}.
For time $0<t<t^*$, the time evolution is given by a macroscopic operation $U_{L}(t,0)$, while the time evolution for $t>t^*$ is written as a relaxation process $e^{-iH_1 (t-t^*)}$.
\end{definition}

In this section, we only consider the case where the operation time $t^*$ is independent of $L$. 
A longer operation time will be discussed in Sec.~\ref{sec:Counterexamples_Timescale}.

\subsection{Law of increasing entropy under an assumption of thermalization}

Let $\rho_{L}(t)$ be the state at time~$t$.
We are interested in $\rho_{L}(t)$ in the relaxation part long after $t^*$.
The relaxation part starts from the state $\rho_{L}(t^*)=U_{L}(t^*,0)\rho_{L}(0)U_{L}^{\dagger}(t^*,0)$.
In the studies of thermalization, it is well established that, under a reasonable assumption, the state evolved by a time-independent Hamiltonian is macroscopically equivalent to its long time average,
\begin{align}
    \overline{\rho_{L}(t>t^*)} 
    :=
    \lim_{\mathcal{T}\to\infty}\frac{1}{\mathcal{T}}\int_{0}^{\mathcal{T}} dt\ e^{-iH_{1}t}\rho_{L}(t^*)e^{iH_{1}t},
    \label{eq:LongTimeAve}
\end{align}
at almost all times~\cite{Tasaki1998,Reimann2008,Linden2009,Short2012,Reimann2012}. 
Therefore we can regard the long time average $\overline{\rho_{L}(t>t^*)}$ as the final state of the relaxation process.

Furthermore, previous 
studies of thermalization have uncovered that realistic systems usually exhibit thermalization, that is, the long-time-averaged state represents thermal equilibrium~\cite{Deutsch1991,Srednicki1994,Rigol2008,D'Alessio2016}.
Note that, when discussing thermalization, the initial state of the relaxation process $\rho_{L}(t^*)$ 
is restricted to states whose energy fluctuation is macroscopically negligible
\begin{align}
    \mathrm{Tr}[\rho_{L}(t^*) (H_1)^2]-\bigl(\mathrm{Tr}[\rho_{L}(t^*) H_1]\bigr)^2=o(N^2).
    \label{eq:NegligibleVarH_1}
\end{align}
Following these, 
we assume that if the initial state satisfies 
this condition,
then the final state $\bigl(\overline{\rho_{L}(t>t^*)}\bigr)_{L\in\mathbb{N}}$ represents iMATE:
\begin{assumpThermalization}[Thermalization]\label{assumption:thermalization}
If the sequence of the initial state of the relaxation process $\bigl(\rho_{L}(t^*)\bigr)_{L\in\mathbb{N}}$ is a macroscopic state that satisfies 
Eq.~\eqref{eq:NegligibleVarH_1},
then 
the sequence of states $\bigl(\overline{\rho_{L}(t>t^*)}\bigr)_{L\in\mathbb{N}}$ defined by Eq.~\eqref{eq:LongTimeAve} represents an iMATE~\eqref{eq:MacroEqState}. 
\end{assumpThermalization}
\noindent
In other words, this assumption says that 
there exists an inverse temperature~$\beta_{1}\in\mathbb{R}$ such that
\begin{align}
    \bigl(\overline{\rho_{L}(t>t^*)}\bigr)_{L\in\mathbb{N}}\maceq \bigl(\rho_{L}^{\mathrm{can}}(\beta_{1}|H_1)\bigr)_{L\in\mathbb{N}}.
    \label{eq:Assump_Thermalization}
\end{align}
This inverse temperature $\beta_{1}\in\mathbb{R}$ can be determined from the energy conservation, 
\begin{align}
    \lim_{L\to\infty}\frac{\mathrm{Tr}[\rho_{L}(t^*) H_1]}{N}
    =\lim_{L\to\infty}\frac{\mathrm{Tr}[\rho_{L}^{\mathrm{can}}(\beta_{1}|H_1) H_1]}{N}
    \label{eq:FiniteEnergyDensity}.
\end{align}
Note that, in the study of thermalization, it is known that such an assumption can be justified by the ETH~\cite{Deutsch1991,Srednicki1994,Rigol2008}.

Now, we give a quantum-mechanical version of the law of increasing entropy
(which is a weaker version of Theorem~\ref{theorem:EntropyIncrease_WithoutThermalization} given in Sec.~\ref{sec:EntropyIncrease_WithoutThermalization}).
\begin{theorem}[Law of increasing quantum macroscopic entropy density
and thermodynamic entropy density]\label{theorem:EntropyIncrease}
Let $t^*$ be an $L$-independent time, 
$\bigl(\rho_{L}(0)\bigr)_{L\in\mathbb{N}}$ be a representation of normal iMATE at an inverse temperature~$\beta_0\in\mathbb{R}$,
and $\rho_{L}(t)$ be the state evolved from $\rho_{L}(0)$ by the Hamiltonian explained immediately after 
Eq.~\eqref{eq:EntropyIncrease_H(t)}. 
Suppose that Assumption~\ref{assumption:thermalization} 
is
satisfied.
Then the sequence of states $\bigl(\overline{\rho_{L}(t>t^*)}\bigr)_{L\in\mathbb{N}}$ defined by Eq.~\eqref{eq:LongTimeAve} represents iMATE~\eqref{eq:MacroEqState} at the inverse temperature~$\beta_1$ determined by Eq.~\eqref{eq:FiniteEnergyDensity}, and satisfies
\begin{align}
    \lim_{\ell\to\infty}s^{\mathrm{mac}}_{\ell}\bigl[\bigl(\rho_{L}(0)\bigr)_{L\in\mathbb{N}}\bigr]\le \lim_{\ell\to\infty}s^{\mathrm{mac}}_{\ell}\bigl[\bigl(\overline{\rho_{L}(t>t^*)}\bigr)_{L\in\mathbb{N}}\bigr]
    \label{eq:Increasing_s^red_Thermalization}
\end{align}
for any macroscopic operation $\bigl(U_L(t^*,0)\bigr)_{L\in\mathbb{N}}$ of operation time $t^*=O(L^0)$.
Here, the existence of the limit $\ell\to\infty$ follows from Theorem~\ref{theorem:s^red=s^TD}.
Furthermore, it holds that
\begin{align}
    s^{\mathrm{TD}}(\beta_{0}|H_{0})\le s^{\mathrm{TD}}(\beta_{1}|H_{1}).
    \label{eq:Increasing_s^TD}
\end{align}
\end{theorem}
\noindent
This theorem 
says that the quantum macroscopic entropy density 
and thermodynamic entropy density 
cannot decrease through \textit{any} finite-time macroscopic operation followed by the relaxation process~\footnote{Note that the inverse temperatures $\beta_0$ and $\beta_1$ in Eq.~\eqref{eq:Increasing_s^TD} are not necessarily positive. In other words, this result applies to transitions between positive and negative inverse temperature states~\cite{Purcell1951,Hakonen1990,Hakonen1992,Braun2013}, 
which are beyond standard thermodynamics.},
and Eq.~\eqref{eq:Increasing_s^TD} is precisely the law of increasing entropy.

Note 
that the long time average in the RHS of Eq.~\eqref{eq:Increasing_s^red_Thermalization} corresponds to focusing on a state at a typical time in a sense similar to the study of thermalization, as explained 
for 
Eq.~\eqref{eq:LongTimeAve}.
Note also that, in Sec.~\ref{sec:Counterexamples_Timescale}, we will show that the constraint on the operation time, $t^*=O(L^0)$, is optimal by constructing a counterexample for any longer timescale $t^*=\omega(L^0)$.

\begin{proof}

There are two ways; one is to use Theorem~\ref{theorem:EntropyIncrease_WithoutThermalization} given in Sec.~\ref{sec:EntropyIncrease_WithoutThermalization}, and the other is to show it directly.
For the former, see Appendix~\ref{sec:Proof_Th3_from_Th4}.
For the latter, see Appendix~\ref{sec:DirectProof_EntropyIncrease}, which provides a much simpler proof than the proof of Theorem~\ref{theorem:EntropyIncrease_WithoutThermalization}.

\end{proof}

\subsection{\label{sec:EntropyIncrease_WithoutThermalization}Law of increasing entropy under assumption weaker than thermalization}

In the previous subsection, we have proved that the quantum macroscopic entropy density does not decrease through finite-time macroscopic operation and infinite-time relaxation process by assuming that the final Hamiltonian induces thermalization.
In this section, we will show the same proposition under a weaker assumption.

First, we introduce the weaker assumption. Consider two initial states $\rho_L$ and $\sigma_L$ whose sequences represent the same iMATE, 
\begin{align}
    (\rho_{L})_{L\in\mathbb{N}}
    \maceq (\sigma_{L})_{L\in\mathbb{N}}
    \maceq \bigl(\rho_{L}^{\mathrm{can}}(\beta_{0}|H_{0})\bigr)_{L\in\mathbb{N}}.
    \label{eq:UniqueFinalState_InitialStates}
\end{align}
From Theorem~\ref{theorem:JAIVTMXN_3}, the states after any macroscopic operation of operation time $t^*=O(L^0)$ remain to represent the same macroscopic state,
\begin{align}
    \bigl(\rho_L(t^*)\bigr)_{L\in\mathbb{N}}\maceq \bigl(\sigma_L(t^*)\bigr)_{L\in\mathbb{N}}.\label{eq:UniqueFinalState_StatesAfterOperation}
\end{align}
We consider the relaxation process induced by $H_1$ starting from these states.
We assume that, by this process, these states relax to the same macroscopic state. Precisely, this assumption is stated as follows:
\begin{assumpThermalization}[Existence of a unique final macroscopic state]\label{assumption:UniqueFinalState}

If two sequences of the initial states 
$\bigl(\rho_L(t^*)\bigr)_{L\in\mathbb{N}}$ and $\bigl(\sigma_L(t^*)\bigr)_{L\in\mathbb{N}}$ 
are macroscopically equivalent as in Eq.~\eqref{eq:UniqueFinalState_StatesAfterOperation}, then
the sequences 
$\bigl(\overline{\rho_{L}(t>t^*)}\bigr)_{L\in\mathbb{N}}$ and 
$\bigl(\overline{\sigma_{L}(t>t^*)}\bigr)_{L\in\mathbb{N}}$ given by the long time average~\eqref{eq:LongTimeAve} of the relaxation process also represent some macroscopic states, and they are macroscopically equivalent,
\begin{align}
    \bigl(\overline{\rho_{L}(t>t^*)}\bigr)_{L\in\mathbb{N}}
    \maceq \bigl(\overline{\sigma_{L}(t>t^*)}\bigr)_{L\in\mathbb{N}}.
    \label{eq:UniqueFinalState_FinalState}
\end{align}
\end{assumpThermalization}
In other words, this assumption requires that the relaxation process induced by $H_1$ has a unique stationary macroscopic state for any macroscopically equivalent representations of a macroscopic state, $\bigl(\rho_{L}(t^*)\bigr)_{L\in\mathbb{N}}$ and $\bigl(\sigma_{L}(t^*)\bigr)_{L\in\mathbb{N}}$. 
This assumption is weaker than Assumption~\ref{assumption:thermalization} because it does not assume that the unique stationary macroscopic state is macroscopically equivalent to the canonical Gibbs state, whereas Assumption~\ref{assumption:thermalization} does.
Furthermore, even in integrable systems,  if relaxation to a unique stationary state described by the generalized Gibbs ensemble~\cite{Rigol2007,Sotiriadis2014,Wouters2014,Pozsgay2014,Ilievski2015,Mierzejewski2015,Ilievski2015a,Ilievski2016,Doyon2017} occurs, then Assumption~\ref{assumption:UniqueFinalState} will be satisfied while Assumption~\ref{assumption:thermalization} is not.

Note that if the operation time is taken as $t^*=\omega(L^0)$, Eq.~\eqref{eq:UniqueFinalState_StatesAfterOperation} is no longer satisfied in general.
This means that such a case is out of the scope of Assumption~\ref{assumption:UniqueFinalState}, and hence, Eq.~\eqref{eq:UniqueFinalState_FinalState} can be violated.

Under such a weaker assumption, we again obtain the statement that prohibits the decrease of the quantum macroscopic entropy density.
\begin{theorem}[Law of increasing quantum macroscopic entropy density without assuming thermalization]\label{theorem:EntropyIncrease_WithoutThermalization}
Let $t^*$ be an $L$-independent time, 
$\bigl(\rho_{L}\bigr)_{L\in\mathbb{N}}$ represent an iMATE,
and $\rho_{L}(t)$ be the density matrix evolved according to Eq.~\eqref{eq:EntropyIncrease_H(t)} starting from $\rho_{L}$. 
Suppose that Assumption~\ref{assumption:UniqueFinalState} is satisfied.
Then, the representation $\bigl(\overline{\rho_{L}(t>t^*)}\bigr)_{L\in\mathbb{N}}$ of the final macroscopic state, defined by Eq.~\eqref{eq:LongTimeAve}, satisfies
\begin{align}
    \lim_{\ell\to\infty}s^{\mathrm{mac}}_{\ell}[(\rho_{L})_{L\in\mathbb{N}}]
    \le \liminf_{\ell\to\infty}s^{\mathrm{mac}}_{\ell}\bigl[\bigl(\overline{\rho_{L}(t>t^*)}\bigr)_{L\in\mathbb{N}}\bigr].
    \label{eq:Increasing_s^red_NoThermalization}
\end{align}
\end{theorem}
\begin{proof}
See Appendix~\ref{sec:Proof_EntropyIncrease}.
\end{proof}
We mention that Eq.~\eqref{eq:Increasing_s^red_Thermalization} of Theorem~\ref{theorem:EntropyIncrease} readily follows from this theorem, as explained in Appendix~\ref{sec:Proof_Th3_from_Th4}.

\section{\label{sec:ComparisonStudies}Comparison with the existing studies}

In this section, we compare our notion of iMATE and our results on the second law with those of the existing studies.

\subsection{\label{sec:MITE}iMATE vs. MITE and KMS}

In this subsection, we will show that our notion of thermal equilibrium, iMATE, is a broader class of states than the well-known notion of thermal equilibrium, MITE. We also explain that, when discussing the connection to thermodynamics, iMATE is more appropriate than MITE.

We start by rewriting the definition of MITE in Ref.~\cite{Mori2018} in terms of our notation.
The condition for MITE on the length scale $\ell$ in Ref.~\cite{Mori2018} can be rephrased as
\begin{align}
    \max_{\bm{r}\in\Lambda_{L}}\|\rho_{L|\ell}^{\bm{r}}- \tau_{L|\ell}^{\bm{r}}\|_{1}\ll 1,
\end{align}
where $\tau_{L}$ is the microcanonical Gibbs state on some energy shell, and  $\rho_{L|\ell}^{\bm{r}}$ (resp. $\tau_{L|\ell}^{\bm{r}}$)  is the reduced density matrix of $\rho_L$ (resp. $\tau_L$) on 
$\Cell+\bm{r}$ 
moved onto $\Cell$,  as defined in Eq.~\eqref{eq:ReducedState_rho^r}.
We interpret the symbol ``$\ll 1$'' in the above condition to mean that the larger $L$ is, the smaller the LHS becomes.
Then
the above condition can be rewritten as
\begin{align}
    \lim_{L\to\infty}\max_{\bm{r}\in\Lambda_{L}}\|\rho_{L|\ell}^{\bm{r}}- \tau_{L|\ell}^{\bm{r}}\|_{1}=0.
\end{align}

It is usually considered that, outside the first-order phase transition points, the microcanonical Gibbs state 
$\tau_L$
in the above condition can be replaced by the canonical Gibbs state 
$\sigma_L=\rho_{L}^{\mathrm{can}}(\beta|H_L)$ with the corresponding inverse temperature~$\beta$
because of the local equivalence of these two Gibbs states (which is rigorously shown under reasonable assumptions~\cite{Brandao2015Equivalence,Tasaki2018}).
Thus, we obtain the following:
\begin{definition}[\label{definition:MITE}MITE]
For any $\ell\in\mathbb{N}$, a sequence of density matrices $(\rho_{L})_{L\in\mathbb{N}}$ is said to represent $\ell$-local MITE of a system described by a Hamiltonian $H_L$ if 
\begin{align}
    \lim_{L\to\infty}\max_{\bm{r}\in\Lambda_{L}}\bigl\|\rho_{L|\ell}^{\bm{r}}-\sigma^{\bm{r}}_{L|\ell}\bigr\|_{1}=0,
    \label{eq:condition_MITE}
\end{align}
where $\sigma_L=\rho_{L}^{\mathrm{can}}(\beta|H_L)$ is the canonical Gibbs state with some inverse temperature $-\infty<\beta<\infty$.
Furthermore, it is said to represent $O(L^0)$-local MITE 
if it represents $\ell$-local MITE 
for all $\ell\in\mathbb{N}$.
\end{definition}

It is worth mentioning that a KMS state in infinite systems will correspond to the thermodynamic limit of a state in $O(L^0)$-local MITE, as it corresponds to the thermodynamic limit of the Gibbs state~\cite{Haag1996,Bratteli2002,Bratteli2002a}.

Our iMATE includes the above-defined $O(L^0)$-local MITE:
\begin{proposition}[MITE $\subset$ iMATE]
If a sequence of density matrices $(\rho_{L})_{L\in\mathbb{N}}$ represents $O(L^0)$-local MITE 
then it represents iMATE.
\end{proposition}
\begin{proof}
From Eq.~\eqref{eq:condition_MITE}, we have
\begin{align}
    \lim_{L\to\infty}\bigl\|\rho_{L|\ell}^{(k)}-\sigma_{L|\ell}^{(k)} \bigr\|_{1}=0
\end{align}
for all $\ell\in \mathbb{N}$.
From Proposition~\ref{proposition:Equiv_rho^red}, this means $(\rho_{L})_{L\in\mathbb{N}}\maceq \bigl(\rho^\mathrm{can}_{L}(\beta|H_L)\bigr)_{L\in\mathbb{N}}$, indicating that $(\rho_{L})_{L\in\mathbb{N}}$ represents iMATE.
\end{proof}
\noindent
On the other hand, 
we have seen in Sec.~\ref{sec:MacroEquilibrium} that there are thermal equilibrium states that are iMATE but not MITE.
Therefore, MITE serves as a narrower definition of thermal equilibrium than the iMATE. 
We have also seen there that
iMATE has a stronger connection to thermodynamics than MITE, as iMATE is characterized by additive observables, which correspond to additive quantities in thermodynamics, while MITE is characterized by local observables.
Moreover, MITE crucially depends on the choice of boundary conditions or the existence of only a single impurity; the Gibbs states under different boundary conditions or with different positions of few impurities have different expectation values of local observables near the boundary or the impurities, meaning that they represent different MITE.
By contrast, they represent the same iMATE. This means that iMATE reflects the insensitivity of thermodynamics and equilibrium statistical mechanics to the choice of boundary conditions and the position of few impurities, while MITE does not.
For these reasons, iMATE is more appropriate than MITE
when discussing the consistency with thermodynamics.

Note that 
iMATE is characterized by 
an $L$-independent number of observables, in contrast to MITE.
To explain this point, let us compare the number of independent
additive observables composed of $\ell$-local observables with the number of independent $\ell$-local observables. 
The former
is bounded from above by $K^dD^{2\ell^d}$, which is independent of $L$, while 
the latter
is bounded from below by $(D^2-1)^{\ell^d}N$. 
Since we take the thermodynamic limit $L\to\infty$ before taking $\ell$ large, 
the latter
grows much faster than 
the former.
In this sense, verification of iMATE requires investigating a much
smaller number of observables than $O(L^0)$-local MITE. 

We mention that
several studies consider the following class of states 
that is even narrower than the $O(L^0)$-local MITE:
Let $(\ell_{L})_{L\in\mathbb{N}}$ be a sequence of positive integers that diverges slower than $\Theta(L)$, that is, it satisfies $\ell_{L}=\omega(L^0)$ and $\ell_{L}=o(L)$, such as $\ell_L=\lfloor\sqrt{L}\rfloor$~\footnote{We do not know exactly how fast the divergence of $\ell_{L}$ can be. However, at least when $\ell_{L}=\Theta(L)$, the local equivalence between the microcanonical and the canonical Gibbs states breaks down~\cite{Garrison2018}, which is outside the scope of our interest. Therefore, $\ell_{L}=o(L)$ should be imposed. }. 
We say that a sequence of density matrices $(\rho_{L})_{L\in\mathbb{N}}$ represents $\ell_L$-local MITE if it satisfies
\begin{align}
    \lim_{L\to\infty}\max_{\bm{r}\in\Lambda_{L}}\bigl\|\rho_{L|\ell_L}^{\bm{r}}-\sigma^{\bm{r}}_{L|\ell_L}\bigr\|_{1}=0,
    \label{eq:condition_MITE_LargerLength}
\end{align}
where $\sigma_L=\rho_{L}^{\mathrm{can}}(\beta|H_L)$ is the canonical Gibbs state with some inverse temperature.
This condition is stronger than Eq.~\eqref{eq:condition_MITE} for any $\ell\in\mathbb{N}$, and hence $(\rho_L)_{L\in\mathbb{N}}$ also satisfies the condition for $O(L^0)$-local MITE.
Although there are some studies~\cite{Brandao2015Equivalence,Tasaki2018,Garrison2018,sugimoto2023bounds} investigating observables with $\omega(L^0)$ locality, which would be related to Eq.~\eqref{eq:condition_MITE_LargerLength}, most of the existing studies~\cite{Kim2014,Beugeling2014} explore $\ell$-local observables with $\ell=O(L^0)$, which is 
consistent with Eq.~\eqref{eq:condition_MITE}. 
In addition, as mentioned above, the KMS condition corresponds to $O(L^0)$-local MITE. This also suggests that $\ell_L$-local MITE is less established than $O(L^0)$-MITE.

\subsection{\label{sec:OrdinaryMATE}iMATE vs. the ordinary MATE}

In this subsection, we briefly review the ordinary notion of MATE~\cite{Neumann1929,Goldstein2015,Goldstein2017,Tasaki2016,Mori2018} and compare it with our iMATE.

The notion of MATE was first introduced by von Neumann~\cite{Neumann1929}; it was later termed MATE by Goldstein \textit{et al.}~\cite{Goldstein2015}. Here, we explain the notion of MATE, following the review article by Mori \textit{et al.}~\cite{Mori2018}. First, they take a fixed set of ``macroobservables'' $\{M^{1}_{L},...,M^{m}_{L}\}$. As examples, they enumerated particle number, energy, momentum, and magnetization on each macroscopic subsystem. Next, they introduced approximation of these macroobservables $\{\tilde{M}_{L}^{1},...,\tilde{M}_{L}^{m}\}$ that commute with each other. Note that although the existence of such approximations is quite nontrivial, when ``macroobservables'' $M_{L}^{1},...,M_{L}^{m}$ (whose definition is not exactly specified in
Refs.~\cite{Neumann1929,Goldstein2015,Goldstein2017,Tasaki2016,Mori2018}) are taken as additive observables and $m$ is fixed to a finite value, the existence is proved mathematically by Ogata~\cite{Ogata2013}. Lastly, they say a quantum state represents MATE if, in the joint distribution of $\tilde{M}_{L}^{1},...,\tilde{M}_{L}^{m}$ in that state, $\tilde{M}_{L}^{\mu}$'s are sufficiently concentrated at the equilibrium values for all $\mu=1,...,m$.

Our iMATE shares a key motivation with the ordinary MATE to focus on observables that are macroscopic in some sense; however, there is a significant difference in the number of observables of interest. As Ref.~\cite{Mori2018} says that $\{M^{1}_{L},...,M^{m}_{L}\}$ is fixed, the number of macroobservables $m$ is fixed 
independently of $L$
in the ordinary MATE. Another formulation of MATE by Tasaki~\cite{Tasaki2016} also requires that ``$m$ is not too large, and is independent of $N$''. Furthermore, these are consistent with the condition for applying the mathematical result by Ogata~\cite{Ogata2013} that $m$ is fixed to a finite value. By contrast, in our formulation of iMATE, the set of (sequences of) additive observables $\SetAdditive$, which we focus on, is an infinite set. Moreover, as shown by Example~\ref{example:ProblemFiniteObs}, focusing only on a finite number of additive observables will cause a contradiction to the second law of thermodynamics.

We also comment on
another difference:
The ordinary MATE requires the concentration in the distribution of macroobservables, while our iMATE focuses only on the expectation values of additive observables. However, we consider this difference noncritical, as our normal iMATE assumes small variances for additive observables, which has a meaning close to the concentration.

\subsection{\label{sec:Hokkyo}Comparison with the results by Hokkyo and Ueda}

Hokkyo and Ueda~\cite{Hokkyo2025} introduced ``thermodynamic passivity," 
which can be compared to our macroscopic passivity (Corollary~\ref{corollary:macroscopic-passivity}).

They showed thermodynamic passivity for states representing MITE.
[More precisely, they assumed the $\ell_L$-local MITE, Eq.~\eqref{eq:condition_MITE_LargerLength}, for some $\ell_L$ satisfying $\ell_{L}=\omega(L^0)$; this is  a stronger assumption than the $O(L^0)$-local MITE, Eq.~\eqref{eq:condition_MITE}, as explained in Sec.~\ref{sec:MITE}.]
Our macroscopic passivity (Corollary~\ref{corollary:macroscopic-passivity}) is also valid for states representing iMATE, including typical METTS~\eqref{eq:METTS} and Examples~\ref{example:NonMITE_MacroEqState} and \ref{example:NonMITE_MacroEqState_FiniteTemp}, to which their results cannot be applied. We consider this difference crucial because, as mentioned in Sec.~\ref{sec:MITE}, MITE has less connection to thermodynamics than iMATE. 
Therefore, 
applicability for iMATE is not a naive extension of MITE but a significant step toward deriving thermodynamic laws solely from quantum mechanics.

We also point out that
they showed thermodynamic passivity by assuming that the inverse temperature of the initial state is greater than some positive value independent of the system size.
On the other hand, our proof based on the Lieb-Robinson bound applies to any nonnegative inverse temperature, including $\beta=0$.
In particular, at $\beta=0$, we have found that macroscopic passivity can be extended to a more general form: the expectation value of not only Hamiltonian but any additive observable cannot be increased nor decreased extensively through any macroscopic operation with a finite operation time (Corollary~\ref{corollary:MacroPassivity_beta0}). 

Note that our contributions 
cover far beyond the passivity;
we have proved the law of increasing entropy (Theorems~\ref{theorem:EntropyIncrease} and \ref{theorem:EntropyIncrease_WithoutThermalization}),
which is quantum-mechanically inequivalent to the macroscopic passivity.
Furthermore, we have introduced many key concepts, including macroscopic equivalence, iMATE, macroscopic operations, and quantum macroscopic entropy density, and obtained rich results,
such as
preservation of macroscopic equivalence (Theorem~\ref{theorem:JAIVTMXN_3}), 
and the
entropy formula (Theorem~\ref{theorem:s^red=s^TD}).
We also give plenty of examples,
including critical examples in Sec.~\ref{sec:Counterexamples_Timescale},
which
show that 
the operation times in Ref.~\cite{Hokkyo2025} and ours are impossible to extend.

\subsection{\label{sec:Meier}Comparison with the results by Meier, \textit{et al.}}

Meier \textit{et al.} recently studied equilibration of the ``Shannon observable entropy'' $S_{\mathrm{Sh}}^{O}$ in isolated quantum systems \cite{Meier2025}. Here, $S_{\mathrm{Sh}}^{O}$ is defined as the Shannon entropy of 
the probability distribution of the measurement outcome of a single observable~$O$.
Under certain assumptions,
they 
found that 
the deviation of $S_{\mathrm{Sh}}^{O}$ from its value in the stationary state
is sufficiently small for typical time $t\in[0,T]$ when $T$ is sufficiently large~\footnote{Note that, as discussed in Refs.~\cite{Short2012,deOliveira2018}, $T$ should be taken exponentially large with respect to $N$ to make their upper bound meaningful.}.
This implies that $S_{\mathrm{Sh}}^{O}$ increases for typical time.
We raise several crucial differences between their results and ours.

First, we are interested in resolving the apparent inconsistency of quantum mechanics and thermodynamics 
that occurs when the thermal equilibrium state is represented not by Gibbs states but by a general quantum state.
Therefore, our quantum macroscopic entropy has been introduced such that it agrees with thermodynamic entropy for any quantum state representing iMATE. In contrast, $S_{\mathrm{Sh}}^{O}$ has no general connection to thermodynamic entropy except when the state is taken as a Gibbs state (with $O$ taken as the Hamiltonian).

Second, 
while our results (Theorems~\ref{theorem:EntropyIncrease} and \ref{theorem:EntropyIncrease_WithoutThermalization}) apply to any initial state representing thermal equilibrium,
they focused on the initial state whose
effective dimension is sufficiently 
larger than the number of possible measurement outcomes of $O$
(which typically grows as $\Omega(N)$~\cite{LandauSymbols} when $O$ is taken as an additive observable such as 
the magnetization or the Hamiltonian).
This means that 
energy eigenstates,
which are one of our primary interests as discussed in Sec.~\ref{sec:introduction},
were excluded
even if they represent thermal equilibrium.

Third, 
they considered isolated systems while we consider closed systems: their time evolution is governed by a \emph{time-independent} Hamiltonian in contrast to ours. In other words, they focused solely on 
relaxation 
in the absence of operations,
whereas we are interested in the 
more general case where the systems are
subject to 
\textit{time-dependent} adiabatic operations.

Finally, 
$S_{\mathrm{Sh}}^{O}$ was defined for a single observable $O$;
extending this to multiple noncommutative observables appears nontrivial,
since the spectral decomposition of $O$ was 
used to define $S_{\mathrm{Sh}}^{O}$.
By contrast, our quantum macroscopic entropy density $s_{\ell}^{\mathrm{mac}}[\bullet]$ is defined using $\rho_{L|\ell}^{(k)}$, which contains complete information about the expectation values of \textit{all} additive observables composed of $\ell$-local observables.

\section{\label{sec:Counterexamples_Timescale}Counterexamples by 
longer-timescale operations}

We have proved two forms of the second law of thermodynamics, i.e., the macroscopic passivity (Corollary~\ref{corollary:macroscopic-passivity}) and the law of increasing entropy (Theorems~\ref{theorem:EntropyIncrease} and \ref{theorem:EntropyIncrease_WithoutThermalization}).
They are proved for any macroscopic operations with the operation 
time $t^*$ of $O(L^0)$.
In this section, we show that this timescale of $t^*$ is the longest 
for these general results, by showing that there exist counterexamples when $t^*$ is longer.

\subsection{Counterexample to the macroscopic passivity}

We first show that a naive extension 
of the macroscopic passivity (Corollary~\ref{corollary:macroscopic-passivity}) 
to longer operation times is impossible:
\begin{proposition}[Counterexample to macroscopic passivity at any longer timescale] \label{proposition:breakdown_macroscopic-passivity}
Consider a system on a lattice of arbitrary spatial dimension with local Hilbert spaces $\cong \mathbb{C}^4$.
For any operation time $T_L$
satisfying both $T_L=\omega(L^0)$ and $T_L=O(L)$,
there exist a Hamiltonian $\initH$,
a sequence of microscopic states $(\rho_L)_{L\in\mathbb{N}}$ representing iMATE described by $\initH$ at the inverse temperature $\beta=0$, and a macroscopic operation $(U_L(\bullet,0))_{L\in\mathbb{N}}$ such that
\begin{align}
    \lim_{L \to\infty} \mathrm{Tr} \left[ \rho_L(T_L) \initH / N \right]
    < \lim_{L \to \infty} \mathrm{Tr} \left[ \rho_L(0) \initH / N \right],
    \label{eq:breakdown_macroscopic-passivity}
\end{align}
where $\rho_L(t) = U_L(t, 0) \rho_L {U_L(t, 0)}^\dagger$.
\end{proposition}
\begin{proof}[Proof outline]
We prove this proposition by an explicit construction of $(\rho_L)_{L\in\mathbb{N}}$, $H_L$, and $(U_L(T_L,0))_{L\in\mathbb{N}}$.
We take $\rho_L$ as a product of $L/2$ Bell pairs. Each pair consists of two spins separated by a distance $\simeq 2T_L$. This implies that $(\rho_L)_{L\in\mathbb{N}}$ represents iMATE at $\beta=0$ (for any $H_L$). We choose the time-dependent Hamiltonian $H_L(t)$ such that $U_L(T_L,0)$ becomes a repetition of swap gates for $\simeq T_L$ times. This maps $\rho_L(T_L)$ to a product of nearest-neighbor Bell pairs. Choosing $H_L$ appropriately, we obtain Eq.~\eqref{eq:breakdown_macroscopic-passivity}.
For the details, see Appendix~\ref{sec:breakdown_macroscopic-passivity}.
\end{proof}
\noindent
Therefore, the timescale presented in the Corollary~\ref{corollary:macroscopic-passivity} is the longest one that can be proven in the physical setting considered in this paper.

We emphasize that this upper bound of the timescale is not peculiar to
our definition of thermal equilibrium, iMATE.
In fact, the sequence of microscopic states used in the construction of the above example is in MITE,
and,
since macroscopic operation $\subset$ local control (Definition~\ref{definition:LocalControl}),
the macroscopic operation is 
also local control. 
That is, the failure of passivity at timescales beyond $O(L^0)$ already occurs within 
the setup of ``MITE $+$ local control'' considered, e.g., by Hokkyo and Ueda~\cite{Hokkyo2025},
and is not specific to
our setup of ``iMATE $+$ macroscopic operation''.

\subsection{Counterexample to the law of increasing entropy}

We then show that a naive extension 
of the law of increasing entropy (Theorems~\ref{theorem:EntropyIncrease} and \ref{theorem:EntropyIncrease_WithoutThermalization}) 
to longer operation times is impossible:

\begin{proposition}[\label{proposition:breakdown_entropy}Counterexample to the law of increasing entropy at a longer timescale]
    Consider one-dimensional spin-$1/2$ systems, i.e., the case of $D=2$ and $d=1$.
    For arbitrary integer $T_L$ that scales as $T_L=\omega(L^0)$ and $T_L=O(L)$, there exist initial and final Hamiltonians $H_0$ and $H_1$, a sequence of the initial states $(\rho_L)_{L\in\mathbb{N}}$ representing iMATE at infinite temperature $\beta=0$, and a macroscopic operation $(U_L(\bullet,0))_{L\in\mathbb{N}}$ such that
    \begin{align}
        \lim_{\ell\to\infty}s^{\mathrm{mac}}_{\ell}[(\rho_{L})_{L\in\mathbb{N}}]&=\log 2\notag\\
        &>\lim_{\ell\to\infty}s^{\mathrm{mac}}_{\ell}[(\sigma_{L})_{L\in\mathbb{N}}]=0,
        \label{eq:Counterex_IncreasingEntropy}
    \end{align}
    where $(\sigma_L)_{L\in\mathbb{N}}$ is a sequence of the final state evaluated by the long-time average~\eqref{eq:LongTimeAve},
    \begin{align}
        \sigma_{L}:=\overline{\rho_L(t>T_L)},
        \label{eq:Counterex_Entropy_FinalState}
    \end{align}
    of the state at $t=T_L$, with $\rho_{L}(T_L):=U_L(T_L,0)\rho_L U_L^\dagger(T_L,0)$.
\end{proposition}
\begin{proof}[Proof outline]
We prove this proposition by an explicit construction of $(\rho_L)_{L\in\mathbb{N}}$, $H_0$, $H_1$, and $(U_L(T_L,0))_{L\in\mathbb{N}}$.
We take $\rho_L$ and $U_L(T_L,0)$ as the same ones as Proposition~\ref{proposition:breakdown_macroscopic-passivity}. Here, note that the choice of $H_0$ is unimportant. As mentioned in the proof outline of Proposition~\ref{proposition:breakdown_macroscopic-passivity}, $\rho_L(T_L)$ is a product of nearest-neighbor Bell pairs. This state is known as the ground state of the Majumdar-Ghosh model~\cite{Majumdar1969}. By choosing $H_1$ as the Majumdar-Ghosh Hamiltonian, $\rho_L(T_L)$ becomes stationary in the relaxation process governed by $H_1$, and hence $\sigma_{L}=\rho_L(T_L)$ holds. Consequently, we obtain Eq.~\eqref{eq:Counterex_IncreasingEntropy}.
For the details, see Appendix~\ref{sec:breakdown_entropy}.
\end{proof}
\noindent
In this example, through a macroscopic operation of operation time $T_L=\omega(L^0)$, the quantum macroscopic entropy density decreases from the maximum value $\log 2$ to the minimum value $0$.
This means that if we permit a macroscopic operation of operation time $\omega(L^0)$, we can construct a counterexample of Theorems~\ref{theorem:EntropyIncrease} and \ref{theorem:EntropyIncrease_WithoutThermalization} with the maximum violation.
In addition, as shown in Appendix~\ref{sec:breakdown_entropy}, $\rho_L(T_L)$ is the ground state of $H_1$, indicating that the macroscopic operation induces state transition from $\beta=0$ to $\beta\to\infty$.

\section{Discussion}\label{sec:Discussion}

\subsection{Perspectives for longer-timescale operations}\label{sec:longer_time_scale}

As shown in the previous section, 
the timescale $O(L^0)$ of the operation time 
in our results corresponding to the second law
(Corollary~\ref{corollary:macroscopic-passivity}, 
Theorems~\ref{theorem:EntropyIncrease} and \ref{theorem:EntropyIncrease_WithoutThermalization})
is optimal in the sense that the straightforward extension
to a longer timescale is impossible.

However, in thermodynamics, we sometimes encounter operations completed on longer timescales.
For instance, consider the extraction of work from a fluid system by moving a piston quasistatically. To extract an amount of energy of $\Theta(L^d)$, which remains relevant in the thermodynamic limit, the piston must be moved by a distance of $\Theta(L)$. On the other hand, since the local
relaxation time of the fluid is independent of the system size, the piston speed must be $O(L^0)$ to ensure that the process remains quasistatic. Consequently, such an operation inevitably requires a time of $\Omega(L)$~\cite{LandauSymbols} and cannot be completed within the timescale of $O(L^0)$.

To extend the consistency with thermodynamics to 
such a longer timescale, we need to introduce some additional condition(s).
We here list three candidates:
\begin{enumerate}[label={(\roman*)},ref={\roman*}]
    \item \label{enum:LongerTimescaleCandidate_PreciseState} It is impossible to precisely prepare a quantum state exactly as specified, such as the pure state~\eqref{eq:breakdown_InitialState} in our counterexample.
    \item \label{enum:LongerTimescaleCandidate_PreciseOperation} It is impossible to precisely execute a macroscopic operation exactly as specified, such as the operation generated by the Hamiltonian~\eqref{eq:breakdown_Hamiltonian} in our counterexample.
    \item \label{enum:LongerTimescaleCandidate_ComplicatedState} When discussing state transitions from thermal equilibrium, we need to restrict quantum states to a subset of iMATE. The subset consists only of quantum states with complex structures, such as the typical energy eigenstates of nonintegrable systems.
\end{enumerate}

In candidates (\ref{enum:LongerTimescaleCandidate_PreciseState}) and (\ref{enum:LongerTimescaleCandidate_PreciseOperation}), `impossible' means that the success probability $\to 0$ as $L \to \infty$.
Candidate (\ref{enum:LongerTimescaleCandidate_ComplicatedState}) sounds reasonable because state transitions 
are \emph{nonequilibrium} processes, while we have defined iMATE to be fully consistent with \emph{equilibrium} statistical mechanics.
For example, while an integrable system has an equilibrium state, it
does not thermalize after quench and hence is abnormal thermodynamically.
Therefore, it seems natural to restrict quantum states
when discussing (non-quasistatic) thermodynamic processes.
Further investigation will be a subject of future research.

Interestingly,
when a longer timescale is considered without imposing these additional
conditions, our counterexamples
provide a rigorous realization of a violation of the second law of thermodynamics.
We stress that our examples do not rely on the measurement-and-feedback approach, which has been widely employed to construct counterexamples to the second law since Maxwell's pioneering thought experiment~\cite{Maxwell1867}. Rather, they are similar to the spirit of Loschmidt's paradox, stating that the second law seems violated because of the reversibility of microscopic dynamics. 
Given the state-of-the-art controllability of many-body dynamics in quantum simulator and computers, it is interesting to experimentally observe the violation of the second law predicted by our counterexamples using these platforms.
Such experiments will also confirm the necessity and importance of the additional condition(s) mentioned above.

\subsection{$K$ dependence}\label{sec:Kdependence}

As explained in Sec.~\ref{sec:Setup_Lattice},  we have taken the partition number $K$ to be a large integer independent of $L$, $K=O(L^0)$.
Physically, our results are expected to be independent of $K$.
An easy way of seeing this independence explicitly is to take 
the $K\to\infty$ limit after taking the thermodynamic limit $L\to\infty$, which would correspond to the situation in which continuum mechanics applies.
When taking this limit, we make the following technical adjustment.

Naively, one may expect that taking $K$ larger means focusing on more observables and considering more operations, since the side length of the primitive macroscopic subsystems, $L/K$, becomes smaller.
However, if we straightforwardly took $K$ larger, the set of (sequences of) additive observables would not necessarily get broader.
For instance, additive observables with $K=2$, in general, cannot be written as additive observables with $K=3$. To avoid such complexities, we can take $K$ as \footnote{
We can also take $K$ smaller. For example, when 
$M^{\mathrm{max}}=6$ it is sufficient to take 
$K=60$ instead of $K=6 \, !$.
}
\begin{align}
K=M^{\mathrm{max}} \, !,
\label{eq:choice_of_K}
\end{align}
where $M^{\mathrm{max}}$ is a positive integer.
To achieve the $K$ independence, 
we take $M^{\mathrm{max}}\to\infty$ after taking $L\to\infty$.

The above form of $K$ 
allows us to consider additive observables on macroscopic subsystems of side length $L/M$ 
for {\emph{all} $M\le M^{\mathrm{max}}$ \footnote{
This property is the physical motivation of the adjustment~\eqref{eq:choice_of_K}. Since the same property  
is not necessarily satisfied for \emph{all} $M$ such that $M^{\mathrm{max}}<M \leq M^{\mathrm{max}} \, !$, 
we focus on the case of $M\le M^{\mathrm{max}}$.}
}.
As a result, we can deal with macroscopic operations controlling the fields conjugate to additive observables on macroscopic subsystems of side length $L/M$ for any $M\le M^{\mathrm{max}}$.
Furthermore, in iMATE with 
$K=M^{\mathrm{max}} \, !$, 
any additive observables on macroscopic subsystems of side length $L/M$ take the equilibrium values independently of $M$, as long as $M\le M^{\mathrm{max}}$.

In this way, the adjustment \eqref{eq:choice_of_K} resolves the above-mentioned complexities, and the limit $M^{\mathrm{max}}\to\infty$ (after taking $L\to\infty$) 
makes our theory independent of $K$, as physically expected.

In the above argument, the existence of 
the $M^{\mathrm{max}}\to\infty$ limit of the quantum macroscopic entropy density seems nontrivial.
We finally address this point by showing that 
the limit indeed exists for any macroscopic state.

Given a sequence of states $(\rho_{L})_{L\in\mathbb{N}}$. Suppose that it represents a macroscopic state, 
defined by Def.~\ref{definition:MacroState}, 
for any $K$ of the form~\eqref{eq:choice_of_K} and for any $M^{\mathrm{max}}\in\mathbb{N}$. This can be rephrased that $(\rho_{L})_{L\in\mathbb{N}}$ satisfies Eq.~\eqref{eq:MacroState} for any $K\in\mathbb{N}$~\footnote{This is because Eq.~\eqref{eq:MacroState} for $K=(M^{\mathrm{max}})!$ implies Eq.~\eqref{eq:MacroState} for $K=M^{\mathrm{max}}$.}.
To indicate the $K$ dependence of the quantum macroscopic entropy density $s^{\mathrm{mac}}_{\ell}[(\rho_{L})_{L\in\mathbb{N}}]$, defined by Eq.~\eqref{eq:DEF_s^red_ell}, 
we add a superscript $K$ to it and write it as $s^{\mathrm{mac},K}_{\ell}[(\rho_{L})_{L\in\mathbb{N}}]$ 
in the remainder of this subsection.
Using the concavity of the von Neumann entropy, we can show
\begin{align}
    s^{\mathrm{mac},(M^{\mathrm{max}})!}_{\ell}[(\rho_{L})_{L\in\mathbb{N}}]
    \ge s^{\mathrm{mac},(M^{\mathrm{max}}+1)!}_{\ell}[(\rho_{L})_{L\in\mathbb{N}}].
    \label{eq:Kdependence_concavity_s^macK}
\end{align}
This means that 
$s^{\mathrm{mac},(M^{\mathrm{max}})!}_{\ell}[(\rho_{L})_{L\in\mathbb{N}}]$ is monotonically decreasing with respect to $M^{\mathrm{max}}$. Furthermore, since it is nonnegative, there exists its $M^{\mathrm{max}}\to\infty$ limit, and we obtain the desired result: 
\begin{proposition}[\label{proposition:Kindependent_s^mac}The existence of $K\to\infty$ limit of quantum macroscopic entropy density]
Given a sequence of states $(\rho_{L})_{L\in\mathbb{N}}$ that represents a macroscopic state for any partition number $K\in\mathbb{N}$, which is used to define additive observables in Def.~\ref{definition:Additive}. The quantum macroscopic entropy density of $(\rho_{L})_{L\in\mathbb{N}}$ has the $M^{\mathrm{max}}\to\infty$ limit,
\begin{align}
    \lim_{M^{\mathrm{max}}\to\infty}s^{\mathrm{mac},(M^{\mathrm{max}})!}_{\ell}[(\rho_{L})_{L\in\mathbb{N}}].
\end{align}
\end{proposition}
\noindent
In other words, $s^{\mathrm{mac},(M^{\mathrm{max}})!}_{\ell}[(\rho_{L})_{L\in\mathbb{N}}]$ becomes almost independent of $M^{\mathrm{max}}$ as $M^{\mathrm{max}}$ is taken large. 
\begin{proof}[Proof of Eq.~\eqref{eq:Kdependence_concavity_s^macK}]
See Appendix~\ref{sec:Proof_Kdependence}.
\end{proof}

\section{\label{sec:summary}Summary}

For closed quantum many-body systems of arbitrary spatial dimension $d$ obeying the unitary time evolution,
we have proposed and proved two forms of 
the second law of thermodynamics:
the macroscopic passivity and the law of increasing entropy,
which are quantum-mechanically inequivalent to each other.

To this end, we first 
introduced \emph{macroscopic subsystems}   
in order to detect macroscopic nonuniformity of 
nonequilibrium states (Sec.~\ref{sec:Setup}).
We then gave clear definitions of \emph{additive observables} and \emph{$\ell$-local observables}.

To formulate the thermodynamic limit rigorously, we have considered sequences of quantum states (Sec.~\ref{sec:MacroEquiv}).
We say that such a sequence 
represents a \emph{macroscopic state} if the expectation values
of densities of all additive observables converge in the thermodynamic limit.
Importantly, we introduced a notion of \emph{macroscopic equivalence};
we say that two quantum-mechanical
representations of macroscopic states are macroscopically equivalent if
densities of all additive observables coincide in the thermodynamic limit
(Definition~\ref{definition:MacroEquiv}).
This notion forms the foundation for the results in this paper.

We then introduced 
\emph{infinite-observable macroscopic thermal equilibrium (iMATE)}
as a notion of thermal equilibrium states (Definition~\ref{definition:MacroEqState}, Sec.~\ref{sec:MacroEquilibrium}).
It corresponds to an extension of an ordinary notion called macroscopic thermal equilibrium (MATE) such that all, thus an infinite number of, additive observables should be consistent with thermodynamics.
This resolves possible contradictions of MATE against thermodynamics (Example~\ref{example:ProblemFiniteObs}, Sec.~\ref{sec:MacroPassivity}). 
Furthermore, our iMATE includes 
well-known quantum states representing equilibrium states,
 such as the minimally entangled typical thermal states 
(Example~\ref{example:TypicalMETTS_represents_iMATE}, Sec.~\ref{sec:Example_iMATE}),
which are not included in another ordinary notion called microscopic 
thermal equilibrium (MITE).
Therefore, we employ iMATE as quantum-mechanical representations of equilibrium states.

We also introduced \emph{macroscopic operations} as unitary time evolutions generated by a time-dependent Hamiltonian that is additive and whose time dependence is governed by a finite number of parameters conjugate to additive observables (Definition~\ref{definition:macroscopic-operation}, Sec.~\ref{sec:MacroOp_MacroEquiv}). 
We then proved that macroscopic equivalence is preserved after any macroscopic operations 
(Theorem~\ref{theorem:JAIVTMXN_3}).
While there are many different quantum states that represent the same iMATE,
this proposition guarantees that all such states exhibit the same 
responses of additive observables to macroscopic operations,
in consistency with thermodynamics (Corollary~\ref{corollary:macroscopically-equivalence_equilibrium-state}).
We also derived a lower bound on the timescale of thermalization (Corollary~\ref{corollary:lower_bound_thermalization}).

We then proposed and proved our first form of quantum-mechanical expressions of the second law, which we call the \emph{macroscopic passivity}
(Corollary~\ref{corollary:macroscopic-passivity}, 
Sec.~\ref{sec:MacroPassivity}). 
It states that 
if the system is prepared in an arbitrary state representing iMATE described by an initial Hamiltonian with nonnegative inverse temperature, no macroscopic operation can decrease the expectation value of the initial Hamiltonian extensively.
In other words, extensive work cannot be extracted from states representing iMATE by any macroscopic operation. 
We also proved that iMATE at infinite temperature, $\beta=0$, 
is stable against any macroscopic operations (Corollary~\ref{corollary:MacroPassivity_beta0}).

In order to proceed to our second expression of the second law, 
we introduced the \emph{$\ell$-local density matrices} 
and their spatial average (Definition~\ref{def:sp.av.ell.dens.mat}, Sec.~\ref{sec:reduced_rho}).
We proved that all macroscopically equivalent representations of a macroscopic state are characterized only by the spatial averages of $\ell$-local density matrices.
Using this result, we unambiguously defined the notion of
\emph{macroscopically uniform} states 
(Definition~\ref{definition:MacroUniform}) and proved that 
the macroscopic uniformity depends only on
a macroscopic state but not on its representation (Corollary~\ref{cor:mac.un.dep.only.on.mcrostate}).
In particular, 
any representation of an iMATE represents a macroscopically uniform state,
under our assumption that the system is not at a
first-order phase transition point
(Corollary~\ref{corollary:MacroscopicallyUniform_iMATE}).
For comparison, we also defined \emph{locally uniform} macroscopic states
(Definition~\ref{definition:LocallyUniform}).
Many other quantum-mechanical definitions of 
thermal equilibrium (such as the Gibbs states and MITE) assumed 
locally uniform macroscopic states,
which are a strictly narrower class than macroscopically uniform states.

Using the spatial average of $\ell$-local density matrices, 
we introduced the \emph{quantum macroscopic entropy density}
(Definition~\ref{definition:MacroEntropy}, Sec.~\ref{sec:Entropy_s^mac}),
and proved that it takes the same value among equivalent macroscopic states
(Corollary~\ref{corollary:Equiv_s^red}).
Most importantly, 
we proved that it agrees with the thermodynamic entropy density for all representations of iMATE (Theorem~\ref{theorem:s^red=s^TD}, Sec.~\ref{sec:Entropy_s^mac_eq}).
This should be contrasted with many other quantum-mechanical entropies, 
such as the von Neumann entropy of the whole system and 
the 
entanglement entropy, 
which do not necessarily agree with the thermodynamic entropy
(subsection~\ref{sec:numerical_entropy} and Example~\ref{example:HalfChainEntanglement}).
Interestingly, our result means that 
the thermodynamic entropy 
can be obtained by measuring appropriate additive observables 
without individual measurements of a huge number of local observables 
or thermodynamic measurements on multiple states along a 
thermodynamic process (subsection~\ref{sec:method.measure.sTD}).

Using these results, 
we proved the law of increasing entropy, 
which is our second form of quantum-mechanical expressions of the second law
(Theorems~\ref{theorem:EntropyIncrease} and \ref{theorem:EntropyIncrease_WithoutThermalization}, Sec.~\ref{sec:EntropyIncrease}).
It states 
that if the system is prepared in a state representing iMATE described by an initial Hamiltonian then the quantum macroscopic entropy density cannot be decreased by any finite-time macroscopic operation followed by an infinite-time relaxation process governed by a final Hamiltonian.
This corresponds to the thermodynamic statement which   
asserts that an adiabatic thermodynamic operation cannot decrease the thermodynamic entropy of the system.
Although it is sometimes said to be equivalent 
to the statement corresponding to our macroscopic passivity, 
the other laws of thermodynamics were assumed in such arguments.
By contrast, we have proved both forms independently from quantum mechanics.
Comparisons with the existing studies were made in 
Sec.~\ref{sec:ComparisonStudies}.

Our two forms of quantum-mechanical expressions of the second law
are proved for any macroscopic operations with the operation 
time $t^*$ 
independent of the system size.
We showed that this timescale of $t^*$ is the longest 
for these general results, by showing that there exist counterexamples when $t^*$ is longer (Sec.~\ref{sec:Counterexamples_Timescale}).
We pointed out that this upper bound of the timescale is not peculiar to
our definition of thermal equilibrium, iMATE, but rather is common 
to the previous study that assumed MITE. 
Since this restriction on the operation time seems 
unsatisfactory from the viewpoint of thermodynamics, 
we discussed possible candidates to 
remove the restriction
in Sec.~\ref{sec:longer_time_scale}.
On the other hand, 
it is interesting to experimentally realize
our counterexamples because they 
provide  
a violation of the second law of thermodynamics in the spirit of Loschmidt's paradox without relying on the measurement-and-feedback approach
(last paragraph of Sec.~\ref{sec:longer_time_scale}).
We believe that such experiments will lead to a deeper understanding 
of the second law.
We also pointed out that our results will be 
independent of the choice of the partition number $K$
in the iterated limit, $K \to \infty$ after $L \to \infty$
(sec~\ref{sec:Kdependence}).

\begin{acknowledgments}
We appreciate 
Honoe Kandabashi, Yuto Kojima, Koki Ono, Katsunori Sugiyama, Shun Umekawa, and Junxuan Wen
for pointing out logical errors in the first draft of this paper.
We also thank A.~Hokkyo, H.~Tasaki, T.~Mori, and K.~Takasan for fruitful discussions.
The MPS calculations in this work are performed using the ITensor library~\cite{itensor}.
This work was supported by JST ERATO Grant Number JPMJER2302, Japan. 
Y.C. and Y.Y. are supported by the Special Postdoctoral Researchers Program at RIKEN.
R.H. was supported by Japan Society for the Promotion of Science (JSPS) KAKENHI Grant No. JP24K16982.
A.S. was supported by the RIKEN TRIP initiative (HIKARI-COOL Tokyo)
and by Japan Society for the Promotion of Science (JSPS) KAKENHI Grant No. 23K22413. 
\end{acknowledgments}

\section*{Data Availability}
The data used to generate Fig.~\ref{fig:MacroscopicEntropydensity} are available from Ref.~\cite{Chiba2026SecondLaw_Data}.

\appendix
\onecolumngrid

\section{\label{sec:notation}Summary of notation and definitions}

This Appendix summarizes the notation and terminology used throughout the main text.
Table~\ref{tbl:notation_geometry} shows the notation related to the lattice, Hilbert spaces, macroscopic subsystems, and locality.
Table~\ref{tbl:notation_obs} shows the notation for translations, local/additive observables, and operator-theoretic conventions.
Table~\ref{tbl:notation_states_dynamics_entropy}
shows the notation for states, equilibrium notions, operations, reduced states, and entropies.

\begin{table*}[h]
\caption{Notation related to the lattice, Hilbert spaces, macroscopic subsystems, and locality.}
\label{tbl:notation_geometry}
\centering
\footnotesize
\begin{ruledtabular}
\begin{tabular}{p{0.25\textwidth}p{0.55\textwidth}p{0.16\textwidth}}
Symbol & Meaning / role & Where introduced \\
\hline
$d$ & Spatial dimension of the hypercubic lattice. & Sec.~\ref{sec:Setup_Lattice} \\
$L$ & Linear system size (side length). & Sec.~\ref{sec:Setup_Lattice} \\
$N=L^{d}$ & Total number of lattice sites. & Sec.~\ref{sec:Setup_Lattice} \\
$\mathbb{Z}_L=\{-\lceil L/2\rceil+1,\dots,\lfloor L/2\rfloor\}$ & Integer coordinates with periodic boundary conditions. & Sec.~\ref{sec:Setup_Lattice} \\
$\Lambda_L=(\mathbb{Z}_L)^d$ & Total lattice of size $L$. & Sec.~\ref{sec:Setup_Lattice} \\
$\bm r,\bm r',\bm 0$ & Lattice sites and the origin, $\bm r\in\Lambda_L$. & Sec.~\ref{sec:Setup_Lattice} \\
$|\bm r-\bm r'|$ & Manhattan distance on $\Lambda_L$ (with periodic identification). & Sec.~\ref{sec:Setup_Lattice} \\
$\mathcal{H}_{\bm r}$ & Local Hilbert space at site $\bm r$. & Sec.~\ref{sec:Setup_Lattice} \\
$D$ & Local Hilbert-space dimension; $\mathcal{H}_{\bm r}\cong\mathbb{C}^{D}$. & Sec.~\ref{sec:Setup_Lattice} \\
$\{\ket{j}_{\bm r}\}_{j=0}^{D-1}$ & Orthonormal basis of $\mathcal{H}_{\bm r}$. & Sec.~\ref{sec:Setup_Lattice} \\
$\mathcal{H}_{\Lambda_L}=\bigotimes_{\bm r\in\Lambda_L}\mathcal{H}_{\bm r}$ & Total Hilbert space on $\Lambda_L$. & Sec.~\ref{sec:Setup_Lattice} \\
$K$ & $L$-independent number of partitions in each 
spatial direction. & Sec.~\ref{sec:Setup_Lattice} \\
$K^{d}$ & Number of primitive macroscopic subsystems. & Sec.~\ref{sec:Setup_Lattice} \\
$\mathcal{S}^{(k)}_L$ & $k$-th primitive macroscopic subsystem ($k\in\{1,\dots,K^d\}$). & Def.~\ref{definition:ProperMacroSubsystem} \\
$\mathcal{S}_L$ & (General) macroscopic subsystem: union of primitive ones. & Def.~\ref{definition:ProperMacroSubsystem} \\
$m$ & Number of primitive macroscopic subsystems forming $\mathcal{S}_L$ ($m=1,\dots,K^d$). & Def.~\ref{definition:ProperMacroSubsystem} \\
$k,k',k_j$ & Labels of primitive macroscopic subsystems. & Def.~\ref{definition:ProperMacroSubsystem} \\
$|\mathcal{S}_L|$ & Number of sites in $\mathcal{S}_L$; scales as $\Theta(L^d)$ for proper sequences. & Def.~\ref{definition:ProperMacroSubsystem} \\
$\ell$ & Locality length (kept $L$-independent, i.e., $\ell=\Theta(L^{0})$). & Def.~\ref{definition:LocalObs} \\
$\Cell$ & $d$-dimensional hypercube (side length $\ell$) centered at $\bm 0$. & Def.~\ref{definition:LocalObs} \\
$\Cell[\bm r]=\Cell+\bm r$ & Translation of the reference cell by $\bm r$. & Def.~\ref{definition:LocalObs} \\
$O(\cdot),o(\cdot),\Theta(\cdot),\Omega(\cdot),\omega(\cdot)$ & Landau asymptotic symbols. & Ref.~\cite{LandauSymbols} \\
\end{tabular}
\end{ruledtabular}
\end{table*}

\begin{table*}[h]
\caption{Notation for translations, local/additive observables, and operator-theoretic conventions.}
\label{tbl:notation_obs}
\centering
\footnotesize
\begin{ruledtabular}
\begin{tabular}{p{0.25\textwidth}p{0.55\textwidth}p{0.16\textwidth}}
Symbol & Meaning / role & Where introduced \\
\hline
$T_{\bm r}: \bigotimes_{\bm r'}\ket{j_{\bm r'}}_{\bm r'}
\mapsto \bigotimes_{\bm r'}\ket{j_{\bm r'}}_{\bm r'+\bm r}$ & Lattice translation by $\bm r$ (unitary on $\mathcal{H}_{\Lambda_L}$). & Sec.~\ref{sec:Setup_Lattice} \\
$\aLocal$ & Local observable. & Defs.~\ref{definition:LocalObs} and \ref{definition:Additive} \\
$\aLocal_{\bm r}=T_{\bm r}\aLocal T_{\bm r}^{\dagger}$ & Translate of $\aLocal$ supported on $\Cell[\bm r]$. & Def.~\ref{definition:Additive} \\
$\supp(\cdot)$ & Support (set of sites) of an operator. & Def.~\ref{definition:Additive} \\
$A_{\mathcal{S}_L}(\aLocal)$ & Additive observable on $\mathcal{S}_L$ obtained from $\aLocal$. & Def.~\ref{definition:Additive} \\
$A_L$ & Generic additive observable (often $A_L=A_{\mathcal{S}_L}(\aLocal)$). & Def.~\ref{definition:Proper_Additive} \\
$\SetAdditive$ & Set of all proper sequences of additive observables. & Def.~\ref{definition:Proper_Additive} \\
$\|\cdot\|$ & Operator norm. & --- \\
$\|\cdot\|_{\mathrm{tr}}$ & Trace norm. & --- \\
\end{tabular}
\end{ruledtabular}
\end{table*}

\begin{table*}[h]
\caption{Notation for states, equilibrium notions, operations, reduced states, and entropies.}
\label{tbl:notation_states_dynamics_entropy}
\centering
\footnotesize
\begin{ruledtabular}
\begin{tabular}{p{0.25\textwidth}p{0.55\textwidth}p{0.16\textwidth}}
Symbol & Meaning / role & Where introduced \\
\hline
$\rho_L,\sigma_L$ & Density matrices on $\mathcal{H}_{\Lambda_L}$. & Sec.~\ref{sec:MacroEquiv} \\
$\bigl(\rho_L\bigr)_{L\in\mathbb{N}}$ & A sequence of states considered in the thermodynamic limit. & Def.~\ref{definition:MacroState} \\
Macroscopic state & Sequences for which $\mathrm{Tr}[\rho_L A_L/N]$ converges for all $(A_L)\in\SetAdditive$. & Def.~\ref{definition:MacroState}, Eq.~\eqref{eq:MacroState} \\
Macroscopic equivalence $\maceq$ & Equivalence relation between macroscopic states. & Def.~\ref{definition:MacroEquiv}, Eq.~\eqref{eq:MacroscopicEquiv} \\
Normal macroscopic state & Macroscopic state with macroscopically negligible variance for all $(A_L)\in\SetAdditive$. & Def.~\ref{definition:MacroState_Normal}, Eq.~\eqref{eq:MacroscopicallyNormal} \\
$H_L$ & Hamiltonian. & Eq.~\eqref{eq:Hamiltonian} \\
$h$ & Local interaction term of the Hamiltonian. & Eq.~\eqref{eq:Hamiltonian} \\
$\beta$ & Inverse temperature. & Def.~\ref{definition:MacroEqState} \\
$\rho_L^{\mathrm{can}}(\beta|H_L)$ & Canonical Gibbs state. & Eq.~\eqref{eq:DEF_rho^can} \\
iMATE & ``Infinite-observable macroscopic thermal equilibrium''. & Def.~\ref{definition:MacroEqState}, Eq.~\eqref{eq:MacroEqState} \\
Normal iMATE & iMATE that is also a normal macroscopic state. & Def.~\ref{definition:MacroEqState_Normal} \\
$U_L(t,0)$ & Unitary time evolution operator for a macroscopic operation. & Def.~\ref{definition:macroscopic-operation} \\
$H_L(t)=\initH-\sum_{\mu=1}^m f^\mu(t)B_L^\mu$ & Controlled Hamiltonian generating a macroscopic operation. & Def.~\ref{definition:macroscopic-operation}, Eq.~\eqref{eq:general_H} \\
$m$ & Number of external fields coupled to additive observables. & Def.~\ref{definition:macroscopic-operation} \\
$B_L^\mu$ & Additive observables coupled to external fields. & Def.~\ref{definition:macroscopic-operation} \\
$f^\mu(t)$ & Values of external fields. & Def.~\ref{definition:macroscopic-operation} \\
$t^*$ & Operation time. & Def.~\ref{definition:macroscopic-operation} \\
$\rho_L(t)=U_L(t,0)\rho_LU_L^\dagger(t,0)$ & Time-evolved state under the operation. & Def.~\ref{definition:macroscopic-operation} \\
$\overline{\rho_L(t>t^*)}$ & Long-time average in the relaxation stage. & Eq.~\eqref{eq:LongTimeAve} \\
$\mathcal{T}$ & Time-window length for defining the long-time average. & Eq.~\eqref{eq:LongTimeAve} \\
$H_0,H_1$ & Initial and final Hamiltonians in an adiabatic-operation protocol. & Eq.~\eqref{eq:EntropyIncrease_H(t)} and Table~\ref{tbl:Operation_EntropyIncrease} \\
$\rho_{L|\ell}^{\bm r}$ & $\ell$-local reduced state around $\bm r$, moved onto $\Cell$. & Eq.~\eqref{eq:ReducedState_rho^r} \\
$\rho_{L|\ell}^{(k)}$ & Spatial average of $\rho_{L|\ell}^{\bm r}$ over $\mathcal{S}^{(k)}_L$. & Eq.~\eqref{eq:DEF_rho_L^red} \\
$\rho_{\infty|\ell}^{(k)}$ & Thermodynamic limit of $\rho_{L|\ell}^{(k)}$. & Eq.~\eqref{eq:DEF_rho_infty^red} \\
$\rho_{L|\ell}^{\mathrm{ave}}$ & Spatial average of $\rho_{L|\ell}^{\bm r}$ over the full lattice $\Lambda_L$. & Eq.~\eqref{eq:DEF_rho_infty^tot} \\
$\rho_{\infty|\ell}^{\mathrm{ave}}$ & Thermodynamic limit of $\rho_{L|\ell}^{\mathrm{ave}}$. & Below Eq.~\eqref{eq:tilde_s} \\
$S_{\mathrm{vN}}[\rho]$ & von Neumann entropy of a density matrix $\rho$. & Def.~\ref{definition:MacroEntropy} \\
$s_{\ell}^{\mathrm{mac}}[(\rho_L)_L]$ & Quantum macroscopic entropy density at scale $\ell$. & Def.~\ref{definition:MacroEntropy}, Eq.~\eqref{eq:DEF_s^red_ell} \\
$\tilde{s}_\ell[(\rho_L)_L]$ & Upper bound on $s_{\ell}^{\mathrm{mac}}$ obtained from $\rho_{\infty|\ell}^{\mathrm{ave}}$. & Eq.~\eqref{eq:tilde_s} \\
$s^{\mathrm{TD}}(\beta|H)$ & Thermodynamic entropy density of the canonical Gibbs state. & Eq.~\eqref{eq:DEF_s^TD} \\
$D(\sigma\Vert\rho)$ & Quantum relative entropy. & Proof of Thm.~\ref{theorem:EntropyIncrease} \\
\end{tabular}
\end{ruledtabular}
\end{table*}

\section{\label{sec:PlanckPrinciple}Planck's principle for adiabatic operations induced by external fields}

According to Refs.~\cite{Lieb1999,Uffink2001}, Planck suggested formulating thermodynamics on the basis that `friction' is an adiabatic operation that is not reversible. Since friction increases the energy of the system, its irreversibility means the impossibility of decreasing the energy by any adiabatic operation.
Following this contribution of Planck, we call the following law Planck's principle, as in Ref.~\cite{Lieb1999}:
\begin{TDLaw}[\label{TDLaw:OriginalPlanck}Original form of Planck's principle]
Consider a simple thermodynamic system whose thermal equilibrium states are specified by internal energy $U$ and several additional additive quantities, so-called ``work coordinate'', denoted by $V^1,...,V^{n}$, such as volume and magnetization.
Consider changing the state of the system from one thermal equilibrium state $(U_0,V^1_0,...,V^{n}_0)$ to another thermal equilibrium state $(U_1,V^1_0,...,V^{n}_0)$, whose work coordinates have the same values.
If the final energy is lower than the initial energy,
\begin{align}
    U_0>U_1,
\end{align}
it is impossible to induce such a change through any adiabatic operation.
\end{TDLaw}
\noindent
Note that Planck's principle (Theorem~3.4) in Ref.~\cite{Lieb1999} contains not only the above statement but also its inverse: If $U_0\le U_1$, then it is \emph{possible} to induce the state change $(U_0,V^1_0,...,V^{n}_0)\to (U_1,V^1_0,...,V^{n}_0)$ through some adiabatic operation. However, which types of adiabatic operations are \emph{possible} crucially depend on the details of the setting, and, to our knowledge, there is no consensus on this. For this reason, we call the above statement Planck's principle and do not include its inverse, as in Refs.~\cite{Tasaki2016SecondLaw,Hokkyo2025}.

This principle easily follows from another law, the law of increasing entropy, as follows:
Let $S(U,V^1,...,V^{n})$ be the thermodynamic entropy function of a simple thermodynamic system in Planck's principle~\ref{TDLaw:OriginalPlanck}.
Since its derivative with respect to $U$ gives the inverse temperature~$\beta$, $S(U,V^1,...,V^{n})$ is monotonically increasing with respect to $U$ (at least in the region $0< \beta <\infty$, to which thermodynamics usually applies).
Furthermore, the law of increasing entropy states that adiabatic operations that decrease the entropy,
\begin{align}
    S(U_0,V^1_0,...,V^{n}_0)> S(U_1,V^1_0,...,V^{n}_0),
\end{align}
are impossible.
From the monotonicity of $S(\bullet,V^1,...,V^{n})$, this is equivalent to the statement that adiabatic operations decreasing the energy $U_0> U_1$ are impossible, which implies Planck's principle~\ref{TDLaw:OriginalPlanck}.

Planck's principle~\ref{TDLaw:OriginalPlanck} is useful when $V^1,...,V^{n}$ can be controlled, as in the case of volume enclosed by a container with a piston.
However, it is inconvenient when their conjugate parameters are controlled, as in the case where $V^\mu$ is magnetization while the magnetic field --- the variable conjugate to $V^\mu$ --- is controlled. Below, we rewrite Planck's principle~\ref{TDLaw:OriginalPlanck} in a form useful for such a case, that is, the case where adiabatic operations are induced by external fields, such as a magnetic field. 

We consider adiabatic operations induced by external fields $f^1,...,f^m$. We write additive quantities conjugate to them by $B^1,...,B^m$. We assume that thermal equilibrium states are specified by $(U,B^1,...,B^m, N)$, where $N$ is the total number of sites and has the same role as the volume of the system. This corresponds to taking $(V^1,...,V^n)$ as $(B^1,...,B^m, N)$. 
Let $U(S,B^1,...,B^m,N)$ be internal energy function,
and $E(S,f^1,...,f^m,N)$ be its Legendre transformation with respect to $B^1,...,B^m$. The latter represents the energy in the presence of external fields.
Solving $E=E(S,f^1,...,f^m,N)$ with respect to $S$, we obtain an entropy function whose arguments are $(E,f^1,...,f^m,N)$, denoted by $\tilde{S}(E,f^1,...,f^m,N)$.
The derivative of $\tilde{S}(E,f^1,...,f^m,N)$ with respect to $E$ coincides with the inverse temperature~$\beta$ because
\begin{align}
    \frac{\partial E}{\partial S}\frac{\partial \tilde{S}}{\partial E}&=1
\end{align}
follows from 
\begin{align}
    E(\tilde{S}(E,f^1,...,f^m,N),f^1,...,f^m,N)&=E.
\end{align}
This means that 
$\tilde{S}(E,f^1,...,f^{m},N)$ is monotonically increasing with respect to $E$ (at least in the region $0< \beta <\infty$, to which thermodynamics usually applies).
Combining this with the law of increasing entropy, we obtain the following:
\begin{TDLaw}[\label{TDLaw:Planck_field}Another form of Planck's principle]
Consider a simple thermodynamic system under external fields~$f^1,...,f^m$ whose thermal equilibrium states are specified by internal energy $U$, additive quantities $B^1,...,B^{m}$ conjugate to $f^1,...,f^m$, and the total number of sites $N$.
We consider thermal equilibrium states outside the first-order phase transition point; then, they can also be specified by $(E,f^1,...,f^m,N)$, where $E$ is the energy in the presence of fields 
\begin{align}
    E=U-\sum_{\mu=1}^{m}f^\mu B^\mu.
    \label{eq:PlanckPrinciple_Energy}
\end{align}

Consider changing the state of the system from one thermal equilibrium state $(E_0,f^1_0,...,f^{m}_0,N)$ to another thermal equilibrium state $(E_1,f^1_0,...,f^{m}_0,N)$, whose external fields take the same values.
If the final energy is lower than the initial energy,
\begin{align}
    E_0>E_1,
\end{align}
it is impossible to induce such a change through any adiabatic operation.
\end{TDLaw}
Note that Eq.~\eqref{eq:PlanckPrinciple_Energy} corresponds to the thermal equilibrium value of the following Hamiltonian under external fields,
\begin{align}
    H_L(f^1,...,f^m)=H_L-\sum_{\mu=1}^{m}f^\mu B^{\mu}_L,
\end{align}
which takes the same form as Eq.~\eqref{eq:general_H}.
In this sense, our definition of macroscopic operation, Definition~\ref{definition:macroscopic-operation}, is consistent with the definition of adiabatic operation in thermodynamics.

In particular, applying Planck's principle~\ref{TDLaw:Planck_field} to the case $f^1_0=...=f^m_0=0$, we obtain the following statement: any adiabatic operation induced by external fields $f^1,...,f^m$ does not decrease the energy
\begin{align}
    U_0=E_0\le E_1=U_1,
\end{align}
as long as the initial and final values of external fields are set to zero.
This indeed coincides with the inequality~\eqref{eq:macroscopic-passivity}, i.e., macroscopic passivity.

\section{\label{sec:Proof_Example_iMATE}Examples of iMATE}

In Appendix~\ref{sec:Proof_METTS}, we will show that the sequence of METTS represents iMATE almost surely.
In addition, in Appendix~\ref{sec:Proof_Example_iMATE_InfiniteTemp} and \ref{sec:Proof_Example_iMATE_FiniteTemp}, we will construct representations of iMATE that contain no random variables in contrast to METTS.

\subsection{\label{sec:Proof_METTS}METTS (details of Example~\ref{example:TypicalMETTS_represents_iMATE})}

In Appendix~\ref{sec:Proof_METTS_1}, we will show that there are appropriate sequences of METTS that represent iMATE.
In Appendix~\ref{sec:Proof_METTS_Typical}, we will reinforce this result by showing that if we construct a sequence of METTS according to the probability distribution $P(i)/Z$, the sequence represents iMATE almost surely, under a certain condition on the probability distribution.
In Appendix~\ref{sec:Proof_METTS_Justification}, we justify this condition from a reasonable physical assumption on the canonical Gibbs state.

\subsubsection{\label{sec:Proof_METTS_1}There are sequences of METTS representing iMATE}

Here, we show that there exist sequences of METTS that represent
iMATE. We are interested in METTS~\eqref{eq:METTS} corresponding to the Gibbs state described by $H_L$ at inverse temperature $\beta$, i.e., $\rho_L^{\mathrm{can}}(\beta|H_L)$.
We assume that the Gibbs state satisfies an assumption slightly stronger than Eq.~\eqref{eq:MacroscopicallyNormal}, which is a part of Assumption~\ref{assumption:GibbsState_Normal}:
\begin{assumpGibbs}[\label{assumption:GibbsState_LargerLocality}Assumption slightly stronger than Eq.~\eqref{eq:MacroscopicallyNormal}]
There is a sequence of positive integers $(\ell_L^{\mathrm{max}})_{L\in\mathbb{N}}$ that diverges in the limit $L\to\infty$ [with the rate of $o(L)$]
and in which the Gibbs state $\rho_L^{\mathrm{can}}(\beta|H_L)$ satisfies the following: 
The variance of any additive observable composed of $\ell_L^{\mathrm{max}}$-local observables is macroscopically negligible, that is,
\begin{align}
    \mathrm{Tr} \Bigl[ \rho_L^{\mathrm{can}}(\beta|H_L)  \frac{\bigl(A_{\mathcal{S}_L}(\aLocal_L)\bigr)^2}{|\mathcal{S}_L|^2} \Bigr]
    -\Bigl(\mathrm{Tr} \Bigl[ \rho_L^{\mathrm{can}}(\beta|H_L)  \frac{A_{\mathcal{S}_L}(\aLocal_L)}{|\mathcal{S}_L|} \Bigr]\Bigr)^2=o(L^0)
    \label{eq:GibbsState_Normal_LargerLocality}
\end{align}
holds for every macroscopic subsystem $(\mathcal{S}_L)_{L\in\mathbb{N}}$ and for every $\ell_{L}^{\mathrm{max}}$-local observable $\aLocal_L$ satisfying $\|\aLocal_L\|_{\infty}= 1$.
\end{assumpGibbs}
Since Eq.~\eqref{eq:MacroscopicallyNormal} refers to additive observables composed of $\ell$-local observables with $\ell=O(L^0)$, the above condition~\eqref{eq:GibbsState_Normal_LargerLocality} is stronger than Eq.~\eqref{eq:MacroscopicallyNormal}. Note that, as long as $\ell_L^{\mathrm{max}}=\omega(L^0)$, its divergence can be arbitrarily slow. In this sense, we consider that the gap between Assumption~\ref{assumption:GibbsState_LargerLocality} and Eq.~\eqref{eq:MacroscopicallyNormal} is physically unimportant.

Under this assumption, we show the following:
\begin{proposition}[\label{proposition:METTS_represent_iMATE}There are sequences of METTS representing iMATE]
Let $\phi_L(i)=\ket{\phi(i)}\bra{\phi(i)}$ with $i\in\{0,1,...,D^N-1\}$ be METTS defined by Eq.~\eqref{eq:METTS}. Suppose that the corresponding Gibbs state $\rho_L^{\mathrm{can}}(\beta|H_L)$ satisfies Assumption~\ref{assumption:GibbsState_LargerLocality}. For each system size $L\in\mathbb{N}$, we choose $i_L\in\{0,1,...,D^N-1\}$ appropriately such that Eq.~\eqref{eq:Cond_METTS_represent_iMATE} given below is satisfied. Then, the sequence of states $(\rho_L)_{L\in\mathbb{N}}$ given by $\rho_L=\phi_L(i_L)$ represents iMATE described by $H_L$ at inverse temperature $\beta$. 
Moreover, according to the probability distribution $P(i_L)/Z$, the probability of choosing $i_L\in\{0,1,...,D^N-1\}$ that satisfies Eq.~\eqref{eq:Cond_METTS_represent_iMATE} is arbitrarily close to $1$ for sufficiently large $L$.
\end{proposition}
This means that if Assumption~\ref{assumption:GibbsState_LargerLocality} is satisfied, then there are sequences of states $(\rho_L)_{L\in\mathbb{N}}$ representing iMATE. 
Below, we prove Proposition~\ref{proposition:METTS_represent_iMATE}.

\begin{proof}
First, we define the following small quantity, which provides an upper bound on the fluctuations of additive observables composed of $\ell_L^{\mathrm{max}}$-local observables: Take one of the primitive macroscopic subsystems $\mathcal{S}_L^{(k)}$, and define
\begin{align}
    \varepsilon_L^{\mathrm{fluc}}&:=\sup_{\aLocal_L(\|\aLocal_L\|_{\infty}=1)}\sqrt{\mathrm{Tr} \Bigl[ \rho_L^{\mathrm{can}}(\beta|H_L)  \frac{\bigl(A_{\mathcal{S}_L^{(k)}}(\aLocal_L)\bigr)^2}{|\mathcal{S}_L^{(k)}|^2} \Bigr]-\Bigl(\mathrm{Tr} \Bigl[ \rho_L^{\mathrm{can}}(\beta|H_L)  \frac{A_{\mathcal{S}_L^{(k)}}(\aLocal_L)}{|\mathcal{S}_L^{(k)}|} \Bigr]\Bigr)^2},
\end{align}
where $\aLocal_L$ in the RHS runs over all $\ell_L^{\mathrm{max}}$-local observables with $\|\aLocal_L\|=1$.
Note that because of the translation invariance of the Gibbs state, it is independent of the choice of $k\in\{1,2,...,K^d\}$.
Assumption~\ref{assumption:GibbsState_LargerLocality} implies $\varepsilon_L^{\mathrm{fluc}}=o(L^0)$.

Second, we introduce sequences of positive integers $(\ell_L)_{L\in\mathbb{N}}$ and positive real values $(\delta_L)_{L\in\mathbb{N}}$ so that they satisfy the following:
\begin{align}
    &\ell_L=\omega(L^0),\quad \ell_L\le \ell_L^{\mathrm{max}}\label{eq:METTS_Large_ell}\\
    &\delta_L=o(L^0)\label{eq:METTS_Small_delta}\\
    &\frac{D^{\ell_L^d}\varepsilon_L^{\mathrm{fluc}}}{\delta_L}=o(L^0).\label{eq:METTS_Small_ell}
\end{align}
For instance, this can be achieved by choosing 
\begin{align}
    \delta_L&=(\varepsilon_L^{\mathrm{fluc}})^{1/3},\label{eq:METTS_delta_Example}\\
    \ell_L&=\min\Bigl\{\Bigl(\frac{-\log(\varepsilon_L^{\mathrm{fluc}})}{3\log D}\Bigr)^{1/d},\ell_L^{\mathrm{max}}\Bigr\}.\label{eq:METTS_ell_Example}
\end{align}

Third, we describe how to choose $i_L\in\{0,1,...,D^N-1\}$.
Let $\phi_{L|\ell_L}^{(k)}(i)$ be the spatial average~\eqref{eq:DEF_rho_L^red} of $\ell_L$-local density matrices of $\phi_{L}(i)$,
and $\rho_{L|\ell_L}^{\mathrm{can}}$ be that of the Gibbs state $\rho_{L}^{\mathrm{can}}(\beta|H_L)$. Here, we note that it is independent of the choice of $k\in\{1,2,...,K^d\}$ because of the translation invariance of the Gibbs state.
We choose $i_L\in\{0,1,...,D^N-1\}$ for each $L$ such that the following condition is satisfied:
\begin{align}
    \max_{k\in\{1,2,...,K^d\}}\|\phi_{L|\ell_L}^{(k)}(i_L)-\rho_{L|\ell_L}^{\mathrm{can}}\|_{1}<\delta_L.
    \label{eq:Cond_METTS_represent_iMATE}
\end{align}
Below, we will show that such an $i_L$ satisfying Eq.~\eqref{eq:Cond_METTS_represent_iMATE} indeed exists, and furthermore, that typical $i_L$'s satisfy Eq.~\eqref{eq:Cond_METTS_represent_iMATE}.

Fourth, we show that if Eq.~\eqref{eq:Cond_METTS_represent_iMATE} is satisfied for all $L\in\mathbb{N}$, then the obtained $(\rho_L)_{L\in\mathbb{N}}$ represents iMATE described by $H_L$ at inverse temperature $\beta$.
Take an arbitrary positive integer $\ell\in\mathbb{N}$ that is of $O(L^0)$.
Because of Eq.~\eqref{eq:METTS_Large_ell}, $\ell\le \ell_L$ holds for a sufficiently large $L$.
This implies
\begin{align}
    \|\rho_{L|\ell}^{(k)}-\rho_{L|\ell}^{\mathrm{can}}\|_{1}
    \le \|\rho_{L|\ell_L}^{(k)}-\rho_{L|\ell_L}^{\mathrm{can}}\|_{1}.
    \label{eq:METTS_Cond_Implies_iMATE_1}
\end{align}
Combining this with Eqs.~\eqref{eq:METTS_Small_delta} and \eqref{eq:Cond_METTS_represent_iMATE}, we have
\begin{align}
    \lim_{L\to\infty}\max_{k\in\{1,2,...,K^d\}}\|\rho_{L|\ell}^{(k)}-\rho_{L|\ell}^{\mathrm{can}}\|_{1}=0.
    \label{eq:METTS_Cond_Implies_iMATE_2}
\end{align}
Since $\ell$ is arbitrary, this holds for any $\ell\in\mathbb{N}$.
From Proposition~\ref{proposition:Equiv_rho^red}, this implies $(\rho_L)_{L\in\mathbb{N}}\maceq \bigl(\rho_L^{\mathrm{can}}(\beta|H_L)\bigr)_{L\in\mathbb{N}}$.

Finally, we show that typical $i_L$'s satisfy Eq.~\eqref{eq:Cond_METTS_represent_iMATE}. To this end, we evaluate the probability that  Eq.~\eqref{eq:Cond_METTS_represent_iMATE} is violated,
\begin{align}
    \mathrm{Prob}_{i}[\max_{k\in\{1,2,...,K^d\}}\|\phi_{L|\ell_L}^{(k)}(i)-\rho_{L|\ell_L}^{\mathrm{can}}\|_{1}\ge \delta_L]
    :=\sum_{i=0}^{D^N-1}\frac{P(i)}{Z}\chi[\max_{k\in\{1,2,...,K^d\}}\|\phi_{L|\ell_L}^{(k)}(i)-\rho_{L|\ell_L}^{\mathrm{can}}\|_{1}\ge \delta_L],
\end{align}
where $\chi[\bullet]$ is the indicator function that takes $1$ when $\bullet$ is true, and $0$ otherwise.
From the Chebyshev inequality, we have
\begin{align}
    \mathrm{Prob}_{i}[\max_{k\in\{1,2,...,K^d\}}\|\phi_{L|\ell_L}^{(k)}(i)-\rho_{L|\ell_L}^{\mathrm{can}}\|_{1}\ge \delta_L]
    \le\max_{k\in\{1,2,...,K^d\}}\sum_{i=0}^{D^N-1}\frac{P(i)}{Z}\frac{\|\phi_{L|\ell_L}^{(k)}(i)-\rho_{L|\ell_L}^{\mathrm{can}}\|_{1}^2}{ \delta_L^2}.
\end{align}
From the relation between the trace norm and the operator norm, we have
\begin{align}
    \|\phi_{L|\ell_L}^{(k)}(i)-\rho_{L|\ell_L}^{\mathrm{can}}\|_{1}
    =\sup_{\aLocal_L(\|\aLocal_L\|_{\infty}=1)}\bigl|\mathrm{Tr}\bigl[\aLocal_L\bigl(\phi_{L|\ell_L}^{(k)}(i)-\rho_{L|\ell_L}^{\mathrm{can}}\bigr)\bigr]\bigr|.
\end{align}
We introduce an orthogonal basis of operators on $C_{\ell_L}$, $\{u^\mu\}_{\mu=1}^{D^{2\ell_L^d}}$, satisfying the following:
\begin{align}
    (u^\mu)^\dagger&=u^\mu,\\
    \mathrm{Tr}[ u^\mu u^{\mu^\prime}]&=D^{\ell_L^d}\delta_{\mu,\mu^\prime},\\
    \|u^\mu\|_{\infty}&\le \sqrt{2}.
\end{align}
For instance, this can be satisfied by $\bigl(v^\mu+(v^\mu)^\dagger\bigr)/\sqrt{2}$ and $\bigl(v^\mu-(v^\mu)^\dagger\bigr)/\sqrt{2}i$~\footnote{Precisely, when $v^\mu$ and $(v^\mu)^\dagger$ are linearly dependent, $(v^\mu)^\dagger=e^{i\theta} v^\mu$, so  we use $e^{i\theta/2}v^\mu$ instead.}, where $v^\mu$'s are the generalized Pauli matrices on $C_{\ell_L}$.
Using this, we expand the operator $\aLocal_L$,
\begin{align}
    \aLocal_L=\sum_{\mu=1}^{D^{2\ell_L^d}}c_\mu u^\mu.
\end{align}
From the condition $\|\aLocal_L\|_{\infty}=1$, the coefficients $c_\mu$'s satisfy
\begin{align}
    \sum_{\mu=1}^{D^{2\ell_L^d}}|c_\mu|^2=\frac{\|\aLocal_L\|_{2}^2}{D^{\ell_L^d}}\le \|\aLocal_L\|_{\infty}^2=1,
\end{align}
where $\|\aLocal_L\|_{2}$ is the Hilbert-Schmidt norm.
Combining these, we have
\begin{align}
    \bigl|\mathrm{Tr}\bigl[\aLocal_L\bigl(\phi_{L|\ell_L}^{(k)}(i)-\rho_{L|\ell_L}^{\mathrm{can}}\bigr)\bigr]\bigr|^2
    &=\Bigl|\sum_{\mu=1}^{D^{2\ell_L^d}}c_\mu \mathrm{Tr}\bigl[u^\mu \bigl(\phi_{L|\ell_L}^{(k)}(i)-\rho_{L|\ell_L}^{\mathrm{can}}\bigr)\bigr]\Bigr|^2\\
    &\le \Bigl(\sum_{\mu=1}^{D^{2\ell_L^d}}|c_\mu|^2\Bigr)\Bigl(\sum_{\mu=1}^{D^{2\ell_L^d}}\bigl|\mathrm{Tr}\bigl[u^\mu \bigl(\phi_{L|\ell_L}^{(k)}(i)-\rho_{L|\ell_L}^{\mathrm{can}}\bigr)\bigr]\bigr|^2\Bigr)\\
    &\le \sum_{\mu=1}^{D^{2\ell_L^d}}\bigl|\mathrm{Tr}\bigl[u^\mu \bigl(\phi_{L|\ell_L}^{(k)}(i)-\rho_{L|\ell_L}^{\mathrm{can}}\bigr)\bigr]\bigr|^2,
\end{align}
where in the second line, we used the Cauchy-Schwarz inequality.
This implies 
\begin{align}
    \|\phi_{L|\ell_L}^{(k)}(i)-\rho_{L|\ell_L}^{\mathrm{can}}\|_{1}^2
    \le \sum_{\mu=1}^{D^{2\ell_L^d}}\bigl|\mathrm{Tr}\bigl[u^\mu \bigl(\phi_{L|\ell_L}^{(k)}(i)-\rho_{L|\ell_L}^{\mathrm{can}}\bigr)\bigr]\bigr|^2,
    \label{eq:Hilbert–Schmidt_expansion_bound}
\end{align}
which enables us to eliminate $\sup_{\aLocal_L(\|\aLocal_L\|_{\infty}=1)}$.
From Eq.~\eqref{eq:<a>^red=<A_Sigma>/Sigma}, the trace in the RHS can be rewritten as 
\begin{align}
    \mathrm{Tr}\bigl[u^\mu \bigl(\phi_{L|\ell_L}^{(k)}(i)-\rho_{L|\ell_L}^{\mathrm{can}}\bigr)\bigr]
    =\mathrm{Tr}\Bigl[\frac{A_{\mathcal{S}_{L}^{(k)}}(u^\mu)}{|\mathcal{S}_L^{(k)}|}\bigl(\phi_{L}(i)-\rho_{L}^{\mathrm{can}}(\beta|H_L)\bigr)\Bigr].
\end{align}
Combining these, we have 
\begin{align}
    \mathrm{Prob}_{i}[\max_{k\in\{1,2,...,K^d\}}\|\phi_{L|\ell_L}^{(k)}(i)-\rho_{L|\ell_L}^{\mathrm{can}}\|_{1}\ge \delta_L]
    \le\max_{k\in\{1,2,...,K^d\}}\frac{1}{\delta_L^2}\sum_{i=0}^{D^N-1}\frac{P(i)}{Z}\sum_{\mu=1}^{D^{2\ell_L^d}}\Bigl|\mathrm{Tr}\Bigl[\frac{A_{\mathcal{S}_{L}^{(k)}}(u^\mu)}{|\mathcal{S}_L^{(k)}|}\bigl(\phi_{L}(i)-\rho_{L}^{\mathrm{can}}(\beta|H_L)\bigr)\Bigr]\Bigr|^2\\
    =\max_{k\in\{1,2,...,K^d\}}\frac{1}{\delta_L^2}\sum_{\mu=1}^{D^{2\ell_L^d}}\Bigl\{\sum_{i=0}^{D^N-1}\frac{P(i)}{Z}\Bigl(\mathrm{Tr}\Bigl[\frac{A_{\mathcal{S}_{L}^{(k)}}(u^\mu)}{|\mathcal{S}_L^{(k)}|}\phi_{L}(i)\Bigr]\Bigr)^2-\Bigl(\mathrm{Tr}\Bigl[\frac{A_{\mathcal{S}_{L}^{(k)}}(u^\mu)}{|\mathcal{S}_L^{(k)}|}\rho_{L}^{\mathrm{can}}(\beta|H_L)\Bigr]\Bigr)^2\Bigr\},
\end{align}
where we used Eq.~\eqref{eq:METTS_canonical}.
Using nonnegativity of the variance of $A_{\mathcal{S}_{L}^{(k)}}(u^\mu)/|\mathcal{S}_L^{(k)}|$ in $\phi_{L}(i)$,
\begin{align}
    \Bigl(\mathrm{Tr}\Bigl[\frac{A_{\mathcal{S}_{L}^{(k)}}(u^\mu)}{|\mathcal{S}_L^{(k)}|}\phi_{L}(i)\Bigr]\Bigr)^2
    \le \mathrm{Tr}\Bigl[\frac{\bigl(A_{\mathcal{S}_{L}^{(k)}}(u^\mu)\bigr)^2}{|\mathcal{S}_L^{(k)}|^2}\phi_{L}(i)\Bigr],
\end{align}
we have
\begin{align}
    \sum_{i=0}^{D^N-1}\frac{P(i)}{Z}\Bigl(\mathrm{Tr}\Bigl[\frac{A_{\mathcal{S}_{L}^{(k)}}(u^\mu)}{|\mathcal{S}_L^{(k)}|}\phi_{L}(i)\Bigr]\Bigr)^2
    \le \sum_{i=0}^{D^N-1}\frac{P(i)}{Z} \mathrm{Tr}\Bigl[\frac{\bigl(A_{\mathcal{S}_{L}^{(k)}}(u^\mu)\bigr)^2}{|\mathcal{S}_L^{(k)}|^2}\phi_{L}(i)\Bigr]
    =\mathrm{Tr}\Bigl[\frac{\bigl(A_{\mathcal{S}_{L}^{(k)}}(u^\mu)\bigr)^2}{|\mathcal{S}_L^{(k)}|^2}\rho_{L}^{\mathrm{can}}(\beta|H_L)\Bigr]
\end{align}
Since $u^\mu$ is an $\ell_L$-local observable with $\|u^\mu\|_{\infty}\le \sqrt{2}$, the variance of $A_{\mathcal{S}_{L}^{(k)}}(u^\mu)/|\mathcal{S}_L^{(k)}|$ in the Gibbs state is bounded from above by $2(\varepsilon_L^{\mathrm{fluc}})^2$,
\begin{align}
    \mathrm{Tr}\Bigl[\frac{\bigl(A_{\mathcal{S}_{L}^{(k)}}(u^\mu)\bigr)^2}{|\mathcal{S}_L^{(k)}|^2}\rho_{L}^{\mathrm{can}}(\beta|H_L)\Bigr]-\Bigl(\mathrm{Tr}\Bigl[\frac{A_{\mathcal{S}_{L}^{(k)}}(u^\mu)}{|\mathcal{S}_L^{(k)}|}\rho_{L}^{\mathrm{can}}(\beta|H_L)\Bigr]\Bigr)^2
    \le 2(\varepsilon_L^{\mathrm{fluc}})^2.
\end{align}
Therefore, we have
\begin{align}
    \mathrm{Prob}_{i}[\max_{k\in\{1,2,...,K^d\}}\|\phi_{L|\ell_L}^{(k)}(i)-\rho_{L|\ell_L}^{\mathrm{can}}\|_{1}\ge \delta_L]
    \le\max_{k\in\{1,2,...,K^d\}}\frac{1}{\delta_L^2}\sum_{\mu=1}^{D^{2\ell_L^d}}2(\varepsilon_L^{\mathrm{fluc}})^2=2\Bigl(\frac{D^{\ell_L^d}\varepsilon_L^{\mathrm{fluc}}}{\delta_L}\Bigr)^2,
\end{align}
which vanishes in the limit $L\to\infty$ due to Eq.~\eqref{eq:METTS_Small_ell}.
This means that the probability of choosing $i_L\in\{0,1,...,D^N-1\}$ that satisfies Eq.~\eqref{eq:Cond_METTS_represent_iMATE} is arbitrarily close to $1$ for sufficiently large $L$.

\end{proof}

\subsubsection{\label{sec:Proof_METTS_Typical}Typical sequence of METTS represents iMATE}

In the previous subsubsection, we have shown that if $i_L$ is chosen appropriately such that Eq.~\eqref{eq:Cond_METTS_represent_iMATE} is satisfied for every $L$, then the sequence of METTS, $\bigl(\phi_L(i_L)\bigr)_{L\in\mathbb{N}}$, represents iMATE. Although we have also shown that the probability of satisfying Eq.~\eqref{eq:Cond_METTS_represent_iMATE} is close to $1$ for each $L$, it is unclear how high the probability is that Eq.~\eqref{eq:Cond_METTS_represent_iMATE} is satisfied for \emph{all} $L$ simultaneously. In other words, we did not evaluate the probability that the sequence $\bigl(\phi_L(i_L)\bigr)_{L\in\mathbb{N}}$ represents iMATE. Here, we show that the sequence $\bigl(\phi_L(i_L)\bigr)_{L\in\mathbb{N}}$ represents iMATE with probability $1$, assuming that the probability of violating Eq.~\eqref{eq:Cond_METTS_represent_iMATE} decays at a sufficiently fast rate with respect to $L$.

Recall that, under Assumption~\ref{assumption:GibbsState_LargerLocality}, $(\delta_L)_{L\in\mathbb{N}}$ and $(\ell_L)_{L\in\mathbb{N}}$ are defined by Eqs.~\eqref{eq:METTS_Large_ell}--\eqref{eq:METTS_Small_ell}. For simplicity of notation, let $p_L$ be the probability that Eq.~\eqref{eq:Cond_METTS_represent_iMATE} is violated,
\begin{align}
    p_L := \mathrm{Prob}_{i}[\max_{k\in\{1,2,...,K^d\}}\|\phi_{L|\ell_L}^{(k)}(i)-\rho_{L|\ell_L}^{\mathrm{can}}\|_{1}\ge \delta_L].
\end{align}
Our upper bound of $p_L$ given in Appendix~\ref{sec:Proof_METTS_1} employs the Chebyshev inequality, which is not sufficiently tight for the purpose of this subsubsection.
Instead, we here assume
\begin{align}
    \sum_{L=1}^{\infty}p_L<\infty.
    \label{eq:METTS_Decay_pL}
\end{align}
This condition is justified by a physically reasonable assumption in the next subsubsection.

Using condition~\eqref{eq:METTS_Decay_pL}, we can show the following:
\begin{proposition}[\label{proposition:TypicalMETTS_represents_iMATE}Typical sequence of METTS represents iMATE]
Suppose that Assumption~\ref{assumption:GibbsState_LargerLocality} and Eq.~\eqref{eq:METTS_Decay_pL} are satisfied. We choose $i_L\in\{0,1,...,D^N-1\}$ according to the probability distribution $P(i_L)/Z$ and define $\rho_L=\phi_L(i_L)$.
For any positive number $\epsilon>0$, we can say that the sequence of states $(\rho_L)_{L\in\mathbb{N}}$ satisfies $(\rho_L)_{L\in\mathbb{N}}\maceq\bigl(\rho_{L}^{\mathrm{can}}(\beta|H_L)\bigr)_{L\in\mathbb{N}}$ with probability larger than $1-\epsilon$.
\end{proposition}
\noindent
In other words, if Assumption~\ref{assumption:GibbsState_LargerLocality} and Eq.~\eqref{eq:METTS_Decay_pL} are satisfied, then $(\rho_L)_{L\in\mathbb{N}}$ represents iMATE almost surely. The essential idea of the following proof is almost the same as that of the Borel-Cantelli lemma.
\begin{proof}
Let $L^*\in\mathbb{N}$ be a sufficiently large integer.
We consider the probability that Eq.~\eqref{eq:Cond_METTS_represent_iMATE} is satisfied for all $L\ge L^*$,
\begin{align}
    P_{L^*}:=\prod_{L=L^*}^{\infty}(1-p_L).
\end{align}
Taking its logarithm, we have
\begin{align}
    \log P_{L^*}=\sum_{L=L^*}^{\infty}\log(1-p_L).
\end{align}
From Eq.~\eqref{eq:METTS_Decay_pL}, since $L^*$ is sufficiently large, $p_L<1/2$ holds for all $L\ge L^*$.
From the concavity of $\log(1-x)$, we have $\log(1-p_L)\ge -2p_L\log 2$ for $p_L<1/2$. Substituting this, we have
\begin{align}
    \log P_{L^*}\ge -2\log 2\sum_{L=L^*}^{\infty}p_L,
\end{align}
or equivalently
\begin{align}
    P_{L^*}\ge \exp\Bigl[-2\log 2\sum_{L=L^*}^{\infty}p_L\Bigr]\ge 1-2\log 2\sum_{L=L^*}^{\infty}p_L.
\end{align}
From Eq.~\eqref{eq:METTS_Decay_pL}, $\sum_{L=L^*}^{\infty}p_L$ becomes sufficiently small by taking $L^*$ sufficiently large.
This implies that 
\begin{align}
    P_{L^*}\ge 1-\epsilon
\end{align}
holds for sufficiently large $L^*$.

The remaining task is to show that if Eq.~\eqref{eq:Cond_METTS_represent_iMATE} is satisfied for all $L\ge L^*$, then the corresponding $(\rho_L)_{L\in\mathbb{N}}$ satisfies $(\rho_L)_{L\in\mathbb{N}}\maceq\bigl(\rho_{L}^{\mathrm{can}}(\beta|H_L)\bigr)_{L\in\mathbb{N}}$.
This readily follows from Eqs.~\eqref{eq:METTS_Cond_Implies_iMATE_1} and \eqref{eq:METTS_Cond_Implies_iMATE_2}.
\end{proof}

\subsubsection{\label{sec:Proof_METTS_Justification}Physical justification of Eq.~\eqref{eq:METTS_Decay_pL}}

Here, we justify Eq.~\eqref{eq:METTS_Decay_pL} from a physically reasonable assumption. First, to introduce a new assumption given below (Assumption~\ref{assumption:GibbsState_Concentration}), which shares the key motivation with Assumption~\ref{assumption:GibbsState_LargerLocality},
we temporarily impose Assumption~\ref{assumption:GibbsState_LargerLocality}.
Let $\aLocal_L$ be an arbitrary $\ell_L$-local observable with $\|\aLocal_L\|_{\infty}=1$, where $\ell_L$ is given by Eq.~\eqref{eq:METTS_Large_ell}. 
From Assumption~\ref{assumption:GibbsState_LargerLocality}, the variance of the additive observable $A_L=A_{\mathcal{S}_{L}^{(k)}}(\aLocal_L)$ scales as
\begin{align}
    \sigma_A^2:=\mathrm{Tr} \Bigl[ \rho_L^{\mathrm{can}}(\beta|H_L)  \bigl(\delta A_L\bigr)^2 \Bigr]
    =o(N^2),
    \label{eq:Concentration_Var_A}
\end{align}
where 
\begin{align}
    \delta A_L:=A_L
    -\mathrm{Tr} \bigl[ \rho_L^{\mathrm{can}}(\beta|H_L)  A_L \bigr].
\end{align}
As suggested from the quantum central limit theorem~\cite{Goderis1989,Goderis1990,Matsui2003,Fujikura2016,Shimizu2017}, the distribution of $A_L$ is close to the Gaussian distribution with the variance $\sigma_A^2$,
\begin{align}
    \mathrm{Tr}[ \rho_L^{\mathrm{can}}(\beta|H_L) \mathcal{P}_{m}^A ]\propto \exp\Bigl[-\frac{(\delta A_m)^2}{2\sigma_A^2}\Bigr],
\end{align}
where $\mathcal{P}_m^A$ is the projection to the eigenspace of $A_L$ with the eigenvalue $A_m$, and $\delta A_m=A_m-\mathrm{Tr} \bigl[ \rho_L^{\mathrm{can}}(\beta|H_L)  A_L \bigr]$.
This indicates that the moment generating function can be approximated by
\begin{align}
    \mathrm{Tr} \Bigl[\rho_L^\mathrm{can}(\beta|H_L) \exp[\tau \delta A_L]\Bigr]
    \simeq \exp\Bigl[\frac{\sigma_A^2\tau^2}{2}  \Bigr]=\exp[\tau^2 o(N^2)].
\end{align}
These discussions should hold for any $\aLocal_L$ and at least for sufficiently small $|\tau|$.
Therefore, the Gaussian concentration bound~\cite{Lenci2005,Ogata2010,Ogata2010a,Kuwahara2020Concentration,DePalma2025} of the following form will be a reasonable assumption for the canonical Gibbs state:
\begin{assumpGibbs}[\label{assumption:GibbsState_Concentration}Gaussian concentration bound]
There are positive real numbers $\tau_0$ and $\nu$ independent of $L$, a sequence of positive integers $(\ell_L^{\mathrm{max}})_{L\in\mathbb{N}}$ and a sequence of positive real numbers $(\chi_L)_{L\in\mathrm{N}}$ that scale as
\begin{align}
    &\chi_L=O(N^{2-\nu}),\label{eq:AssumptionConcentration_chiL=o(N^2)}\\
    &\chi_L=\Omega(N),\label{eq:AssumptionConcentration_chiL=Omega(N)}\\
    &\ell_L^{\mathrm{max}}=\omega(L^0)\label{eq:AssumptionConcentration_ellL^max}
\end{align}
and in which 
\begin{align}
    &\log \mathrm{Tr} \bigl[\rho_L^\mathrm{can}(\beta|H_L) \exp[\tau \delta A_{\mathcal{S}_{L}^{(k)}}(\aLocal_L)]\bigr]
    \leq \frac{1}{2} \tau^2 \chi_L
    \label{eq:AssumptionConcentration}
\end{align}
is satisfied for any $\tau\in(-\tau_0,\tau_0)$, $\ell_L^{\mathrm{max}}$-local observable $\aLocal_L$ with $\|\aLocal_L\|_{\infty}=1$, $L\in\mathbb{N}$, and $k\in\{1,2,...,K^d\}$.
\end{assumpGibbs}
\noindent
Note that, although there is a rigorous proof of the Gaussian concentration bound~\cite{Kuwahara2020Concentration} in the case where $\ell_L^{\mathrm{max}}$ is independent of $L$ and $\beta$ is sufficiently small, Assumption~\ref{assumption:GibbsState_Concentration} imposes more strongly than such a bound. We expect Assumption~\ref{assumption:GibbsState_Concentration} to be valid because of the discussions given above Assumption~\ref{assumption:GibbsState_Concentration}.
In addition, when $\tau_0$ is very small, Eq.~\eqref{eq:AssumptionConcentration} will be more reasonable because the LHS becomes close to $\sigma_A^2\tau^2/2$.
We remark that, though Eq.~\eqref{eq:AssumptionConcentration_chiL=o(N^2)} is slightly stronger than the scaling of $o(N^2)$ which corresponds to Eq.~\eqref{eq:Concentration_Var_A}, their difference will be unimportant for sufficiently small $\nu$. 
Note also that since Eq.~\eqref{eq:AssumptionConcentration} is an upper bound, Eq.~\eqref{eq:AssumptionConcentration_chiL=Omega(N)} is not crucial for Assumption~\ref{assumption:GibbsState_Concentration}; if one finds $\chi_L$ satisfying the above conditions other than Eq.~\eqref{eq:AssumptionConcentration_chiL=Omega(N)}, we can replace it by a larger $\chi_L$ to meet Eq.~\eqref{eq:AssumptionConcentration_chiL=Omega(N)}.

Using this assumption, we show the following:
\begin{proposition}[\label{proposition:METTS_Justification_Decay_pL}Physical justification of Eq.~\eqref{eq:METTS_Decay_pL}]
Suppose that Assumption~\ref{assumption:GibbsState_Concentration} is satisfied. We take sequences $(\delta_L)_{L\in\mathbb{N}}$ and $(\ell_L)_{L\in\mathbb{N}}$ as 
\begin{align}
    \delta_L&=\frac{1}{\log L}=o(L^0)\label{eq:Concentration_deltaL}\\
    \ell_L&=\min\{\Bigl\lfloor\Bigl(\frac{\log\log L}{\log D}\Bigr)^{1/d}\Bigr\rfloor,\ell_L^{\mathrm{max}}\}=\omega(L^0)\label{eq:Concentration_ellL}
\end{align}
instead of Eqs.~\eqref{eq:METTS_Large_ell}--\eqref{eq:METTS_Small_ell}.
Then, Eq.~\eqref{eq:METTS_Decay_pL} is satisfied.
This implies that, for any positive number $\epsilon>0$, we can say that the sequence of states $(\rho_L)_{L\in\mathbb{N}}$ defined in Proposition~\ref{proposition:TypicalMETTS_represents_iMATE} satisfies $(\rho_L)_{L\in\mathbb{N}}\maceq\bigl(\rho_{L}^{\mathrm{can}}(\beta|H_L)\bigr)_{L\in\mathbb{N}}$ with probability larger than $1-\epsilon$.
\end{proposition}
In other words, if Assumption~\ref{assumption:GibbsState_Concentration} is satisfied, then $(\rho_L)_{L\in\mathbb{N}}$ represents iMATE almost surely.

\begin{proof}

From Jensen's inequality, for any $\tau\in\mathbb{R}$, we have
\begin{align}
    \exp [\tau \mathrm{Tr}[\phi_L(i) \delta A_{\mathcal{S}_{L}^{(k)}}(\aLocal_L)]]\le \mathrm{Tr} \Bigl[\phi_L(i) \exp[\tau \delta A_{\mathcal{S}_{L}^{(k)}}(\aLocal_L)]\Bigr],
\end{align}
which implies
\begin{align}
    \sum_i \frac{P(i)}{Z} \exp [\tau \mathrm{Tr}[\phi_L(i) \delta A_{\mathcal{S}_{L}^{(k)}}(\aLocal_L)]]
    \leq \mathrm{Tr} [\rho_L^\mathrm{can} \exp[\tau \delta A_{\mathcal{S}_{L}^{(k)}}(\aLocal_L)]].
\end{align}
Combining this with Assumption~\ref{assumption:GibbsState_Concentration}, we have
\begin{align}
    \sum_i \frac{P(i)}{Z} \exp [\tau \mathrm{Tr}[\phi_L(i) \delta A_{\mathcal{S}_{L}^{(k)}}(\aLocal_L)]]
    \leq \exp\Bigl[\frac{\chi_L\tau^2}{2} \Bigr]
\end{align}
for any $\tau\in(-\tau_0,\tau_0)$.

Next, using these inequalities, we give an upper bound for the probability that the expectation value of $A_{\mathcal{S}_{L}^{(k)}}(\aLocal_L)$ in $\phi_{L}(i)$ deviates from that in the canonical Gibbs state. For any $\tau>0$ and $\epsilon>0$, the following identity holds:
\begin{align}
    \mathrm{Prob}_i \Bigl[ \mathrm{Tr}[\phi_L(i) \frac{\delta A_{\mathcal{S}_{L}^{(k)}}(\aLocal_L)}{|\mathcal{S}_{L}^{(k)}|}] \ge \epsilon\Bigr]
    &=\mathrm{Prob}_i \Bigl[ \exp [\tau \mathrm{Tr}[\phi_L(i) \delta A_{\mathcal{S}_{L}^{(k)}}(\aLocal_L)] \ge e^{\epsilon\tau |\mathcal{S}_{L}^{(k)}|}\Bigr].
\end{align}
Applying the Markov inequality, we have
\begin{align}
    \mathrm{Prob}_i \Bigl[ \mathrm{Tr}[\phi_L(i) \frac{\delta A_{\mathcal{S}_{L}^{(k)}}(\aLocal_L)}{|\mathcal{S}_{L}^{(k)}|}] \ge \epsilon\Bigr]
    &\le e^{-\epsilon\tau |\mathcal{S}_{L}^{(k)}|}\sum_i \frac{P(i)}{Z} \exp [\tau \mathrm{Tr}[\phi_L(i) \delta A_{\mathcal{S}_{L}^{(k)}}(\aLocal_L)]]\\
    &\le e^{-\epsilon\tau |\mathcal{S}_{L}^{(k)}|} \exp\Bigl[\frac{\chi_L\tau^2}{2} \Bigr].
\end{align}
The same calculation applies to $\mathrm{Prob}_i \Bigl[ -\mathrm{Tr}[\phi_L(i) \frac{\delta A_{\mathcal{S}_{L}^{(k)}}(\aLocal_L)}{|\mathcal{S}_{L}^{(k)}|}] \ge \epsilon\Bigr]$.
Therefore, we have
\begin{align}
    \mathrm{Prob}_i \Bigl[ \Bigl|\mathrm{Tr}[\phi_L(i) \frac{\delta A_{\mathcal{S}_{L}^{(k)}}(\aLocal_L)}{|\mathcal{S}_{L}^{(k)}|}]\Bigr| \ge \epsilon\Bigr]
    &\le 2e^{-\epsilon\tau |\mathcal{S}_{L}^{(k)}|} \exp\Bigl[\frac{\chi_L\tau^2}{2} \Bigr].
\end{align}
Substituting $\epsilon=\delta_L/D^{{\ell_L}^d}$ and $\tau=(\delta_L/D^{{\ell_L}^d})\times(|\mathcal{S}_{L}^{(k)}|/\chi_L)$, we have
\begin{align}
    \mathrm{Prob}_i \Bigl[ \Bigl|\mathrm{Tr}[\phi_L(i) \frac{\delta A_{\mathcal{S}_{L}^{(k)}}(\aLocal_L)}{|\mathcal{S}_{L}^{(k)}|}]\Bigr| \ge \frac{\delta_L}{D^{{\ell_L}^d}}\Bigr]
    &\le 2\exp\Bigl[-\frac{|\mathcal{S}_L^{(k)}|^2}{2\chi_L} (\delta_L/D^{{\ell_L}^d})^2\Bigr].
\end{align}
Note that $\tau=(\delta_L/D^{2{\ell_L}^d})\times(|\mathcal{S}_{L}^{(k)}|/\chi_L)=o(L^0)$ holds from Eqs.~\eqref{eq:AssumptionConcentration_chiL=Omega(N)}--\eqref{eq:Concentration_ellL}, and hence, $\tau\in(0,\tau_0)$ is satisfied.

Now, we show Eq.~\eqref{eq:METTS_Decay_pL}.
From Eq.~\eqref{eq:Hilbert–Schmidt_expansion_bound}, we have
\begin{align}
    \mathrm{Prob}_i \left[ \left\|
         \phi_L^{(k)}(i) - \rho_L^\mathrm{can}
    \right\|_1 \geq \delta_L\right]
    \le \sum_{\mu=1}^{D^{2{\ell_L}^d}}\mathrm{Prob}_i \Bigl[ \Bigl|\mathrm{Tr}[\phi_L(i) \frac{\delta A_{\mathcal{S}_{L}^{(k)}}(u^\mu)}{|\mathcal{S}_{L}^{(k)}|}]\Bigr|^2 \ge \frac{\delta_L^2}{D^{2{\ell_L}^d}}\Bigr]
    \leq 2 D^{2 {\ell_L}^d} \exp\Bigl[-\frac{|\mathcal{S}_L^{(k)}|^2}{2\chi_L} (\delta_L/D^{{\ell_L}^d})^2\Bigr].
\end{align}
Taking a sum over all $k\in\{1,2,...,K^d\}$, we have
\begin{align}
    p_L \leq 2 K^d D^{2 {\ell_L}^d} \exp\Bigl[-\frac{|\mathcal{S}_L^{(k)}|^2}{2\chi_L} (\delta_L/D^{{\ell_L}^d})^2\Bigr].
\end{align}
From Eqs.~\eqref{eq:AssumptionConcentration_chiL=o(N^2)}, \eqref{eq:Concentration_deltaL}, \eqref{eq:Concentration_ellL}, the RHS scales as
\begin{align}
    2 K^d D^{2 {\ell_L}^d} \exp\Bigl[-\frac{|\mathcal{S}_L^{(k)}|^2}{2\chi_L} (\delta_L/D^{{\ell_L}^d})^2\Bigr]
    =O\bigl((\log L)^2\bigr)\times\exp\bigl[-\Omega(N^\nu)/(\log L)^4\bigr].
\end{align}
This implies that, at sufficiently large $L$, $p_L$ decays faster than, e.g., $O(1/L^2)$, which satisfies Eq.~\eqref{eq:METTS_Decay_pL}.
Therefore, we have
\begin{align}
    \sum_{L=1}^{\infty} p_L
    \leq \sum_{L=1}^{\infty} 2 K^d D^{2 {\ell_L}^d} \exp\Bigl[-\frac{|\mathcal{S}_L^{(k)}|^2}{2\chi_L} (\delta_L/D^{{\ell_L}^d})^2\Bigr]
    < \infty,
\end{align}
that is, Eq.~\eqref{eq:METTS_Decay_pL} is satisfied.

Finally, we briefly show that $(\rho_L)_{L\in\mathbb{N}}$ represents iMATE almost surely, imposing Assumption~\ref{assumption:GibbsState_Concentration} instead of Assumption~\ref{assumption:GibbsState_LargerLocality}. 
We emphasize that Assumption~\ref{assumption:GibbsState_LargerLocality} is not apparently used in the proof of Proposition~\ref{proposition:TypicalMETTS_represents_iMATE}.
Therefore, we can apply the proof of Proposition~\ref{proposition:TypicalMETTS_represents_iMATE} almost directly, replacing Eqs.~\eqref{eq:METTS_Large_ell}--\eqref{eq:METTS_Small_ell} by Eqs.~\eqref{eq:Concentration_deltaL} and \eqref{eq:Concentration_ellL} as the definitions for $\ell_L$ and $\delta_L$, which enable to obtain the desired result.

\end{proof}

\subsection{\label{sec:Proof_Example_iMATE_InfiniteTemp}Infinite temperature iMATE without random variables (details of Example~\ref{example:NonMITE_MacroEqState})}

In this subsection, we provide a simple example of representations of iMATE at $\beta=0$ that contains no random variables in contrast to the typical sequence of METTS.

\begin{proposition}[\label{proposition:NonMITE_MacroEqState}Infinite temperature iMATE without random variables]
For simplicity, we consider a one-dimensional spin-$1/2$ system.
The extension to systems in higher spatial dimensions is straightforward.
Let the system size $L$ be expressed as $L = 2^n n$ for an integer $n$.
We assume that the number of subdivisions $K$ used to define the primitive macroscopic subsystem $\mathcal{S}^{(k)}$ (see Sec.~\ref{sec:Setup_Lattice}) can be written as $K=2^\kappa$, where $\kappa$ can be taken arbitrarily large as long as it is independent of $L$.

We write the computational basis of a system composed of $n$ sites by $\{\ket{b_n(i)}\}_{i=0,1,\dots,2^n-1}$, where $b_n(i)$ maps each integer $0 \leq i \leq 2^n-1$ to an $n$-bit binary sequence.
Thus, the basis states are enumerated as
\begin{align}
    \ket{b_n(0)} &= \underbrace{\ket{0 \dots 0 0}}_{n},\quad
    \ket{b_n(1)} = \underbrace{\ket{0 \dots 0 1}}_{n}, \nonumber\\
    \ket{b_n(2)} &= \underbrace{\ket{0 \dots 1 0}}_{n},\quad
    \ket{b_n(3)} = \underbrace{\ket{0 \dots 1 1}}_{n},\quad
    \dots
    \label{eq:basis_binary}
\end{align}
Using this, we consider a pure state $\rho_L = \ket{\psi_L}\bra{\psi_L}$ of the system on $\Lambda_L$, where
\begin{align}
    \ket{\psi_L} = \bigotimes_{i=0}^{2^n-1} \ket{b_n(i)}_{ni+1}.
    \label{eq:Example_MacroEqState_infty}
\end{align}
Here, $\ket{b_n(i)}_{ni+1}$ is defined on $\{ni+1,...,ni+n\}$.
In this state, the expectation values of local observables, such as Pauli-z operator $\sigma_{j}^{z}$, depend on the site. Therefore, this state does not represent MITE of a system described by a translation-invariant Hamiltonian.
By contrast, 
$(\rho_L)_L$ represents iMATE at inverse temperature $\beta=0$ for any Hamiltonian.

\end{proposition}

\begin{proof}
We divide the total lattice $\Lambda_L$ into two types of disjoint subsets 
\begin{align}
    \Lambda_L&=\Bigl(\bigcup_{i=0}^{2^n-1}\mathcal{I}_{n,i}\Bigr)\cup \Bigl(\bigcup_{i=0}^{2^n-1}\mathcal{J}_{n,i}\Bigr),\\
    \mathcal{I}_{n,i}&:=\{ni+1,...,ni+\kappa\},\\
    \mathcal{J}_{n,i}&:=\{ni+\kappa+1,...,ni+n\}\quad \text{ for }i=0,...,2^n-1.
\end{align}
Because they satisfy $|\mathcal{I}_{n,i}|=\kappa$ and $|\mathcal{J}_{n,i}|=n-\kappa$, $\mathcal{I}_{n,i}$'s occupy a negligibly small fraction of the system and $\mathcal{J}_{n,i}$'s occupy the majority of the system,
\begin{align}
    \Bigl|\bigcup_{i=0}^{2^n-1}\mathcal{I}_{n,i}\Bigr|\Bigm/N=\kappa/n \to 0,
    \qquad \Bigl|\bigcup_{i=0}^{2^n-1}\mathcal{J}_{n,i}\Bigr|\Bigm/N=(n-\kappa)/n\to 1,
    \qquad \text{as } L\to\infty.
\end{align}
Using these, primitive macroscopic subsystems can be written as
\begin{align}
    &\mathcal{S}_{L}^{(k)}=\{(k-1)n2^{n-\kappa}+1, (k-1)n2^{n-\kappa}+2,...,kn2^{n-\kappa}\}
    =\mathcal{I}_{L}^{(k)}\cup\mathcal{J}_L^{(k)},
\end{align}
where
\begin{align}
    \mathcal{I}_{L}^{(k)}:=\bigcup_{i=(k-1)2^{n-\kappa}}^{k2^{n-\kappa}-1}\mathcal{I}_{n,i},
    \quad\mathcal{J}_{L}^{(k)}:=\bigcup_{i=(k-1)2^{n-\kappa}}^{k2^{n-\kappa}-1}\mathcal{J}_{n,i}\qquad\text{ for }k=1,2,...,2^\kappa.
\end{align}
This implies that the additive observable on $\mathcal{S}_{L}^{(k)}$ obtained from $\aLocal$ can be well approximated by the sum of translations of $\aLocal$ on $\mathcal{J}_{L}^{(k)}$, 
\begin{align}
    \Bigl\|A_{\mathcal{S}_{L}^{(k)}}(\aLocal)-\sum_{j\text{ s.t. }\mathrm{Supp}(\aLocal_j)\subset\mathcal{J}_L^{(k)}}\aLocal_j\Bigr\|
    \le \|\aLocal\|\sum_{i=(k-1)2^{n-\kappa}}^{k2^{n-\kappa}-1}\bigl(|\mathcal{I}_{n,i}|+2|\mathrm{Supp}(\aLocal)|\bigr)= \|\aLocal\|\bigl(\kappa+2|\mathrm{Supp}(\aLocal)|\bigr)2^{n-\kappa}=N\times O(1/n),
    \label{eq:Example_iMATE_Infinite_OpNormApprox}
\end{align}
where $\aLocal_j=T_j \aLocal T_j^{-1}$ is the translation of $\aLocal$ by $j$ sites.
By taking the expectation value in the state~\eqref{eq:Example_MacroEqState_infty}, we have 
\begin{align}
    \bra{\psi_L}A_{\mathcal{S}_{L}^{(k)}}(\aLocal)\ket{\psi_L}/N
    &=\frac{1}{N}\sum_{j\text{ s.t. }\mathrm{Supp}(\aLocal_j)\subset\mathcal{J}_L^{(k)}} \bra{\psi_L}\aLocal_j\ket{\psi_L}+O(1/n)\\
    &=\frac{1}{N}\sum_{i=(k-1)2^{n-\kappa}}^{k2^{n-\kappa}-1}\sum_{j\text{ s.t. }\mathrm{Supp}(\aLocal_j)\subset\mathcal{J}_{n,i}} \bra{\psi_L}\aLocal_j\ket{\psi_L}+O(1/n)
    \\
    &=\frac{1}{N}\sum_{i=(k-1)2^{n-\kappa}}^{k2^{n-\kappa}-1}\sum_{j\text{ s.t. }\mathrm{Supp}(\aLocal_j)\subset\mathcal{J}_{n,i}} {}_{ni+1}\bra{b_n(i)}\aLocal_j\ket{b_n(i)}_{ni+1}+O(1/n),
\end{align}
where the definition~\eqref{eq:Example_MacroEqState_infty} is used in the last equality. 
Now we introduce a symbol $\tilde{i}$ as an integer satisfying 
\begin{align}
    \tilde{i}\in\{0,1,2,...,2^{n-\kappa}-1\}, \quad \tilde{i}\equiv i\ (\mathrm{mod} \ 2^{n-\kappa}),
    \quad \text{for }i\in\{0,1,...,2^n\}.
    \label{eq:Example_iMATE_Infinite_tilde_i}
\end{align}
Using this, the basis~\eqref{eq:basis_binary} can be written as
\begin{align}
    \ket{b_n(i)}_{ni+1}=\ket{b_{\kappa}(k)}_{\mathcal{I}_{n,i}}\otimes\ket{b_{n-\kappa}(\tilde{i})}_{\mathcal{J}_{n,i}},
\end{align}
where $\ket{b_{\kappa}(k)}_{\mathcal{I}_{n,i}}$ is defined on $\mathcal{I}_{n,i}$ and $\ket{b_{n-\kappa}(\tilde{i})}_{\mathcal{J}_{n,i}}$ on $\mathcal{J}_{n,i}$.
Inserting this into the above expression, we have 
\begin{align}
    \bra{\psi_L}A_{\mathcal{S}_{L}^{(k)}}(\aLocal)\ket{\psi_L}/N
    &=\frac{1}{N}\sum_{i=(k-1)2^{n-\kappa}}^{k2^{n-\kappa}-1}\sum_{j\text{ s.t. }\mathrm{Supp}(\aLocal_j)\subset\mathcal{J}_{n,i}} {}_{\mathcal{J}_{n,i}}\bra{b_{n-\kappa}(\tilde{i})}\aLocal_j\ket{b_{n-\kappa}(\tilde{i})}_{\mathcal{J}_{n,i}}+O(1/n),
\end{align}
which follows from $\mathrm{Supp}(\aLocal_j)\subset \mathcal{J}_{n,i}$.
Since $\mathcal{J}_{n,i}$ is the translation of $\mathcal{J}_{n,0}$ by $n(i+1)$ sites, we can rewrite as
\begin{align}
    \sum_{j\text{ s.t. }\mathrm{Supp}(\aLocal_j)\subset\mathcal{J}_{n,i}} {}_{\mathcal{J}_{n,i}}\bra{b_{n-\kappa}(\tilde{i})}\aLocal_j\ket{b_{n-\kappa}(\tilde{i})}_{\mathcal{J}_{n,i}}
    =\sum_{j\text{ s.t. }\mathrm{Supp}(\aLocal_j)\subset\mathcal{J}_{n,0}} {}_{\mathcal{J}_{n,0}}\bra{b_{n-\kappa}(\tilde{i})}\aLocal_j\ket{b_{n-\kappa}(\tilde{i})}_{\mathcal{J}_{n,0}}.
\end{align}
Inserting this, we have 
\begin{align}
    \bra{\psi_L}A_{\mathcal{S}_{L}^{(k)}}(\aLocal)\ket{\psi_L}/N
    &=\frac{1}{N}\sum_{i=(k-1)2^{n-\kappa}}^{k2^{n-\kappa}-1}\sum_{j\text{ s.t. }\mathrm{Supp}(\aLocal_j)\subset\mathcal{J}_{n,0}} {}_{\mathcal{J}_{n,0}}\bra{b_{n-\kappa}(\tilde{i})}\aLocal_j\ket{b_{n-\kappa}(\tilde{i})}_{\mathcal{J}_{n,0}}+O(1/n)\\
    &=\frac{1}{N} \sum_{i^\prime=0}^{2^{n-\kappa}-1} \sum_{j\text{ s.t. }\mathrm{Supp}(\aLocal_j)\subset\mathcal{J}_{n,0}} {}_{\mathcal{J}_{n,0}}\bra{b_{n-\kappa}(i^\prime)}\aLocal_j\ket{b_{n-\kappa}(i^\prime)}_{\mathcal{J}_{n,0}}+O(1/n).
    \label{eq:Example_iMATE_Orthonormal}
\end{align}
Furthermore, since $\{\ket{b_{n-\kappa}}(i^\prime)\}_{i^\prime=0}^{2^{n-\kappa}-1}$ forms an orthonormal basis on $\mathcal{J}_{n,0}$, the last expression can be rewritten using the trace on $\mathcal{J}_{n,0}$, $\mathrm{Tr}_{\mathcal{J}_{n,0}}[\bullet]$, as 
\begin{align}
    \bra{\psi_L}A_{\mathcal{S}_{L}^{(k)}}(\aLocal)\ket{\psi_L}/N
    =\frac{1}{N} \sum_{j\text{ s.t. }\mathrm{Supp}(\aLocal_j)\subset\mathcal{J}_{n,0}} \mathrm{Tr}_{\mathcal{J}_{n,0}}[\aLocal_j]+O(1/n).
\end{align}
Since the trace does not depend on the site $j$, we have
\begin{align}
    \bra{\psi_L}A_{\mathcal{S}_{L}^{(k)}}(\aLocal)\ket{\psi_L}/N
    &=\frac{\mathrm{Tr}_{\mathcal{J}_{n,0}}[\aLocal]}{N} |\{j\in\Lambda_L|\mathrm{Supp}(\aLocal_j)\subset\mathcal{J}_{n,0}\}| +O(1/n)\\
    &=\frac{1}{K}\frac{\mathrm{Tr}_{\mathcal{J}_{n,0}}[\aLocal]}{2^{n-\kappa}} +O(1/n).
    \label{eq:Example_iMATE_ExpValue_psi}
\end{align}
On the other hand, recalling that $\rho_{L}^{\mathrm{can}}(\beta=0|H_L)=I_{\Lambda_L}/2^{N}$, we have
\begin{align}
    \mathrm{Tr}[\rho_{L}^{\mathrm{can}}(\beta=0|H_L)A_{\mathcal{S}_{L}^{(k)}}(\aLocal)]/N
    =\frac{|\mathcal{S}_{L}^{(k)}|}{N}\frac{\mathrm{Tr}[\aLocal]}{2^N}=\frac{1}{K}\frac{\mathrm{Tr}[\aLocal]}{2^N}.
    \label{eq:Example_iMATE_ExpValue_can}
\end{align}
Combining Eqs.~\eqref{eq:Example_iMATE_ExpValue_psi} and \eqref{eq:Example_iMATE_ExpValue_can}, we have 
\begin{align}
    \bra{\psi_L}A_{\mathcal{S}_{L}^{(k)}}(\aLocal)\ket{\psi_L}/N
    =\mathrm{Tr}[\rho_{L}^{\mathrm{can}}(\beta=0|H_L)A_{\mathcal{S}_{L}^{(k)}}(\aLocal)]/N+O(1/n).
\end{align}
Because of the additivity property~\eqref{eq:Additive_S_approximation}, the same equality holds for all macroscopic subsystems $\mathcal{S}_L$.
This implies that $(\rho_L)_{L\in\mathbb{N}}$ with $\rho_L=\ket{\psi_L}\bra{\psi_L}$ represents iMATE at $\beta=0$.
\end{proof}

\subsection{\label{sec:Proof_Example_iMATE_FiniteTemp}Finite temperature iMATE without random variables 
(details of Example~\ref{example:NonMITE_MacroEqState_FiniteTemp})}

In this subsection, we explicitly construct a representation of iMATE at finite temperature. Importantly, it contains no random variables in contrast to METTS.

\begin{proposition}[\label{proposition:NonMITE_MacroEqState_FiniteTemp}Finite temperature iMATE without random variables]
For simplicity, we consider a one-dimensional spin-$1/2$ system $(d=1,D=2)$; the extension to larger $d$ and $D$ is straightforward.
We restrict the system size to the form
\begin{align}
    L = 2^n n \qquad (n\in\mathbb{N}),
\end{align}
and assume that the number of subdivisions used to define the primitive macroscopic subsystems $\mathcal{S}^{(k)}$ (see Sec.~\ref{sec:Setup_Lattice}) is
\begin{align}
    K = 2^\kappa,
\end{align}
where $\kappa$ is fixed (independent of $L$) and can be chosen arbitrarily large.

We can construct pure quantum states $\{\ket{\varphi(i)}\}_{i=0}^{2^{n}-1}$ defined on $(n-\kappa)$ consecutive qubits such that the product state
\begin{align}
    \ket{\psi_L;\beta}
    = 
    \bigotimes_{i=0}^{2^n-1}\bigl(\ket{0}^{\otimes \kappa}\otimes\ket{\varphi(i)}\bigr),
    \label{eq:Example_MacroEqState_finiteT}
\end{align}
represents iMATE described by $H_L$ at inverse temperature $\beta$,
\begin{align}
    \bigl( \ket{\psi_L;\beta}\bra{\psi_L;\beta} \bigr)_L \maceq \left( \rho_{L}^{\mathrm{can}}(\beta|H_L) \right)_L.
\end{align}
\end{proposition}
Note that since $\ket{\psi_L;\beta}$ is a product state, $( \ket{\psi_L;\beta}\bra{\psi_L;\beta} )_L$ does not represent MITE.

\begin{proof}
We work in one dimension and restrict the system size to $L=2^n n$ with $n\in\mathbb{N}$.
We fix $\kappa$ independent of $n$ and set $K=2^\kappa$.

We partition the chain $\Lambda_L=\{1,2,\cdots,L\}$ into $\mathcal{I}_{n,i}$'s and $\mathcal{J}_{n,i}$'s as in Sec.~\ref{sec:Proof_Example_iMATE_InfiniteTemp}.
The Hamiltonian on the subset $\mathcal{J}_{n,i}$ is given by
\begin{align}
    H_{n,i}:=\sum_{j(\mathrm{Supp}(h_j)\subset\mathcal{J}_{n,i})}h_j,
\end{align}
where $h_{j}=T_{j}hT_{j}^{-1}$. 
We define the canonical Gibbs state on $\mathcal{J}_{n,i}$ by
\begin{align}
    \rho_{n,i}^{\mathrm{can}}=\frac{e^{-\beta H_{n,i} }}{\mathrm{Tr}[e^{-\beta H_{n,i} }]}.
\end{align}
Let $\{\ket{E_\nu}_{\mathcal{J}_{n,i}}\}_{\nu=1}^{2^{n-\kappa}}$ be eigenbasis of $H_{n,i} $ with eigenvalue $E_\nu$.
The canonical Gibbs state is expanded as
\begin{align}
    \rho_{n,i}^{\mathrm{can}}
    &= \sum_{\nu=1}^{2^{n-\kappa}} \lambda_\nu \ket{E_\nu}_{\mathcal{J}_{n,i}}\bra{E_\nu}_{\mathcal{J}_{n,i}},\\
    \lambda_{\nu}&=\frac{e^{-\beta E_\nu}}{\mathrm{Tr}[e^{-\beta H_{n,i} }]}.
\end{align}
Let $\omega = e^{2\pi i/2^{n-\kappa}}$ and perform the Fourier transform,
\begin{align}
    \ket{q}_{\mathcal{J}_{n,i}}= \frac{1}{\sqrt{2^{n-\kappa}}} \sum_{\nu=1}^{2^{n-\kappa}} \omega^{q\nu} \ket{E_\nu}_{\mathcal{J}_{n,i}}
    \qquad \text{for }q\in\{0,1,...,2^{n-\kappa}-1\}.
\end{align}
Then, for every $q$,
\begin{align}
    \bra{q}_{\mathcal{J}_{n,i}} \rho_{n,i}^{\mathrm{can}} \ket{q}_{\mathcal{J}_{n,i}}
    = \sum_{\nu=1}^{2^{n-\kappa}} \lambda_\nu \left| \bra{E_\nu}_{\mathcal{J}_{n,i}}\ket{q}_{\mathcal{J}_{n,i}} \right|^2
    = \frac{1}{2^{n-\kappa}} \sum_{\nu=1}^{2^{n-\kappa}} \lambda_\nu
    = \frac{1}{2^{n-\kappa}}.
    \label{eq:flat-diagonal-element}
\end{align}

For $q\in\{0,1,...,2^{n-\kappa}-1\}$, define
\begin{align}
    \ket{\varphi(q)}_{\mathcal{J}_{n,i}} = \sqrt{2^{n-\kappa}} (\rho_{n,i}^{\mathrm{can}})^{1/2} \ket{q}_{\mathcal{J}_{n,i}}.
\end{align}
From Eq.~\eqref{eq:flat-diagonal-element}, each $\ket{\varphi(q)}$ is normalized:
\begin{align}
    \|\ket{\varphi(q)}_{\mathcal{J}_{n,i}}\|^2
    = 2^{n-\kappa} \bra{q}_{\mathcal{J}_{n,i}} \rho_{n,i}^{\mathrm{can}} \ket{q}_{\mathcal{J}_{n,i}}
    = 1.
\end{align}
Importantly, it holds that
\begin{align}
    \frac{1}{2^{n-\kappa}} \sum_{q=1}^{2^{n-\kappa}} \ket{\varphi(q)}_{\mathcal{J}_{n,i}}\bra{\varphi(q)}_{\mathcal{J}_{n,i}}
    &= (\rho_{n,i}^{\mathrm{can}})^{1/2} \left( \sum_{q=1}^{2^{n-\kappa}} \ket{q}_{\mathcal{J}_{n,i}}\bra{q}_{\mathcal{J}_{n,i}} \right) (\rho_{n,i}^{\mathrm{can}})^{1/2}
    = \rho_{n,i}^{\mathrm{can}}.
    \label{eq:block-averaging-identity}
\end{align}
Using this, we define
\begin{align}
    \ket{\psi_L;\beta}
    = \bigotimes_{i=0}^{2^{n}-1} \ket{0}^{\otimes \mathcal{I}_{n,i}}\otimes\ket{\varphi(\tilde{i})}_{\mathcal{J}_{n,i}},
\end{align}
where $\tilde{i}$ is defined as Eq.~\eqref{eq:Example_iMATE_Infinite_tilde_i}.

Let $\aLocal$ be an arbitrary $\ell$-local observable for an arbitrary $\ell\in\mathbb{N}$. 
By taking $n$ sufficiently large, $|\mathcal{J}_{n,i}|=n-\kappa$ can be taken sufficiently larger than $\ell=O(L^0)$.
By the same calculation as Eqs.~\eqref{eq:Example_iMATE_Infinite_OpNormApprox}--\eqref{eq:Example_iMATE_Orthonormal}, we have 
\begin{align}
    \bra{\psi_L;\beta}A_{\mathcal{S}_{L}^{(k)}}(\aLocal)\ket{\psi_L;\beta}/N
    &=\frac{1}{N}\sum_{j\text{ s.t. }\mathrm{Supp}(\aLocal_j)\subset\mathcal{J}_L^{(k)}} \bra{\psi_L;\beta}\aLocal_j\ket{\psi_L;\beta}+O(1/n)\\
    &=\frac{1}{N} \sum_{i^\prime=0}^{2^{n-\kappa}-1} \sum_{j\text{ s.t. }\mathrm{Supp}(\aLocal_j)\subset\mathcal{J}_{n,0}} {}_{\mathcal{J}_{n,0}}\bra{\varphi(i^\prime)}\aLocal_j\ket{\varphi(i^\prime)}_{\mathcal{J}_{n,0}}+O(1/n).
\end{align}
Since $\{\ket{q}_{\mathcal{J}_{n,0}}\}_{q=0}^{2^{n-\kappa}-1}$ forms an orthonormal basis on $\mathcal{J}_{n,0}$, it holds that 
\begin{align}
    \bra{\psi_L;\beta}A_{\mathcal{S}_{L}^{(k)}}(\aLocal)\ket{\psi_L;\beta}/N
    =\frac{2^{n-\kappa}}{n 2^{n}} \sum_{j\text{ s.t. }\mathrm{Supp}(\aLocal_j)\subset\mathcal{J}_{n,0}} \mathrm{Tr}_{\mathcal{J}_{n,0}}[\rho_{n,0}^{\mathrm{can}}\aLocal_j]+O(1/n).
\end{align}
Taking the thermodynamic limit, we have
\begin{align}
    \lim_{L\to\infty}\bra{\psi_L;\beta}A_{\mathcal{S}_{L}^{(k)}}(\aLocal)\ket{\psi_L;\beta}/N
    =\frac{1}{K} \lim_{n\to\infty}\frac{1}{n}\sum_{j\text{ s.t. }\mathrm{Supp}(\aLocal_j)\subset\mathcal{J}_{n,0}} \mathrm{Tr}_{\mathcal{J}_{n,0}}[\rho_{n,0}^{\mathrm{can}}\aLocal_j].
\end{align}
The RHS represents the expectation value of the density of the additive observable obtained from $\aLocal$ in the canonical Gibbs state in the system of size $n$. Therefore, rewriting $n$ by $L$, we have
\begin{align}
    \frac{1}{K} \lim_{n\to\infty}\frac{1}{n}\sum_{j\text{ s.t. }\mathrm{Supp}(\aLocal_j)\subset\mathcal{J}_{n,0}} \mathrm{Tr}_{\mathcal{J}_{n,0}}[\rho_{n,0}^{\mathrm{can}}\aLocal_j]
    &=\frac{1}{K} \lim_{L\to\infty} \mathrm{Tr}_{\Lambda_L}[\rho_L^{\mathrm{can}}(\beta|H_L)\frac{A_{\Lambda_L}(\aLocal)}{N}]\\
    &=\lim_{L\to\infty} \mathrm{Tr}_{\Lambda_L}[\rho_L^{\mathrm{can}}(\beta|H_L)\frac{A_{\mathcal{S}_L^{(k)}}(\aLocal)}{N}].
\end{align}
This means that 
\begin{align}
    \lim_{L\to\infty}\bra{\psi_L;\beta}\frac{A_{\mathcal{S}_{L}^{(k)}}(\aLocal)}{N}\ket{\psi_L;\beta}
    =\lim_{L\to\infty} \mathrm{Tr}_{\Lambda_L}[\rho_L^{\mathrm{can}}(\beta|H_L)\frac{A_{\mathcal{S}_L^{(k)}}(\aLocal)}{N}]
\end{align}
holds for arbitrary $k\in\{1,2,...,K^d\}$ and for an arbitrary $\ell$-local observable $\aLocal$ with $\ell=O(L^0)$.
From Eq.~\eqref{eq:Additive_S_approximation}, this implies
\begin{align}
    (\ket{\psi_L;\beta}\bra{\psi_L;\beta})_{L\in\mathbb{N}}\maceq \bigl(\rho_L^{\mathrm{can}}(\beta|H_L)\bigr)_{L\in\mathbb{N}},
\end{align}
i.e., $(\ket{\psi_L;\beta}\bra{\psi_L;\beta})_{L\in\mathbb{N}}$ represents iMATE described by $H_L$ at inverse temperature $\beta$.

Note that $\ket{\psi_L;\beta}$ is the product of states on $\mathcal{I}_{n,i}$'s and those on $\mathcal{J}_{n,i}$'s. This means that the adjacent $\mathcal{I}_{n,i}$'s and $\mathcal{J}_{n,i}$'s have no correlation. Furthermore, the states on $\mathcal{I}_{n,i}$ can be taken arbitrarily; hence, for simplicity, we took them as the tensor product of $\ket{0}$'s. These facts indicate that $(\ket{\psi_L;\beta}\bra{\psi_L;\beta})_{L\in\mathbb{N}}$ does not represent MITE.

\end{proof}

\section{Proof of Propositions~\ref{proposition:JAIVTMXN_1}, \ref{proposition:JAIVTMXN_2}, and Theorem~\ref{theorem:JAIVTMXN_3}} \label{sec:JAIVTMXN}
Our proof relies crucially on the Lieb-Robinson bound, which describes the emergent light cone that limits information propagation in locally interacting many-body systems. Before explaining this bound, we introduce some notation.
Here, we consider a general setting where the Hamiltonian $H(t)$ can be written as
\begin{align}
    H(t) = \sum_{X \subset \Lambda_L} H_X(t),
\end{align}
where $H_X(t)$ is a Hermitian operator supported on $X$.
Note that $L$ dependence is omitted for the simplicity of notation. To represent regular locally interacting systems, we impose two conditions on this Hamiltonian. One is that this Hamiltonian $H(t)$ is composed of $r_{H}$-local observables for some $r_{H}\in\mathbb{N}$ independent of $L$ and $t$, that is, 
\begin{align}
    H_{X}(t)=0\qquad \text{for }l_X>r_{H},
    \label{eq:r_H_local}
\end{align} 
where $l_X = \max_{\bm{r},\bm{r}^\prime \in X} |\bm{r}-\bm{r}^\prime|$ is the diameter of $X$.
The other is that the operator norm of $H_{X}(t)$ is bounded from above by some constant $C_H$ independent of $t,X,L$,
\begin{align}
    \|H_{X}(t)\|\le C_{H}\qquad \text{for all }X\subset\Lambda_L\text{ and all }t\in\mathbb{R}.
\end{align}
Importantly, for the time-dependent Hamiltonian~\eqref{eq:general_H}, which generates a macroscopic operation, both conditions are satisfied.
Under these conditions, the following norm~$\|H(t)\|_{\mu}$
for $\mu \geq 0$ becomes of $O(L^0)$, 
\begin{align}
    \| H(t) \|_{\mu} := \max_{\bm{r} \in \Lambda_L} \sum_{X \ni \bm{r}} \| H_X(t) \| e^{\mu l_X}=O(L^0).
    \label{eq:LRbound_norm}
\end{align}

Now we present the Lieb-Robinson bound for the time-dependent Hamiltonian.
\begin{lemma}[Theorem~2 of Ref.~\cite{Gong2020}] \label{lemma:Lieb-Robinson-bound}
Given a (time-dependent) Hamiltonian $H(t)$
with well-defined $\|H(t)\|_{\mu} < \infty$ for all $\mu < \mu^*$ and all $t \in \mathbb{R}$,
then for two arbitrary local operators $O_X$ and $O_Y$ with supports $X$ and $Y$, respectively,
and any $\kappa, \eta > 0$ with $\kappa + \eta < \mu^*$,
we have
\begin{align}
    &\| [ {U_L(t, t_0)}^\dagger O_X U_L(t, t_0), O_Y ] \|
    \leq 2 \min \{ |X|, |Y| \} \|O_X\| \|O_Y\|
    \exp\Bigl\{- \kappa [\mathrm{dist}(X,Y) - \frac{2 C_d}{\kappa} \left( \frac{d}{e \eta} \right)^d \int_{t_0}^{t} ds \| H(s) \|_{\kappa + \eta} ]\Bigr\},
\end{align}
where $U_L(t, t_0)$ is the unitary time-evolution operator generated by $H(t)$ during $[t_0, t]$,
$\mathrm{dist}(X,Y)$ is the distance between $X, Y \subset \Lambda_L$,
and $C_d$ is a constant (determined solely by the lattice geometry) that validates $|X| \leq C_d (l_X + 1)^d$ for all $X \subset \Lambda_L$.
\end{lemma}
\noindent
By virtue of the Lieb-Robinson bound, as we will do below, we can study the properties of the Heisenberg operator for an additive observable from the properties of additive observables.

We first observe that an additive observable on a macroscopic subsystem that  can be expressed as a composite system of macroscopic subsystems can be approximated by the sum of additive observables on each subsystem. Since we assume that any macroscopic subsystem can be expressed as a composite system of primitive macroscopic subsystems, it suffices to examine the properties of additive observables on primitive macroscopic subsystems (see Eq.~\eqref{eq:Additive_S_approximation}).
Therefore, in the following, we investigate the properties of the Heisenberg operator for an additive observable $A$ on a primitive macroscopic subsystem $\mathcal{S}^{(k)}$.
To be concrete, 
consider the additive observable $A$ on $\mathcal{S}^{(k)}$ obtained from an $\ell_A$-local observable $\aLocal$ for some $\ell_A\in\mathbb{N}$,
\begin{align}
    A
    = \sum_{\supp(\aLocal)+\bm{r} \subset \mathcal{S}^{(k)}} 
    T_{\bm{r}} \aLocal T_{\bm{r}}^\dagger
    = \sum_{\supp(\aLocal_{\bm{r}}) \subset \mathcal{S}^{(k)}} \aLocal_{\bm{r}}.
\end{align}
Here we write $T_{\bm{r}} \aLocal T_{\bm{r}}^\dagger$ as $\aLocal_{\bm{r}}$.
Then the Heisenberg operator $A_H(t)$ for $A$ can be written as
\begin{align}
    A^H(t)
    = {U_L(t, 0)}^\dagger A U_L(t, 0)
    = \sum_{\supp(\aLocal_{\bm{r}}) \subset \mathcal{S}^{(k)}} {U_L(t, 0)}^\dagger \aLocal_{\bm{r}} U_L(t, 0)
    = \sum_{\supp(\aLocal_{\bm{r}}) \subset \mathcal{S}^{(k)}} \aLocal_{\bm{r}}^H(t),
\end{align}
where $\aLocal_{\bm{r}}^H(t) = {U_L(t, 0)}^\dagger \aLocal_{\bm{r}} U_L(t, 0)$ is the Heisenberg operator for $\aLocal_{\bm{r}}$.
As a preliminary step in proving the theorems, we prove that $A^H(t)$ can be approximated by an additive observable on $\mathcal{S}^{(k)}$.

\begin{figure}
    \centering
    \includegraphics[width=0.5\linewidth]{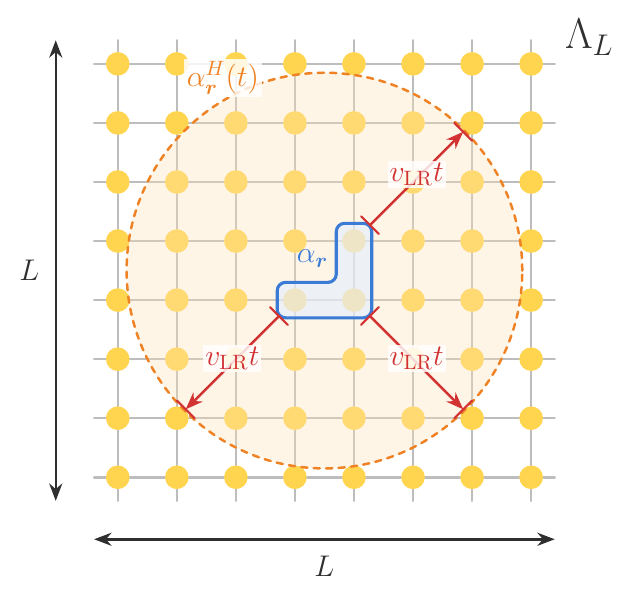}
    \caption{\label{fig:LRvelocity} Schematic illustration of the time evolution $\aLocal_{\bm{r}}^{H}(t)$ of an $\ell_A$-local observable $\aLocal_{\bm{r}}$ in the Heisenberg picture. It can be approximated by a $(2\ell+\ell_A)$-local observable $\aLocal_{\bm{r},\ell}^{\tilde{H}}(t)$ within a precision $\simeq e^{- R_2 (\ell-v_{\mathrm{LR}}t)}$, as in Eq.~\eqref{eq:LRbound_LightCone}. As a result, $\aLocal_{\bm{r}}^{H}(t)$ can be effectively considered as a local observable whose support grows in $t$ as $O(v_{\mathrm{LR}}t)+\ell_A$.
    }
\end{figure}

Since $U_L(t, 0)$ is the time evolution operator for the entire system, $\aLocal_{\bm{r}}^H(t)$ is generally nonlocal. We then introduce a local observable that approximates $\aLocal_{\bm{r}}^H(t)$.
See also Fig.~\ref{fig:LRvelocity} for the rough sketch of the time evolution of $\aLocal_{\bm{r}}^H(t)$.
Let $\ell$ be a positive integer independent of $L$, and let $\Lambda_{\bm{r},\ell}$ be a set of lattice sites within distance $\ell$ from $\supp(\aLocal_{\bm{r}})$,
\begin{align}
    \Lambda_{\bm{r},\ell} = \{ \bm{r}' \in \Lambda_L | \mathrm{dist}(\bm{r}', \supp(\aLocal_{\bm{r}})) \leq \ell \}.
    \label{eq:Lambda_r-ell}
\end{align}
We define $H_{\bm{r},\ell}^0(t)$ as the sum of all local observables that constitute $H(t)$ whose supports overlap with $\Lambda_{\bm{r},\ell}$.
Clearly, $H_{\bm{r},\ell}^0(t)$ is local, and the Heisenberg operator $\aLocal_{\bm{r},\ell}^{\tilde{H}}(t) = {U_{\bm{r},\ell}^0(t,0)}^\dagger \aLocal_{\bm{r}} U_{\bm{r},\ell}^0(t,0)$ for $\aLocal_{\bm{r}}$ obtained by applying the unitary evolution $U_{\bm{r},\ell}^0(t,0)$ generated by $H_{\bm{r},\ell}^0(t)$ is also local.

Now, let us show that $\aLocal_{\bm{r},\ell}^{\tilde{H}}(t)$ approximates $\aLocal_{\bm{r}}^H(t)$.
First, using $H_{\bm{r},\ell}^0(t)$ defined above, we decompose the Hamiltonian $H(t)$ as
\begin{align}
    H(t) = H_{\bm{r},\ell}^0(t) + H_{\bm{r},\ell}^1(t) + H_{\bm{r},\ell}^2(t),
\end{align}
where $H_{\bm{r},\ell}^1(t)$ is an operator acting only on $\Lambda_L \setminus \supp(H_{\bm{r},\ell}^0(t))$, and $H_{\bm{r},\ell}^2(t)$ is an interaction acting on both $\supp(H_{\bm{r},\ell}^0(t))$ and $\Lambda_L \setminus \supp(H_{\bm{r},\ell}^0(t))$.
We then have
\begin{align}
  &\left\| \aLocal_{\bm{r}}^{H}(t) - \aLocal_{\bm{r},\ell}^{\tilde{H}}(t) \right\| \nonumber\\
  &= \left\|
    {U_L(t, 0)}^\dagger \aLocal_{\bm{r}} U_L(t, 0)
    - {U_{\bm{r},\ell}^0(t, 0)}^\dagger \aLocal_{\bm{r}} U_{\bm{r},\ell}^0(t, 0)
  \right\| \nonumber\\
  &= \left\| \int_{0}^{t} ds \ \frac{d}{ds} \left\{
    {U_L(t-s, 0)}^\dagger {U_{\bm{r},\ell}^0(t, t-s)}^\dagger
    \aLocal_{\bm{r}}
    U_{\bm{r},\ell}^0(t, t-s) U_L(t-s, 0)
  \right\} \right\| \nonumber\\
  &= \left\| \int_{0}^{t} ds \ \left\{
    {U_L(t-s, 0)}^\dagger \left[ i H(t-s) - i H_{\bm{r},\ell}^0(t-s) \right] {U_{\bm{r},\ell}^0(t, t-s)}^\dagger
    \aLocal_{\bm{r}}
    U_{\bm{r},\ell}^0(t, t-s) U_L(t-s, 0)
  \right. \right. \nonumber\\
  &\left. \left. \hspace{2.0cm}
    + {U_L(t-s, 0)}^\dagger {U_{\bm{r},\ell}^0(t, t-s)}^\dagger
    \aLocal_{\bm{r}}
    U_{\bm{r},\ell}^0(t, t-s) \left[ i H_{\bm{r},\ell}^0(t-s) - i H(t-s) \right] U_L(t-s, 0)
  \right\} \right\| \nonumber\\
  &= \left\| \int_{0}^{t} ds \ 
    {U_L(t-s, 0)}^\dagger
    \left[
        {U_{\bm{r},\ell}^0(t, t-s)}^\dagger \aLocal_{\bm{r}} U_{\bm{r},\ell}^0(t, t-s),
        H_{\bm{r},\ell}^1(t-s) + H_{\bm{r},\ell}^2(t-s)
    \right]
    U_L(t-s, 0) \right\| \nonumber\\
  &= \left\| \int_{0}^{t} ds \ 
    {U_L(t-s, 0)}^\dagger
    \left[
        {U_{\bm{r},\ell}^0(t, t-s)}^\dagger \aLocal_{\bm{r}} U_{\bm{r},\ell}^0(t, t-s),
        H_{\bm{r},\ell}^2(t-s)
    \right]
    U_L(t-s, 0) \right\| \nonumber\\
  &\leq \int_{0}^{t} ds \ 
    \left\| {U_L(t-s, 0)}^\dagger
    \left[
        {U_{\bm{r},\ell}^0(t, t-s)}^\dagger \aLocal_{\bm{r}} U_{\bm{r},\ell}^0(t, t-s),
        H_{\bm{r},\ell}^2(t-s)
    \right]
    U_L(t-s, 0) \right\| \nonumber\\
  &= \int_{0}^{t} ds \ 
    \left\| \left[
        {U_{\bm{r},\ell}^0(t, t-s)}^\dagger \aLocal_{\bm{r}} U_{\bm{r},\ell}^0(t, t-s),
        H_{\bm{r},\ell}^2(t-s)
    \right] \right\|.
    \label{eq:Bound_TEvol_Comm}
\end{align}
By applying Lemma~\ref{lemma:Lieb-Robinson-bound}, for any $\kappa, \eta > 0$, we have
\begin{align}
  &\left\| \aLocal_{\bm{r}}^{H}(t) - \aLocal_{\bm{r},\ell}^{\tilde{H}}(t) \right\| \nonumber\\
  &\leq \int_{0}^{t} ds \
    2 \min \{ |\supp(\aLocal_{\bm{r}})|, |\supp(H_{\bm{r},\ell}^2(t-s))| \}
    \|\aLocal_{\bm{r}}\| \|H_{\bm{r},\ell}^2(t-s)\| \nonumber\\
  &\hspace{2.0cm}
    \times \exp \left[
      - \kappa \left\{
        \mathrm{dist}(\supp(\aLocal_{\bm{r}}),\supp(H_{\bm{r},\ell}^2(t-s)))
        - \frac{2 C_d}{\kappa} \left( \frac{d}{e \eta} \right)^d \int_{t-s}^{t} du \ \| H_{\bm{r},\ell}^0(u) \|_{\kappa + \eta}
      \right\}
    \right] \nonumber\\
    &\le 2 |\supp(\aLocal_{\bm{r}})| \, \|\aLocal_{\bm{r}}\| \sup_{0\le s\le t}\|H_{\bm{r},\ell}^2(s)\|\exp\Bigl[-\kappa\inf_{0\le s\le t}\mathrm{dist}(\supp(\aLocal_{\bm{r}}),\supp(H_{\bm{r},\ell}^2(s)))\Bigr] \notag\\
    &\qquad\times \int_{0}^{t}d s \ \exp\Bigl[2C_d s\left( \frac{d}{e \eta} \right)^d \sup_{0\le u\le t}\| H_{\bm{r},\ell}^0(u) \|_{\kappa + \eta} \Bigr].
\end{align}
Since $H(t)$ is a sum of $r_H$-local observables, as assumed in Eq.~\eqref{eq:r_H_local}, $H_{\bm{r},\ell}^2(t)$ satisfies $\|H_{\bm{r},\ell}^2(t)\| = O(\ell^{d-1})$  for $\ell\gg r_H$. 
Furthermore, by the definition of $H_{\bm{r},\ell}^0(t)$, $\mathrm{dist}(\supp(\aLocal_{\bm{r}}), \supp(H_{\bm{r},\ell}^2(t))) =\Theta(\ell)$ for $\ell\gg r_H$. 
In addition, from the definition~\eqref{eq:LRbound_norm} of the norm~$\|\bullet\|_{\mu}$, $\| H_{\bm{r},\ell}^0(u) \|_{\kappa+\eta}$ is bounded from above by $\|H(u)\|_{\kappa+\eta}$, which is independent of $\ell$ and is of $O(L^0)$. Using these, we have
\begin{align}
    &\left\| \aLocal_{\bm{r}}^{H}(t) - \aLocal_{\bm{r},\ell}^{\tilde{H}}(t) \right\| \notag\\
    &\le 2 |\supp(\aLocal_{\bm{r}})|\, \|\aLocal_{\bm{r}}\| \sup_{0\le s\le t}\|H_{\bm{r},\ell}^2(s)\|\exp\Bigl[-\kappa\inf_{0\le s\le t}\mathrm{dist}(\supp(\aLocal_{\bm{r}}),\supp(H_{\bm{r},\ell}^2(s)))\Bigr] \notag\\
    &\qquad\times \int_{0}^{t}d s \ \exp\Bigl[2C_d s\left( \frac{d}{e \eta} \right)^d \sup_{0\le u\le t}\| H(u) \|_{\kappa + \eta} \Bigr]\\
    &\leq R_1 \ell^{d-1} e^{- R_2 (\ell-v_{\mathrm{LR}}t)}.
    \label{eq:LRbound_LightCone}
\end{align}
Here, $R_1$, $R_2$, and $v_{\mathrm{LR}}$ are some positive constants that are independent of $L$, $t$, and $\ell$.
Thus, we see that $\aLocal_{\bm{r},\ell}^{\tilde{H}}(t)$ is an $(2\ell+\ell_A)$-local observable approximating $\aLocal_{\bm{r}}^H(t)$ with an error of $O(\ell^{d-1} e^{-\Theta(\ell-v_{\mathrm{LR}}t)})$, which becomes sufficiently small when $\ell\gg v_{\mathrm{LR}}t$.
Consequently, by letting
\begin{align}
    B_{t,\ell}=\sum_{\supp(\aLocal_{\bm{r}}) \subset \mathcal{S}^{(k)}} \aLocal_{\bm{r},\ell}^{\tilde{H}}(t),
\end{align}
we have
\begin{align}
    \| A^H(t) / |\mathcal{S}^{(k)}| - B_{t,\ell} / |\mathcal{S}^{(k)}| \|
    \leq \frac{1}{|\mathcal{S}^{(k)}|} \sum_{\supp(\aLocal_{\bm{r}}) \subset \mathcal{S}^{(k)}} \| \aLocal_{\bm{r}}^H(t) - \aLocal_{\bm{r},\ell}^{\tilde{H}}(t) \|
    \le R_1 \ell^{d-1} e^{- R_2 (\ell-v_{\mathrm{LR}}t)}.
    \label{eq:A-B}
\end{align}
Note that this error bound does not depend on $L$. This shows that $A^H(t)/|\mathcal{S}^{(k)}|$ can be approximated by a sum of local observables uniformly in $L$. However, $B_{t,\ell}$ is still not an additive observable since $\aLocal_{\bm{r},\ell}^{\tilde{H}}(t)$ and $\aLocal_{\bm{r}',\ell}^{\tilde{H}}(t)$ are generally not related to each other via spatial translation. Moreover, the support of $B_{t,\ell}$ may even extend beyond $\mathcal{S}^{(k)}$.

Let us further examine the properties of $\aLocal_{\bm{r},\ell}^{\tilde{H}}(t)$. Assume that the support of $\aLocal_{\bm{r}}$ is sufficiently far from the surface of the subsystem $\mathcal{S}^{(k)}$. More precisely, we consider the case where
\begin{align}
    \mathrm{dist} (\supp(\aLocal_{\bm{r}}), \Lambda_L \setminus \mathcal{S}^{(k)}) > \ell + \delta.
\end{align}
Here $\delta$ is the maximum of diameters of the local terms that constitute $H(t)$ and is a constant independent of $L$.
In this case, the support of $H_{\bm{r},\ell}^0$ is entirely within $\mathcal{S}^{(k)}$.
Since $H(t)$ is defined as a sum of additive observables, which are sums of translations of local observables in macroscopic subsystems, if $H_{\bm{r},\ell}^0$ has support entirely within the primitive macroscopic subsystem, it can be expressed as a sum of translations of a local observable $h(t)$ independent of $\bm{r}$ as
\begin{align}
    H_{\bm{r},\ell}^0(t) = \sum_{\supp(h(t))+\bm{r}' \cap \Lambda_{\bm{r},\ell} \neq \emptyset} T_{\bm{r}'} h(t) T_{\bm{r}'}^\dagger.
\end{align}
By the definition of $\Lambda_{\bm{r},\ell}$ \eqref{eq:Lambda_r-ell} and $\supp(\aLocal_{\bm{r}}) = \supp(\aLocal)+{\bm{r}}$, we obtain
\begin{align}
    H_{\bm{r},\ell}^0(t)
    &= \sum_{\supp(h(t))+\bm{r}'-\bm{r} \cap \Lambda_{\bm{0},\ell} \neq \emptyset} T_{\bm{r}'} h(t) T_{\bm{r}'}^\dagger \nonumber\\
    &= \sum_{\supp(h(t))+\bm{r}' \cap \Lambda_{\bm{0},\ell} \neq \emptyset} T_{\bm{r}'+\bm{r}} h(t) T_{\bm{r}'+\bm{r}}^\dagger
    = T_{\bm{r}} H_{\ell}^0(t) T_{\bm{r}}^\dagger,
\end{align}
where
\begin{align}
    H_{\ell}^0(t) = \sum_{\supp(h(t))+\bm{r}' \cap \Lambda_{\bm{0},\ell} \neq \emptyset} T_{\bm{r}'} h(t) T_{\bm{r}'}^\dagger.
\end{align}
Therefore, using the Heisenberg operator $c_{t,\ell} = {U_{\ell}^0(t,0)}^\dagger \aLocal U_{\ell}^0(t,0)$ for $\aLocal$ obtained by applying the unitary evolution $U_{\ell}^0(t,0)$ generated by $H_{\ell}^0(t)$, we have
\begin{align}
    \aLocal_{\bm{r},\ell}^{\tilde{H}}(t) = T_{\bm{r}} c_{t,\ell} T_{\bm{r}}^\dagger.
\end{align}
In other words, among the $\aLocal_{\bm{r},\ell}^{\tilde{H}}(t)$ constituting $B_{\bm{r},\ell}$, those associated with $\bm{r}$ for which the support of $\aLocal_{\bm{r}}$ is at a distance greater than $\ell + \delta$ from the surface of $\mathcal{S}^{(k)}$ can be represented as translations of $c_{t,\ell}$. Then we introduce an additive observable $C_{t,\ell}$ on $\mathcal{S}^{(k)}$ obtained from 
a $(2\ell+\ell_A)$-local observable $c_{t,\ell}$,
\begin{align}
    C_{t,\ell} = \sum_{\supp(c_{t,\ell})+\bm{r}\, \subset \mathcal{S}^{(k)}} T_{\bm{r}} c_{t,\ell} T_{\bm{r}}^\dagger.
\end{align}
Recalling that $\ell+\delta$ is a constant independent of $L$, we obtain
\begin{align}
    \| B_{t,\ell} / |\mathcal{S}^{(k)}| - C_{t,\ell} / |\mathcal{S}^{(k)}| \|
    \leq R_3 \ell L^{-1},
    \label{eq:B-C}
\end{align}
where $R_3$ is some positive constant that is independent of both $L$ and $\ell$.

From the above, we get
\begin{align}
    &\| A^H(t) / |\mathcal{S}^{(k)}| - C_{t,\ell} / |\mathcal{S}^{(k)}| \| \nonumber\\
    &\leq \| A^H(t) / |\mathcal{S}^{(k)}| - B_{t,\ell} / |\mathcal{S}^{(k)}| \|
    + \| B_{t,\ell} / |\mathcal{S}^{(k)}| - C_{t,\ell} / |\mathcal{S}^{(k)}| \| \nonumber\\
    &\leq R_1 \ell^{d-1} e^{- R_2 (\ell-v_{\mathrm{LR}}t)} + R_3 \ell L^{-1}.
    \label{eq:A-C}
\end{align}
That is, $C_{t,\ell}$ is an additive observable on the macroscopic subsystem $\mathcal{S}^{(k)}$ that approximates $A^H(t)$~\footnote{This result implies that the completion of $\SetAdditive$, $\overline{\SetAdditive}$, is invariant under macroscopic operations.}.

Now we are ready to prove Propositions~\ref{proposition:JAIVTMXN_1} and \ref{proposition:JAIVTMXN_2} and Theorem~\ref{theorem:JAIVTMXN_3}.
\begin{proof}[Proof of Proposition~\ref{proposition:JAIVTMXN_1}]
We prove that the expectation value of the Heisenberg operator $A^H(t)/|\mathcal{S}|$ for the additive observable density $A/|\mathcal{S}|$ converges in the thermodynamic limit for every macroscopic subsystem $\mathcal{S}$ and additive observable $A$ on $\mathcal{S}$. To this end, we prove that $\left( \mathrm{Tr} \left[ \rho_{L} A^H(t) / |\mathcal{S}| \right] \right)_{L\in\mathbb{N}}$ is a Cauchy sequence.

Suppose that the macroscopic subsystem $\mathcal{S}$ is a composite system of $J$ primitive macroscopic subsystems $\{\mathcal{S}^{(k_1)}, \mathcal{S}^{(k_2)}, \dots, \mathcal{S}^{(k_J)}\}$. Then, as shown in Eq.~\eqref{eq:Additive_S_approximation}, it holds that
\begin{align}
    \left\| A^H(t) / |\mathcal{S}| - \sum_{j=1}^J \lambda^{(k_j)} A_{\mathcal{S}^{(k_j)}}^H(t) / |\mathcal{S}^{(k_j)}| \right\| = o(L^0),
\end{align}
where $\displaystyle \lambda^{(k_j)} = \lim_{L\to\infty} |\mathcal{S}^{(k_j)}|/|\mathcal{S}|$.
Therefore, if $\left( \mathrm{Tr} \left[ \rho_{L} A_{\mathcal{S}^{(k_j)}}^H(t) / |\mathcal{S}^{(k_j)}| \right] \right)_{L\in\mathbb{N}}$ converges in the thermodynamic limit for all $j=1, 2, \dots, J$, $\left( \mathrm{Tr} \left[ \rho_{L} A^H(t) / |\mathcal{S}| \right] \right)_{L\in\mathbb{N}}$ also converges.
Thus, we consider the case where the macroscopic subsystem $\mathcal{S}$ is a primitive macroscopic subsystem.

Since $A$ is an additive observable on the primitive macroscopic subsystem, there exists an additive observable $C_{t,\ell}$ on the macroscopic subsystem that approximates $A^H(t)$ in the sense of Eq.~\eqref{eq:A-C}.
Then, by the triangle inequality and H\"older's inequality, for any positive integer $\ell_0$, we have
\begin{align}
    &\left|
        \left. \mathrm{Tr} \left[ \rho_L A^H(t) / |\mathcal{S}| \right] \right|_{L=L_1}
        - \left. \mathrm{Tr} \left[ \rho_L A^H(t) / |\mathcal{S}| \right] \right|_{L=L_2}
    \right| \nonumber\\
    &\leq \left|
        \left. \mathrm{Tr} \left[ \rho_L A^H(t) / |\mathcal{S}| \right] \right|_{L=L_1}
        - \left. \mathrm{Tr} \left[ \rho_L C_{t,\ell_0} / |\mathcal{S}| \right] \right|_{L=L_1}
    \right| \nonumber\\
    &+ \left|
        \left. \mathrm{Tr} \left[ \rho_L C_{t,\ell_0} / |\mathcal{S}| \right] \right|_{L=L_1}
        - \left. \mathrm{Tr} \left[ \rho_L C_{t,\ell_0} / |\mathcal{S}| \right] \right|_{L=L_2}
    \right| \nonumber\\
    &+ \left|
        \left. \mathrm{Tr} \left[ \rho_L C_{t,\ell_0} / |\mathcal{S}| \right] \right|_{L=L_2}
        - \left. \mathrm{Tr} \left[ \rho_L A^H(t) / |\mathcal{S}| \right] \right|_{L=L_2}
    \right| \nonumber\\
    &= \left| \left. \mathrm{Tr} \left[ \rho_L \left\{ A^H(t) / |\mathcal{S}| - C_{t,\ell_0} / |\mathcal{S}| \right\} \right] \right|_{L=L_1} \right| \nonumber\\
    &+ \left| \left. \mathrm{Tr} \left[ \rho_L \left\{ A^H(t) / |\mathcal{S}| - C_{t,\ell_0} / |\mathcal{S}| \right\} \right] \right|_{L=L_2} \right| \nonumber\\
    &+ \left|
        \left. \mathrm{Tr} \left[ \rho_L C_{t,\ell_0} / |\mathcal{S}| \right] \right|_{L=L_1}
        - \left. \mathrm{Tr} \left[ \rho_L C_{t,\ell_0} / |\mathcal{S}| \right] \right|_{L=L_2}
    \right| \nonumber\\
    &\leq \underbrace{
        \left. \left\| A^H(t) / |\mathcal{S}| - C_{t,\ell_0} / |\mathcal{S}| \right\| \right|_{L=L_1}
    }_{=(\ast 1)} \nonumber\\
    &+ \underbrace{
        \left. \left\| A^H(t) / |\mathcal{S}| - C_{t,\ell_0} / |\mathcal{S}| \right\| \right|_{L=L_2}
    }_{=(\ast 2)} \nonumber\\
    &+ \underbrace{\left|
        \left. \mathrm{Tr} \left[ \rho_L C_{t,\ell_0} / |\mathcal{S}| \right] \right|_{L=L_1}
        - \left. \mathrm{Tr} \left[ \rho_L C_{t,\ell_0} / |\mathcal{S}| \right] \right|_{L=L_2}
    \right|}_{=(\ast\ast)}.
\end{align}
Let $\ell_\ast \in \mathbb{N}$ be such that
\begin{align}
    \ell^{d-1}
    \leq \exp \left[\frac{R_2}{2}(\ell-v_{\mathrm{LR}}t)\right]
    \qquad \text{for all $\ell \geq \ell_\ast$}.
\end{align}
Then, from Eq.~\eqref{eq:A-C}, for any given $\epsilon>0$,
choosing $\ell_0 = \max \{ \ell_\ast, \lceil v_{\mathrm{LR}}t + \frac{2}{R_2} \log \frac{8 R_1}{\epsilon} \rceil \}$
and $L_0^{\ast} = \lceil \frac{8 R_3 \ell_0}{\epsilon} \rceil$,
for every $L_1, L_2 > L_0^{\ast}$, it holds that
\begin{align}
    (\ast 1)
    &\leq R_1 {\ell_0}^{d-1} e^{- R_2 (\ell_0-v_{\mathrm{LR}}t)} + R_3 \ell_0 {L_1}^{-1}
    \leq R_1 e^{- \frac{R_2}{2} (\ell_0-v_{\mathrm{LR}}t)} + R_3 \ell_0 {L_0^{\ast}}^{-1}
    \leq \epsilon/4,\\
    (\ast 2)
    &\leq R_1 {\ell_0}^{d-1} e^{- R_2 (\ell_0-v_{\mathrm{LR}}t)} + R_3 \ell_0 {L_2}^{-1}
    \leq R_1 e^{- \frac{R_2}{2} (\ell_0-v_{\mathrm{LR}}t)} + R_3 \ell_0 {L_0^{\ast}}^{-1}
    \leq \epsilon/4.
\end{align}
In addition, since $C_{t,\ell_0}$ is an additive observable on the macroscopic subsystem,
its expectation value in the sequence of the microscopic states $(\rho_L)_{L\in\mathbb{N}}$ representing the macroscopic state
converges in the thermodynamic limit.
Therefore, for any $\epsilon>0$, there exists $L_0^{\ast\ast} \in \mathbb{N}$ such that for every $L_1, L_2 > L_0^{\ast\ast}$, it holds that
\begin{align}
    (\ast\ast)
    = \left|
        \left. \mathrm{Tr} \left[ \rho_L C_{t,\ell_0} / |\mathcal{S}| \right] \right|_{L=L_1}
        - \left. \mathrm{Tr} \left[ \rho_L C_{t,\ell_0} / |\mathcal{S}| \right] \right|_{L=L_2}
    \right|
    < \epsilon/2.
\end{align}
Thus, from the above, for any $\epsilon>0$, there exists $L_0 \in \mathbb{N}$ such that for every $L_1, L_2 > L_0$, it holds that
\begin{align}
    \left| 
        \left. \mathrm{Tr} \left[ \rho_L A^H(t) / |\mathcal{S}| \right] \right|_{L=L_1}
        - \left. \mathrm{Tr} \left[ \rho_L A^H(t) / |\mathcal{S}| \right] \right|_{L=L_2}
    \right| < \epsilon.
\end{align}
This implies that $\left( \mathrm{Tr} \left[ \rho_{L} A^H(t) / |\mathcal{S}| \right] \right)_{L\in\mathbb{N}}$ is a Cauchy sequence.
This completes the proof of Proposition~\ref{proposition:JAIVTMXN_1}.
\end{proof}

\begin{proof}[Proof of Proposition~\ref{proposition:JAIVTMXN_2}]
We prove that the variance of the Heisenberg operator $A^H(t)/|\mathcal{S}|$ for the additive observable density $A/|\mathcal{S}|$ is $o(L^0)$ for every macroscopic subsystem $\mathcal{S}$ and additive observable $A$ on $\mathcal{S}$.

Suppose that the macroscopic subsystem $\mathcal{S}$ is a composite system of $J$ primitive macroscopic subsystems $\{\mathcal{S}^{(k_1)}, \mathcal{S}^{(k_2)}, \dots, \mathcal{S}^{(k_J)}\}$. Then, as shown in Eq.~\eqref{eq:Additive_S_approximation}, it holds that
\begin{align}
    \left\| A^H(t) / |\mathcal{S}| - \sum_{j=1}^J \lambda^{(k_j)} A_{\mathcal{S}^{(k_j)}}^H(t) / |\mathcal{S}^{(k_j)}| \right\| = o(L^0).
\end{align}
Thus, we have
\begin{align}
    &\mathrm{Tr} \left[ \rho_L \left( A^H(t) / |\mathcal{S}| \right)^2 \right]
    - \left( \mathrm{Tr} \left[ \rho_L A^H(t) / |\mathcal{S}| \right] \right)^2 \nonumber\\
    &= \sum_{j,j'=1}^J \lambda^{(k_j)}\lambda^{(k_{j'})} \mathrm{Tr} \left[ \rho_L \left( \delta A_{\mathcal{S}^{(k_j)}}^H(t) / |\mathcal{S}^{(k_j)}| \right) \left( \delta A_{\mathcal{S}^{(k_{j'})}}^H(t) / |\mathcal{S}^{(k_{j'})}| \right) \right]
    + o(L^0),
\end{align}
where $\delta A_{\mathcal{S}^{(k_j)}}^H(t) = A_{\mathcal{S}^{(k_j)}}^H(t) - \mathrm{Tr} \left[ \rho_L A_{\mathcal{S}^{(k_j)}}^H(t) \right]$.
Using the triangle inequality and Cauchy--Schwarz inequality, we have
\begin{align}
    &\left| \mathrm{Tr} \left[ \rho_L \left( A^H(t) / |\mathcal{S}| \right)^2 \right]
    - \left( \mathrm{Tr} \left[ \rho_L A^H(t) / |\mathcal{S}| \right] \right)^2 \right| \nonumber\\
    &\leq \sum_{j,j'=1}^J \lambda^{(k_j)}\lambda^{(k_{j'})} \left| \mathrm{Tr} \left[ \rho_L \left( \delta A_{\mathcal{S}^{(k_j)}}^H(t) / |\mathcal{S}^{(k_j)}| \right) \left( \delta A_{\mathcal{S}^{(k_{j'})}}^H(t) / |\mathcal{S}^{(k_{j'})}| \right) \right] \right|
    + o(L^0) \nonumber\\
    &\leq \sum_{j,j'=1}^J \lambda^{(k_j)}\lambda^{(k_{j'})} \sqrt{ \mathrm{Tr} \left[ \rho_L \left( \delta A_{\mathcal{S}^{(k_j)}}^H(t) / |\mathcal{S}^{(k_j)}| \right)^2 \right] \mathrm{Tr} \left[ \rho_L \left( \delta A_{\mathcal{S}^{(k_{j'})}}^H(t) / |\mathcal{S}^{(k_{j'})}| \right)^2 \right] }
    + o(L^0),
\end{align}
Therefore, if the variance of $A_{\mathcal{S}^{(k_j)}}^H(t) / |\mathcal{S}^{(k_j)}|$ is $o(L^0)$ for all $j=1, 2, \dots, J$, the variance of $A^H(t)$ is also $o(L^0)$.
Thus, we consider the case where the macroscopic subsystem $\mathcal{S}$ is a primitive macroscopic subsystem.

Since $A$ is an additive observable on the primitive macroscopic subsystem, there exists an observable $B_{t,\ell}$ that approximates $A^H(t)$ in the sense of Eq.~\eqref{eq:A-B} and an additive observable $C_{t,\ell}$ on the macroscopic subsystem that approximates $B_{t,\ell}$ in the sense of Eq.~\eqref{eq:B-C}.
Therefore, we have
\begin{align}
    \lim_{L\to\infty} \mathrm{Tr} \left[ \rho_L A^H(t) / |\mathcal{S}| \right]
    = \lim_{\ell\to\infty} \lim_{L\to\infty} \mathrm{Tr} \left[ \rho_L B_{t,\ell} / |\mathcal{S}| \right]
    = \lim_{\ell\to\infty} \lim_{L\to\infty} \mathrm{Tr} \left[ \rho_L C_{t,\ell} / |\mathcal{S}| \right],
    \label{eq:Proof_JAIVTMXN_1_A=B=C_DoubleLimits}
\end{align}
where 
the first equality follows from the fact that Eq.~\eqref{eq:A-B} holds even after $L\to\infty$,
and the second equality from Eq.~\eqref{eq:B-C} combined with $L\to\infty$.
In a similar manner, we have
\begin{align}
    \lim_{L\to\infty} \mathrm{Tr} \left[ \rho_L \left( A^H(t) / |\mathcal{S}| \right)^2 \right]
    = \lim_{\ell\to\infty} \lim_{L\to\infty} \mathrm{Tr} \left[ \rho_L \left( B_{t,\ell} / |\mathcal{S}| \right)^2 \right]
    = \lim_{\ell\to\infty} \lim_{L\to\infty} \mathrm{Tr} \left[ \rho_L \left( C_{t,\ell} / |\mathcal{S}| \right)^2 \right].
\end{align}
Thus, we obtain
\begin{align}
    \lim_{L\to\infty} \left\{
        \mathrm{Tr} \left[ \rho_L \left( A^H(t) / |\mathcal{S}| \right)^2 \right]
        - \left( \mathrm{Tr} \left[ \rho_L A^H(t) / |\mathcal{S}| \right] \right)^2
    \right\}
    = \lim_{\ell\to\infty} \lim_{L\to\infty} \left\{
        \mathrm{Tr} \left[ \rho_L \left( C_{t,\ell} / |\mathcal{S}| \right)^2 \right]
        - \left( \mathrm{Tr} \left[ \rho_L C_{t,\ell} / |\mathcal{S}| \right] \right)^2
    \right\}.
\end{align}
Since $C_{t,\ell}$ is an additive observable composed of $(2\ell+\ell_A)$-local observables, the variance of $C_{t,\ell} / |\mathcal{S}|$ in the sequence of microscopic states representing a normal macroscopic state $(\rho_L)_{L\in\mathbb{N}}$ is $o(L^0)$. Therefore, we obtain
\begin{align}
    \lim_{L\to\infty} \left\{
        \mathrm{Tr} \left[ \rho_L \left( A^H(t) / |\mathcal{S}| \right)^2 \right]
        - \left( \mathrm{Tr} \left[ \rho_L A^H(t) / |\mathcal{S}| \right] \right)^2
    \right\}
    = 0.
\end{align}
This completes the proof of Proposition~\ref{proposition:JAIVTMXN_2}.
\end{proof}

\begin{proof}[Proof of Theorem~\ref{theorem:JAIVTMXN_3}]
We prove that for every macroscopic subsystem $\mathcal{S}$ and additive observable $A$ on $\mathcal{S}$, we have
\begin{align}
    \lim_{L\to\infty} \mathrm{Tr} \left[ \rho_L A^H(t)/|\mathcal{S}| \right]
    = \lim_{L\to\infty} \mathrm{Tr} \left[ \sigma_L A^H(t)/|\mathcal{S}| \right].
    \label{eq:JAIVTMXN_2}
\end{align}

Suppose that the macroscopic subsystem $\mathcal{S}$ is a composite system of $J$ primitive macroscopic subsystems $\{\mathcal{S}^{(k_1)}, \mathcal{S}^{(k_2)}, \dots, \mathcal{S}^{(k_J)}\}$. Then, as shown in Eq.~\eqref{eq:Additive_S_approximation}, it holds that
\begin{align}
    \left\| A^H(t) / |\mathcal{S}| - \sum_{j=1}^J \lambda^{(k_j)} A_{\mathcal{S}^{(k_j)}}^H(t) / |\mathcal{S}^{(k_j)}| \right\| = o(L^0).
\end{align}
Therefore, if Eq.~\eqref{eq:JAIVTMXN_2} is satisfied when $A=A_{\mathcal{S}^{(k_j)}}$ for all $j=1, 2, \dots, J$, it is also satisfied for $A$.
Thus, we consider the case where the macroscopic subsystem $\mathcal{S}$ is a primitive macroscopic subsystem.

Since $A$ is an additive observable on the primitive macroscopic subsystem, there exists an observable $B_{t,\ell}$ that approximates $A^H(t)$ in the sense of Eq.~\eqref{eq:A-B} and an additive observable $C_{t,\ell}$ on the macroscopic subsystem that approximates $B_{t,\ell}$ in the sense of Eq.~\eqref{eq:B-C}.
Therefore, in the same manner as Eq.~\eqref{eq:Proof_JAIVTMXN_1_A=B=C_DoubleLimits}, we have
\begin{align}
    \lim_{L\to\infty} \mathrm{Tr} \left[ \rho_L A^H(t) / |\mathcal{S}| \right]
    = \lim_{\ell\to\infty} \lim_{L\to\infty} \mathrm{Tr} \left[ \rho_L B_{t,\ell} / |\mathcal{S}| \right]
    = \lim_{\ell\to\infty} \lim_{L\to\infty} \mathrm{Tr} \left[ \rho_L C_{t,\ell} / |\mathcal{S}| \right],
\end{align}
and
\begin{align}
    \lim_{L\to\infty} \mathrm{Tr} [ \sigma_L A^H(t) / |\mathcal{S}| ]
    = \lim_{\ell\to\infty}\lim_{L\to\infty} \mathrm{Tr} \left[ \sigma_L C_{t,\ell} / |\mathcal{S}| \right].
\end{align}
On the other hand, since $C_{t,\ell}$ is an additive observable composed of $(2\ell+\ell_A)$-local observables with $2\ell+\ell_A=O(L^0)$, macroscopically equivalent states $(\rho_L)_{L\in\mathbb{N}}$ and $(\sigma_L)_{L\in\mathbb{N}}$ satisfy
\begin{align}
    \lim_{L\to\infty} \mathrm{Tr} [ \rho_L C_{t,\ell} / |\mathcal{S}| ]
    = \lim_{L\to\infty} \mathrm{Tr} \left[ \sigma_L C_{t,\ell} / |\mathcal{S}| \right].
\end{align}
By combining the above, we finally obtain
\begin{align}
    \lim_{L\to\infty} \mathrm{Tr} [ \rho_L A^H(t) / |\mathcal{S}| ]
    = \lim_{L\to\infty} \mathrm{Tr} \left[ \sigma_L A^H(t) / |\mathcal{S}| \right].
\end{align}
This completes the proof of Theorem~\ref{theorem:JAIVTMXN_3}~\footnote{The above proofs also imply that we can replace $\SetAdditive$ in Definitions~\ref{definition:MacroState}--\ref{definition:MacroState_Normal} by its completion $\overline{\SetAdditive}$ without causing any change to these definitions.}.
\end{proof}

\section{\label{sec:Proof_s^red=s^TD}Proof of Proposition~\ref{proposition:Equiv_rho^red}, Theorems~\ref{theorem:s^red=s^TD}, \ref{theorem:EntropyIncrease}, and \ref{theorem:EntropyIncrease_WithoutThermalization}}

\subsection{\label{sec:Proof_Equiv_rho^k}Proof of Proposition~\ref{proposition:Equiv_rho^red}}

Here, we prove Proposition~\ref{proposition:Equiv_rho^red}.
We divide the proof of this proposition into two parts: proof for the necessity and sufficiency of Eq.~\eqref{eq:Equiv_rho^red}.
\begin{proof}[Proof for the necessity of Eq.~\eqref{eq:Equiv_rho^red}.]
Take arbitrary $k\in\{1,...,K^d\}$ and $\ell\in\mathbb{N}$. As explained above, Eq.~\eqref{eq:<a>^red=<A_Sigma>/Sigma} holds for any observable $\aLocal$ on $\Cell$. If $\{\rho_L\}_L$ and $\{\sigma_L\}_L$ are macroscopically equivalent, we can show
\begin{align}
    \lim_{L\to\infty}\mathrm{Tr}\Bigl[\rho_{L}\frac{A_{\mathcal{S}^{(k)}}(\aLocal)}{|\mathcal{S}^{(k)}|}\Bigr]
    =\lim_{L\to\infty}\mathrm{Tr}\Bigl[\sigma_{L}\frac{A_{\mathcal{S}^{(k)}}(\aLocal)}{|\mathcal{S}^{(k)}|}\Bigr]
    \label{eq:Proof_Equiv_rho^red_1}
\end{align}
by multiplying $N/|\mathcal{S}^{(k)}|$ to the definition of macroscopic equivalence~\eqref{eq:MacroscopicEquiv}. Here $A_{\mathcal{S}^{(k)}}$ is the additive observable on $\mathcal{S}^{(k)}$ obtained from a local observable $\aLocal$. 
By combining Eqs.~\eqref{eq:<a>^red=<A_Sigma>/Sigma} and \eqref{eq:Proof_Equiv_rho^red_1}, we have
\begin{align}
    \mathrm{Tr}\bigl[\rho_{\infty|\ell}^{(k)}\aLocal\bigr]
    =\mathrm{Tr}\bigl[\sigma_{\infty|\ell}^{(k)}\aLocal\bigr]
    \label{eq:Proof_Equiv_rho^red_2}
\end{align}
for any observable $\aLocal$ on $\Cell$. This implies Eq.~\eqref{eq:Equiv_rho^red}.
\end{proof}

\begin{proof}[Proof for the sufficiency of Eq.~\eqref{eq:Equiv_rho^red}.] 
Assume that Eq.~\eqref{eq:Equiv_rho^red} holds for all $k\in\{1,...,K^d\}$ and for all $\ell\in\mathbb{N}$. First we will show Eq.~\eqref{eq:Proof_Equiv_rho^red_1} for arbitrary additive observables on an arbitrary primitive macroscopic subsystem. Let $A_{\mathcal{S}^{(k)}}$ be an additive observable on a primitive macroscopic subsystem $\mathcal{S}^{(k)}$ obtained from a local observable~$\aLocal$. By taking a sufficiently large positive number $\ell\in\mathbb{N}$, the local observable~$\aLocal$ can be regarded as an observable on $\Cell$. From Eq.~\eqref{eq:Equiv_rho^red}, we have Eq.~\eqref{eq:Proof_Equiv_rho^red_2}. Combining this with Eq.~\eqref{eq:<a>^red=<A_Sigma>/Sigma}, we obtain Eq.~\eqref{eq:Proof_Equiv_rho^red_1}. 

Finally, we
show Eq.~\eqref{eq:MacroscopicEquiv} for arbitrary additive observables on an arbitrary macroscopic subsystem. Let $A_{S}$ be an additive observable on a macroscopic subsystem~$\mathcal{S}$ obtained from the local observable~$\aLocal$. From Eqs.~\eqref{eq:Additive_S_approximation} and \eqref{eq:Proof_Equiv_rho^red_1}, we have
\begin{align}
    \lim_{L\to\infty}\mathrm{Tr}\Bigl[\rho_{L}\frac{A_{S}(\aLocal)}{N}\Bigr]
    =\lim_{L\to\infty}\mathrm{Tr}\Bigl[\sigma_{L}\frac{A_{S}(\aLocal)}{N}\Bigr].
\end{align}
This implies macroscopic equivalence~\eqref{eq:MacroscopicEquiv}.
\end{proof}

\subsection{Upper bound on the quantum macroscopic entropy density}

\begin{lemma}[Principle of maximum entropy]\label{lemma:s^red<=s^TD}
Let $(\rho_{L})_{L\in\mathbb{N}}$ represent a macroscopic state. If its energy density coincides with that of an equilibrium state described by a Hamiltonian $H$ at some 
inverse temperature $\beta^*$,
\begin{align}
    \lim_{L\to\infty}\mathrm{Tr}[\rho_{L} H]/N
    =\lim_{L\to\infty}\mathrm{Tr}[\rho^{\mathrm{can}}_{L}(\beta^*|H) H]/N,
    \label{eq:s^red<=s^TD_DEF_beta}
\end{align}
then the quantum macroscopic entropy density $s^{\mathrm{mac}}[\bullet]$ is bounded from above by thermodynamic entropy 
\begin{align}
    \limsup_{\ell\to\infty}s^{\mathrm{mac}}_{\ell}[(\rho_{L})_{L\in\mathbb{N}}]\le s^{\mathrm{TD}}(\beta^*|H).
    \label{eq:s^red<=s^TD}
\end{align}
\end{lemma}
\begin{proof}
We will utilize the fact that $s^{\mathrm{mac}}_{\ell}[\bullet]$ is bounded from above by $\tilde{s}_{\ell}[\bullet]$, as shown in Eq.~\eqref{eq:s^red<=tilde_s}, to obtain Eq.~\eqref{eq:s^red<=s^TD}.
Therefore, we first investigate $\rho_{\infty|\ell}^{\mathrm{ave}}$, which is defined by Eq.~\eqref{eq:DEF_rho_infty^tot}.
Let $H_\ell^{\mathrm{p}}$ ($H_\ell^{\mathrm{o}}$) be the Hamiltonian on the hypercube $\Cell$ with the periodic (resp. open) boundary conditions and let $\rho^{\mathrm{can}}_{\ell}(\beta|H)$ be the canonical Gibbs state $e^{-\beta H_\ell^{\mathrm{p}}}/Z_{\ell}$.
From the definition of the quantum relative entropy $D(\rho\Vert\sigma)=\mathrm{Tr}\bigl[\rho\log\rho\bigr]-\mathrm{Tr}\bigl[\rho\log\sigma\bigr]$, we have
\begin{align}
    D\Bigl(
    \rho_{\infty|\ell}^{\mathrm{ave}}
    \Bigm|\rho^{\mathrm{can}}_{\ell}(\beta^*|H)\Bigr)
    &=\beta^*\bigl(\mathrm{Tr}[
    \rho_{\infty|\ell}^{\mathrm{ave}}
    H_\ell^{\mathrm{p}}]-\mathrm{Tr}[\rho^{\mathrm{can}}_{\ell}(\beta^*|H) H_\ell^{\mathrm{p}}]\bigr)
    \notag\\
    &\quad +S_{\mathrm{vN}}\bigl[\rho^{\mathrm{can}}_{\ell}(\beta^*|H)\bigr]-S_{\mathrm{vN}}\bigl[
    \rho_{\infty|\ell}^{\mathrm{ave}}
    \bigr].
    \label{eq:RelativeEnt_red|can}
\end{align}
For the first term of Eq.~\eqref{eq:RelativeEnt_red|can}, we have the following equality:
\begin{align}
    &\lim_{\ell\to\infty}\frac{\mathrm{Tr}\bigl[
    \rho_{\infty|\ell}^{\mathrm{ave}}
    H_\ell^{\mathrm{p}}\bigr]}{\ell^d}\notag\\
    &=\lim_{\ell\to\infty}\frac{\mathrm{Tr}\bigl[
    \rho_{\infty|\ell}^{\mathrm{ave}}
    H_\ell^{\mathrm{o}}\bigr]}{\ell^d}\quad \Bigl(\because \|H_\ell^{\mathrm{p}}-H_\ell^{\mathrm{o}}\|=O(\ell^{d-1})\Bigr)\\
    &=\lim_{\ell\to\infty}\lim_{L\to\infty}\frac{1}{N \ell^d}\sum_{r\in\Lambda_{L}}\sum_{r^\prime (\mathrm{Supp}[h]+r^\prime\subset \Cell)}\mathrm{Tr}[\rho_{L} T_{r+r^\prime}hT_{r+r^\prime}^{\dagger}]\\
    &=\lim_{\ell\to\infty}\lim_{L\to\infty}\frac{1}{N \ell^d}\sum_{r^\prime (\mathrm{Supp}[h]+r^\prime\subset \Cell)}\sum_{r\in\Lambda_{L}}\mathrm{Tr}[\rho_{L} T_{r}hT_{r}^{\dagger}]\\
    &=\lim_{\ell\to\infty}\lim_{L\to\infty}\frac{1}{N \ell^d}\sum_{r^\prime (\mathrm{Supp}[h]+r^\prime\subset \Cell)}\mathrm{Tr}[\rho_{L} H_L^{\mathrm{p}}]\\
    &=\lim_{\ell\to\infty}\frac{1}{\ell^d}\sum_{r^\prime (\mathrm{Supp}[h]+r^\prime\subset \Cell)}\lim_{L\to\infty}\frac{\mathrm{Tr}[\rho_{L} H_L^{\mathrm{p}}]}{N}\\
    &=\lim_{L\to\infty}\frac{\mathrm{Tr}[\rho_{L} H_L^{\mathrm{p}}]}{N}\label{eq:<H>^red=<H>}\\
    &=\lim_{L\to\infty}\frac{\mathrm{Tr}[\rho^{\mathrm{can}}_{L}(\beta^*|H) H_L^{\mathrm{p}}]}{N}\quad \Bigl(\because \text{ Eq.~\eqref{eq:s^red<=s^TD_DEF_beta} }\Bigr)\\
    &=\lim_{\ell\to\infty}\frac{\mathrm{Tr}[\rho^{\mathrm{can}}_{\ell}(\beta^*|H) H_\ell^{\mathrm{p}}]}{\ell^d}
    \label{eq:canEnergy_TDL}
\end{align}
Note that the above calculation can be performed in almost the same manner for other boundary conditions.

Since the relative entropy is nonnegative, RHS of Eq.~\eqref{eq:RelativeEnt_red|can} is also nonnegative.
Combining this with Eqs.~\eqref{eq:canEnergy_TDL} and \eqref{eq:s^red<=tilde_s}, we have
\begin{align}
    s^{\mathrm{mac}}_{\ell}[(\rho_{L})_{L\in\mathbb{N}}]
    \le \tilde{s}_{\ell}[(\rho_{L})_{L\in\mathbb{N}}]
    \le s^{\mathrm{TD}}(\beta^*|H).
\end{align}
By taking $\limsup_{\ell\to\infty}$, we obtain Eq.~\eqref{eq:s^red<=s^TD}.

\end{proof}

\subsection{Lower bound on the quantum macroscopic entropy density after a unital CPTP map}

Next, we investigate how $s^{\mathrm{mac}}[\bullet]$ can be decreased by a unital completely positive trace-preserving (CPTP) map. A CPTP map $\mathcal{U}(\bullet)$ on the total Hilbert space is said to be unital when it satisfies
\begin{align}
    \mathcal{U}(I_{\Lambda_L})=I_{\Lambda_{L}},
    \label{eq:DEF_unital}
\end{align}
where $I_{\Lambda_{L}}$ is the identity operator on the total Hilbert space.
For instance, a similarity transformation by a unitary $U$, $\mathcal{U}(\bullet)=U\bullet U^{\dagger}$, is a unital CPTP map.
For arbitrary unital CPTP maps, we can show the following proposition.
\begin{proposition}[Lower bound on the quantum macroscopic entropy density after a unital CPTP map]\label{proposition:s^red>=SvN}
Let $(\mathcal{U}_{L})_{L\in\mathbb{N}}$
be a sequence of unital CPTP maps  and 
$(\rho_{L})_{L\in\mathbb{N}}$ be a representation of 
a macroscopic state. 
Suppose that $\bigl(\mathcal{U}_L(\rho_{L})\bigr)_{L\in\mathbb{N}}$ also represents some macroscopic state.
Then, it holds that
\begin{align}
    s^{\mathrm{mac}}_{\ell}\bigl[\bigl(\mathcal{U}_L(\rho_{L})\bigr)_{L\in\mathbb{N}}\bigr]
    \ge \limsup_{L\to\infty}\frac{S_{\mathrm{vN}}[\rho_L]}{N}\quad \text{for all }\ell\in\mathbb{N}.
    \label{eq:s^red>=svN}
\end{align}
\end{proposition}

\begin{proof}
Let us consider a combination of a unital CPTP map $\mathcal{U}_L(\bullet)$ and the translation $T_{\boldsymbol{r}}^{\dagger}\bullet T_{\boldsymbol{r}}$ by some $\bm{r}\in\Lambda_L$,
\begin{align}
    \mathcal{U}_L^{\bm{r}}(\bullet):=T_{\boldsymbol{r}}^{\dagger}\mathcal{U}_L(\bullet) T_{\boldsymbol{r}},
\end{align}
which is also a unital CPTP map.
From the monotonicity of relative entropy~\cite{NielsenChuang2010}, we have
\begin{align}
    D(\rho_{L}\Vert I_{\Lambda_L}/D^N)
    \ge D(\mathcal{U}_L^{\bm{r}}(\rho_{L})\Vert \mathcal{U}_L^{\bm{r}}(I_{\Lambda_L}/D^N))
    =D(\mathcal{U}_L^{\bm{r}}(\rho_{L})\Vert I_{\Lambda_L}/D^N),
    \label{eq:Monotonicity_RelEnt}
\end{align}
where $D$ is the dimension of the Hilbert space on each site.
In addition, the relative entropy is related to the von Neumann entropy by
\begin{align}
    D(\rho_{L}\Vert I_{\Lambda_L}/D^N)=N\log D-S_{\mathrm{vN}}[\rho_{L}].
    \label{eq:RelEnt_SvN}
\end{align}
Combining Eqs.~\eqref{eq:Monotonicity_RelEnt} and \eqref{eq:RelEnt_SvN}, 
we obtain the monotonicity of the von Neumann entropy against a unital CPTP map
\begin{align}
    S_{\mathrm{vN}}[\rho_{L}]
    \le S_{\mathrm{vN}}[T_{\boldsymbol{r}}^{\dagger}\tilde{\rho}_{L}T_{\boldsymbol{r}}],
\end{align}
where we write $\tilde{\rho}_{L}=\mathcal{U}_L(\rho_{L})$ for simplicity of notation.
This inequality is used to obtain the lower bound of the quantum macroscopic entropy density of a macroscopic state $(\tilde{\rho}_{L})_{L\in\mathbb{N}}$.

To evaluate the quantum macroscopic entropy density of $(\tilde{\rho}_{L})_{L\in\mathbb{N}}$, let us consider partitioning the system into many hypercubes of side length $\ell=O(L^0)$. 
Above Eq.~\eqref{eq:DEF_rho_L^red} of Sec.~\ref{sec:Entropy_CGLS}, we have introduced a $d$-dimensional hypercube $\Cell\subset\Lambda_L$ with side length $\ell$ centered at site $\boldsymbol{0}\in\Lambda_L$.
Now we introduce hypercubes of the same size and partition the system by them.
Let $M:=(\lfloor L/K\ell\rfloor)^d$ be the number of hypercubes in a primitive macroscopic subsystem $\mathcal{S}^{(k)}$ and
$\Cell[\bm{r}_{j}^{k}]\subset\mathcal{S}^{(k)}$ ($j=1,...,M$) be disjoint $d$-dimensional hypercubes with side length $\ell$.  
Suppose that each $\Cell[\bm{r}_{j}^{k}]$ is centered at site $\bm{r}_{j}^{k}\in\mathcal{S}^{(k)}$.Then, $\Cell[\bm{r}_{j}^{k}]$ can be written as a translation of $\Cell$,
\begin{align}
    \Cell[\bm{r}_{j}^{k}]= \Cell+\bm{r}_{j}^{k}\quad (\bm{r}_{j}^{k}\in \mathcal{S}^{(k)}, j=1,...,M).
    \label{eq:C^k,j=C+r^k,j}
\end{align}
Let $R^{k}$ be the set of remaining sites in $\mathcal{S}^{(k)}$.
They satisfy
\begin{align}
    \mathcal{S}^{(k)}=\Bigl(\bigcup_{j=1}^{M}\Cell[\bm{r}_{j}^{k}]\Bigr)\cup R^k
    \label{eq:S^k=C^k,j_cup_R^k}
\end{align}
by definition.
Note that the number of sites in $R^k$,
\begin{align}
    |R^k|=|\mathcal{S}^{(k)}|-M\ell^d
    =(L/K)^d-(\lfloor L/K\ell\rfloor)^d\ell^d
    =O(L^{d-1}),
\end{align}
is negligibly small compared to the total number of sites $N=L^d$ when $L$ is large.

For any density matrix $\rho_{L}$ on $\Lambda_{L}$ and for any subset $S\subset \Lambda_{L}$, let $(\rho_{L})|_{S}:=\mathrm{Tr}_{\Lambda_L\setminus S}\bigl[\rho_{L}\bigr]$ be the reduced density matrix on $S$.
Note that $(\rho_{L})|_{\Cell[\bm{r}]}$ corresponds to $\rho_{L|\ell}^{\bm{r}}$ defined in Eq.~\eqref{eq:ReducedState_rho^r}.
Now, we consider the state given by the tensor product of the reduced density matrices of $T_{\boldsymbol{r}}^{\dagger}\tilde{\rho}_{L}T_{\boldsymbol{r}}$ on $\Cell[\bm{r}_{j}^{k}]$'s as a reference state,
\begin{align}
    &\tau_{L}^{\bm{r}}:= 
    \Bigl(\bigotimes_{k=1}^{K^d}\bigotimes_{j=1}^{M}(T_{\bm{r}}^{\dagger}\tilde{\rho}_{L}T_{\bm{r}})|_{\Cell[\bm{r}_{j}^{k}]}\Bigr)
    \otimes
    \Bigl(\bigotimes_{k=1}^{K^d}(T_{\bm{r}}^{\dagger}\tilde{\rho}_{L}T_{\bm{r}})|_{R^k}\Bigr).
\end{align}
The nonnegativity of the quantum relative entropy $D(T_{\bm{r}}^{\dagger}\tilde{\rho}_{L}T_{\bm{r}}\Vert \tau_{L}^{\bm{r}})$ results in the following inequality,
\begin{align}
    S_{\mathrm{vN}}[T_{\boldsymbol{r}}^{\dagger}\tilde{\rho}_{L}T_{\boldsymbol{r}}]
    &\le \sum_{k=1}^{K^d}\sum_{j=1}^{M}S_{\mathrm{vN}}[(T_{\boldsymbol{r}}^{\dagger}\tilde{\rho}_{L}T_{\boldsymbol{r}})|_{\Cell[\bm{r}_{j}^{k}]}]
    +\sum_{k=1}^{K^d}S_{\mathrm{vN}}[(T_{\boldsymbol{r}}^{\dagger}\tilde{\rho}_{L}T_{\boldsymbol{r}})|_{R^k}],
    \label{eq:SvN_Partition_1}
\end{align}
which is known as the subadditivity of the von Neumann entropy.
For the first term of the RHS of Eq.~\eqref{eq:SvN_Partition_1}, it holds that 
\begin{align}
    \sum_{j=1}^{M}S_{\mathrm{vN}}[(T_{\boldsymbol{r}}^{\dagger}\tilde{\rho}_{L}T_{\boldsymbol{r}})|_{\Cell[\bm{r}_{j}^{k}]}]
    =\sum_{j=1}^{M}S_{\mathrm{vN}}[(T_{\boldsymbol{r}+\bm{r}_{j}^{k}}^{\dagger}\tilde{\rho}_{L}T_{\boldsymbol{r}+\bm{r}_{j}^{k}})|_{\Cell}]
\end{align}
because of Eq.~\eqref{eq:C^k,j=C+r^k,j}.
Summing up the above quantity for all $\boldsymbol{r}\in \Cell$, we have
\begin{align}
    &\sum_{\boldsymbol{r}\in \Cell}\sum_{j=1}^{M}S_{\mathrm{vN}}[(T_{\boldsymbol{r}}^{\dagger}\tilde{\rho}_{L}T_{\boldsymbol{r}})|_{\Cell[\bm{r}_{j}^{k}]}]\notag\\
    &\quad=\sum_{\boldsymbol{r}\in \Cell}\sum_{j=1}^{M}S_{\mathrm{vN}}[(T_{\boldsymbol{r}+\bm{r}_{j}^{k}}^{\dagger}\tilde{\rho}_{L}T_{\boldsymbol{r}+\bm{r}_{j}^{k}})|_{\Cell}]\\
    &\quad=\sum_{\boldsymbol{r}^\prime\in \mathcal{S}^{(k)}\setminus R^k}S_{\mathrm{vN}}[(T_{\boldsymbol{r}^\prime}^{\dagger}\tilde{\rho}_{L}T_{\boldsymbol{r}^\prime})|_{\Cell}]\\
    &\quad\le \sum_{\boldsymbol{r}\in \mathcal{S}^{(k)}}S_{\mathrm{vN}}[(T_{\boldsymbol{r}}^{\dagger}\tilde{\rho}_{L}T_{\boldsymbol{r}})|_{\Cell}],
    \label{eq:SvN_S^k}
\end{align}
where the second equality follows from Eq.~\eqref{eq:S^k=C^k,j_cup_R^k} and the last inequality from the nonnegativity of the von Neumann entropy.
Furthermore, from the concavity of $S_{\mathrm{vN}}[\bullet]$ and the definition of the coarse-grained local state, Eq.~\eqref{eq:DEF_rho_L^red}, the above upper bound satisfies
\begin{align}
    \sum_{\boldsymbol{r}\in \mathcal{S}^{(k)}}S_{\mathrm{vN}}[(T_{\boldsymbol{r}}^{\dagger}\tilde{\rho}_{L}T_{\boldsymbol{r}})|_{\Cell}]
    \le |\mathcal{S}^{(k)}| \ S_{\mathrm{vN}}[
    \tilde{\rho}_{L|\ell}^{(k)}].
    \label{eq:SvN_Concavity}
\end{align}
Here, recall that $\tilde{\rho}_{L|\ell}^{(k)}$ is defined by Eq.~\eqref{eq:DEF_rho_L^red}.
Combining Eqs.~\eqref{eq:SvN_Partition_1}, \eqref{eq:SvN_S^k} and \eqref{eq:SvN_Concavity}, we have
\begin{align}
    \frac{S_{\mathrm{vN}}[\rho_{L}]}{N}
    &\le \frac{1}{\ell^d}\sum_{\boldsymbol{r}\in \Cell}
    \frac{S_{\mathrm{vN}}[T_{\boldsymbol{r}}^{\dagger}\tilde{\rho}_{L}T_{\boldsymbol{r}}]}{N}
    \\
    &\le \frac{1}{K^d}\sum_{k=1}^{K^d}\frac{S_{\mathrm{vN}}[
    \tilde{\rho}_{L|\ell}^{(k)}
    ]}{\ell^d}
    +\frac{1}{\ell^d}\sum_{\boldsymbol{r}\in \Cell}\sum_{k=1}^{K^d}
    \frac{S_{\mathrm{vN}}[(T_{\boldsymbol{r}}^{\dagger}\tilde{\rho}_{L}T_{\boldsymbol{r}})|_{R^k}]}{N}.
    \label{eq:SvN_Partition_2}
\end{align}
For the second term of Eq.~\eqref{eq:SvN_Partition_2}, we have
\begin{align}
    S_{\mathrm{vN}}[(T_{\boldsymbol{r}}^{\dagger}\tilde{\rho}_{L}T_{\boldsymbol{r}})|_{R^k}]
    \le |R^k|\log D=O(L^{d-1}).
\end{align}
Therefore, 
taking $\limsup_{L\to\infty}$ of Eq.~\eqref{eq:SvN_Partition_2}, we obtain
\begin{align}
    \limsup_{L\to\infty}\frac{S_{\mathrm{vN}}[\rho_L]}{N}
    \le 
    \frac{1}{K^d}\sum_{k=1}^{K^d}\frac{S_{\mathrm{vN}}[
    \tilde{\rho}_{\infty|\ell}^{(k)}
    ]}{\ell^d},
\end{align}
where we interchanged $\limsup_{L\to\infty}$ and $S_{\mathrm{vN}}[\bullet]$ using the continuity of the von Neumann entropy~\cite{NielsenChuang2010}.
The RHS of the above inequality corresponds to $s^{\mathrm{mac}}_{\ell}\bigl[\bigl(\mathcal{U}_L(\rho_{L})\bigr)_{L\in\mathbb{N}}\bigr]$.

\end{proof}

\subsection{Proof of Eq.~\eqref{eq:s^red_rho^can=s^TD} which is used in the proof of Theorem~\ref{theorem:s^red=s^TD}}

\begin{proof}
Inserting $\rho_{L}=\rho^{\mathrm{can}}_{L}(\beta|H)$ into Lemma~\ref{lemma:s^red<=s^TD},
we have
\begin{align}
    \limsup_{\ell\to\infty}s^{\mathrm{mac}}_{\ell}\bigl[\bigl(\rho^{\mathrm{can}}_{L}(\beta|H)\bigr)_{L\in\mathbb{N}}\bigr]
    \le s^{\mathrm{TD}}(\beta|H).
\end{align}
On the other hand, inserting $\rho_{L}=\rho^{\mathrm{can}}_{L}(\beta|H)$ into Proposition~\ref{proposition:s^red>=SvN} and taking $\mathcal{U}_{L}$ in Proposition~\ref{proposition:s^red>=SvN} as the identity map, we have
\begin{align}
    s^{\mathrm{mac}}_{\ell}\bigl[\bigl(\rho^{\mathrm{can}}_{L}(\beta|H)\bigr)_{L\in\mathbb{N}}\bigr]
    \ge \limsup_{L\to\infty}\frac{S_{\mathrm{vN}}[\rho^{\mathrm{can}}_{L}(\beta|H)]}{N}
    =s^{\mathrm{TD}}(\beta|H)
    \quad \text{for all }\ell\in\mathbb{N},
\end{align}
where the last equality follows by definition~\eqref{eq:DEF_s^TD}.
Combining these inequalities and taking $\liminf_{\ell\to\infty}$ of the latter inequality, 
we obtain 
\begin{align}
    s^{\mathrm{TD}}(\beta|H)
    \le \liminf_{\ell\to\infty}s^{\mathrm{mac}}_{\ell}\bigl[\bigl(\rho^{\mathrm{can}}_{L}(\beta|H)\bigr)_{L\in\mathbb{N}}\bigr]
    \le \limsup_{\ell\to\infty}s^{\mathrm{mac}}_{\ell}\bigl[\bigl(\rho^{\mathrm{can}}_{L}(\beta|H)\bigr)_{L\in\mathbb{N}}\bigr]
    \le s^{\mathrm{TD}}(\beta|H),
\end{align}
which corresponds to Eq.~\eqref{eq:s^red_rho^can=s^TD}.

\end{proof}

\subsection{\label{sec:Proof_EntropyIncrease}Proof of Theorem~\ref{theorem:EntropyIncrease_WithoutThermalization}}

\begin{proof}
We introduce
\begin{align}
    \mathcal{U}_L(\bullet):=\lim_{\mathcal{T}\to\infty}\frac{1}{\mathcal{T}}\int_{0}^{\mathcal{T}} dt\ e^{-iH_{1}t}U(t^*,0)\bullet U^{\dagger}(t^*,0)e^{iH_{1}t},
\end{align}
which is a unital CPTP map.
The existence of the limit $\lim_{\mathcal{T}\to\infty}$ can be shown from the fact that the Hilbert space is finite dimensional. 
Inserting $\rho_{L}=\rho^{\mathrm{can}}_{L}(\beta|H)$ into Proposition~\ref{proposition:s^red>=SvN} and taking $\mathcal{U}_L(\bullet)$ in Proposition~\ref{proposition:s^red>=SvN} as the above one,
we have
\begin{align}
    s^{\mathrm{mac}}_{\ell}\Bigl[\Bigl(\mathcal{U}_L\bigl(\rho_{L}^{\mathrm{can}}(\beta_{0}|H_{0})\bigr)\Bigr)_{L\in\mathbb{N}}\Bigr]\ge s^{\mathrm{TD}}(\beta_{0}|H_{0})\quad \text{for all }\ell\in\mathbb{N}.
\end{align}
Furthermore, from Assumption~\ref{assumption:UniqueFinalState}, 
for any representation $(\rho_{L})_{L\in\mathbb{N}}$ of iMATE described by $H_{0}$ at inverse temperature $\beta_{0}$, i.e., 
$(\rho_{L})_{L\in\mathbb{N}}\maceq \bigl(\rho_{L}^{\mathrm{can}}(\beta_{0}|H_{0})\bigr)_{L\in\mathbb{N}}$, 
it holds that
\begin{align}
    \bigl(\overline{\rho_{L}(t>t^*)}\bigr)_{L\in\mathbb{N}}
    =\bigl(\mathcal{U}_L(\rho_{L})\bigr)_{L\in\mathbb{N}}
    \maceq \Bigl(\mathcal{U}_L\bigl(\rho_{L}^{\mathrm{can}}(\beta_0|H_0)\bigr)\Bigr)_{L\in\mathbb{N}}.
\end{align}
Combining these with Corollary~\ref{corollary:Equiv_s^red}, we have
\begin{align}
    s^{\mathrm{mac}}_{\ell}\bigl[\bigl(\overline{\rho_{L}(t>t^*)}\bigr)_{L\in\mathbb{N}}\bigr]
    \ge s^{\mathrm{TD}}(\beta_{0}|H_{0})\quad \text{for all }\ell\in\mathbb{N}.
    \label{eq:s^red_final>=s^TD_0}
\end{align}
In addition, using Theorem~\ref{theorem:s^red=s^TD}, the RHS of the above inequality can be written as
\begin{align}
    s^{\mathrm{TD}}(\beta_{0}|H_{0})=\lim_{\ell\to\infty}s^{\mathrm{mac}}_{\ell}\bigl[(\rho_{L})_{L\in\mathbb{N}}\bigr].
\end{align}
Taking $\liminf_{\ell\to\infty}$ of Eq.~\eqref{eq:s^red_final>=s^TD_0}, we obtain
\begin{align}
    \liminf_{\ell\to\infty}s^{\mathrm{mac}}_{\ell}\bigl[\bigl(\overline{\rho_{L}(t>t^*)}\bigr)_{L\in\mathbb{N}}\bigr]
    \ge\lim_{\ell\to\infty}s^{\mathrm{mac}}_{\ell}\bigl[(\rho_{L})_{L\in\mathbb{N}}\bigr],
\end{align}
which corresponds to Theorem~\ref{theorem:EntropyIncrease_WithoutThermalization}.

\end{proof}

\subsection{\label{sec:Proof_Th3_from_Th4}Proof of Theorem~\ref{theorem:EntropyIncrease} using Theorem~\ref{theorem:EntropyIncrease_WithoutThermalization}}

Equation \eqref{eq:Increasing_s^red_Thermalization} of Theorem~\ref{theorem:EntropyIncrease} readily follows 
from Theorem~\ref{theorem:EntropyIncrease_WithoutThermalization}
as follows.

\begin{proof}
Suppose that the assumptions of Theorem~\ref{theorem:EntropyIncrease} are satisfied. First, we will show that Assumption~\ref{assumption:UniqueFinalState} is also satisfied. Consider two representations of the same iMATE, i.e., $(\rho_{L})_{L\in\mathbb{N}}$ and $(\sigma_{L})_{L\in\mathbb{N}}$ satisfying Eq.~\eqref{eq:UniqueFinalState_InitialStates}. From Corollary~\ref{corollary:macroscopically-equivalence_equilibrium-state}, the expectation values of $H_1/N$ in $\rho_L(t^*)$ and $\sigma_L(t^*)$ coincide, which implies that the inverse temperatures $\beta_1$ determined from Eq.~\eqref{eq:FiniteEnergyDensity} also coincide. From Assumption~\ref{assumption:thermalization}, $\bigl(\overline{\rho_{L}(t>t^*)}\bigr)_{L\in\mathbb{N}}$ and $\bigl(\overline{\sigma_{L}(t>t^*)}\bigr)_{L\in\mathbb{N}}$ represent the same iMATE, $\bigl(\rho_{L}^{\mathrm{can}}(\beta_{1}|H_1)\bigr)_{L\in\mathbb{N}}$. This implies Eq.~\eqref{eq:UniqueFinalState_FinalState}, indicating that Assumption~\ref{assumption:UniqueFinalState} is satisfied.

Next, we 
show that Eq.~\eqref{eq:Increasing_s^red_NoThermalization} reduces to Eq.~\eqref{eq:Increasing_s^red_Thermalization}.
Since $\bigl(\overline{\rho_{L}(t>t^*)}\bigr)_{L\in\mathbb{N}}$ represents iMATE, $\lim_{\ell\to\infty}s^{\mathrm{mac}}_{\ell}\bigl[\bigl(\overline{\rho_{L}(t>t^*)}\bigr)_{L\in\mathbb{N}}\bigr]$ exists, as explained below Eq.~\eqref{eq:Increasing_s^red_Thermalization}. This means that the RHS of Eqs.~\eqref{eq:Increasing_s^red_Thermalization} and \eqref{eq:Increasing_s^red_NoThermalization} coincide, implying that Eq.~\eqref{eq:Increasing_s^red_NoThermalization} reduces to Eq.~\eqref{eq:Increasing_s^red_Thermalization}.

\end{proof}

\subsection{\label{sec:DirectProof_EntropyIncrease}Direct proof of Theorem~\ref{theorem:EntropyIncrease}}

We here give a more direct proof of Theorem~\ref{theorem:EntropyIncrease}.

\begin{proof}
First, we will show that 
$\bigl(\overline{\rho_{L}(t>t^*)}\bigr)_{L\in\mathbb{N}}$ represents iMATE.
Since $\bigl(\rho_{L}(0)\bigr)_{L\in\mathbb{N}}$ represents normal iMATE, Proposition~\ref{proposition:JAIVTMXN_2} implies that $\bigl(\rho_{L}(t^*)\bigr)_{L\in\mathbb{N}}$ represents a normal macroscopic state. This means that Eq.~\eqref{eq:NegligibleVarH_1} is satisfied. This fact combined with Assumption~\ref{assumption:thermalization} implies that $\bigl(\overline{\rho_{L}(t>t^*)}\bigr)_{L\in\mathbb{N}}$ represents iMATE.

Next, we will show Eq.~\eqref{eq:Increasing_s^TD}.
Let $\sigma_{L}=\rho_{L}^{\mathrm{can}}(\beta_{0}|H_{0})$ be the canonical Gibbs state described by the initial Hamiltonian, and $\sigma_{L}(t^*)=U_{L}(t^*,0)\sigma_{L}U_{L}^{\dagger}(t^*,0)$ be its time evolution by a macroscopic operation.
Consider the quantum relative entropy $D(\sigma\Vert \rho)=\mathrm{Tr}[\sigma\log\sigma]-\mathrm{Tr}[\sigma\log\rho]$ of $\sigma_{L}(t^*)$ with respect to $\rho_{L}^{\mathrm{can}}(\beta_{1}|H_{1})$,
\begin{align}
    &D\bigl(\sigma_{L}(t^*)||\rho_{L}^{\mathrm{can}}(\beta_{1}|H_{1})\bigr)
    \notag\\
    &\quad  =\beta_{1}\bigl(\mathrm{Tr}[\sigma_{L}(t^*) H_{1}]-\mathrm{Tr}[\rho_{L}^{\mathrm{can}}(\beta_{1}|H_{1}) H_{1}]\bigr)
    \notag\\
    &\quad +S_{\mathrm{vN}}[\rho_{L}^{\mathrm{can}}(\beta_{1}|H_{1})]-S_{\mathrm{vN}}[\sigma_{L}(t^*)].
\end{align}
From Theorem~\ref{theorem:JAIVTMXN_3}, $\bigl(\sigma_{L}(t^*)\bigr)_{L\in\mathbb{N}}$ is macroscopically equivalent to $\bigl(\rho_{L}(t^*)\bigr)_{L\in\mathbb{N}}$. This implies $\sigma_L(t^*)$ has the same energy density $H_1/N$ as that of $\rho_L(t^*)$.
Moreover, from Eq.~\eqref{eq:FiniteEnergyDensity}, energy density of $\rho_{L}^{\mathrm{can}}(\beta_{1}|H_1)$ also coincides.
In addition, the unitary invariance of the von Neumann entropy implies $S_{\mathrm{vN}}[\sigma_{L}(t^*)]=S_{\mathrm{vN}}[\rho_{L}^{\mathrm{can}}(\beta_{0}|H_{0})]$.
Thus, combining these with nonnegativity of relative entropy, we obtain Eq.~\eqref{eq:Increasing_s^TD}.

Finally, note that since $\bigl(\overline{\rho_{L}(t>t^*)}\bigr)_{L\in\mathbb{N}}$ represents iMATE, Eq.~\eqref{eq:Increasing_s^TD} is equivalent to Eq.~\eqref{eq:Increasing_s^red_Thermalization}.
\end{proof}

\subsection{\label{sec:Proof_Kdependence}Proof of Proposition~\ref{proposition:Kindependent_s^mac}}

The core idea of the proof of Proposition~\ref{proposition:Kindependent_s^mac} is given in Sec.~\ref{sec:Kdependence}. The remaining task is to prove Eq.~\eqref{eq:Kdependence_concavity_s^macK}.
\begin{proof}[Proof of Eq.~\eqref{eq:Kdependence_concavity_s^macK}]
For simplicity of explanation, 
we here write the spatial average of $\ell$-local density matrices~\eqref{eq:DEF_rho_L^red} as 
\begin{align}
    \rho_{L|\ell}^{\bm{k},K},\qquad \bm{k}=(k^1,k^2,...,k^d)\in\{1,2,...,K\}^d
\end{align}
instead of the original symbol $\rho_{L|\ell}^{(k)}$ ($k\in\{1,2,...,K^d\}$).
Here, the label $\bm{k}$ divided by $K$ expresses the relative position of the primitive subsystem.
Furthermore, for notational simplicity, we write $K_1=(M^{\mathrm{max}})!$ and $K_2=(M^{\mathrm{max}}+1)!$.
In a similar manner as Eq.~\eqref{eq:DEF_rho_infty^tot_1}, we can show that
\begin{align}
    \rho_{L|\ell}^{\bm{k},K_1}=\frac{1}{(M^{\mathrm{max}}+1)^d}\sum_{\bm{q}\text{ near } \bm{k}}\rho_{L|\ell}^{\bm{q},K_2},
    \label{eq:Kdependence_ConvComb}
\end{align}
where $\bm{k}\in\{1,2,...,K_1\}^d$ and $\bm{q}\in\{1,2,...,K_2\}^d$.
Here, the sum $\sum_{\bm{q}\text{ near } \bm{k}}$ runs over all tuples of integers $\bm{q}=(q^1,q^2,...,q^d)$ satisfying
\begin{align}
    (k^\mu-1)(M^{\mathrm{max}}+1)+1\le q^{\mu}\le k^{\mu}(M^{\mathrm{max}}+1)
    \label{eq:Kdependence_Cond_q}
\end{align}
for all $\mu=1,2,...,d$.
The condition~\eqref{eq:Kdependence_Cond_q} roughly states that $\bm{q}/K_2$ is close to $\bm{k}/K_1$ [within the precision of $1/K_1$].
Combining Eq.~\eqref{eq:Kdependence_ConvComb} with the concavity of the von Neumann entropy, we have
\begin{align}
    S_{\mathrm{vN}}[\rho_{\infty|\ell}^{\bm{k},K_1}]\ge
    \frac{1}{(M^{\mathrm{max}}+1)^d}\sum_{\bm{q}\text{ near } \bm{k}}S_{\mathrm{vN}}[\rho_{\infty|\ell}^{\bm{q},K_2}].
\end{align}
This results in 
\begin{align}
    \frac{1}{K_1^d}\sum_{\bm{k}\in\{1,2,...,K_1\}^d}S_{\mathrm{vN}}[\rho_{\infty|\ell}^{\bm{k},K_1}]
    \ge 
    \frac{1}{K_2^d}\sum_{\bm{q}\in\{1,2,...,K_2\}^d}S_{\mathrm{vN}}[\rho_{\infty|\ell}^{\bm{q},K_2}],
\end{align}
which reduces to Eq.~\eqref{eq:Kdependence_concavity_s^macK}.
\end{proof}

\section{\label{sec:Proof_Example_OrdinaryMATE}Difficulty in characterizing thermal equilibrium by only a few additive observables (details of Example~\ref{example:ProblemFiniteObs})}

\begin{proposition}[\label{proposition:ProblemFiniteObs}Contradiction to Planck's principle when thermal equilibrium were characterized by only a
few additive observables]
Consider the system described by the Hamiltonian~\eqref{eq:Counterex_XYmodel} and take $h^x$ as a control parameter.
There exist a sequence of states $(\rho_L)_{L\in\mathbb{N}}$ and a macroscopic operation $\bigl(U_L(t^*,0)\bigr)$ of operation time $t^*=O(L^0)$ such that, although $(\rho_L)_{L\in\mathbb{N}}$ represents thermal equilibrium with respect to $\tilde{\mathcal{A}}=\{(M_L^x)_{L\in\mathbb{N}},(M_L^y)_{L\in\mathbb{N}},(H_L^{XY})_{L\in\mathbb{N}}\}$ in the sense that 
\begin{align}
    \lim_{L\to\infty}\mathrm{Tr}[\rho_L\frac{A_L}{N}]=\lim_{L\to\infty}\mathrm{Tr}[\rho_{L}^{\mathrm{can}}(\beta|H_L)\frac{A_L}{N}]\quad \text{for any }(A_L)_{L\in\mathbb{N}}\in\tilde{\mathcal{A}}
    \label{eq:Indistinguishable_tildeA}
\end{align}
is satisfied, Planck's principle is violated,
\begin{align}
    \lim_{L\to\infty}\mathrm{Tr}[\rho_L(t^*)\frac{H_L}{N}]<\lim_{L\to\infty}\mathrm{Tr}[\rho_L\frac{H_L}{N}],
    \label{eq:ProblemFiniteObs}
\end{align}
where $\rho_L(t^*)=U_L(t^*,0)\rho_L U_L^{\dagger}(t^*,0)$.
Moreover, $(\rho_L)_{L\in\mathbb{N}}$ satisfies the condition for the ordinary MATE formulated by Tasaki~\cite{Tasaki2016} in addition to Eq.~\eqref{eq:Indistinguishable_tildeA}.

\end{proposition}

\noindent 
In other words, the choice $\tilde{\mathcal{A}}$ is inappropriate in the sense that it cannot distinguish a macroscopically-nonpassive state $\rho_L$ from $\rho_{L}^{\mathrm{can}}(\beta|H_L)$.

\begin{proof}
Consider a sequence of states $(\ket{\psi_{L}}\bra{\psi_{L}})_{L\in\mathbb{N}}$ defined by $\ket{\psi_{L}}:=\ket{0}^{\otimes N}$, which is the eigenstate of $\sigma_{j,k}^z$ with the eigenvalue $+1$. 
Because the expectation values of the additive observables $M_L^x,M_L^y,H_L^{XY}$ in the state $\ket{\psi_{L}}$ are zero coinciding with those in $\rho_{L}^{\mathrm{can}}(\beta=0|H_L)$, Eq.~\eqref{eq:Indistinguishable_tildeA} is satisfied at $\beta=0$.
In addition, since $\ket{\psi_{L}}$ is a product state, the distributions of additive observables $\tilde{\mathcal{A}}$ have a sufficiently sharp peak, indicating that it satisfies the condition for the ordinary MATE formulated by Tasaki~\cite{Tasaki2016}.

However, $\ket{\psi_{L}}$ responds to a macroscopic operation differently than $\rho_{L}^{\mathrm{can}}(\beta=0|H_L)$, and it does not satisfy macroscopic passivity~\eqref{eq:macroscopic-passivity}. We first provide a rough sketch of the proof.
For simplicity, we take $h^y=0$. We try to realize the $\pi/2$ pulse around the $x$ axis. To this end, at $t=0$, we suddenly change $h^x$ to some positive value $1/\delta$ that is sufficiently larger than $|J^x|$ and $|J^y|$. The time evolution $U_L(\bullet,0)$ generated by such a Hamiltonian is a macroscopic operation, and at $t=t^*:=\pi\delta/2$, the unitary $U_L(t^*,0)$ is close to the $\pi/2$ pulse around the $x$ axis, which induces the state transition to the eigenstate of $\sigma_{j,k}^{y}$ with the eigenvalue~$-1$.
As a result, we have 
\begin{align}
    &\lim_{L\to\infty}\bra{\psi_{L}}U_{L}^{\dagger}(t^*,0)\frac{H_L}{N} U_{L}(t^*,0)\ket{\psi_{L}}= -2J^y+O(\delta)
    <0
    \ \Bigl(=\lim_{L\to\infty}\bra{\psi_{L}}\frac{H_L}{N} \ket{\psi_{L}}\Bigr),
    \label{eq:ProblemFiniteObs_EnergyChange}
\end{align}
which implies Eq.~\eqref{eq:ProblemFiniteObs}.

Now, we will show Eq.~\eqref{eq:ProblemFiniteObs_EnergyChange} exactly.
The above macroscopic operation $U_L(\bullet,0)$ is generated by
\begin{align}
    H_L(t)=\begin{cases}
        H_L=H_L^{XY}-h^xM_L^x &\qquad t\le 0\\
        H_L^{XY}-\frac{1}{\delta} M_L^x &\qquad 0<t<t^*=\pi\delta/2
    \end{cases}.
\end{align}
Below, we consider approximating the macroscopic operation $U_L(t,0)$ by $e^{\frac{i}{\delta}M_L^x t}$ for small $\delta$.
By a calculation similar to Eq.~\eqref{eq:Bound_TEvol_Comm}, we can show that an arbitrary observable $A_L$ satisfies 
\begin{align}
    \|U_L^{\dagger}(t^*,0)A_L U_L(t^*,0)-\tilde{A}_{L}(t^*)\|_{\infty}
    \le \int_{0}^{t^*} dt\ \|[\tilde{A}_L(t),H_L^{XY}]\|_{\infty}
\end{align}
where $\tilde{A}_L(t):=e^{-i\frac{t}{\delta}M_L^x} A_L e^{i\frac{t}{\delta}M_L^x}$.
Inserting $A_L=H_L$, we have 
\begin{align}
    \|U_L^{\dagger}(t^*,0)H_L U_L(t^*,0)-\tilde{H}_{L}(t^*)\|_{\infty}
    &\le \int_{0}^{t^*} dt\ \|[\tilde{H}_L(t),H_L^{XY}]\|_{\infty}\\
    &\le t^*\times\sup_{t\in[0,t^*]}\|[\tilde{H}_L(t),H_L^{XY}]\|_{\infty}.
    \label{eq:Bound_EnergyChange}
\end{align}
Here, $\tilde{H}_{L}(t)$ can be written as $\tilde{H}_{L}(t)=\tilde{H}_L^{XY}(t)-h^xM_L^x$, 
where 
\begin{align}
    &\tilde{H}_L^{XY}(t)=e^{-i\frac{t}{\delta}M_L^x}H_L^{XY}e^{i\frac{t}{\delta}M_L^x}=-\sum_{j,k=1}^{L}\Bigl(J^x(\sigma_{j,k}^{x}\sigma_{j+1,k}^{x}+\sigma_{j,k}^{x}\sigma_{j,k+1}^{x})+J^y\bigl(\sigma_{j,k}^{y}(t)\sigma_{j+1,k}^{y}(t)+\sigma_{j,k}^{y}(t)\sigma_{j,k+1}^{y}(t)\bigr)\Bigr),\\
    &\sigma_{j,k}^{y}(t)=\cos(t/\delta)\sigma_{j,k}^{y}-i\sin(t/\delta)\sigma_{j,k}^z.
\end{align}
Due to the additivity of $\tilde{H}_L(t)$, the commutator in the RHS of Eq.~\eqref{eq:Bound_EnergyChange} is also additive, and hence is of $O(N)$.
This implies that 
\begin{align}
    \Bigl|\lim_{L\to\infty}\bra{\psi_L}U_L^{\dagger}(t^*,0)\frac{H_L}{N}U_L(t^*,0)\ket{\psi_L}
    -\lim_{L\to\infty}\bra{\psi_L}\frac{\tilde{H}_L(t^*)}{N}\ket{\psi_L}\Bigr|
    \le t^*\times\lim_{L\to\infty}\frac{\sup_{t\in[0,t^*]}\|[\tilde{H}_L(t),H_L^{XY}]\|_{\infty}}{N}
    =O(\delta).
\end{align}
Combining this with $\bra{\psi_L}\frac{\tilde{H}_L(t^*)}{N}\ket{\psi_L}=-2J^y$, we obtain Eq.~\eqref{eq:ProblemFiniteObs_EnergyChange}.

\end{proof}

\section{\label{sec:Local_control}Local controls}

\subsection{Passivity to local controls}

In this Appendix, we
extend our results for the setup ``iMATE $+$ macroscopic operations'' to the setup ``MITE $+$ local control''. Precisely, MITE considered here corresponds to $O(L^0)$-local MITE defined by Definition~\ref{definition:MITE} and local control considered here is given as follows:
\begin{definition}[\label{definition:LocalControl}Local control]
Let $r_H$ be a positive integer independent of $L$ and consider a time-dependent Hamiltonian 
\begin{align}
    H_L(t)=\sum_{\bm{r}\in\Lambda_L}h^{\bm{r}}_L(t),
\end{align}
where $h^{\bm{r}}_L(t)$'s are $r_H$-local observables.
Here, $h^{\bm{r}}_L(t)$'s can be arbitrarily chosen for all $\bm{r}\in\Lambda_L$, and different $h^{\bm{r}}_L(t)$'s are not necessarily related by translation in contrast to Eqs.~\eqref{eq:Hamiltonian} and \eqref{eq:general_H}.
Assume that $\|h^{\bm{r}}(t)\|$'s are bounded from above by some constant independent of $L$, $\bm{r}$, and $t$.
Let $U_{L}(T,0)$ be the unitary operator generated by $H(t)$ from $t=0$ to $t=T$.
The sequence of these unitary operators $\bigl(U_{L}(T,0)\bigr)_{L\in\mathbb{N}}$ is called local control of operation time $T$.
\end{definition}
\noindent
Note that the Hamiltonian here does not have any spatial uniformity, unlike Eqs.~\eqref{eq:Hamiltonian} and \eqref{eq:general_H}. Therefore, ``local control'' is a broader class of operations than ``macroscopic operation''.

With 
these definitions of $O(L^0)$-local MITE, Definition~\ref{definition:MITE}, and local control, Definition~\ref{definition:LocalControl}, we 
obtain the following proposition by a straightforward extension of our proof for our original setup of ``iMATE $+$ macroscopic operations'':
\begin{proposition}[\label{corollary:Hokkyo_Passivity}Passivity of MITE to local control]
Let $(\rho_{L})_{L\in\mathbb{N}}$ be a sequence of density matrices that represents $O(L^0)$-local MITE described by a Hamiltonian $H_L$ with inverse temperature $\beta\ge 0$.
Let $t^*$ be an arbitrary positive number independent of $L$.
For any local control $\bigl(U_{L}(t^*,0)\bigr)_{L\in\mathbb{N}}$ of operation time $t^*$, the expectation value of $H_L$ does not decrease extensively,
\begin{align}
    \lim_{L \to \infty} \mathrm{Tr} \left[ \rho_L(t^*) \initH / N \right]
    \geq \lim_{L \to \infty} \mathrm{Tr} \left[ \rho_L(0) \initH / N \right],
    \label{eq:LocalPassivity}
\end{align}
where $\rho_L(t) = U_L(t, 0) \rho_L {U_L(t, 0)}^\dagger$.

\end{proposition}
\noindent
This means that extensive work cannot be extracted from $O(L^0)$-local MITE by any local control of operation time independent of $L$.
This extends thermodynamic passivity 
by Hokkyo and Ueda \cite{Hokkyo2025} 
to the case of $\beta=0$ and to the case where $(\rho_{L})_{L\in\mathbb{N}}$ represents $O(L^0)$-local MITE but does not satisfy 
Eq.~\eqref{eq:condition_MITE_LargerLength}.

\subsection{\label{sec:Proof_Problem_CombiningiMATEandLocalControl}Proof of Eq.~\eqref{eq:Problem_CombiningiMATEandLocalControl_nonPassive} in Example~\ref{example:Problem_CombiningiMATEandLocalControl}
}

Thus, the passivity holds true for the ``MITE $+$ local control'' setup, 
in addition to the ``iMATE $+$ macroscopic operation'' setup (Corollary~\ref{corollary:macroscopic-passivity}).
One might expect that they could be extended to a broader setup of ``iMATE $+$ local control.'' 
However, Example~\ref{example:Problem_CombiningiMATEandLocalControl} shows that such an extension is impossible. Below, we prove Eq.~\eqref{eq:Problem_CombiningiMATEandLocalControl_nonPassive} in Example~\ref{example:Problem_CombiningiMATEandLocalControl}.

\begin{proof}[Proof of Eq.~\eqref{eq:Problem_CombiningiMATEandLocalControl_nonPassive}]
Let us consider a pure product state $\rho_L=\ket{\psi_L}\bra{\psi_L}$ given by Eq.~\eqref{eq:Example_MacroEqState_infty} and the initial Hamiltonian by Eq.~\eqref{eq:Problem_CombiningiMATEandLocalControl_initH}.
The sequence $(\rho_L)_{L\in\mathbb{N}}$ represents iMATE at $\beta=0$.
Consider flipping all up spins by a local control $U_L(t^*,0)$ generated by
\begin{align}
    H_L(t)=\sum_{j=1}^{L}b_L^{j}\sigma_{j}^x,
\end{align}
where $b_{L}^{j}\in\{0,1\}$ are taken such that $b_{L}^{j}=0$ holds if the spin on site~$j$ in $\ket{\psi_L}$ points downward, and $b_{L}^{j}=1$ if that spin points upward.
By taking the operation time as $t^*=\pi/2$, $U_L(t^*,0)$ becomes the $\pi$-pulse around $x$-axis for any site $j$ satisfying $b_{L}^{j}=1$. 
Applying this local control $\rho_L(t^*) = U_L(t^*, 0) \rho_L {U_L(t^*, 0)}^\dagger$, we have
\begin{align}
    \lim_{L \to \infty} \mathrm{Tr} \left[ \rho_L(t^*) \initH / N \right]=-1
    < \lim_{L \to \infty} \mathrm{Tr} \left[ \rho_L \initH / N \right]=0,
\end{align}
which means that Eq.~\eqref{eq:Problem_CombiningiMATEandLocalControl_nonPassive} holds.
\end{proof}

\section{Construction of counterexamples at a longer operation time}

\subsection{\label{sec:breakdown_macroscopic-passivity}Proof of Proposition~\ref{proposition:breakdown_macroscopic-passivity}}
We prove Proposition~\ref{proposition:breakdown_macroscopic-passivity} by an explicit construction. Since we can obtain a $d$-dimensional system by bundling $L^{d-1}$ copies of a one-dimensional system, it suffices to provide an example for the case of $d = 1$. Moreover, since the timescale can be freely rescaled by a factor independent of $L$ through multiplying the Hamiltonian by a constant, we can assume without loss of generality that $T_L < L/8$.

While we consider a system with a local Hilbert space of dimension $4$, for the sake of explanation, we use an equivalent setting where the local Hilbert space is of dimension $2$ and $L$ is even. By grouping two adjacent sites ($2j-1$ and $2j$ for $j=1,2,\dots L/2$) into a single site, we can obtain a system with a local Hilbert space dimension of $4$. It is important to note that a one-site translation in the system with a local Hilbert space dimension of $4$ corresponds to a two-site translation in the system with a local Hilbert space dimension of $2$.

\subsubsection{Setup}
We choose the initial Hamiltonian $\initH$ as
\begin{align}
    \initH = \frac{\pi}{4} \sum_{j=1}^L \vec{\sigma}_{j} \cdot \vec{\sigma}_{j+1} + \frac{\pi}{4} \sum_{j=1}^L \vec{\sigma}_{j} \cdot \vec{\sigma}_{j+3},
\end{align}
where $\vec{\sigma}_{i} \cdot \vec{\sigma}_{j} = \sum_{\mu=x,y,z}\sigma_{i}^\mu \sigma_{j}^\mu$.

In addition, for the initial state $\rho_L$, we select a pure state $\rho_L = \ket{\psi_L}\bra{\psi_L}$, where
\begin{align}
    \ket{\psi_L} &= \bigotimes_{j=1}^{L/2} \ket{\Phi}_{2j-1-2M_L \mathrm{\ mod\ } L,\ 2j+2M_L \mathrm{\ mod\ } L}.
    \label{eq:breakdown_InitialState}
\end{align}
Here, $M_L = \lfloor T_L \rfloor$ is the largest integer not greater than $T_L$, and $\ket{\Phi} = \frac{1}{\sqrt{2}} (\ket{01}-\ket{10})$ is a spin singlet state.
Then $(\rho_L)_{L\in\mathbb{N}}$
represents 
iMATE at infinite temperature because for any subsystem with diameter of $O(L^0)$ the reduced density matrix is maximally mixed.

Furthermore, we choose the macroscopic operation $(U_L(\bullet, 0))_{L\in\mathbb{N}}$ to be determined, according to Definition~\ref{definition:macroscopic-operation}, by the following additive observables and the external fields coupled to them:
\begin{align}
    B^1 = \frac{\pi}{4} \sum_{j=1}^{L} (-1)^j \vec{\sigma}_{j} \cdot \vec{\sigma}_{j+1}, \qquad
    B^2 = \frac{\pi}{4} \sum_{j=1}^{L} \vec{\sigma}_{j} \cdot \vec{\sigma}_{j+3},
\end{align}
and
\begin{align}
    f^1(t) = \begin{cases}
        + 1 & (n \leq t < n+1/2)\\
        - 1 & (n+1/2 \leq t < n+1)
    \end{cases} \quad (n \in \mathbb{Z}_{\geq 0}), \qquad
    f^2(t) = 1.
\end{align}
Then, from Eq.~\eqref{eq:general_H}, the generator of the time evolution $H(t)$ satisfies
\begin{align}
    H(t) = \begin{cases}
        \displaystyle \frac{\pi}{2} \sum_{j=1}^{L/2} \vec{\sigma}_{2j-1} \cdot \vec{\sigma}_{2j} & (n \leq t < n+1/2)\\ \\
        \displaystyle \frac{\pi}{2} \sum_{j=1}^{L/2} \vec{\sigma}_{2j} \cdot \vec{\sigma}_{2j+1} & (n+1/2 \leq t < n+1)
    \end{cases} \quad (n \in \mathbb{Z}_{\geq 0}).
    \label{eq:breakdown_Hamiltonian}
\end{align}

\subsubsection{Proof of Eq.~\eqref{eq:breakdown_macroscopic-passivity}}
We now prove that Eq.~\eqref{eq:breakdown_macroscopic-passivity} is satisfied in the setting provided above.

First, note that for any $n \in \mathbb{Z}_{\geq 0}$, it holds that
\begin{align}
    U(n+\delta t, n)
    &= U(\delta t, 0),\\
    U(n+1, n)
    &\propto \prod_{j=1}^{L/2} \mathrm{SWAP}_{2j,2j+1} \prod_{j=1}^{L/2} \mathrm{SWAP}_{2j-1,2j}.
\end{align}
From this, we can rewrite $U_L(T_L,0)$ as
\begin{align}
    U_L(T_L,0)
    &= U_L(T_L,M_L) U_L(M_L,0)
    \propto U_L(T_L-M_L, 0) \left( \prod_{j=1}^{L/2} \mathrm{SWAP}_{2j,2j+1} \prod_{j=1}^{L/2} \mathrm{SWAP}_{2j-1,2j} \right)^{M_L}.
\end{align}
Thus, we have
\begin{align}
    U_L(T_L,0) \ket{\psi_L}
    = U_L(T_L-M_L, 0) \bigotimes_{j=1}^{L/2} \ket{\Phi}_{2j-1,2j}.
    \label{eq:breakdown_FinalState}
\end{align}
Direct calculation verifies that
\begin{align}
    &\left(\bigotimes_{j=1}^{L/2} \bra{\Phi}_{2j-1,2j}\right) {U_L(T_L-M_L, 0)}^\dagger \initH U_L(T_L-M_L, 0) \left(\bigotimes_{j=1}^{L/2} \ket{\Phi}_{2j-1,2j}\right)\\
    &= \begin{cases}
        \displaystyle - \frac{3}{8} \pi L & (0 \leq T_L-M_L < 1/2)\\ \\
        \displaystyle - \frac{3}{8} \pi L \left\{ \sin^4 \pi (T_L-M_L-1/2) + \cos^4 \pi (T_L-M_L-1/2) \right\} & (1/2 \leq T_L-M_L < 1)\\
    \end{cases}
\end{align}
for sufficiently large $L$.
Therefore, since the minimum value of $\sin^4 x + \cos^4 x$ on $[0, \pi/2)$ is $1/2$, we find
\begin{align}
    \lim_{L\to\infty} \braket{\psi_L|U_L(T_L, 0)^\dagger \initH U_L(T_L, 0)/N|\psi_L}
    \leq - \frac{3}{16} \pi.
\end{align}
Here, we have used the fact that in a one-dimensional system, $N = L$.
On the other hand, since $(\ket{\psi_L})_{L\in\mathbb{N}}$ is locally indistinguishable from the maximally mixed state, we have
\begin{align}
    \lim_{L\to\infty} \braket{\psi_L|\initH / N|\psi_L}
    = 0.
\end{align}
Hence, we obtain 
\begin{align}
    \lim_{L\to\infty} \braket{\psi_L|U_L(T_L, 0)^\dagger \initH U_L(T_L, 0)/N|\psi_L}
    < \lim_{L\to\infty} \braket{\psi_L|\initH / N|\psi_L}.
\end{align}
This is what we aimed to prove.

\subsection{\label{sec:breakdown_entropy}Proof of Proposition~\ref{proposition:breakdown_entropy}}

In this subsection, we prove Proposition~\ref{proposition:breakdown_entropy} by an explicit construction.
The basic strategy is almost the same as that of Proposition~\ref{proposition:breakdown_macroscopic-passivity}.

\subsubsection{Setup}
We choose the initial Hamiltonian $H_0$ as
\begin{align}
    H_0 = \frac{\pi}{4} \sum_{j=1}^L \vec{\sigma}_{j} \cdot \vec{\sigma}_{j+1}
\end{align}
and the final Hamiltonian $H_1$ as
\begin{align}
    H_1 = \frac{\pi}{4} \sum_{j=1}^L \vec{\sigma}_{j} \cdot \vec{\sigma}_{j+1} + \frac{\pi}{8} \sum_{j=1}^L \vec{\sigma}_{j} \cdot \vec{\sigma}_{j+2}.
\end{align}
The latter is known as the Majumdar-Ghosh model~\cite{Majumdar1969}.
Furthermore, we select the initial state $\rho_L$ as Eq.~\eqref{eq:breakdown_InitialState}
with $M_L = T_L$.
Moreover, we choose the macroscopic operation $(U_L(\bullet, 0))_{L\in\mathbb{N}}$ to be determined, according to Definition~\ref{definition:macroscopic-operation}, by the following additive observable and the external fields coupled to it:
\begin{align}
    B^1 = \frac{\pi}{4} \sum_{j=1}^{L} (-1)^j \vec{\sigma}_{j} \cdot \vec{\sigma}_{j+1}, \qquad
\end{align}
and
\begin{align}
    f^1(t) = \begin{cases}
        + 1 & (n \leq t < n+1/2)\\
        - 1 & (n+1/2 \leq t < n+1)
    \end{cases} \quad (n \in \mathbb{Z}_{\geq 0}), \qquad
\end{align}
Then, from Eq.~\eqref{eq:general_H}, the generator of the time evolution $H(t)$ is written as
\begin{align}
    H(t) = \begin{cases}
        H_0 & (t\le 0)\\
        \displaystyle \frac{\pi}{2} \sum_{j=1}^{L/2} \vec{\sigma}_{2j-1} \cdot \vec{\sigma}_{2j} & (0<t<T_L,\ n \leq t < n+1/2)\\ \\
        \displaystyle \frac{\pi}{2} \sum_{j=1}^{L/2} \vec{\sigma}_{2j} \cdot \vec{\sigma}_{2j+1} & (0<t<T_L,\ n+1/2 \leq t < n+1)\\
        H_1 & (t\ge T_L)
    \end{cases} \quad (n \in \mathbb{Z}_{\geq 0}).
\end{align}

Under this setup, we prove Eq.~\eqref{eq:Counterex_IncreasingEntropy} in the next subsubsection.
Note that, since the initial state is at $\beta=0$, Proposition~\ref{proposition:breakdown_entropy} is independent of the choice of the initial Hamiltonian $H_0$.
We choose the above $H_0$ such that $B^{1}$ and $f^{1}(t)$ become simple.

\subsubsection{Proof of Eq.~\eqref{eq:Counterex_IncreasingEntropy}}
We now prove that Eq.~\eqref{eq:Counterex_IncreasingEntropy} is satisfied for the setting provided in the previous subsubsection.

It is easy to see that the reduced density matrix of $\rho_L$ on a subset with diameter less than $4T_L+1$ is maximally mixed. This indicates that $s^{\mathrm{mac}}_{\ell}[(\rho_L)_{L\in\mathbb{N}}]=\log 2$.
Therefore, what we need to show below is only $\lim_{\ell\to\infty}s^{\mathrm{mac}}_{\ell}[(\sigma_L)_{L\in\mathbb{N}}]=0$.
As shown in Sec.~\ref{sec:breakdown_macroscopic-passivity}, the state at time $t=T_L$ is given by a pure state
\begin{align}
    U_L(T_L,0) \ket{\psi_L}
    =  \bigotimes_{j=1}^{L/2} \ket{\Phi}_{2j-1,2j}=:\ket{\phi_L}.
    \label{eq:Counterex_Entropy_State_T}
\end{align}
This state is known as the ground state of the Majumdar-Ghosh Hamiltonian $H_1$~\cite{Majumdar1969}.
This shows that, at $t>T_L$, the state remains unchanged under the unitary time evolution generated by $H_1$, indicating that the long-time averaged  state~\eqref{eq:Counterex_Entropy_FinalState} is given by
\begin{align}
    \sigma_{L}=\ket{\phi_L}\bra{\phi_L}.
\end{align}
Since this state is an eigenstate of the two-site shift operator $\mathcal{T}^2$, the coarse-grained local state~\eqref{eq:DEF_rho_infty^red} is given by 
\begin{align}
    (\sigma_L)|^{k}_{\Cell}=\frac{1}{2}\bigl(\sigma_L)|_{\Cell}+\frac{1}{2}\bigl(\mathcal{T}\sigma_L\mathcal{T}^{\dagger})|_{\Cell}.
\end{align}
Because of the structure of $\ket{\phi_L}$, we can show that the rank of each of $(\sigma_L)|_{\Cell}$ and $(\mathcal{T}\sigma_L\mathcal{T}^{\dagger})|_{\Cell}$ is bounded by $4$.
This implies that
\begin{align}
    \mathrm{rank}\ (\sigma_L)|^{k}_{\Cell}
    \le \mathrm{rank}\ (\sigma_L)|_{\Cell}
    +\mathrm{rank}\ (\mathcal{T}\sigma_L\mathcal{T}^{\dagger})|_{\Cell}
    \le 8
\end{align}
for $k=1,...,K^d$ and for arbitrary $\ell\ll L$.
Since the von Neumann entropy is bounded from above by the rank of the density matrix, we have
\begin{align}
    s^{\mathrm{mac}}_{\ell}[(\sigma_L)_{L\in\mathbb{N}}]=\frac{1}{K^d}\sum_{k=1}^{K^d}\frac{S_{\mathrm{vN}}[(\sigma_\infty)|^{k}_{\Cell}]}{\ell^d}\le \frac{\log 8}{\ell^d},
\end{align}
which indicates $\lim_{\ell\to\infty}s^{\mathrm{mac}}_{\ell}[(\sigma_L)_{L\in\mathbb{N}}]=0$.

\section{\label{sec:MeasurementDeviation}Smallness of deviation in measurement outcomes}

In this section, we show that the deviation of the measurement outcomes in Corollary~\ref{corollary:Measurement_s^mac} from the $L\to\infty$ limit of the expectation value of the density of additive observables is negligibly small for obtaining quantum macroscopic entropy density.

\begin{corollary}[\label{corollary:MeasurementDeviation_s^mac}Deviation of the measurement outcomes does not affect the quantum macroscopic entropy density]
Consider the measurement protocol given in Corollary~\ref{corollary:Measurement_s^mac}. If it is applied to a system with $L<\infty$, the deviation of its prediction from the $L\to\infty$ limit, Eq.~\eqref{eq:Measurement_s^mac_Add=rho}, is given by the difference between the sample average of $M$ measurement outcomes, $a^{i_1},a^{i_2},...,a^{i_M}$, of $A_{\mathcal{S}^{(k)}}(\ket{\bm{j}}\bra{\bm{j}^\prime})/|\mathcal{S}^{(k)}|$ and the LHS of Eq.~\eqref{eq:Measurement_s^mac_Add=rho}, 
\begin{align}
    \delta_{\bm{j},\bm{j}^\prime}^{(k)}(\bm{i}):=\frac{a^{i_1}+...+a^{i_M}}{M}-\lim_{L\to\infty}\mathrm{Tr}\Bigl[\rho_L \frac{A_{\mathcal{S}^{(k)}}(\ket{\bm{j}}\bra{\bm{j}^\prime})}{|\mathcal{S}^{(k)}|}\Bigr].
\end{align}
Its mean square over all outcomes, denoted by $\mathbb{E}_{\bm{i}}\bigl[\bigl|\delta_{\bm{j},\bm{j}^\prime}^{(k)}(\bm{i})\bigr|^2\bigr]$, can be taken sufficiently small by taking $M$ and $L$ sufficiently large and choosing the measurement device appropriately. [If $(\rho_L)_{L\in\mathbb{N}}$ represents a normal macroscopic state, $M=1$ is sufficient.]
This means that in such measurements, we can evaluate $\bra{\boldsymbol{j}^\prime}_{\Cell}\rho_{\infty|\ell}^{(k)}\ket{\boldsymbol{j}}_{\Cell}$ in a sufficiently high precision.

Moreover, the precision of quantum macroscopic entropy density is given as follows.
Let $\delta^{(k)}$ be a $D^{\ell^d}\times D^{\ell^d}$ matrix whose matrix elements are given by $\delta_{\bm{j},\bm{j}^\prime}^{(k)}(\bm{i})$, 
and $\tilde{\rho}_{\ell}^{(k)}$ be a $D^{\ell^d}\times D^{\ell^d}$ density matrix whose matrix elements $\bra{\bm{j}^\prime}\tilde{\rho}_{\ell}^{(k)}\ket{\bm{j}}$ are given by the sample average $(a^{i_1}+...+a^{i_M})/M$ of $A_{\mathcal{S}^{(k)}}(\ket{\bm{j}}\bra{\bm{j}^\prime})/|\mathcal{S}^{(k)}|$. It satisfies
\begin{align}
    \|\tilde{\rho}_{\ell}^{(k)}-\rho_{\infty|\ell}^{(k)}\|_{1}=\|\delta^{(k)}\|_{1},
\end{align}
where $\|\bullet\|_{1}$ is the trace norm.
From the Fannes-Audenaert inequality~\cite{Audenaert2007}, we have 
\begin{align}
    \frac{|S_{\mathrm{vN}}[\tilde{\rho}_{\ell}^{(k)}]-S_{\mathrm{vN}}[\rho_{\infty|\ell}^{(k)}]|}{\ell^d}\le \|\delta^{(k)}\|_{1}\frac{\log D}{2}+\frac{H_2(\|\delta^{(k)}\|_{1}/2)}{\ell^d},
\end{align}
where $H_2(x)=-x\log x-(1-x)\log (1-x)$.
Since all elements of $\delta^{(k)}$ can be taken sufficiently small when taking $M$ and $L$ sufficiently large and choosing the measurement device appropriately,
$\tilde{\rho}_{\ell}^{(k)}$ and $S_{\mathrm{vN}}[\tilde{\rho}_{\ell}^{(k)}]$ becomes sufficiently close to $\rho_{\infty|\ell}^{(k)}$ and $S_{\mathrm{vN}}[\rho_{\infty|\ell}^{(k)}]$, respectively.
Therefore, taking the average of $S_{\mathrm{vN}}[\tilde{\rho}_{\ell}^{(k)}]/\ell^d$ over all primitive macroscopic subsystems $k$, we obtain quantum macroscopic entropy $s^{\mathrm{mac}}_{\ell}[(\rho_{L})_{L\in\mathbb{N}}]$.
\end{corollary}
\begin{proof}[Proof outline.]
The mean square of the deviation is composed of three contributions,
\begin{align}
    \mathbb{E}_{\bm{i}}[|\delta_{\bm{j},\bm{j}^\prime}^{(k)}(\bm{i})|^2]
    =|\delta_{\bm{j},\bm{j}^\prime}^{(k),\mathrm{fluc}}|^2
    +|\delta_{\bm{j},\bm{j}^\prime}^{(k),\mathrm{error}}|^2
    +|\delta_{\bm{j},\bm{j}^\prime}^{(k),\mathrm{size}}|^2,
\end{align}
where 
the first, second, and third terms represent the quantum fluctuation of the additive observable, measurement error, and finite-size effect, respectively.
For the first term, if measurements are repeated $M$ times, it scales as $\delta_{\bm{j},\bm{j}^\prime}^{(k),\mathrm{fluc}}=O(1/\sqrt{M})$, which becomes sufficiently small when $M$ is sufficiently large.
[In addition, if $(\rho_L)_{L\in\mathbb{N}}$ represents a normal macroscopic state, it scales as $\delta_{\bm{j},\bm{j}^\prime}^{(k),\mathrm{fluc}}=o(L^0)$ even for a single-shot measurement.]
For the second term, choosing an appropriate measurement device makes $\delta_{\bm{j},\bm{j}^\prime}^{(k),\mathrm{error}}$ sufficiently small~\footnote{Note that, since post-measurement state is not used, the disturbance of measurement is not problematic, and $\delta_{\bm{j},\bm{j}^\prime}^{(k),\mathrm{error}}$ can be taken arbitrarily small.}, such as $\delta_{\bm{j},\bm{j}^\prime}^{(k),\mathrm{error}}=O(1/\sqrt{N})$~\cite{Fujikura2016,Shimizu2017}.
For the third term, it scales as $\delta_{\bm{j},\bm{j}^\prime}^{(k),\mathrm{size}}=o(L^0)$ by the following definition,
\begin{align}
    \delta_{\bm{j},\bm{j}^\prime}^{(k),\mathrm{size}}:=\mathrm{Tr}\Bigl[\rho_L \frac{A_{\mathcal{S}^{(k)}}(\ket{\bm{j}}\bra{\bm{j}^\prime})}{|\mathcal{S}^{(k)}|}\Bigr]-\lim_{L\to\infty}\mathrm{Tr}\Bigl[\rho_L \frac{A_{\mathcal{S}^{(k)}}(\ket{\bm{j}}\bra{\bm{j}^\prime})}{|\mathcal{S}^{(k)}|}\Bigr].
\end{align}
The second and third terms (and the first term in the case of normal macroscopic states) become sufficiently small when $L$ is taken sufficiently large.
\end{proof}

\twocolumngrid
\bibliography{document}

\begin{thebibliography}{140}
\bibitem [{Lan()}]{LandauSymbols}%
  \BibitemOpen
  \href@noop {} {}\bibinfo {note} {{ In this paper, we make use of the standard Bachmann-Landau notations $O(\bullet)$, $o(\bullet)$, $\Omega(\bullet)$, $\omega(\bullet)$, and $\Theta(\bullet)$ to describe the asymptotic behavior of functions. Let $f$ and $g$ be real functions, and let $g(x)$ be strictly positive for all large enough values of $x$. We say $f(x) = O(g(x))$ if $\limsup_{x\to\infty} |f(x)|/g(x) < \infty$ and $f(x) = o(g(x))$ if $\lim_{x\to\infty} f(x)/g(x) = 0$. We also say $f(x) = \Omega(g(x))$ if $\liminf_{x\to\infty} f(x)/g(x) > 0$ and $f(x) = \omega(g(x))$ if $\lim_{x\to\infty} f(x)/g(x) = \infty$. If both $f(x) = O(g(x))$ and $f(x) = \Omega(g(x))$, we write $f(x) = \Theta(g(x))$}}\BibitemShut {NoStop}%
\end{thebibliography}
\end{document}